%% file: arxiv_v3.tex
\documentclass[journal,onecolumn,11pt]{IEEEtran}
\include{Packages}
\include{Notations}

\begin{document}

\title{\huge A Unified Approach to Quantum Contraction and Correlation Coefficients}
\author[1]{Ian George}
\author[1,2]{Marco Tomamichel}
\affil[1]{\small Centre for Quantum Technologies, National University of Singapore, Singapore}
\affil[2]{\small Department of Electrical and Computer Engineering, National University of Singapore, Singapore}

\maketitle

\begin{abstract}
        The maximal correlation coefficient measures the linear correlation in a bipartite distribution and contraction coefficients measure how much information is lost under a noisy channel. Remarkably, Raginsky established a close relation between these two concepts by showing that the $\chi^2$ contraction coefficient equals the maximal correlation coefficient of the joint input/output distribution of the channel.
        In quantum theory, several generalizations of these concepts have been proposed, but none recover all the classical properties.
        Here we construct a framework in which the classical theory extends to the quantum setting. We introduce families of quantum maximal correlation coefficients and show that many impose limits on converting quantum states under local operations. We establish a family of quantum contraction coefficients are efficiently computable, yielding a generic efficient algorithm for mixing times of quantum channels with a full rank fixed point. Furthermore, we establish a quantum analogue of Raginsky's classical correspondence that relates these two families of quantities.
        To do this, we develop the operator-theoretic approach to Petz's family of non-commutative $L^{2}(p)$ spaces that extend the data processing inequality for variance to quantum theory. 
\end{abstract}

\tableofcontents

\section{Introduction}
The limits of information processing are established through measures of correlation and their data processing inequalities, which govern how these measures change when information is manipulated. For example, data processing inequalities formally capture that modifying local data can only destroy correlation in a distributed setting. While most traditional limits of information processing are investigated through divergences and divergence-induced quantities, such as conditional entropy and mutual information, there do exist other relevant measures such as the (Hirschfeld-Gebelein-R\'{e}nyi) maximal correlation coefficient $\mu(X:Y)_{p}$ \cite{Hirschfeld-1935a,Gebelein-1941a,Renyi-1959sieve,Renyi-1959measuresofdep} which in classical information theory establishes fundamental limits of processing information in both a distributed parallel setting and in a sequential setting \cite{Witsenhausen-1975a,Anantharam-2013a, Makur-2020a,GZB-preprint-2024}.

In the parallel setting, the maximal correlation limits the ability to convert an arbitrary number of a joint distribution $p_{XY}$ into a single copy of a different distribution $q_{X'Y'}$  under local operations as seen in Fig~\ref{fig:classical-parallel-processing}, which represents a theorem by Witsenhausen~\cite{Witsenhausen-1975a}:
\begin{figure}[H]
    \begin{center}
    \begin{tikzpicture}
        \tikzstyle{porte} = [draw=black!50, fill=black!20]
            \draw
            (-3.5,0.25) node (exist) {$\exists n \in \mbb{N}$, $\cW_{X^{n} \to X'}$, $\cR_{Y^{n} \to Y'}$}
            (-3.5,-0.25) node (exists2) { such that}
            (0,0) node (input) {$p_{XY}^{\otimes n}$}
            ++(1.5,-0.75) node[porte] (Decoder-2) {$\cR_{Y^{n} \to Y'}$}
            ++(0,1.5) node[porte] (Decoder-1) {$\cW_{X^{n} \to X'}$}
            ++(1.5,0) node (out-A) {$X'$}
            ++(0,-1.5) node (out-C) {$Y'$}
            ++(0.75,0.75) node (out) {$=q_{X'Y'}$}
            ++(3.5,0) node (cond) {~~~~only if~~~$\mu(X:Y)_{p} \geq \mu(X':Y')_{q}$.}
            ;
            \path[draw=black, -] (input) -- (0,-0.75);
            \path[draw=black, ->] (0,-0.75) -- (Decoder-2) node[midway, below,sloped] {$Y^{n}$};
            \path[draw=black, -] (input) -- (0,0.75);
            \path[draw=black, ->] (0,0.75) -- (Decoder-1) node[midway, above, sloped] {$X^{n}$};
            \path[draw=black, ->] (Decoder-1) -- (out-A);
            \path[draw=black, ->] (Decoder-2) -- (out-C);
        \end{tikzpicture}
        \end{center}
        \caption{A necessary condition for the existence of a conversion for some integer $n$ under local operations with stochastic maps $\cW_{X^{n} \to X'}$ and $\cR_{Y^{n} \to Y'}$ is given in terms of the maximal correlation coefficient in the classical setting.}
         \label{fig:classical-parallel-processing}
\end{figure}

\noindent In other words, the maximal correlation coefficient is a `strong' monotone in the sense that it implies limits that are independent of the number of input copies. Moreover, the maximal correlation coefficient satisfies $\mu(X:Y)_{p} = 1$ if and only if a bit of perfect correlation can be distilled under local operations. We call this the Witsenhausen property. In fact, Witsenhausen~\cite{Witsenhausen-1975a} showed something even stronger: one can approximate perfect correlation to vanishing error from $p_{XY}^{\otimes n}$ as $n$ grows using local operations only if $\mu(X:Y)_{p} = 1$. Satisfying this asymptotic property was called ``$p_{XY}$ having asymptotic common data (ACD)" in \cite{Beigi-2013a}. All of these powerful results on conversions under local operations follow from the maximal correlation coefficient $\mu(X:Y)_{p}$ satisfying two properties: data processing under local operations, $\mu(X':Y')_{(\cW_{X \to X'} \otimes \cR_{Y \to Y'})(p)} \leq \mu(X:Y)_{p}$, and tensorization, $\mu(XX':YY')_{p^{\otimes 2}} = \mu(X:Y)_{p}$. The tensorization property distinguishes it from the standard correlation measure in information theory, the mutual information, which is why the maximal correlation coefficient is useful for establishing these strong limits.\footnote{In particular, recall that the mutual information is additive over tensor products, $I(A^{n}:B^{n})_{\rho^{\otimes n}} = nI(A:B)_{\rho}$, so it is a measure of correlation that grows as $n$ grows.}

In the sequential setting, the maximal correlation upper bounds the rate at which a time-homogeneous, discrete time Markov chain converges to its unique distribution (see Fig.~\ref{fig:classical-sequential-processing}). 
\begin{figure}[H]
    \begin{center}
    \begin{tikzpicture}
        \tikzstyle{porte} = [draw=black!50, fill=black!20]
            \draw
            (0,0) node (input) {$p_{X}$}
            ++(1.5,0) node[porte] (Chan-1) {$\cW_{X \to X}$}
            ++(2,0) node[porte] (Chan-2) {$\cW_{X \to X}$}
            ++(2,0) node[porte] (Chan-3) {$\cW_{X \to X}$}
            ++(2,0) node (ongoing)
            {$\cdots \pi$}
            ++(3,0) node (speed) {at rate at most $\mu(\pi_{XY})$}
            ;
            \path[draw=black, ->] (input) -- (Chan-1);
            \path[draw=black, ->] (Chan-1) -- (Chan-2);
            \path[draw=black, ->] (Chan-2) -- (Chan-3);
            \path[draw=black, ->] (Chan-3) -- (ongoing);
        \end{tikzpicture}
        \end{center}
        \caption{Bound on speed of convergence to stationary distribution in the classical setting in terms of maximal correlation coefficie.}
         \label{fig:classical-sequential-processing}
\end{figure}

\noindent That is, given $\cW_{X \to X}$ with unique, full rank stationary distribution $\pi_{X} = \sum_{x} \pi(x)\dyad{x}$, the rate at which $\cW^{n} \coloneq \circ_{i \in [n]} \cW$ converges an arbitrary input to $\pi$ is upper bounded by $\mu(\pi_{XY})$ where $\pi_{XY} \coloneq \sum_{x,y} \pi(x)\cW(y\vert x) \dyad{x} \otimes \dyad{y}$. Moreover, this upper bound can be tight \cite{Makur-2020a} and is largely independent of the divergence used to measure the distinguishability of the process and the stationary distribution \cite{GZB-preprint-2024}. Beyond the asymptotic setting, because $\mu(X:Y)_{p}$ is in fact efficient to compute in many cases, one may also use it to practically compute mixing times. Formally, these sequential results stem from the fact that $\mu(X:Y)_{\pi}$ is the input-dependent contraction coefficient of the $\chi^{2}$-divergence \cite{Raginsky-2016a}, i.e. 
\begin{align}\label{eq:correlation-to-contraction}
     \eta_{\chi^{2}}(\cW,\pi)=\mu(X:Y)_{\pi}^{2} \ . 
\end{align} 
Note by the tensorization property of $\mu(p_{XY})$ this further implies 
\begin{align}\label{eq:cl-contraction-tensorize}
    \eta_{\chi^{2}}(\cW_{X \to Y}^{\otimes 2},\pi_{X}^{\otimes 2}) = \eta_{\chi^{2}}(\cW,\pi) \ . 
\end{align}
In total, while the maximal correlation coefficient is not itself a divergence in the traditional sense, it captures a fundamental aspect of the data processing inequality of (a family of classical) divergences.

A key aspect of the above results is that both the maximal correlation coefficient and $\chi^{2}$-divergence may be expressed in terms of expectation and variance. Indeed, the maximal correlation coefficient may be expressed as
\begin{align}\label{eq:intro-max-corr}
    \mu(X:Y)_{p} = \sup\{ \mbb{E}_{p_{XY}}[fg] : \mbb{E}_{p_{X}}[f] = 0 = \mbb{E}_{p_{Y}}[g] \, \& \, \Var_{p_{X}}[f] = 1 = \Var_{p_{Y}}[g] \} \ ,  
\end{align}
where the supremum is with respect to real-valued functions and subscripts denote the relevant probability measure for the given expectation or variance. Moreover, for probability distributions $p$ and $q$ and defining the likelihood ratio $\frac{p}{q}$, 
\begin{align}\label{eq:chi-sq-as-variance}
    \chi^{2}(p \Vert q) = \Var_{q}\left[\frac{p}{q}\right] \ .
\end{align}
The reason these probabilistic identifications are important is that they allow the quantities to be written in terms of the $L^{2}(p)$ space that induces the expectation and variance. Witsenhausen used this ``operator-theoretic interpretation" to establish his seminal results \cite{Witsenhausen-1975a}. Similarly, by ``following Witsenhausen," Raginsky established \eqref{eq:correlation-to-contraction} in a simple manner \cite[Theorem 3.2]{Raginsky-2016a}. In particular, his proof identifies $\eta_{\chi^{2}}(\cW,q)$ as the minimal decrease of the variance of the likelihood ratio under $\cW$, re-expresses this as an operator norm between subspaces of $L^{2}(q)$ and $L^{2}(\cW(q))$, and notes this operator norm is equivalent to the optimization in \eqref{eq:intro-max-corr} when applied to $q_{XY} = \sum_{x,y} q(x)\cW(y\vert x) \dyad{x} \otimes \dyad{y}$. These identifications are thus made possible by working with $L^{2}(p)$ spaces, which is why the probabilistic identification is critical. The total of these ideas and results in classical information theory are summarized in Fig.~\ref{subfig:CIT-properties}.

In quantum information theory, the analogue of the relations above is not well understood. To address the parallel setting above, Beigi introduced a quantum maximal correlation coefficient \cite{Beigi-2013a}. He established it as a strong monotone for local operations, but also showed that it does not satisfy the Witsenhausen property \cite[Example 7]{Beigi-2013a} and thus could not be used to detect whether $\rho_{AB}$ has asymptotic common data. The sequential setting is further muddled. Temme \emph{et al.}~\cite{Temme-2010a} introduced a family of quantum $\chi^{2}$-divergences, and their results suggest most of the family of quantum $\chi^{2}$-divergences may be relevant for studying quantum time-homogeneous Markov chains. Recent work \cite{GZB-preprint-2024,Beigi-2025a}
suggests multiple quantum $\chi^{2}$-divergences are relevant even in the asymptotic setting. \cite{Cao-2019a} introduced a quantum maximal correlation coefficient that extended the correspondence in \eqref{eq:correlation-to-contraction} with a specific quantum $\chi^{2}$-divergence under certain conditions.  However, the authors did not determine why this quantity in particular extends \eqref{eq:correlation-to-contraction} nor were they able to unify their quantity with Beigi's. As such, the current extensions of maximal correlation coefficients to the quantum theory are disparate and their relations to each other and quantum $\chi^{2}$-divergences are not understood.

This paper unifies the various previous extensions and relations. We introduce new families of quantum maximal correlation coefficients, $\mu_f$, parametrized by operator monotone functions $f$, and establish operational interpretations for many instances of these quantities. We identify their relation to the contraction coefficients of $\chi^{2}_{f}$-divergences. We establish equivalent conditions to these divergences saturating the data processing inequality, and construct efficient algorithms for computing mixing times of many quantum channels through their $\chi^{2}_{f}$ contraction coefficients. Central to our approach is viewing these quantities as being induced via non-commutative extensions of the $L^{2}(p)$ space, the $L^{2}_{f}(\sigma)$ spaces parameterized by operator monotone function $f$. In total, our methodology and results extend the entire classical framework to quantum theory except the asymptotic common data property, which depends on a new question about the stability under tensor products of two of the quantum maximal correlation coefficients (see Section \ref{sec:conclusion} for the formal statement of the open question).

\begin{figure}
    \begin{center}
    \centering
    \begin{subfigure}{\textwidth}
        \begin{center}
        \begin{tikzpicture}
        \tikzstyle{porte} = [draw=black!50, fill=black!20]
            \draw
            (0,0) node[porte, fill=black!20] (L2-space) {$L^{2}(p)$ space}
            (-3,-2.5) node[circle, fill=black!20] (MCC) {$\mu(X:Y)_{p}$};
            \draw (-0.5,-3) node[align=center] (effic) {\small Efficient to \\ \small Compute};
            \draw (-4.5,-1) node[align=center] (Wits-p) {\small Witsenhausen \\ \small Property \cite{Witsenhausen-1975a}};
            \draw (-6.5,-2.5) node[align=center] (ACD-p) {\small ACD iff \\ $\mu(X:Y)_{p} = 1$ \cite{Witsenhausen-1975a}};
            \draw (-5.5,-4) node[align=center] (DPI) {\small DPI under \\ LO};
            \draw (-3,-4) node[align=center] (tens) {\small Tensorizes \cite{Witsenhausen-1975a}};
            \draw (-4,-5.25) node[align=center] (strong-MO) {\small Strong \\ \small Monotone \cite{Witsenhausen-1975a}};
            \draw (4, -2.5) node[align=center,circle,fill=black!20] (contract) {\small $\eta_{\chi^{2}}(\cW,\pi)$};
            \draw (6, -0.75) node[align=center] (Rate) {\small Bound on \\ \small Rate of Contraction \\ \small under $D_{f}(p \Vert q)$ \cite{GZB-preprint-2024}};
            \draw (7, -2.5) node[align=center] (Mixing-Bound) {\small TVD Mixing Time \\ \small Bounds \cite{Diaconis-1991a,Fill-1991a}};
            \draw (6, -4) node[align=center] (Raginsky) {\small $=\mu((\id_{X} \otimes \cW)(\chi^{\vert \pi}))$ \cite{Raginsky-2016a}};
            \draw (2.5, -5.25) node[align=center] (Comp-Mixing-Bound) {\small Computable Mixing \\ \small Time Bounds};
            \draw (-5.5, 0.5) node[align=center] (CIT-summ) {\underline{Classical Information Theory}};
            \path[draw=blue, ->] (L2-space) -- (contract) node[midway] {$\chi^{2}(p \Vert q) = \Var_{q}[\frac{p}{q}]$ \cite{Raginsky-2016a}};
            \path[draw=blue, ->] (L2-space) -- (MCC);
            \path[draw=black, -] (contract) -- (Raginsky);
            \path[draw=black, -] (contract) -- (Mixing-Bound);
            \path[draw=black, -] (contract) -- (Rate);
            \path[draw=black, -] (MCC) -- (tens);
            \path[draw=black, -] (MCC) -- (effic);
            \path[draw=black, -] (MCC) -- (DPI);
            \path[draw=black, -] (MCC) -- (ACD-p);
            \path[draw=black, -] (MCC) -- (Wits-p);
            \path[draw=blue, ->] (DPI) -- (strong-MO);
            \path[draw=blue, ->] (tens) -- (strong-MO);
            \path[draw=blue, ->] (effic) -- (Comp-Mixing-Bound);
            \draw [draw = blue, ->]
    (Mixing-Bound) edge [bend left=40] (Comp-Mixing-Bound);
        \end{tikzpicture}
        \end{center}
        \caption{\footnotesize Summary of known relations of maximal correlation coefficient $\mu(p_{XY})$ and the contraction coefficient of the $\chi^{2}$-divergence $\eta_{\chi^{2}}(\cW,\pi)$ in classical information theory. Some stated known results are not made explicit in the cited works. Claims without citations are believed to have been well-understood. To the best of our knowledge, it is not until \cite{GZB-preprint-2024} that the efficiency of computing the mixing time bounds via the contraction coefficient $\eta_{\chi^{2}}(\cW,\pi)$ was explicated although it follows from previous work.}
        \label{subfig:CIT-properties}
    \end{subfigure}
    \\[5mm]
    \begin{subfigure}{\textwidth}
        \begin{center}
        \begin{tikzpicture}
        \tikzstyle{porte} = [draw=black!50, fill=black!20]
            \draw
            (0,0) node[porte, fill=black!20, align=center] (L2-space) {\small $L_{f}^{2}(\sigma)$ spaces \\ \small (Section \ref{sec:non-commutative-L2-spaces})}
            (-3,-2.5) node[circle, fill=black!20, align=center] (MCC) {$\mu_{f}(A:B)_{\rho}$ \\
            $\mu_{f}^{\text{Lin}}(A:B)_{\rho}$};
            \draw (-5,-0.75) node[align=center] (Wits-p) {\small Witsenhausen Property \\ \small for $\mu_{AM}(A:B)_{\rho}$ (Thm.~\ref{thm:extreme-values-summary})};
            \draw (-6.5,-2.5) node[align=center] (ACD-p) {\small ACD only if \\ \small $\mu_{GM}(A:B)_{\rho} = 1$ \\ \small
            (Thm.~\ref{thm:asymptotic-data})};
            \draw (-5.5,-4) node[align=center] (DPI) {\small DPI under LO \\ \small (Prop.~\ref{prop:DPI-for-f-correlation})};
            \draw (-3,-5) node[align=center] (tens) {\small $\mu_{f}^{\text{Lin}}$ tensorizes for \\ \small $f = x^{k}$, $k \in [0,1]$ \\ \small
            (Lemma ~\ref{lem:k-correlation-tensorize})};
            \draw (-4,-6.75) node[align=center] (strong-MO) {{\small Strong Monotone} \\ {\small for $\mu^{\text{Lin}}_{x^{k}}$, $k \in [0,1]$}};
            \draw (4, -2.5) node[align=center,circle,fill=black!20] (contract) {\small $\eta_{\chi_{f}^{2}}(\cE,\sigma)$};
            \draw (6, -0.5) node[align=center] (Rate) {\small Bound on \\ \small Rate of Contraction (Thm.~\ref{thm:mixing-rate}): \\ 
            $D_{f} \to \eta_{\chi_{LM}^{2}}$, $\widetilde{D}_{\alpha} \to \eta_{\chi^{2}_{GM}}$, \\
            ($\ol{D}_{f} \to \eta_{\chi^{2}_{HM}}$ \cite{GZB-preprint-2024})};
            \draw (7, -2.5) node[align=center] (Mixing-Bound) {\small TD Mixing Time \\ \small Bounds \cite{Temme-2010a}};
            \draw (6, -4) node[align=center] (Raginsky) {\small $=\mu_{f}(A:B)_{H}$ (Thm.~\ref{thm:correspondence-between-contraction-coeffs-and-max-corr-coeffs}) \\ \small Cor.~\ref{cor:contraction-for-sandwiched-case}; See also \cite{Cao-2019a}};
            \draw (4.5, -5.3) node[align=center] (sufficiency) {\small Sufficiency \& Recovery  \\ 
            \small New: Thm.~\ref{thm:DPI-with-equality} \\ 
            \small (Cors.~\ref{cor:suff-for-chi-sq} \& \ref{cor:suff-for-SRD}; See also \cite{jenvcova2012reversibility,Jencova-2017a,Gao-2023-sufficient-fisher})};
            \draw (-5.5, 0.5) node[align=center] (QIT-summ) {\underline{Quantum Information Theory}};
            \draw (0.5,-4) node[align=center] (effic) {{\small Often Efficient} \\ {\small to Compute} \\ {\small
            (Section \ref{sec:computability})}};
            \draw (2.5, -6.75) node[align=center] (Comp-Mixing-Bound) {\small Computable Quantum  \\ \small Mixing Time Bounds (Thm.~\ref{thm:computable-mixing-times})};
            \path[draw=blue, ->] (L2-space) -- (contract) node[midway] {$\chi^{2}_{f}(\rho \Vert \sigma)$ \cite{Temme-2010a}};
            \path[draw=blue, ->] (L2-space) -- (MCC) node[midway] {(Section \ref{sec:q-maximal-correlation-coeff})};
            \path[draw=black, -] (contract) -- (Raginsky);
            \path[draw=black, -] (contract) -- (Mixing-Bound);
            \path[draw=black, -] (contract) -- (Rate);
            \path[draw=black, -] (3.65,-3.36) -- (3.4,-4.6);
            \path[draw=black, -] (MCC) -- (tens);
            \path[draw=black, -] (MCC) -- (DPI);
            \path[draw=black, -] (MCC) -- (ACD-p);
            \path[draw=black, -] (MCC) -- (Wits-p);
            \path[draw=blue, ->] (DPI) -- (strong-MO);
            \path[draw=blue, ->] (tens) -- (strong-MO);
            \draw [draw = blue, ->]
    (Mixing-Bound) edge [bend left=38] (Comp-Mixing-Bound);
            \draw [draw = blue, ->] (effic) -- (Comp-Mixing-Bound);
            \draw [draw = black, -] (MCC) -- (effic);
            \draw [draw = black, -] (contract) -- (effic);
        \end{tikzpicture}
        \end{center}
        \caption{\footnotesize Summary of relations of quantum maximal correlation coefficients and $\chi^{2}$-divergences in quantum information theory as established in this work. AM, GM, LM, HM stand for arithmetic, geometric, logarithm, and harmonic mean respectively. We note that many of the results on the quantum maximal correlation coefficients listed build on results of \cite{Beigi-2013a}, which is not referenced directly in the figure.}
        \label{subfig:QIT-properties}
    \end{subfigure}
    \end{center}
    \caption{{\small Summary of the relation between the maximal correlation coefficient and $\chi^{2}$-divergence in (a) classical information theory and (b) how it changes when generalized to quantum theory as established in this work. TVD, TD, DPI, LO, and ACD stand for total variational distance, trace distance, data processing inequality, local operations, and asymptotic common data respectively. A `strong' monotone means that conversion is limited independent of the number of input copies. If a quantum result includes a citation to a result outside this work, it is because we have obtained an improvement or deeply related result to one in that work}.}
    \label{fig:CIT-vs-QIT}
\end{figure}

\subsection{Overview of Paper and Results}
We briefly overview the structure of the paper and the key contributions of this work. The structure of the paper as well as many of the key results may also be gleaned from Fig.~\ref{subfig:QIT-properties}. 

In Section \ref{sec:non-commutative-L2-spaces}, we establish that, using inner products introduced by Petz \cite{Petz-1996a}, we can identify non-commutative generalizations of the $L^{2}(p)$ space that induce non-commutative notions of variance that admit a data processing inequality (Proposition \ref{prop:variance-properties}). We call such generalizations of the $L^{2}(p)$ space the $L^{2}_{f}(\sigma)$ spaces indexed by operator monotone function $f$. This identification is central to or approach for two reasons. First, as the $L^{2}_{f}(\sigma)$ spaces provide a systematic generalization of expectation and variance to the non-commutative setting, we can use the framework to establish non-commutative generalizations of the identities in Eqs.~\eqref{eq:intro-max-corr} and \eqref{eq:chi-sq-as-variance}. Second, because the $L^{2}_{f}(\sigma)$ spaces are inner product spaces, this framework allows us to generalize the operator-theoretic approach of Witsenhausen to the quantum setting. Thus it is through establishing the $L^{2}_{f}(\sigma)$ spaces that we will be able to unify quantum maximal correlation coefficients and $\chi^{2}$-divergences.

In Section \ref{sec:functional-analysis-tools}, we relate a variety of functional analytic quantities that we use in this work. In particular, in Lemma \ref{lem:map-norms-as-optimizations}, we make explicit a simple method for identifying when an optimization problem is implicitly an operator norm between two Hilbert spaces. This identification, in tandem with the $L^{2}_{f}(\sigma)$ spaces, is used in a subsequent section to prove the correspondence between quantum $\chi^{2}_{f}$ contraction coefficients and quantum maximal correlation coefficients (Theorem \ref{thm:correspondence-between-contraction-coeffs-and-max-corr-coeffs}).

In Section \ref{sec:q-maximal-correlation-coeff}, we introduce two families of quantum maximal correlation coefficients, $\mu_{f}(A:B)_{\rho}$ and $\mu_{f}^{\text{Lin}}(A:B)_{\rho}$. These quantities are indexed by the operator monotone function $f$ that determines the $L^{2}_{f}(\sigma)$ space they are induced by as well as whether they are expressed as optimization problems over Hermitian or linear operators. We identify which of these quantities satisfy the key classical properties such as data processing, being bounded above and below by 0 and 1 respectively, and the operational interpretations of saturating said extreme values (Theorems \ref{thm:properties-of-maximal-corr-coeffs} and \ref{thm:extreme-values-summary}). We identify which quantum maximal correlation coefficient recovers the Witsenhausen property and establish an algebraic equivalent condition (Item 3 of Theorem \ref{thm:extreme-values-summary}). This generalizes the classical case and identifies a quantum generalization of the notion of a `decomposable distribution' introduced by Ahlswede and G\'{a}cs \cite{Ahlswede-1976a}. Moreover, we establish a family of these maximal correlation coefficients are strong monotones for conversion under local operations, thereby extending Fig.~\ref{fig:classical-parallel-processing} to the quantum setting and generalizing the main result of \cite{Beigi-2013a}, which corresponds to the case $k\in \{0,1\}$ in the following statement.
\begin{tcolorbox}[width=\linewidth, sharp corners=all, colback=white!95!black, boxrule=0pt,frame hidden]
\begin{result}(Thm.~\ref{thm:k-correlation-nec-for-local-processing} Simplified)
    Define $f_{k}(x) \coloneq x^{k}$. If there exists $k \in [0,1]$ such that 
    $$\mu^{\text{Lin}}_{f_{k}}(A:B)_{\rho}~<\mu^{\text{Lin}}_{f_{k}}(A':B')_{\sigma} \ , $$ 
    then $\rho_{AB}^{\otimes n}$ cannot be converted to $\sigma_{A'B'}$  for any $n \in \mbb{N}$ using local operations.
\end{result}
\end{tcolorbox}
\noindent  This allows us to identify many new maximal correlation coefficients equaling one as necessary conditions for a state to have asymptotic common data (Theorem \ref{thm:asymptotic-data}).

In Section \ref{sec:quantum-chi-squared}, we study the quantum $\chi_{f}^{2}$-divergences from \cite{Temme-2010a} via the non-commutative probabilistic framework of $L_{f}^{2}(\sigma)$ spaces. Specifically, for a large class of operator monotone $f$, we establish the non-commutative generalization of \eqref{eq:chi-sq-as-variance}. To the best of our knowledge, this was not known previously and is critical to establishing our results. Critically, this allows us to establish the generic correspondence between quantum $\chi^{2}$-contraction coefficients and quantum maximal correlation coefficients, generalizing \eqref{eq:correlation-to-contraction} to the quantum setting.
\begin{tcolorbox}[width=\linewidth, sharp corners=all, colback=white!95!black, boxrule=0pt,frame hidden]
\begin{result}(Thm.~\ref{thm:correspondence-between-contraction-coeffs-and-max-corr-coeffs} Simplified). For every `standard' operator monotone function $f$, positive, trace-preserving map $\cE_{A \to B}$, and quantum state $\rho_{A}$,
\begin{align*}
        \sqrt{\eta_{\chi^{2}_{f}}}(\cE,\rho) = \mu_{f}(A:B)_{H} \ ,
\end{align*}
where $H_{AB}$ is a Hermitian operator satisfying $H_{A} = \rho_{A}$ and $H_{B} = \cE(\rho_{A})$, i.e. a `relaxed quantum coupling.' Moreover, for $f_{GM} \coloneq \sqrt{x}$, we also have that for any quantum state $\rho_{AB}$, there exists a quantum channel $\cE_{A \to B}$ such that $\eta_{\chi^{2}_{f_{GM}}}(\cE,\rho_{A}) = \mu_{f_{GM}}(A:B)_{\rho}$.
\end{result}
\end{tcolorbox}
\noindent This result unifies the quantum $\chi^{2}$ contraction coefficients and quantum maximal correlation coefficients. The `moreover' statement for $f(x) = \sqrt{x}$ generalizes the correspondence established in \cite[Theorem 18]{Cao-2019a} by removing rank conditions and showing that the identification holds for the whole joint state space. In this case, the correspondence arises because $H_{AB}$ is the standard coupling in quantum Shannon theory--- the channel $\cE_{A \to B}$ applied to the canonical purification of $\rho_{A}$. For other choices of $f$, the canonical purification is replaced by some ``$f$-twisted" Hermitian extension of $\rho_{A}$ that the channel then acts on. Furthermore, and possibly of independent interest, the relaxed quantum couplings constructed for this result may be identified as `quantum states over time' \cite{Leifer_2013,Fullwood_2022} constructed by choosing an operator monotone function $f$ to decide upon a specific non-commutative $L^{2}_{f}(\sigma)$ space (See Remark \ref{rem:QSOT-from-NC-Prob}).

We also establish equivalent conditions under which the $\chi_{f}^{2}$-divergence saturates the data processing inequality that are analogous to those known for the quantum relative entropy. In particular, a seminal result of Petz shows that the quantum relative entropy satisfies $D(\cE(\rho) \Vert \cE(\sigma)) = D(\rho \Vert \sigma)$ if and only if the `Petz recovery map' $\cP_{\cE,\sigma}$ recovers $\rho$, i.e. $(\cP_{\cE,\sigma} \circ \cE)(\rho) = \rho$ while it always holds $(\cP_{\cE,\sigma} \circ \cE)(\sigma) = \sigma$ \cite{petz1986sufficient}. More generally, there exists a family of linear maps $\cS_{f,\cE,\sigma}$ indexed by operator monotone function $f$ such that $(\cS_{f,\cE,\sigma} \circ \cE)(\sigma) = \sigma$ and $\cS_{f_{GM},\cE,\sigma} = \cP_{\cE,\sigma}$. These maps play the same role as the Petz recovery map in our characterization.
\begin{tcolorbox}[width=\linewidth, sharp corners=all, colback=white!95!black, boxrule=0pt,frame hidden]
        \begin{result}(Thm.~\ref{thm:DPI-with-equality}, Simplified)
        For a large class of operator monotone $f$, if $\rho \ll \sigma$ and $\cE$ is a quantum channel, then $\chi^{2}_{f}(\cE(\rho)\Vert\cE(\sigma)) = \chi^{2}_{f}(\rho \Vert \sigma)$ if and only $(\cS_{f,\cE,\sigma} \circ \cE)(\rho) = \rho$.
\end{result}
\end{tcolorbox}
\noindent Prior to this work, to the best of our knowledge, the only case of the above result that was known was for $f$ being the geometric mean, which is a special case of \cite[Theorem 4.6]{Gao-2023-sufficient-fisher} and was established with a distinct proof method.

In Section \ref{sec:computability}, we show that the quantum $\chi_{f}^{2}$ input-dependent contraction coefficients and the relevant quantum maximal correlation coefficients are efficient to compute in a generic sense (Theorems \ref{thm:contraction-coeff-computability} and \ref{thm:f-max-corr-computability}). This in particular implies an efficient method for bounding the mixing time of a channel $\cE_{A \to A}$ with unique fixed point $\pi \in \Density(A)$ under dissimilarity measure $\Delta$, denoted $t^{\Delta}_{\min}(\cE,\delta) \coloneq \min\{n \in \mbb{N}: \max_{\sigma \in \Density(A)} \Delta(\cE^{n}(\sigma), \pi) \leq \delta\}$. 
    \begin{tcolorbox}[width=\linewidth, sharp corners=all, colback=white!95!black, boxrule=0pt,frame hidden, breakable]
    \begin{result}(Thm.~\ref{thm:computable-mixing-times}, Simplified)
    For $\Delta$ being trace distance or relative entropy and any channel $\cE$ with unique, full rank fixed point $\pi$, an upper bound on $t^{\Delta}_{\min}(\cE,\delta)$ can be efficiently determined using Algorithm \ref{alg:contraction-coeff}. Moreover, either the the upper bound of this method will be finite or all well-defined contraction coefficients for $\cE$ are one.
    \end{result}
    \end{tcolorbox}
This result is appealing for a variety of reasons. First, to the best of our knowledge, prior to this work the generic computability of any these quantities was not shown in the quantum setting. In fact, establishing the computability of a specific quantity in the considered family was an open problem of \cite{GZB-preprint-2024}. Second, the `moreover' statement in the result shows that, under the above conditions, any method that is more generally applicable for mixing times of time-homogeneous Markov chains with a full rank fixed point formally relies on more structure than knowledge of the classical description of the channel and the unique fixed point. Third, while not obvious in this simplified form, because we show one can compute the $\chi^{2}_{f}$ input-dependent contraction coefficient for many choices of operator monotone $f$, many of these choices of $f$ can bound the mixing time, and one may interpret using the $\chi^{2}_{f}$ input-dependent contraction coefficient as a spectral gap method for mixing times, our result is the ability to efficiently construct a \textit{family} of spectral gaps, which one could then optimize over choice of $f$ to get the tightest bounds from the family of spectral gaps (See Remark \ref{rem:relation-to-spectral-gap-methods}). Lastly, this result may be appealing as the simple classical method for computing mixing times under trace distance does not generically extend to the quantum setting as the trace distance contraction coefficient can be NP-hard to compute as established in concurrent work \cite{delsol2025computationalaspectstracenorm}.
\subsection{Relation to Previous Work}
The consideration of operator monotone functions and their relation to statistical quantities appears to be first addressed by Petz \cite{Petz-1996a} and later related to quantum Fisher information \cite{Petz-2011a}, a quantity which we do not consider in this work (see \cite{scandi2023quantum} for a recent review on quantum Fisher information). Of particular relevance, \cite{Petz-2011a} introduced a non-commutative extension of variance for Hermitian operators, noted it admitted a data processing inequality, and related it to the quantum Fisher information. Hiai and Petz subsequently generalized the definition for covariance \cite{hiai-2012quasi}, which aligns with this work's definition. As mentioned, Beigi introduced a quantum maximal correlation coefficient \cite{Beigi-2013a}, and Section \ref{sec:q-maximal-correlation-coeff} extends many of the results from that work. We remark that \cite{Beigi-2023maximal-gaussian} analyzed Beigi's maximal correlation coefficient for Gaussian states, whereas this work remains working in finite dimensions. Many of the tools used in Section \ref{sec:q-maximal-correlation-coeff} build on tools and ideas from \cite{Beigi-2013a,Delgosha-2014a}. 

The quantum $\chi^{2}$-divergences, which we will denote $\chi^{2}_{f}$ for relevant operator monotone functions $f$, were introduced in \cite{Temme-2010a} and have been studied subsequently \cite{jenvcova2012reversibility,Cao-2019a,Gao-2023-sufficient-fisher}. Jen\v{c}ov\'{a} established exact conditions for the saturation of the data processing inequality of $\chi^{2}$-divergences in terms of the Petz recovery map \cite[Proposition 4]{jenvcova2012reversibility}. Among other things, \cite{Gao-2023-sufficient-fisher} established approximate recoverability bounds for many quantum $\chi^{2}$-divergences and recovered a subset of the exact case established by Jen\v{c}ov\'{a}. We recover a special case of \cite[Proposition 4]{jenvcova2012reversibility} as a corollary to our new saturation results.

With regards to mixing times, \cite{Temme-2010a} studied mixing times under $\chi^{2}$-divergences. \cite{CARLEN20171810} considers a superset of Petz's inner product spaces, but does not develop the probabilistic interpretation. \cite{Gao-2022a} studies the convergence of continuous and discrete time time-homogeneous Markov chains under the relative entropy, where the latter type are the focus of Section \ref{sec:time-homogeneous-Markov-chains}. In their study, they focus on the power functions as was done in \cite{CARLEN20171810}. Their key lemma bounds a specific $\chi^{2}$-divergence to the relative entropy in an input-dependent manner, which makes it related to results in \cite{GZB-preprint-2024,Hirche-2024a} of which we make use. They also consider bounds on the input-\textit{independent} contraction coefficient of the relative entropy in terms of the input-\textit{independent} contraction coefficient of this specific $\chi^{2}$-divergence. As such, many of the ideas and results of \cite{Gao-2022a} are related, but none seem to be in direct correspondence. In \cite{george2025quantumdoeblincoefficientsinterpretations}, the authors studied quantum Doeblin coefficients and showed one may be used to efficiently bound the mixing time of time-inhomogeneous discrete time quantum Markov chains whenever the quantum Doeblin coefficient is non-zero. 

Finally, \cite{Cao-2019a} also considered the input-dependent contraction coefficients of the quantum $\chi^{2}$-divergences and proposed a family of objects they called maximal correlation coefficients. With the exception of a quantity we denote $\mu^{\text{Lin}}_{GM}$, the maximal correlation coefficients of this work and those in \cite{Cao-2019a} appear to be distinct. This is likely because those proposed in \cite{Cao-2019a} were not defined to satisfy the data processing inequality, which our definitions guarantee through our framework of $L^{2}_{f}(\sigma)$ spaces. As previously alluded to, \cite{Cao-2019a} established $\eta_{\chi^{2}_{GM}}(\cE,\sigma) = \mu_{GM}^{\text{Lin}}(A:B)_{(\id_{A} \otimes \cE)(\psi_{\sigma})}$, but under the assumptions $\sigma$ is full rank, $\cE(\sigma)$ is full rank, and $\cE$ is a quantum channel (completely positive, trace-preserving map), which are significantly stronger assumptions than our Corollary \ref{cor:contraction-for-sandwiched-case} which only requires $\sigma$ is full rank and $\cE$ is a positive, trace-preserving map. Finally, we note that \cite{Cao-2019a} aimed to generalize \eqref{eq:cl-contraction-tensorize} to all quantum $\chi^{2}_{f}$-divergences. For arbitrary channels, they were only able to establish this for $\chi^{2}_{GM}$ where $f_{GM} \coloneq f(x) = x^{1/2}$, though they established it more generally for entanglement-breaking channels. Our methodology arguably shows that the tensorization of $\chi^{2}_{GM}(\cE,\sigma)$ is inherited by the `more fundamental' tensorization property of $\mu_{GM}(\rho_{AB})$. We believe it is fair to say that the variety of improvements we make to results in \cite{Cao-2019a} stem from our systematic approach to analyzing the maximal correlation coefficient and $\chi^{2}$-contraction coefficient through the use of non-commutative $L_{f}^{2}(\sigma)$ spaces.

\section{Petz's Non-Commutative \texorpdfstring{$L^{2}(p)$ Spaces from Operator Monotone Functions}{} }\label{sec:non-commutative-L2-spaces}

In this section, we introduce the non-commutative extensions of the $L^{2}(p)$ space for a probability distribution $p$. These will be induced by a quantum state $\sigma$ and a choice of operator monotone function $f$, so we will ultimately denote these spaces as the $L^{2}_{f}(\sigma)$ spaces. These spaces will be used to unify and analyze the contraction and correlation coefficients considered in this work as highlighted in Fig.~\ref{fig:CIT-vs-QIT}. Many results in this section are minor extensions of results of Petz that may be found in \cite{Petz-1996a,Petz-2011a} and when this is the case we acknowledge it throughout. In Section \ref{sec:Preliminaries} we introduce the basic notation we will use throughout this work. In Section \ref{sec:J-operator}, we identify and establish the relevant properties of a family of linear operators induced by operator monotone functions $f$, which we will use to construct the $L^{2}_{f}(\sigma)$ spaces. In Section \ref{subsec:non-comm-L2-inner-products}, we show how we can construct the $L^{2}_{f}(\sigma)$ spaces via these operators and establish properties of the notion of non-commutative variance for these spaces. 

\subsection{Preliminaries} \label{sec:Preliminaries}
For a finite dimensional Hilbert space $A \cong \mbb{C}^{d}$, we denote $\Pos(A)$ and $\Pd(A)$ as the positive and positive definite operators respectively, which may also be denoted via the L\"{o}wner order, $A \geq B$, e.g. $P \in \Pd(A)$ if and only if $P > 0$. We denote the set of linear operators from $A$ to $B$, $\Lin(A,B)$ and define $\Lin(A) \equiv \Lin(A,A)$. The entry-wise conjugate, transpose, and complex conjugate of $X \in \Lin(A,B)$ are denoted $\ol{X}$, $X^{T}$, and $X^{\ast}$ respectively. A linear transformation from $\Lin(A)$ to $\Lin(B)$ is called a linear `map.' We denote the identity map from $\Lin(A)$ to itself by $\id_{A}$. We denote the Hilbert-Schmidt (HS) inner product $\langle X,Y \rangle \coloneq \Tr[X^{\ast}Y]$. For a given linear map $\Phi: \Lin(A) \to \Lin(A)$, we will make use of inner products induced by the linear map and the HS inner product, $\langle X,Y \rangle_{\Phi} \coloneq \langle X, \Phi(Y) \rangle$. We reserve the notation $\cE^{\ast}$ for the adjoint of a map $\cE$ with respect to the Hilbert-Schmidt inner product. 

A map $\cE_{A \to B}$ is $k$-positive if $(\id_{R} \otimes \cE)(P_{RA}) \geq 0$ for all $P_{RA} \geq 0$ where $R \cong \mbb{C}^{k}$. A map is positive if it is $1$-positive and completely-positive (CP) if it is $k$-positive for all $k \in \mbb{N}$. We denote the set of quantum channels (completely positive, trace-preserving (TP) maps, i.e. CPTP maps) from $\Lin(A)$ to $\Lin(B)$ by $\Channel(A,B)$. Finally, a class of maps related to $1$-positive and $2$-positive maps that will be referred to are ``Schwarz maps." A 1-positive map $\cE$ is a Schwarz map if it satisfies the ``Schwarz inequality":\footnote{Generally Schwarz maps are stated for a chosen ``$C^*$-algebra $\cA \subseteq \Lin(A)$," but we present the special case of using the entire matrix algebra for accessibility to a general quantum information theory audience.}
\begin{align}\label{eq:Schwarz-Inequality}
    \cE(x^{\ast}x) \geq \cE(x)^{\ast}\cE(x) \quad \forall x \in \Lin(A) 
\end{align} For the purposes of this paper, the relevant facts are that all unital $2$-positive maps are Schwarz maps \cite[Proposition 3.3]{paulsen2002completely}, but that Schwarz maps are a more general class.\footnote{In particular, \cite[Proposition 1]{jenvcova2012reversibility} establishes $2$-positivity holds if and only if a generalized notion of the Schwarz inequality for maps holds for the given map.} Thus, for generality, many results will be stated in terms of the adjoint of unital Schwarz maps, which are a class of trace-preserving maps more general than $2$-positive, trace-preserving maps, which are themselves a more general class of maps than quantum channels.

We denote $\Density(A)$, $\Density_{+}(A)$ as the quantum states and positive definite quantum states respectively. Given $\sigma \in \Density$, we define $\supp(\sigma)$ as the space spanned by the support of $\sigma$ and denote the projector on to this support as $\Pi_{\supp(\sigma)}$. We write $\rho \ll \sigma$ if $\supp(\rho) \subseteq \supp(\sigma)$. We denote a perfectly correlated classical state distributed according to distribution $p$ as $\chi^{\vert p}_{XX'} \coloneq \sum_{x} p(x)\dyad{x}_{X} \otimes \dyad{x}_{X'}$.

We introduce some taxonomy on functions. Some naming conventions will only make sense once the monotone metrics are introduced in the subsequent subsection.
\begin{definition}\label{def:function-taxonomy}
    Given a function $f: \mbb{R}_{+} \to \mbb{R}$, we say it is
    \begin{enumerate}
        \item (\textit{Operator Monotone}) $f(A) \leq f(B)$ for all $A,B \in \Pos(A)$ such that $A \leq B$,
        \item (\textit{Normalized}) $f(1) = 1$,
        \item (\textit{Symmetry-Inducing}) $xf(x^{-1}) = f(x)$ for $x \in \mbb{R}_{+}$.
    \end{enumerate}
    Following \cite{Petz-2011a}, we say $f$ is `standard monotone' if it is continuous on $\mbb{R}_{+}$, operator monotone, normalized, and symmetry-inducing. We denote the set of standard monotone functions by $\cM_{\text{St}}$.
\end{definition}
The following operator monotone functions which belong to $\cM_{\text{St}}$ will be of particular interest:
\begin{enumerate}
    \item \textit{Arithmetic mean (AM)}: $f_{AM}(x) = \frac{x+1}{2}$ 
    \item \textit{Harmonic mean (HM)}: $f_{HM}(x) = \frac{2x}{x+1}$
    \item \textit{Logarithmic mean (LM)}: $f_{LM}(x) = \frac{x-1}{\log(x)}$
    \item \textit{Geometric mean (GM)}: $f_{GM}(x) =\sqrt{x}$, 
\end{enumerate}
where we remark that $f_{LM}$ is normalized is by continuous extension. We also note the useful relation
\begin{align}\label{eq:ordering-of-means}
    f_{HM} \leq f_{GM} \leq f_{LM} \leq f_{AM} \ ,
\end{align}
which is within the ordering \cite{Kubo-1980a}
\begin{align}\label{eq:ordering-of-standard-monotones}
    f_{HM} \leq f \leq f_{AM} \quad  \forall f \in \cM_{\text{St}} \ .
\end{align}
Beyond the above functions, we will primarily consider the following family of normalized, operator monotone functions:
\begin{align}\label{eq:power-functions}
    f_{k}(x) \coloneq x^{k} \quad k \in [0,1] \ .
\end{align}
It is straightforward to verify from the definition of symmetry-inducing that $f_{k}(x)$ is symmetry inducing only if $k = 1/2$.

Throughout this work we use the standard information theory conventions $0f(0/0) \coloneq 0$, $0\cdot \infty \coloneq 0$,
\begin{align}
    f(0^{+}) \coloneq \lim_{x \downarrow 0} f(x) \, , \quad \text{and} \quad f'(+\infty) \coloneq \lim_{x \to +\infty} \frac{f(x)}{x} \ . 
\end{align}
Following \cite{Hiai-2017a}, we define the perspective function of continuous $f$ extended to $[0,+\infty) \times [0,+\infty)$ by
\begin{equation}\label{eq:perspective-function}
    P_{f}(x,y) \coloneq \begin{cases}
        yf(x/y) & x,y > 0 \, , \\
        yf(0^{+}) & x = 0 \, , \\
        xf'(+\infty) & y = 0 
    \end{cases}
\end{equation}
Note that if $f:\mbb{R}_{+} \to \mbb{R}_{+}$ is a monotonic function such that $f(0+) \geq 0$, then $P_{f}$ is a non-negative function. Moreover if $f_{1},f_{2}$ are monotonic functions such that $f_{1} \leq f_{2}$, then $0 \leq P_{f_{1}} \leq P_{f_{2}}$.

\subsection{The Linear Operator \texorpdfstring{$\mbf{J}_{f,\sigma}$}{} and its Properties}\label{sec:J-operator}
In this subsection we introduce the family of linear operators induced by a choice of operator monotone function $f$ and quantum state $\sigma$, $\mbf{J}_{f,\sigma}$, and establish relevant properties for this work. These linear operators will be used to construct our $L^{2}_{f}(\sigma)$ spaces. These operators have been studied previously--- see in particular~\cite{Petz-2011a} and the review in \cite[Section IV]{scandi2023quantum}. For clarity, we briefly motivate our approach. 

As discussed in the introduction, the classical versions of the maximal correlation coefficient and the $\chi^{2}$-divergence can be defined in terms of expectation and variance and thus can be addressed with operator-theoretic methods via the $L^{2}(p)$ space (See Eqs.~\eqref{eq:intro-max-corr} and \eqref{eq:chi-sq-as-variance} and subsequent discussion). We thus want non-commutative extensions of the expectation, variance, and $L^{2}(p)$ space. Classically, clearly one way of defining expectation and variance is through the $L_{2}$-inner product on (real-valued) functions $\cF(\cX)$ with respect to a distribution $p$: $\langle f, g \rangle_{p} \coloneq \sum_{x \in \cX} p(x)f(x)g(x)$. One may then define the expectation of $f$ and covariance between $f$ and $g$ through this inner product. Namely, 
\begin{align}
    \mbb{E}_{p}[f] \coloneq \langle \mbf{1} , f \rangle_{p} \quad \text{Cov}_{p}[f,g] \coloneq \langle f - \mbb{E}_{p}[f], g - \mbb{E}_{p}[g]\rangle_{p} \ , 
\end{align} 
where $\mbf{1}$ is the unit function $\mbf{1}(x) = 1$ for all $x \in \cX$. Our goal is to extend this to the non-commutative setting. To this end, we observe that one may relate the $L_{2}$-inner product with respect to $p$ to the Euclidean inner product via a linear operator $\hat{J}_{p}: \mbb{R}^{\vert \cX \vert} \to \mbb{R}^{\vert \cX \vert}$: $\langle f , g \rangle_{p} = \langle f , \hat{J}_{p}(g) \rangle_{\text{Euc}}$ where $\hat{J}_{p}(g)[x] = p(x)g(x)$. We thus wish to generalize the linear operator $\hat{J}_{p}: \mbb{R}^{\vert \cX \vert} \to \mbb{R}^{\vert \cX \vert}$ to $\mbf{J}_{\sigma}: \Lin(A) \to \Lin(A)$ for quantum state $\sigma \in \Density(A)$. 
The following definition does exactly this by also making use of an operator monotone function and the standard left and right multiplication operators:
\begin{align}\label{eq:left-and-right-mult-operators}
L_{W}(X) \coloneq WX \quad R_{W}(X) \coloneq XW \quad \forall X,W \in \Lin(A) \ . 
\end{align}
\begin{definition}\label{def:J-operator}
    Let $f$ be operator monotone and $\sigma \in \Density(A)$ with  spectral decomposition $\sigma = \sum_{i \in [d_{A}]} \lambda_{i}Q_{i}$ such that $f(\lambda_{i}/\lambda_{j})$ is defined for all $i,j \in [d_{A}]$ so that we have
    \begin{align}\label{eq:func-expansion}
        f(L_{\sigma}R_{\sigma}^{-1}) = \sum_{i,j \in [d_{A}]} f(\lambda_{i}/\lambda_{j})L_{Q_{i}}R_{Q_{j}} \ .
    \end{align}
    We then define the linear operator
    \begin{align}\label{eq:J-operator-func-expansion}
        \mbf{J}_{f,\sigma} \coloneq f(L_{\sigma}R_{\sigma}^{-1})R_{\sigma} = \sum_{i,j \in [d_{A}]} P_{f}(\lambda_{i},\lambda_{j})L_{Q_{i}} R_{Q_{j}} \ , 
    \end{align}
    where the equality in \eqref{eq:J-operator-func-expansion} makes use of \eqref{eq:func-expansion} and the perspective function as given in \eqref{eq:perspective-function}.
\end{definition}
\begin{remark}
    \eqref{eq:func-expansion} follows from functional calculus for linear maps. Works on quantum information theory generally omit the derivation \cite{HIAI_2011,Hiai-2017a}, so we provide a proof in Appendix \ref{app:func-expansion}.
\end{remark}

From the $\mbf{J}_{f,\sigma}$ operators, we induce two inner products:
\begin{align}\label{eq:inner-product definitions}
    \langle X, Y \rangle_{f,\sigma} \coloneq \langle X, \mbf{J}_{f,\sigma}(Y) \rangle \quad \langle X, Y \rangle_{f,\sigma}^{\star} \coloneq \langle X, \mbf{J}^{-1}_{f,\sigma}(Y) \rangle \ .
\end{align}
The former inner product will be used to induce the $L^{2}_{f}(\sigma)$ spaces motivated above in Section \ref{subsec:non-comm-L2-inner-products}. The latter inner product however is relevant for monotone metrics and the data processing inequality as was first established by Petz \cite{Petz-1996a} as the following definition and fact summarize.
\begin{definition}
    A metric\footnote{Petz defined a metric on the linear operators $\gamma_{\sigma}: \Lin(A) \times \Lin(A) \to \mbb{C}$ to have to satisfy the following properties:
    \begin{enumerate}[itemsep=0pt]
        \item $(X,Y) \mapsto \gamma_{\sigma}(X,Y)$ is sesquilinear,
        \item $\gamma_{\sigma}(X,X) \geq 0$ with equality if and only if $X = 0$,
        \item $\sigma \mapsto \gamma_{\sigma}(X,X)$ is continuous on $\sigma \in \Density_{+}(X)$ for all $X \in \Lin(A)$.
    \end{enumerate}} $\gamma_{\sigma}(X,Y): \Lin(A) \times \Lin(A) \to \mbb{C}$ is monotone if 
    \begin{align}\label{eq:metric-monotonicity}
        \gamma_{\cE(\sigma)}(\cE(X),\cE(X)) \leq \gamma_{\sigma}(X,X) \quad  \forall (\cE,\sigma,X) \in \Channel(A,B) \times \Density_{+}(A) \times \Lin(A) \, . 
    \end{align}
\end{definition}
\begin{fact}\label{fact:monotone-metric} \cite[Theorems 4 and 5]{Petz-1996a}
     A metric $\gamma_{\sigma}$ is monotone if and only if there exists an operator monotone function $f:\mbb{R}_{+} \to \mbb{R}_{\geq 0}$ such that
     \begin{align}
     \gamma_{\sigma}(X,Y) = \langle X , \mbf{J}^{-1}_{f,\sigma}(Y) \rangle = \langle X, Y \rangle^{\star}_{f,\sigma} \, . 
     \end{align}
\end{fact}
\begin{remark}
    In \cite{Petz-1996a}, it is stated as $f: \mbb{R}_+ \to \mbb{R}$ such that $f(0^{+}) \geq 0$. However, by monotonicity, this implies the function is non-negative valued on $\mbb{R}_{+}$.
\end{remark}

The above fact establishes that we may identify monotone metrics with the corresponding linear operator  $\mbf{J}_{f,\sigma}^{-1}$ and inner product, $\langle \cdot, \cdot \rangle_{f,\sigma}^{\star}$. Note this shows all monotone metrics are non-commutative weighted variants of the HS inner product, which will be a common idea in this work. Moreover, using this identification and that $\cE^{\ast}$ denotes the adjoint map with respect to the HS inner product, one may directly re-express the monotonicity property \eqref{eq:metric-monotonicity} as $\langle X , (\cE^{\ast} \circ \mbf{J}_{f,\sigma}^{-1} \circ \cE) X \rangle \leq \langle X , X \rangle$ for all $X \in \Lin(A)$. Note that by the non-negativity property of a metric, we may conclude $\cE^{\ast} \circ \mbf{J}_{\sigma}^{-1} \circ \cE$ is a positive definite operator on the Hilbert space $(\Lin(A),\langle \cdot, \cdot \rangle)$. We may then re-express monotonicity as a positive semidefinite relation, 
\begin{align}\label{eq:operator-form-of-DPI}
    \cE^{\ast} \circ \mbf{J}_{f,\cE{\sigma}}^{-1} \circ \cE \leq \mbf{J}_{f,\sigma}^{-1} \iff \cE \circ \mbf{J}_{f,\sigma} \circ \cE^{\ast} \leq \mbf{J}_{f,\cE(\sigma)} \ , 
\end{align}
where the equivalence is \cite[Lemma 2]{Petz-1996a}. 

The above positive semidefinite relation in fact holds more generally. Namely, as noted in \cite{Petz-1996a} it follows from \cite{Petz-1986a}, that it holds for any map $\cE$ that is the adjoint of a unital Schwarz map.
\begin{proposition}\cite{Petz-1996a} \label{prop:DPI-for-J-op}
    A monotone metric $\gamma_{f,\sigma}(A,B)$ for $\sigma \in \Density_{+}$ is monotonic for any map $\cE$ that is the adjoint of a unital Schwarz map:
    \begin{align}\label{eq:DPI-expression-for-Jfsigma}
        \cE^{\ast} \circ \mbf{J}_{f,\cE(\sigma)}^{-1} \circ \cE \leq \mbf{J}_{f,\sigma}^{-1} \quad \cE \circ \mbf{J}_{f,\sigma} \circ \cE^{\ast} \leq \mbf{J}_{f,\cE(\sigma)} \ .
    \end{align}
\end{proposition}

\paragraph{Hadamard Product Representation} To establish further properties of $\mbf{J}_{f,\sigma}$ it will be useful to know the following that follows immediately from \eqref{eq:J-operator-func-expansion} by writing an operator $X$ in the eigenbasis of $\sigma$. \footnote{This identification was observed in \cite{Petz-2011a} without identification with the perspective function. It is also implicitly used in \cite{Lesniewski-1999a}.} For an operator monotone (continuous) function $f:\mbb{R}_{+} \to \mbb{R}_{+}$,
\begin{align}\label{eq:Hadamard-prod-form}
    \mbf{J}_{f,\sigma}X = Z_{f,\sigma} \odot X \quad \text{and} \quad \mbf{J}^{-1}_{f,\sigma}X = W_{f,\sigma} \odot X \ , 
\end{align}
where $\odot$ is the Hadamard product, the basis for $i,j$ is determined by the eigenvectors of $\sigma$ and $(Z_{f,\sigma})_{ij} = P_{f}(\lambda_{i},\lambda_{j})$, $(W_{f,\sigma})_{ij} = P_{f}(\lambda_{i},\lambda_{j})^{-1}$. The inverse can be extended in the pseudoinverse sense:
\begin{align}\label{eq:J-pseudoinverse}
    [\mbf{J}_{f,\sigma}^{-1}X]_{ij} = \begin{cases}
        P_{f}(\lambda_{i},\lambda_{j})^{-1} X_{ij} & P_{f}(\lambda_{i},\lambda_{j}) \neq 0 \\
        0 & \text{o.w.}
    \end{cases}
\end{align}
More generally, the real-valued powers $p \in (-\infty,\infty)$ of the operator $\mbf{J}_{f,\sigma}$ can be expressed in this form: 
\begin{align}\label{eq:general-Hadamard-prod-form}
[(\mbf{J}_{f,\sigma}^{p})X]_{ij} = P_{f}(\lambda_{i},\lambda_{j})^{p} X_{ij} \quad \forall i,j \ , 
\end{align}
where we define the negative powers in the pseudoinverse sense. Note this tells us that $\mbf{J}_{f,\sigma}^{p}$ are Hadamard maps. Since a Hadamard map is CP if and only $Y$ is positive semidefinite \cite[Proposition 4.17]{WatrousBook}, this gives us a means to know when $\mbf{J}_{f,\sigma}$ is CP or not. In particular, as claimed in \cite{Petz-2011a} and detailed in \cite[Section IV.C]{scandi2023quantum} the following is known.
\begin{proposition}\cite{Petz-2011a,scandi2023quantum}\label{prop:nec-cond-for-CP-ness}
    For $f \in \cM_{\text{St}}$, $\mbf{J}_{f,\sigma}$ is completely positive only if $f \leq \sqrt{x}$ and $\mbf{J}_{f,\sigma}^{-1}$ is completely positive only if $f \geq \sqrt{x}$.
\end{proposition}
Beyond the property above, we make extensive use of the Hadamard product representation of these operators throughout this work, including using it to show it is algorithmically efficient to compute $\mbf{J}_{f,\sigma}$ and the inner product $\langle \cdot , \cdot \rangle_{\mbf{J}_{f,\sigma}^{p}}$ on an arbitrary operator, which we utilize in Section \ref{sec:computability}. It also allows us to see $\mbf{J}_{f,\sigma}^{p} = \sum_{i,j \in [d_{A}]} P_{f}(\lambda_{i},\lambda_{j})^{p}L_{Q_{i}}R_{Q_{j}}$. A direct calculation using this expression establishes the following to which we later appeal.
\begin{proposition}\label{prop:J-operator-self-adjoint}
    For operator monotone $f$, $\sigma \in \Density_{+}$, and $p \in (-\infty,\infty)$, $\mbf{J}^{p}_{f,\sigma}$ is self-adjoint with respect to the HS inner product.
\end{proposition}

\paragraph{Other Properties} 
We now can turn to establishing the remaining properties we will need. We begin by explaining the naming conventions for the monotone functions previously introduced. A direct calculation via \eqref{eq:Hadamard-prod-form} and Fact \ref{fact:monotone-metric} will verify the condition $\gamma_{f,\sigma}(X,X) = \Tr[\sigma^{-1}X^{\ast}X]$ when $[\sigma,X] = 0$ is equivalent to $f(1)=1$, which justifies the term `normalized' in Definition \ref{def:function-taxonomy}. The condition `symmetry-inducing' implies that the perspective function is symmetric in its arguments and the corresponding inner product $\langle \cdot , \cdot \rangle_{f,\sigma}$ is symmetric on the space of Hermitian operators. The following proposition shows these are all equivalent statements.
\begin{proposition}\label{prop:symmetry-inducing-equivalences}
    For an operator monotone function and $\sigma \in D_{+}(A)$, the following are equivalent
    \begin{enumerate}
        \item $f(x) = xf(x^{-1})$ for all $x \in (0,\infty)$,
        \item $P_{f}(x,y) = P_{f}(y,x)$ for all $x,y \in [0,\infty)^{\times 2}$,
        \item $\mbf{J}^{p}_{f,\sigma}$ is Hermitian-preserving for any $p \in (-\infty,\infty)$, and
        \item For any $p \in (-\infty,\infty)$, $\langle A, B\rangle_{\mbf{J}_{f,\sigma}^{p}} = \langle B, A \rangle_{\mbf{J}_{f,\sigma}^{p}} \in \mbb{R}$ for all $A,B \in \Herm(A)$. In particular, this means the metric is symmetric on Hermitian operators, i.e. $\gamma_{f,\sigma}(A,B) = \gamma_{f,\sigma}(B,A)$.
    \end{enumerate}
\end{proposition}
\begin{proof}
    ($1 \iff 2$) Follows immediately from \eqref{eq:perspective-function}. \\
    ($2 \iff 3$) We focus on the power $p = 1$ as all other cases have the same argument. Recall that to be Hermitian-preserving $\mbf{J}_{f,\sigma}(H)$ needs to be Hermitian for all $H \in \Herm(A)$. By \eqref{eq:Hadamard-prod-form}, this would require for all $i,j$, $P_{f}(\lambda_{i},\lambda_{j})H_{i,j} = P_{f}(\lambda_{j},\lambda_{i})\ol{H}_{j,i} =  P_{f}(\lambda_{j},\lambda_{i})H_{i,j}$ where we used that $H$ is Hermitian. Simplifying, this is equivalent to $P_{f}(\lambda_{i},\lambda_{j}) = P_{f}(\lambda_{j},\lambda_{i})$ for all $\lambda_{i},\lambda_{j} \in [0,\infty)^{\times 2}$. \\  ($2 \iff 4$) We consider $p = -1$ as this is the relevant case for $\gamma_{f,\sigma}(A,B)$, but all choices of $p$ use the same argument. By direct calculation, for arbitrary $X,Y \in \Lin(A)$,
    \begin{align}\label{eq:matrix-expansion}
        \gamma_{f,\sigma}(X,Y) = \sum_{i,j} \ol{X}_{i,j} P_{f}^{-1}(\lambda_{i},\lambda_{j})Y_{i,j} \ .
    \end{align}
    It follows that if $A,B$ are Hermitian so that $A_{i,j} = \ol{A}_{j,i}$, $B_{i,j} = \ol{B}_{j,i}$ then we have 
    $$ \gamma_{f,\sigma}(A,B) = \sum_{i,j} \ol{A}_{i,j} P_{f}^{-1}(\lambda_{i},\lambda_{j})B_{i,j} = \sum_{i,j} A_{j,i} P_{f}^{-1}(\lambda_{i},\lambda_{j}) \ol{B}_{j,i} = \sum_{i,j} A_{i,j} P_{f}^{-1}(\lambda_{j},\lambda_{i}) \ol{B}_{i,j} $$ where the last equality is relabeling $i,j$. Again using \eqref{eq:matrix-expansion}, $\gamma_{f,\sigma}(B,A) = \sum_{i,j} \ol{B}_{i,j} P_{f}^{-1}(\lambda_{i},\lambda_{j})A_{i,j}$. As these two summations only differ in the indexing of the perspective function, setting them equal to each other and re-ordering gets the equivalent condition $0 = \sum_{i,j} [P_{f}^{-1}(\lambda_{i},\lambda_{j}) - P_{f}^{-1}(\lambda_{j},\lambda_{i})] \ol{B}_{i,j}A_{i,j}$. For this condition to hold for all choices of Hermitian $A,B$, it must be the case $P_{f}^{-1}(\lambda_{i},\lambda_{j}) = P_{f}^{-1}(\lambda_{j},\lambda_{i})$ for all $i,j$. Similarly, if this condition holds, then $\gamma_{f,\sigma}(A,B) = \gamma_{f,\sigma}(B,A)$ for all $A,B \in \Herm(A)$. Lastly, as $\mbf{J}^{p}_{f,\sigma}$ is Hermitian preserving, $\langle A , \mbf{J}_{f,\sigma}^{p}(B) \rangle \in \mbb{R}$ as it is the HS inner product of two Hermitian operators. This completes the proof.
\end{proof}
These appealing properties of symmetry-inducing functions is part of the reason the ordering given in \eqref{eq:ordering-of-standard-monotones} can be so relevant. We also note a known ordering on the operator as $f$ is varied and $\sigma$ is left fixed \cite{Petz-2011a} of which we provide a proof using an identity we will use throughout this work.
\begin{proposition}\label{prop:ordering-of-NC-mult}
    If $0 \leq f_{1} \leq f_{2}$ are operator monotone, $\mbf{J}_{f_{2},\sigma} \geq \mbf{J}_{f_{1},\sigma} \geq 0$ and $0 \leq \mbf{J}_{f_{2},\sigma}^{-1} \leq \mbf{J}_{f_{1},\sigma}^{-1}$.
\end{proposition}
\begin{proof}
    By direct calculation using \eqref{eq:matrix-expansion},
\begin{align}\label{eq:MC-form}
    \langle A , \mbf{J}_{f,\sigma}A \rangle = \sum_{i,j} P_{f}(\lambda_{i},\lambda_{j}) \vert A_{ij} \vert^{2} \ . 
\end{align}
Recalling the properties of the perspective function for monotone functions given at the start of the note, we may conclude for operator monotone functions $f_{1},f_{2}$ such that $0 \leq f_{1} \leq f_{2}$, we have $0 \leq \mbf{J}_{f_{1},\sigma} \leq \mbf{J}_{f_{2},\sigma}$. The argument for the inverse is identical using \eqref{eq:J-pseudoinverse}.
\end{proof}

Next, we establish all cases such that $\mbf{J}_{f,\sigma}$ is multiplicative over tensor products of states.
\begin{proposition}\label{prop:multiplicativity-of-J}
    $\mbf{J}_{f,\rho \otimes \sigma} = \mbf{J}_{f,\rho} \otimes \mbf{J}_{f,\sigma}$ for all $\rho_{A},\sigma_{A'} \in \Density(A) \times \Density(A')$ if and only if 
    $f(x) = x^{k}$ for some $k \in [0,1]$. Moreover, the geometric mean, $f_{GM}(x) = x^{1/2}$, is the unique symmetry-inducing operator monotone function satisfying multiplicativity.
\end{proposition}
\begin{proof}
    The proof proceeds in steps. First, we establish $\mbf{J}_{f,\rho \otimes \sigma} = \mbf{J}_{f,\rho} \otimes \mbf{J}_{f,\sigma}$ for all $\rho_{A},\sigma_{A'} \in \Density(A) \times \Density(A')$ is equivalent to $f$ being multiplicative over the positive reals. We then combine this with Cauchy's multiplicative functional equation \cite{Aczel-1966a} and the characterization of operator monotone functions \cite{Bhatia-1997a} to complete the proof.
    
    We begin by showing $f$ must be multiplicative. Let $\rho_{A} = \sum_{i \in \Sigma} \lambda_{i} \dyad{\nu_{i}}$ and $\sigma_{A'} = \sum_{j \in \Lambda} \kappa_{j} \dyad{\omega_{j}}$. We may express $X_{AB}$ in this basis via 
    \begin{align}
        X_{AB} = \sum_{\substack{i,i' \in \Sigma \\ j,j' \in \Lambda}} X_{i,i',j,j'} \ket{\nu_{i}}\bra{\nu_{i'}} \otimes \ket{\omega_{i}}\bra{\omega_{j}} \ ,
    \end{align} 
    and $\rho_{A} \otimes \sigma_{A'} = \sum_{i,j} \lambda_{i}\kappa_{j} \dyad{\nu_{i}} \otimes \dyad{\omega_{j}}$.
    By \eqref{eq:Hadamard-prod-form} and the definition of the perspective function, we have 
    \begin{align}
        \mbf{J}_{f,\rho \otimes \sigma}(X)_{i,i',j,j'} &= (\lambda_{i'}\cdot \kappa_{j'})f\left(\frac{\lambda_{i} \cdot \kappa_{j}}{\lambda_{i'} \cdot \kappa_{j'}} \right) X_{i,i',j,j'} \\ 
        \mbf{J}_{f,\rho} \otimes \mbf{J}_{f,\sigma}(X)_{i,i',j,j'} &=\lambda_{i'}f\left(\frac{\lambda_{i}}{\lambda_{i'}} \right) \cdot \kappa_{j'}f\left(\frac{\kappa_{j}}{ \kappa_{j'}} \right) X_{i,i',j,j'} \ .
    \end{align}
    Since the equality would need to hold for all linear operators and choices of eigenvalues, setting the two equal to each other, the above simplifies to $f\left(a \cdot b \right) =f(a)f(b)$ for all $a,b \in (0,+\infty)$.
    
    Now, by Cauchy's multiplicative functional equation, $f(a \cdot b) = f(a)f(b)$ for all $a,b \in (0,+\infty)$ if and only if $f(x) \coloneq x^{k}$ for some $k \in \mbb{R}$ (see e.g. \cite[Section 2.1, Theorem 3]{Aczel-1966a}). Moreover, as $\mbf{J}_{f,\sigma}$ is defined via an operator monotone function (Definition \ref{def:J-operator}) and the power function is operator monotone if and only if $k \in [0,1]$ \cite[Theorem V.2.10]{Bhatia-1997a}, this establishes the first claim of the proposition. To obtain the moreover statement, recall that being symmetry-inducing means $f(x) = xf(x^{-1})$. Thus, a power function is symmetry-inducing if and only if $x^{k} = x^{1-k} \iff x^{2k-1} = 1$ for all $x \in (0,\infty)$, which only holds for $k = 1/2$. This completes the proof.
\end{proof}
\noindent We stress the above shows many operator monotone functions that induce means that we are interested in, e.g. $f \in \{HM,LM,AM\}$, do not induce $\mbf{J}_{f,\sigma}$ that are multiplicative over tensor products.

We also make use of the following which shows that for all operator monotone $f$, under the trace $\mbf{J}_{f,\sigma}$ and its inverse act as multiplying or dividing by $\sigma$. We note part of these identities was previously established in \cite{Lesniewski-1999a} via different methods.
\begin{proposition}\label{prop:mult-and-div-under-trace} 
    For any $\sigma$ and normalized monotone $f$, for all $X \in \Lin(A)$, the following identities hold 
    \begin{align*} 
        \Tr[\mbf{J}_{f,\sigma}(X)] = \Tr[\sigma X] = \Tr[\sigma^{1/2}\mbf{J}_{f,\sigma}^{1/2}(X)] \ ,
    \end{align*}
    \begin{align*}
    \Tr[\mbf{J}_{f,\sigma}^{-1}(X)] = \Tr[\sigma^{-1}X] \, , \quad \Tr[\sigma \mbf{J}_{f,\sigma}^{-1}(X)] = \Tr[\Pi_{\supp(\sigma)}X\Pi_{\supp(\sigma)}]    \ . 
    \end{align*}
\end{proposition}
\begin{proof}
    We prove the relation involving the square root and the inverse case since this makes it clear how to prove the others. We begin with the square root identity:
    \begin{align*}
        \Tr[\sigma^{1/2}\mbf{J}_{f,\sigma}^{1/2}(X)] =& \Tr\left[\left(\sum_{i} \lambda_{i}^{1/2} \dyad{\nu_{i}}\right)\left(\sum_{j,k} P_{f}(\lambda_{j},\lambda_{k})^{1/2} X_{i,j}\ket{\nu_{j}}\bra{\nu_{k}}\right)\right] \\
        =& \sum_{i,j,k} \lambda_{i}^{1/2} P_{f}(\lambda_{j},\lambda_{k})^{1/2} \delta_{i,j}\delta_{i,k} \\
        =& \sum_{i} \lambda_{i}^{1/2} P_{f}^{1/2}(\lambda_{i},\lambda_{i})X_{ii} = \sum_{i} \lambda_{i} X_{ii} = \Tr[\sum_{i} \lambda_{i} \dyad{\nu_{i}}X] = \Tr[\sigma X] \ , 
    \end{align*}
    where the fourth equality uses that $P_{f}^{1/2}(\lambda_{i},\lambda_{i}) = \lambda_{i}^{1/2}$ by normalization of $f$. The inverse identity is similar:
    \begin{align*}
        \Tr[\mbf{J}_{f,\sigma}^{-1}(X)] = \sum_{i,j} \Tr[P_{f}^{-1}(\lambda_{i},\lambda_{j})X_{ij}] 
        = \sum_{i} \Tr[P_{f}^{-1}(\lambda_{i},\lambda_{i})X_{ii}] 
        =& \sum_{i : \lambda_{i} > 0 } \lambda_{i}^{-1} X_{ii} \\
        =& \sum_{i : \lambda_{i} > 0 } \Tr[\lambda_{i} ^{-1}\dyad{\nu_{i}}X] \\
        =& \Tr[\sigma^{-1}X]  \ , 
    \end{align*}
    where the second equality evaluates the trace in the basis of spectral decomposition of $\sigma$, the third uses the normalization of $f$ and \eqref{eq:perspective-function}, the fourth is by the spectral decomposition, and the last is by the Moore-Penrose pseudoinverse. From this it is straightforward to prove the multiplication case. The inverse multiplied by $\sigma$ follows the same argument as above except the eigenvalues will now cancel.
\end{proof}

Lastly, we note the following useful fact, which applies to special cases we care about, e.g. $f_{GM},f_{LM}$, and allows us to identify cases where we can restrict to the support of $\sigma$ without loss of generality. This will be useful in Section \ref{sec:quantum-chi-squared}.
\begin{proposition}\label{prop:suff-conds-for-restricting-support}
    If $f$ is operator monotone and such that $\lim_{x \to \infty} f(x) = +\infty$, $\lim_{x \to \infty} f'(x) = 0$, and $\lim_{x \downarrow 0} f(x) = 0$, then 
    $$\mbf{J}_{f,\sigma}(X) = \Pi_{\supp(\sigma)}\mbf{J}_{f,\sigma}(X)\Pi_{\supp(\sigma)} = \mbf{J}_{f,\sigma}(\Pi_{\supp(\sigma)}X\Pi_{\supp(\sigma)}) \ , $$
    and similarly for $\mbf{J}_{f,\sigma}^{-1}$.
\end{proposition}
\begin{proof}
    Under the assumptions on $f$, $f'(+\infty) = 0$ via L'H\^{o}pital's and $f(0^{+}) = 0$. It follows $P_{f}(x,y) = 0$ if and only if $\lambda_{i}$ or $\lambda_{j}$ is zero. Thus, in this setting $\mbf{J}_{f,\sigma}(X)$ projects $X$ onto the support of $\sigma$ by \eqref{eq:Hadamard-prod-form}. By the definition of the pseudoinverse generalization of $\mbf{J}_{f,\sigma}^{-1}$ given in \eqref{eq:J-pseudoinverse}, we may conclude the same for $\mbf{J}_{f,\sigma}^{-1}$.
\end{proof}

\subsection{Non-Commutative Extensions of \texorpdfstring{$L^{2}(p)$, \texorpdfstring{$L^{2}_{f}(\sigma)$}{}}{}}\label{subsec:non-comm-L2-inner-products} 
With the theory of the operator $\mbf{J}_{f,\sigma}$ explicated, we now may define the non-commutative extensions of $L^{2}(p)$. Given an operator monotone function $f$ and quantum state $\sigma \in \Density_{+}$, we replace the inner product on real-valued functions $\langle \cdot ,\cdot \rangle_{p}$ with the inner product on linear operators $\langle X , Y \rangle_{f,\sigma} = \langle X, \mbf{J}_{f,\sigma}(Y) \rangle$ as introduced in \eqref{eq:inner-product definitions}. These will be definite inner products on $\Lin(A)$ as is straightforward to see or may be verified via the following section.

Recalling both that $\mbb{E}_{p}[f] = \langle \mbf{1} , f \rangle_{p}$ and that for an observable $X \in \Herm(A)$, we take the expectation of $X$ with respect to $\sigma$ to be $\Tr[\sigma X]$, Proposition \ref{prop:mult-and-div-under-trace} motivates the definition of the expectation of a linear operator with respect to a state $\sigma$ for a normalized operator monotone $f$:
\begin{align}\label{eq:f-expectation}
    \mbb{E}_{f,\sigma}[X] \coloneq \langle \mbb{1}, X\rangle_{f,\sigma} = \Tr[\sigma X] \ .
\end{align}
Note that this is independent of the choice of normalized operator monotone function $f$. Furthering this identification, we can define the general notion of covariance with respect to $\sigma$ for any normalized operator monotone $f$:\footnote{A direct calculation will verify that \eqref{eq:f-covariance} can be simplified to the definition introduced in \cite{hiai-2012quasi}. It also recovers the simpler form in \cite{Petz-2011a} in the case $X,Y$ are Hermitian.}
\begin{align}\label{eq:f-covariance}
    \text{Cov}_{f,\sigma}(X,Y) \coloneq \langle X - \langle \mbb{1}, X \rangle_{f,\sigma}\mbb{1} , Y - \langle \mbb{1} , Y \rangle_{f,\sigma}\mbb{1} \rangle_{f,\sigma}  \ .
\end{align}
One may then obtain the non-commutative variance: $\Var_{f,\sigma}(X) \coloneq \text{Cov}_{f,\sigma}(X,X)$. Motivated by \eqref{eq:f-expectation} and \eqref{eq:f-covariance}, for any normalized operator monotone $f$, we may then view the linear operators equipped with $\langle \cdot, \cdot \rangle_{f,\sigma}$ as a non-commutative extension of $L^{2}(p)$, which we denote $L^{2}_{f}(\sigma)$ for notational clarity. We will formally verify this as an inner product space in Section \ref{sec:functional-analysis-tools}. 

While \eqref{eq:f-expectation} and \eqref{eq:f-covariance} justify our claim that $L^{2}_{f}(\sigma)$ generalizes $L^{2}(p)$, we further verify this by establishing some basic properties of the standard variance are satisfied for our non-commutative extensions. These properties will be useful in subsequent sections. Of particular importance for our work is the data processing inequality for variance (Item 3). At least for Hermitian observables and CPTP maps, this data processing inequality was stated previously \cite[Page 10]{Petz-2011a}.\footnote{To the best of our knowledge, the data processing inequality for non-commutative variance has not been utilized elsewhere. We suspect this is because its identification does not seem relevant for quantum Fisher information. This is likely why it was left as an offhand comment in that work whereas in this work it is a central idea.}
\begin{proposition}\label{prop:variance-properties}
    For any normalized operator monotone $f:\mbb{R}_{+} \to \mbb{R}_{\geq 0}$ and $\sigma \in \Density_{+}(A)$, 
    \begin{enumerate}
    \item If $X \in \Herm$,\begin{align}
        \Var_{f,\sigma}(X) = \langle X , X \rangle_{f,\sigma} - \left( \langle \mbb{1}, X \rangle_{f,\sigma}  \right)^{2} \ .
    \end{align}
    \item If $f \geq 0$, then $\Var_{f,\sigma}(X) \geq 0$ for all $X \in \Lin(X)$.
    \item For any $X \in \Lin(A)$ and any $\cE$ that is the adjoint of a unital Schwarz map, 
    \begin{align}
        \Var_{f,\sigma}(\cE^{\ast}(X)) \leq \Var_{f,\cE(\sigma)}(X) \ .
    \end{align}
    \item If $X \in \Lin(A)$ and $\langle \mbb{1}, X \rangle_{f,\sigma} = 0$, then $\Var_{f,\sigma}(X) = \langle X , X \rangle_{f,\sigma}$.
    \end{enumerate}
\end{proposition}
\begin{proof}
    For Item 1, let $\mu \coloneq \langle \mbb{1}, X \rangle_{f,\sigma} $. Then, by linearity,
    \begin{align}
        \Var_{f,\sigma}(X) =& \langle X , X \rangle_{f,\sigma} - \mu \langle X ,\mbb{1} \rangle_{f,\sigma} - \mu^{2} + \mu^{2}\langle \mbb{1}, \mbb{1}\rangle_{f,\sigma} \\
        =& \langle X , X \rangle_{f,\sigma} - \mu \langle X ,\mbb{1} \rangle_{f,\sigma} - \mu^{2} + \mu^{2}\langle \mbb{1}, \mbb{1}\rangle_{f,\sigma} \\
        =& \langle X , X \rangle_{f,\sigma} - \mu \langle X , \sigma \rangle - \mu^{2} + \mu^{2}\langle \mbb{1}, \sigma \rangle \label{eq:variance-property-step-1} \\
        =& \langle X , X \rangle_{f,\sigma} - \mu \langle \sigma, X \rangle 
    \end{align}
    where the third equality uses that if $[\sigma,X] = 0$, $\mbf{J}_{f,\sigma}(X) = \sigma X$ and that $X = \mbb{1}$ in this case, and the fourth uses the assumption $X$ is Hermitian. Using $\mu = \Tr[\sigma X]$ completes the proof of this item. Item 2 follows directly from the definition by using \eqref{eq:f-covariance} and Proposition \ref{prop:ordering-of-NC-mult}. Item 3 follows from
    \begin{align}
        \Var_{f,\cE(\sigma)}(X) =& \langle X - \langle \mbb{1}, X \rangle_{f,\cE(\sigma)}\mbb{1}, \mbf{J}_{f,\cE(\sigma)}[X - \langle \mbb{1}, X \rangle_{f,\cE(\sigma)}\mbb{1}] \rangle \\
        \geq & \langle X - \langle \mbb{1}, X \rangle_{f,\cE(\sigma)}\mbb{1}, (\cE \circ \mbf{J}_{f,\sigma} \circ \cE^{\ast}) [X - \langle \mbb{1}, X \rangle_{f,\cE(\sigma)}\mbb{1}] \rangle \\
        =& \langle \cE^{\ast}(X) - \langle \mbb{1}, X \rangle_{f,\cE(\sigma)}\mbb{1},[\cE^{\ast}(X) - \langle \mbb{1}, X \rangle_{f,\cE(\sigma)}\mbb{1}] \rangle_{f,\sigma} \\
        =& \Var_{f,\sigma}(\cE^{\ast}(X))
    \end{align}
    where the inequality is Proposition \ref{prop:DPI-for-J-op}, the second equality uses that the adjoint of a trace-preserving map is unital, and the final equality uses 
    $$ \langle \mbb{1} , X \rangle_{f,\cE(\sigma)} = \Tr[\mbf{J}_{f,\cE(\sigma)}(X)] = \Tr[\cE(\sigma)X] = \Tr[\sigma \cE^{\ast}(X)] = \langle \mbb{1}, \cE^{\ast}(X) \rangle_{f,\sigma} \ , $$
    where we have used Proposition \ref{prop:mult-and-div-under-trace} and the definition of adjoint map. Item 4 follows from  $\langle \mbb{1}, X \rangle_{f,\sigma} = \Tr[\mbb{1}\mbf{J}_{f,\sigma}(X)] = \Tr[\sigma X] = \Tr[\sigma X^{\ast}] =0$ and plugging this into \eqref{eq:variance-property-step-1}.
\end{proof}

\section{The Choi Isomorphism, Operator Norms, and Hilbert Spaces}\label{sec:functional-analysis-tools}
The rest of this paper will make use of the $L^{2}_{f}(\sigma)$ spaces we just introduced. It will however make use of two other key pre-existing tools: the Choi isomorphism and maps between Hilbert spaces. The latter we mean in a more general sense than commonly used in contemporary quantum information theory, so we provide extra background for clarity of the work. The key point we aim to clarify is the relation between eigenvalues of an operator, operator norms of the operator, and certain optimization problems.

\paragraph{Inner Product and Hilbert Spaces}
Recall that an inner product space over field $\mbb{F}$, $(V,\langle \cdot , \cdot \rangle_{e})$ is a vector space $V$ with an inner product $\langle \cdot, \cdot \rangle_{e} : V \times V \to \mbb{F}$. Also recall that a Hilbert space over field $\mbb{F}$ is an inner product space $(V,\langle \cdot , \cdot \rangle_{e})$ that also is a complete metric space $(V,\Vert \cdot \Vert_{e})$ where $\Vert X \Vert_{e} \coloneq \sqrt{\langle X , X \rangle_{e}}$ is the canonical norm. Recalling that any finite dimensional normed space is complete, see e.g. \cite{Kreyszig-1991a}, as long as $\langle \cdot, \cdot \rangle_{e}$ is in fact an inner product on finite-dimensional vector space $V$, $(V, \Vert \cdot \Vert_{e})$ is a Hilbert space. 

As stated in Section \ref{sec:Preliminaries}, given a linear map $\Phi_{A \to A}$, we may define what we hope to be an inner product on a subset of linear operators $\Lin(A)$, $\langle X , Y \rangle_{\Phi} \coloneq \langle X, \Phi(Y) \rangle$. In this work, we will be exclusively interested in the case $\Phi = \mbf{J}^{p}_{f,\sigma}$ for $f$ being operator monotone, $\sigma \in \Density(A)$, and $p \in \{-1,1\}$, i.e. the inner products $\langle \cdot, \cdot \rangle_{f,\sigma}$ and $\langle \cdot, \cdot \rangle_{f,\sigma}^{\star}$ defined in \eqref{eq:inner-product definitions}. The relevant space of linear operators will generally be the following:
\begin{align}\label{eq:lin-op-inner-prod-space}
    \Lin(A|\sigma) := \{X \in \Lin(A) : \Pi_{\supp(\sigma)}X = X\Pi_{\supp(\sigma)} \} \ .
\end{align}
\sloppy This is the relevant space due to the following proposition, which establishes $L^{2}_{f}(\sigma) \coloneq (\Lin(A),\langle \cdot , \cdot \rangle_{f,\sigma})$ is a complex Hilbert space in the same manner $L^{2}(p)$ is a real Hilbert space.
\begin{proposition}
    For any operator monotone $f$, density matrix $\sigma$, and subspace $X \subset L(A\vert \sigma)$ $(X, \Vert \cdot \Vert_{f,\sigma})$ is a (complex) Hilbert space.
\end{proposition}
\begin{proof}
    As already explained, it suffices to check $\Vert \cdot \Vert_{f,\sigma}$ is an inner product. All properties are straightforward to verify from the definition except positive-definiteness. However, by \eqref{eq:MC-form}, $\langle X, \mbf{J}_{f,\sigma}(X)\rangle = 0$ implies that $\Pi_{\supp(\sigma)}X\Pi_{\supp(\sigma)} = 0$. It follows that positive definiteness holds on $\Lin(A \vert \sigma)$ and thus any subspace. Thus it is an inner product space and thus a Hilbert space. The subspace case then follows as it is itself a vector space and the same inner product works.
\end{proof}

We will want to often work with Hermitian operators, which are arguably the non-commutative extension of real-valued functions. The Hermitian operators are a \textit{real} vector space, so in that case we will also need to guarantee the inner product is real-valued on Hermitian operators to guarantee we have a real Hilbert space.\footnote{Note this means that if we instead defined $L^{2}_{f}(\sigma) \coloneq (\Herm(A),\langle \cdot , \cdot \rangle_{f,\sigma})$, we would have that it even is the same type of Hilbert space as $L^{2}(p)$. We do not do this because we will define a set of maximal correlation coefficients that use linear operators, thereby recovering the key definition of \cite{Beigi-2013a}.}
\begin{proposition}\label{prop:orthogonal-Herm-is-Hilbert} Let $f$ be symmetry-inducing, $\sigma \in \Density(A)$, and $p \in (-\infty,+\infty)$. For any subspace of $\Lin(A \vert \sigma) \cap \Herm(A)$ is a real Hilbert space when equipped with the inner product induced by the HS inner product and linear map $\mbf{J}^{p}_{f,\sigma}$.
\end{proposition}
\begin{proof}
    By Proposition \ref{prop:symmetry-inducing-equivalences}, $\mbf{J}_{f,\sigma}^{p}$ is Hermitian preserving. Thus $\langle A , \mbf{J}^{p}_{f,\sigma}(B) \rangle \in \mbb{R}$ as it is the Hilbert-Schmidt inner product between two Hermitian operators. Moreover, we already know it satisfies the conditions for being an inner product on linear operators, so this completes the proof.
\end{proof}

\paragraph{Choi Representation}
The Choi representation is a bijection between $\Trans(A,B)$ and $\Lin(A) \otimes \Lin(B)$. There is some freedom in our choice of bijection, and we choose the following. For $\cE \in \Trans(A,B)$
\begin{align}\label{eq:action-of-Choi}
    \Omega_{\cE} = (\id_{A} \otimes \cE)(\Phi^{+}) = \sum_{i,j} \ket{i}\bra{j} \otimes \cE(\ket{i}\bra{j}) \ ,
\end{align}
where $\Phi^{+} = \dyad{\Phi^{+}}$ and $\ket{\Phi^+} = \sum_{i} \ket{i}\ket{i}$ is the unnormalized maximally entangled state. A direct calculation will find the action of a channel may be expressed in terms of the Choi operator:
\begin{align}
    \cE(X) = \Tr_{A}[X^{T} \otimes I_{B} \Omega_{\cE}] \ . 
\end{align}
We will make use of this identity subsequently.

\paragraph{Connecting Operator Norms and Optimization}

Recall that given a linear map $\Lambda \in \Trans(A,B)$, if $V \subset \Lin(A), W \subset (B)$ are vector spaces and $\mathscr{X} = (\Lin(A),\Vert \cdot \Vert_{X})$, $\mathscr{Y} = (\Lin(B), \Vert \cdot \Vert_{Y})$ are normed spaces, then one may define the (map) norm \cite[Proposition 2.1]{Conway-1985a} 
\begin{align}
    \Vert \Lambda \Vert_{\mathscr{X} \to \mathscr{Y}} =& \sup\{ \Vert \Lambda(x) \Vert_{Y}: 0 \neq x \in V \, , \, \Vert x \Vert_{X} \leq 1 \} \\
    =& \sup\left\{ \frac{\Vert \Lambda(x) \Vert_{Y}}{\Vert x \Vert_{X}}: 0 \neq x \in V \right\} \ . \label{eq:map-norm-ratio-form} 
\end{align}
In the case the normed spaces are given the further structure that they are Hilbert spaces, one obtains the following, which is an immediate consequence of the Riesz-Representation theorem (a proof is provided in the appendix for completeness).
\begin{lemma}\label{lem:map-norms-as-optimizations}
    Let $X,Y$ be finite-dimensional vector spaces. Let $\mathscr{X} = (X,\langle \cdot, \cdot \rangle_{X})$, $\mathscr{Y} = (Y,\langle \cdot, \cdot \rangle_{Y})$ be Hilbert spaces. Let $\Lambda: X \to Y$ be a linear map. Then
    \begin{align}\label{eq:map-norm-as-optimization}
        \Vert \Lambda \Vert_{\mathscr{X} \to \mathscr{Y}} = \max\{ \vert \langle y, \Lambda(x) \rangle_{Y} \vert : x \in X, \, , y \in Y, \, \Vert x \Vert_{X} = 1 \, , \, \Vert y \Vert_{Y} = 1 \} \ .
    \end{align}
\end{lemma}
\noindent While simple, this variational form of an operator norm will be our connection to optimization problems and will be central our approach to relating maximal correlation and $\chi^{2}$-contraction coefficients.

Note Lemma \ref{lem:map-norms-as-optimizations} combined with \eqref{eq:action-of-Choi} tells us that if the output inner product space uses the Hilbert-Schmidt inner product, it may be re-expressed in terms of the Choi operator:
\begin{align}\label{eq:map-norm-in-terms-of-Choi}
    \Vert \Lambda \Vert_{\mathscr{X} \to L_{2}} = \max\{ \Tr[x^{\Trans} \otimes  y \Omega_{\Lambda} ] \vert : x \in X, \, , y \in Y, \, \Vert x \Vert_{X} = 1 \, , \, \Vert y \Vert_{2} = 1 \} \ .
\end{align}
This has been noted and utilized previously in \cite{Delgosha-2014a}. In the case the input norm is invariant under the transpose, such as when it uses the HS inner product, we may also remove the transpose.

\paragraph{Singular Values and Schmidt Decompositions}
Lastly, we note a relation between the singular values of a linear map $\Lambda: X \to Y$ and the Schmidt decomposition of its Choi operator in $(X, \langle \cdot , \cdot \rangle)$ that was observed in \cite{Beigi-2013a,Delgosha-2014a}. 
\begin{proposition}\label{prop:Schmidt-to-sing}
Let $X \subset \Lin(A), Y \subset \Lin(B)$ such that  $\mathscr{X} = (X,\langle \cdot, \cdot \rangle)$, $\mathscr{Y} = (Y,\langle \cdot, \cdot \rangle)$ are Hilbert spaces. Consider a linear map $\Lambda: X \to Y$. The Schmidt decomposition of $\Omega_{\Lambda}$ with respect to $\mathscr{X}$, $\mathscr{Y}$ are the singular values of $\Lambda$ and vice-versa.
\end{proposition}
\begin{proof}
By assumption $\mathscr{X},\mathscr{Y}$ are Hilbert spaces. Therefore, we may apply the Schmidt decomposition  of $\Omega_{\Lambda}$, $\Omega_{\Lambda} = \sum_{i} \lambda_{i} M_{i} \otimes N_{i}$ where $\{\lambda_{i}\}_{i}$ are real and non-negative, and $\{M_{i}\}_{i}$, $\{N_{i}\}_{i}$ form orthonormal bases (with respect to HS inner product) for $X$ and $Y$ respectively. It is known that the adjoint map (with respect to the Hilbert Schmidt inner product) $\Lambda^{\ast}$ satisfies $\Omega_{\Lambda^{\ast}} = \mbb{F}(\Omega_{\Lambda})^{T}\mbb{F}$ \cite{George-2024a}, but we verify it more generally here:
\begin{align}
    \langle y, \Lambda(x) \rangle = \Tr[x^{T} \otimes y \Omega_{\Lambda}] = \sum_{i} \lambda_{i} \Tr[x^{T} \otimes y (M_{i} \otimes N_{i})] =& \sum_{i} \lambda_{i} \Tr[x \otimes y^{T} (M_{i}^{T} \otimes N_{i}^{T})] \\
    =& \Tr[y^{T} \otimes x \mbb{F}(\Omega_{\Lambda})^{T}\mbb{F}] \\
    =& \langle \Lambda^{\ast}(y), x \rangle \ .
\end{align}
Then we have the equations
    \begin{align}
        \Lambda(X) =& \Tr_{A}[(X^{T} \otimes I)\Omega_{\Lambda}] = \sum_{i} \lambda_{i} \Tr_{A}[(X^{T} \otimes I) M_{i} \otimes N_{i}] = \sum_{i} \lambda_{i} \langle X, M_{i}^{T} \rangle N_{i} \label{eq:action-in-terms-of-sing-values} \\
        \Lambda^{\ast}(Y) =& \Tr_{A}[(Y^{T} \otimes I) \Omega_{\Lambda^{\ast}}] = \sum_{i} \lambda_{i} \Tr_{A}[(Y^{T} \otimes I) N^{T}_{i} \otimes M^{T}_{i}] = \sum_{i} \lambda_{i} \langle Y^{\ast}, N_{i} \rangle M_{i}^{T} \ .
    \end{align}
    From this, we obtain 
    \begin{align}(\Lambda^{\ast} \circ \Lambda)(X) = \sum_{i'} \lambda_{i'} \langle \Lambda(X)^{\ast}, N_{i'} \rangle M_{i'}^{T} = \sum_{i,i'} \lambda_{i'}\lambda_{i} \ol{\langle X, M_{i}^{T} \rangle} \langle N_{i}^{\ast}, N_{i'} \rangle M_{i'}^{T} = \sum_{i} \lambda_{i}^{2} \ol{\langle X, M_{i}^{T} \rangle} M_{i}^{T} \ ,
    \end{align}
    where we have used the orthonormality of the $N_{i}$. As the $\{M_{i}\}$ are an orthonormal basis (ONB), so too are $\{M_{i}^{T}\}$, which thus are the eigenvectors of $\Lambda^{\ast} \circ \Lambda$, making $\lambda_{i}^{2}$ the eigenvalues. This completes the proof.
\end{proof}

\section{Quantum Maximal Correlation Coefficients}\label{sec:q-maximal-correlation-coeff}
In this section, we develop a framework of quantum maximal correlation coefficients using the non-commutative $L^{2}_{f}(\sigma)$ spaces with the aim of extending the operational interpretations of the classical maximal correlation coefficient in the quantum setting (See Fig.~\ref{fig:CIT-vs-QIT}). We begin by introducing families of quantum maximal correlation coefficients and establishing their fundamental properties, including consistency with the classical case, invariance under isometries, and satisfying the data processing inequality (Theorem \ref{thm:properties-of-maximal-corr-coeffs}). 
We then analyze the conditions under which these coefficients attain their extreme values (Section \ref{subsec:extreme-values-and-ACD}). While these extremal cases are mathematically natural, they also carry operational significance. In the classical setting, $\mu(X:Y)_{p}$ characterizes independence and $\mu(X:Y)_{p}=1$ characterizes the ability to extract perfect classical correlations via local operations both in one-shot and asymptotic regimes. In the quantum setting, we show that all maximal correlation coefficients vanish if and only if the systems $A$ and $B$ are independent and that for all $f \geq f_{GM}$, they are bounded above by one. We establish the equivalent conditions for perfect classical correlations extracted in the one-shot setting using local operations, which includes $\mu_{AM}(A:B)_{\rho} = 1$ (Theorem \ref{thm:extreme-values-summary}).
Finally, in Section \ref{subsec:strong-monotones}, we identify a family of quantum maximal correlation coefficients that tensorize, and hence define strong monotones for state transformations under local operations (Theorem \ref{thm:k-correlation-nec-for-local-processing}). This result can be further strengthened for the coefficient corresponding to the geometric mean $f=f_{GM}$ (Theorem \ref{thm:mu-GM-as-monotone}). As an application, we derive new necessary conditions for asymptotically extracting perfect correlations (Theorem \ref{thm:asymptotic-data}).In Table \ref{table:q-max-corr-coeff}, we summarize many of the results of this section and where to find them.
\begin{table}
\centering
\begin{tabular}{>{\centering\arraybackslash}p{8.5cm} | >{\centering\arraybackslash}p{8.5cm}}\hline\hline
Property & Operator Monotone Functions \\\hline\hline
Data Processing Under Local Operations & $f$ such that $f(1) = 1$ (Theorem \ref{thm:properties-of-maximal-corr-coeffs}) \\  \hline
$\mu_{f}^{\text{Lin}}(A:B)_{\rho},\mu_{f}(A:B)_{\rho} = 0$ iff $\rho_{AB} = \rho_{A} 
\otimes \rho_{B}$ & $f$ such that $f(1) = 1$ (Theorem \ref{thm:extreme-values-summary}) \\ \hline
$\mu_{f}(A:B)_{\rho},\mu_{f}^{\text{Lin}}(A:B)_{\rho} \leq 1$ & $f \geq f_{GM}$ (Theorem \ref{thm:extreme-values-summary}), $\{f_{k}\}_{k \in [0,1]}$ (Lemma \ref{lem:map-norm-is-1-for-k-correlation-coeff}) \\ \hline
Characterizes Exact Correlation Extraction & $f_{AM}$ (Theorem \ref{thm:extreme-values-summary}) \\ \hline
$\mu^{\text{Lin}}_{f}(AA':BB')_{\rho \otimes \sigma} = \max\{\mu^{\text{Lin}}_{f}(A:B)_{\rho},\mu^{\text{Lin}}_{f}(A':B')_{\sigma}\}$ & $\{f_{k}\}_{k \in [0,1]}$ (Lemma \ref{lem:k-correlation-tensorize})  \\ \hline
Necessary Conditions for ACD & $\{f_{k}\}_{k \in [0,1]}$ (Theorem \ref{thm:asymptotic-data}) \\
\hline\hline
\end{tabular}
\caption{Summary of important properties of quantum maximal correlation coefficients, which operator monotone functions we establish them for, and where they are proven in this section.
}
\label{table:q-max-corr-coeff}
\end{table}

\subsection{Quantum Maximal Correlation Coefficients and Data Processing}\label{subsec:quantum-maximal-corr-coeff-and-DPI}
\sloppy To motivate our quantum definitions, we begin by recalling the definition of the classical maximal correlation coefficient. Given a joint distribution $p_{XY}$, the classical (Hirschfeld-Gebelein-R\'{e}nyi) maximal correlation coefficient \cite{Hirschfeld-1935a,Gebelein-1941a,Renyi-1959sieve,Renyi-1959measuresofdep} is given by 
\begin{equation}\label{eq:maximal-correlation-coefficient}
    \begin{aligned}
    \mu(X:Y)_{p} \coloneq& \sup \vert \mbb{E}_{p_{XY}}[f(X)g(Y)] \vert \\
    & \mbb{E}_{p_{X}}[f(X)] = \mbb{E}_{p_{Y}}[g(Y)] = 0 \\
    & \mbb{E}_{p_{X}}[f(X)^{2}] = \mbb{E}_{p_{Y}}[g(Y)^{2}] = 1 \ , 
    \end{aligned}
\end{equation}
where $f,g$ are real-valued, measurable functions on the corresponding spaces. We have used subscripts to clarify what distributions control which expectations. We remark that because the expectation is required to be zero, the second set of constraints is equivalent to $\Var_{p_{X}}[f(X)] = 1 = \Var_{p_{Y}}[g(Y)]$. Furthermore, note that $\mbb{E}_{p}[f(X)g(Y)] = \Tr[F \otimes G^{\ast} p_{XY}]$ where $F(i,j) = \begin{cases} f(x) & i = j = x \\ 0 & \text{o.w.} 
\end{cases}$
and $G$ is similarly defined. Given Propositions \ref{prop:mult-and-div-under-trace} and \ref{prop:variance-properties}, we can define quantum maximal correlation coefficients with our non-commutative expectation and variance.
\begin{tcolorbox}[width=\linewidth, sharp corners=all, colback=white!95!black, boxrule=0pt,frame hidden]
\begin{definition}\label{def:f-quantum-max-corr}
    For any normalized operator monotone $f:\mbb{R}_{+} \to \mbb{R}_{\geq 0}$, we define the $f$-quantum maximal correlation coefficient of $\rho_{AB}$ as
    \begin{equation}
    \begin{aligned}
    \label{eq:f-quantum-max-corr}
        \mu_{f}(A:B)_{\rho} \coloneq \max \; & \vert \Tr[X \otimes Y^{\ast}\rho_{AB}] \vert \\
        \text{s.t.} \; & \langle \mbb{1}, X \rangle_{f,\rho_{A}} = \langle \mbb{1}, Y \rangle_{f,\rho_{B}} = 0 \\
        \; &  \langle X, X \rangle_{f,\rho_{A}} = \langle Y, Y \rangle_{f,\rho_{B}} = 1 \ ,
    \end{aligned}
    \end{equation}
    where the maximization is over $X \in \Herm(A)$, $Y \in \Herm(B)$. We may also relax it to supremizing over the linear operators, which we denote by $\mu_{f}^{\text{Lin}}(A:B)_{\rho}$.
\end{definition}
\end{tcolorbox}

For any operator monotone $f$, the constraints linear in $X,Y$ are equivalent to $\Tr[\rho_{A}X] = 0 = \Tr[\rho_{B}Y]$ by Proposition \ref{prop:mult-and-div-under-trace}. The quadratic terms are equivalent to $\Var_{f,\rho_{A}}[X] = 1 = \Var_{f,\rho_{B}}[Y]$ via Item 4 of Proposition \ref{prop:variance-properties}. Thus, considering \eqref{eq:maximal-correlation-coefficient}, we may view these as natural generalizations of the maximal correlation coefficient.
The following shows our family recovers the two previously defined quantum maximal correlation coefficients that were shown to have any operational relevance.
\begin{proposition}\label{prop:previous-q-max-corr-coeffs}
    The quantum maximal correlation coefficient defined in \cite{Beigi-2013a}, which we denote $\mu_{B}(A:B)_{\rho}$ satisfies $\mu_{B}(A:B)_{\rho} = \mu_{f_{1}}^{\text{Lin}}(A:B)_{\rho}$. The quantum maximal correlation coefficient denoted $\mu_{\kappa^{1/2}}$ in \cite{Cao-2019a} is $\mu_{f_{1/2}}^{\text{Lin}}(A:B)_{\rho}$.
\end{proposition}
\begin{proof}
    By a direct calculation using \eqref{eq:J-operator-func-expansion}, $\langle X , X \rangle_{f_{0},\rho} = \Tr[X^{\ast}X\rho]$ and $\langle X, X \rangle_{f_{1},\rho} = \Tr[\rho XX^{\ast}]$. Thus $f_{1}$ directly recovers the definition in \cite{Beigi-2013a}. Similarly, a straightforward calculation will verify the operator defined in \cite{Cao-2019a} $\mho_{\rho}^{\kappa^{1/2}}$ is equal to $ \mbf{J}_{f_{GM},\rho}$, which is sufficient to recover $\mu_{f_{GM}}^{\text{Lin}}$ from \cite[Definition 16]{Cao-2019a}.
\end{proof}
\noindent We remark that the ``quantum maximal correlation coefficients" proposed in \cite{Cao-2019a} other than $\mu_{GM}^{\text{Lin}}$ likely are not recovered. This is because in \cite{Cao-2019a} the variance terms are defined by the inner products that may be expressed as $\langle X, \mbf{J}_{f_{GM},\sigma} \circ \mbf{J}_{f,\sigma}^{-1} \circ \mbf{J}_{f_{GM},\sigma}[X]\rangle$. While it is clear in the case $f = f_{GM}$ is chosen for the middle map this inner product simplifies to $\langle X , \mbf{J}_{f_{GM},\sigma}[X] \rangle$ and Proposition \ref{prop:DPI-for-J-op} may be appealed to, it is unclear what can be appealed to otherwise.

The above discussion justifies the quantum maximal correlation coefficients as generalizations of the classical maximal correlation coefficient from a probabilistic perspective. To further justify the definition, we aim to establish 3 properties: that their value can only decrease under local actions, that they are isometrically invariant, and that the quantities recover the classical maximal correlation coefficient when the joint state is classical. The first property is desirable as it is a normal axiom for any correlation measure, because acting independently on a system ought to only decrease how correlated it is with another system. Moreover, this property is critical for applications in information processing. The second property, the isometric invariance, is needed to justify that the mathematical quantity is \textit{physically} relevant as it shows the measure is independent of the matrix representation of the state $\rho_{AB}$, and the final property formally verifies that quantum maximal correlation coefficients are indeed quantum probabilistic generalizations of the classical maximal correlation coefficient. For simplicity of presentation, the following theorem summarizes that we establish these properties.
\begin{tcolorbox}[width=\linewidth, sharp corners=all, colback=white!95!black, boxrule=0pt,frame hidden,breakable]
\begin{theorem}\label{thm:properties-of-maximal-corr-coeffs}
    Let $f:\mbb{R}_{+} \to \mbb{R}_{\geq 0}$ be a normalized operator monotone function.
    \begin{enumerate}
        \item $\mu_{f}(A:B)_{\rho}$ decreases under the local action of the adjoint of unital Schwarz maps, i.e. for quantum state $\rho_{AB}$ and $\cE_{A \to A'}$ and $\cF_{B \to B'}$ being the adjoints of unital Schwarz maps, 
        \begin{align}
            \mu_{f}(A':B')_{(\cE \otimes \cF)(\rho)} \leq \mu_{f}(A:B)_{\rho} \ . 
        \end{align}
        \item $\mu_{f}(A:B)_{\rho}$ is isometrically invariant, i.e. for all density matrices $\rho_{AB}$ and isometries $V_{A \to A'}$, $W_{B \to B'}$,
        \begin{align}
            \mu_{f}(A:B)_{\rho} = \mu_{f}(A':B')_{V \otimes W\rho_{AB}V^{\ast}\otimes W^{\ast}} \ .
        \end{align}
        \item Let $\rho_{AB} = p_{XY}$, i.e.~the state is classical. Then $\mu_{f}(X:Y)_{p} = \mu_{f}^{\text{Lin}}(X:Y)_{p} = \mu(X:Y)_{p}$ where $\mu(X:Y)_{p}$ is given in \eqref{eq:maximal-correlation-coefficient}.
    \end{enumerate}
    All of the above statements also hold for $\mu_{f}^{\text{Lin}}(A:B)_{\rho}$.
\end{theorem}
\end{tcolorbox}

All three of the aforementioned properties are either statements about data-processing of $\mu_{f}(A:B)_{\rho}$ or a direct application of data-processing. The data processing of the quantum maximal correlation coefficients are controlled by the data processing inequality for $\Var_{f,\sigma}[X]$ established in Items 3 and 4 of Proposition \ref{prop:variance-properties}. To see this and avoid redundancy in subsequent proofs, the following captures the unifying proof idea:

Using the definition of adjoint map, that the adjoint of a $k$-positive map is $k$-positive, and that positive maps are Hermitian preserving, we have
\begin{gather}
    \vert \Tr[X \otimes Y^{\ast}\sigma_{\wt{A}'\wt{B}'}] \vert = \vert \Tr[\cE^{\ast}(X) \otimes (\cF^{\ast}(Y))^{\ast} \rho_{\wt{A}\wt{B}}] \vert \label{eq:DPId-obj-func} \\
    \Tr[X\sigma_{\wt{A}'}] = \Tr[\cE^{\ast}(X)\rho_{\wt{A}}] \quad \Tr[Y\sigma_{\wt{B}'}] = \Tr[\cF^{\ast}(Y)\rho_{\wt{B}}] \label{eq:DPId-expect-const} \ .
\end{gather}
These equations relate the objective function and expectation constraints in Definition \ref{def:f-quantum-max-corr} for linear operators $X,Y$ and $\cE^{\ast}(X),\cF^{\ast}(Y)$ with regards to the processed and unprocessed state respectively. The remaining issue is then always the variance terms. If $\Var_{f,\rho_{A}}[\cE^{\ast}(X)] \leq \Var_{f,\cE(\rho_{A})}[X]$, then one could rescale $\cE^{\ast}(X)$ by a scalar greater than or equal to one to increase the objective function and satisfy the constraints as $\Tr[\cE^{\ast}(X)\rho] = 0$ will be invariant under scaling, and thus one can construct better optimizers for the unprocessed case. As the needed inequality is the data processing for variance (Item 3 of Proposition \ref{prop:variance-properties}), this will generally hold. This argument is the gist of the proof for data processing under local operations and the other proofs are variants upon this idea.

We begin by observing without loss of generality the optimizers are restricted to the support of the marginal states $\rho_{A}$ and $\rho_{B}$, so we can restrict the optimization to the supports.
\begin{proposition}\label{prop:non-neg-func-need-not-restrict-support}
    As $f:\mbb{R}_{+} \to \mbb{R}_{\geq 0}$ is a non-negative valued function, the optimization problem in \eqref{eq:maximal-correlation-coefficient} may be restricted to the support of the marginals.
\end{proposition}
\begin{proof}
    As the feasible set where the operators are not restricted can only be larger, it suffices to show we can construct an optimizer from the unrestricted case. Let $(X,Y)$ be feasible for $\mu_{f}(A:B)_{\rho}$. As $\rho_{AB} = \Pi_{\supp(\rho_{A})} \otimes \Pi_{\supp(\rho_{B})}\rho_{AB}\Pi_{\supp(\rho_{A})} \otimes \Pi_{\supp(\rho_{B})}$, the objective function is 
    $$\Tr[\Pi_{\supp(\rho_{A})}X\Pi_{\supp(\rho_{A})} \otimes (\Pi_{\supp(\rho_{B})}Y\Pi_{\supp(\rho_{A})})^{\ast}\rho_{AB}] \ . $$
    For the same reason, the expectation terms only depend on the operators restricted to the support of $\rho_{A}$ and $\rho_{B}$ respectively. Thus, it suffices to show we can only increase the objective value by restricting to the support. Let $\wt{A} \coloneq \supp(\rho_{A})$ and $X = X_{i,j}\ket{v_{i}}\bra{v_{j}}$ where $\{\ket{v_{i}}\}_{i \in [d_{A}]}$ is an eigenbasis for $\rho_{A}$. Then, by \eqref{eq:MC-form}, 
    \begin{align}
        \langle X , X \rangle_{f,\rho_{A}} = \sum_{i,j \in [d_{A}]} P_{f}(\lambda_{i},\lambda_{j})\vert X_{i,j} \vert^{2} \geq \sum_{i,j \in [d_{\wt{A}}]} P_{f}(\lambda_{i},\lambda_{j}) \vert X_{i,j} \vert^{2} \ , 
    \end{align}
    where the inequality uses that as $f$ is non-negative, the perspective function is non-negative. Thus, restricting to the support of $\rho_{A}$ and then re-scaling the optimizer can only increase the objective value. This completes the proof.
\end{proof}

We now begin by establishing the invariance of the $f$-maximal correlation coefficients under local isometries. In effect this follows because the isometries don't change the relevant structure on the support of the marginals.
\begin{proposition}\label{prop:iso-inv-of-q-max-corr}
    Let $f$ be a normalized monotone function. Let $V_{A \to A'}$, $W_{B \to B'}$ be isometries. Then  
    \begin{align*}
        \mu_{f}(A:B)_{\rho} = \mu_{f}(A':B')_{V \otimes W\rho_{AB}V^{\ast}\otimes W^{\ast}} \ .
    \end{align*}
    The same holds for $\mu_{f}^{\text{Lin}}(A:B)_{\rho}$.
\end{proposition}

\begin{proof}
    We prove the $\mu_{f}$ case, the argument is identical for $\mu_{f}^{\text{Lin}}$. The proof follows from establishing the maximal correlation coefficient evaluated on $\rho_{AB}$ and $\sigma_{A'B'} \coloneq V \otimes W\rho_{AB}V^{\ast}\otimes W^{\ast}$ upper bound each other.
    
    Let $\wt{A} \coloneq \supp(\rho_{A})$. Then the spectral decomposition takes the form $\rho_{A} = \sum_{i \in [d_{\wt{A}}]} \lambda_{i} \dyad{\nu_{i}}$ and so $\sigma_{A'} = \sum_{i \in [d_{\wt{A}}]} \lambda_{i} \dyad{\phi_{i}}$ where $\ket{\phi_{i}} = V\ket{v_{i}}$ for all $i \in [d_{\wt{A}}]$. It follows $\wt{A}' \coloneq \supp(\sigma_{A'}) = \linspan(\{\ket{\phi_{i}}\}_{i \in [d_{\wt{A}}]})$. The same argument can be made for the $B$ space. Let $X \in \Lin(\wt{A}')$ and thus $X = \sum_{i,j \in [d_{\wt{A}}]} X_{i,j} \ket{\phi_{i}}\bra{\phi_{j}}$. It follows $V^{\ast}XV = \sum_{i,j \in [d_{\wt{A}}]} X_{i,j} \ket{v_{i}}\bra{v_{j}}$. Therefore, using Item 4 of Proposition \ref{prop:variance-properties} and \eqref{eq:MC-form},
    \begin{align}
        \Var_{f,\sigma_{A'}}[X] = \langle X , X \rangle_{f,\sigma_{\wt{A}'}} = \sum_{i,j \in [d_{\wt{A}}]} P_{f}(\lambda_{i},\lambda_{j}) \vert X_{i,j} \vert^{2} = \langle V^{\ast}XV , V^{\ast}XV \rangle_{f,\rho_{\wt{A}}} = \Var_{f,\rho_{\wt{A}}}[V^{\ast}XV] \ . 
    \end{align}
    The same argument holds on the $B$ space. As without loss of generality the optimizer is restricted to the support of the marginals (Proposition \ref{prop:suff-conds-for-restricting-support}), for an optimizer $(X_{\wt{A}'},Y_{\wt{B}'})$ of $\mu_{f}(A':B')_{\sigma}$, $(V^{\ast}X_{\wt{A}'}V,W^{\ast}Y_{\wt{B}'}W)$ is feasible for $\mu_{f}(A:B)_{\rho}$ and, by a direct calculation using \eqref{eq:DPId-obj-func}, achieves the same value. As \eqref{eq:f-quantum-max-corr} is a maximization, $\mu_{f}(A':B')_{\sigma} \leq \mu_{f}(A:B)_{\rho}$.

    To prove the inequality in the other direction. Let $\hat{X} \in \Lin(\wt{A})$. Then $\hat{X}_{i,j} = \sum_{i,j} \hat{X}_{i,j} \ket{v_{i}}\bra{v_{j}}$ and $V\hat{X}V^{\ast} = \sum_{i,j} \hat{X}_{i,j} \ket{\phi_{i}}\bra{\phi_{j}}$. It follows
    \begin{align}
        \Var_{f,\sigma_{\wt{A}'}}[VXV^{\ast}] = \langle VXV^{\ast}, VXV^{\ast} \rangle_{f,\sigma_{\wt{A}'}} = \sum_{i,j} P_{f}(\lambda_{i},\lambda_{j}) \vert \hat{X}_{i,j} \vert = \langle \hat{X} , \hat{X} \rangle_{f,\rho_{\wt{A}}} = \Var_{f,\rho_{\wt{A}}}[\widehat{X}] \ .
    \end{align}
    The same argument holds on the $B$ space. Thus, for an optimizer $(\widehat{X}_{\wt{A}},\widehat{Y}_{\wt{B}})$ of $\mu_{f}(A:B)_{\rho}$, $(V\widehat{X}V^{\ast},WYW^{\ast})$ is feasible for $\mu_{f}(A':B')_{\sigma}$ and, by a direct calculation using \eqref{eq:DPId-obj-func}, achieves the same value. As \eqref{eq:f-quantum-max-corr} is a maximization, $\mu_{f}(A':B')_{\sigma} \geq \mu_{f}(A:B)_{\rho}$.
\end{proof}

We now turn to establishing data processing. We begin with a special case, which we subsequently lift using isometric invariance.
\begin{proposition}\label{prop:DPI-for-f-correlation-with-supp-restriction}
    Let $f:\mbb{R}_{+} \to \mbb{R}_{\geq 0}$ be a normalized operator monotone. Let $\sigma_{A'B'} = (\cE \otimes \cF)(\rho_{AB})$ where $\cE_{A \to A'},\cF_{B \to B'}$ are the adjoint maps of unital Schwarz maps and $A \coloneq \supp(\rho_{A}), B \coloneq \supp(\rho_{B})$. Then $\mu_{f}(A':B')_{\sigma} \leq \mu_{f}(A:B)_{\rho}$. The same holds for $\mu_{f}^{\text{Lin}}(A:B)_{\rho}$.
\end{proposition}
\begin{proof}
    Let $X,Y$ optimize $\mu_{f}(\sigma_{A'B'})$. By our assumptions on the marginals of $\rho_{AB}$ being full rank, by Item 3 of Proposition \ref{prop:variance-properties}, $\Var_{f,\rho_{A}}(\cE^{\ast}(X)) \leq \Var_{f,\cE(\rho)_{A'}}(X)$ and similarly for the $B'$ system. Thus one may rescale $\cE^{\ast}(X)$ and $\cF^{\ast}(Y)$ with scalars greater than or equal to one such that the variance terms are returned to one and increase the objective function given \eqref{eq:DPId-obj-func}. As $0 = \Tr[X\cE(\rho_{A})] = \Tr[\cE^{\ast}(X)\rho_{A}]$, the rescaling does not change the expectation constraints, so we have constructed a feasible point for $\mu_{f}(\rho_{AB})$ that results in a value that can only be larger than $\mu_{f}(\sigma_{A'B'})$. This completes the proof of the $\mu_{f}$ case, and the $\mu_{f}^{\text{Lin}}$ proof is identical because Item 3 of Proposition \ref{prop:variance-properties} holds for arbitrary linear operators.
\end{proof}
We now use the local isometric invariance to generalize the above.
\begin{proposition}\label{prop:DPI-for-f-correlation}
    Let $f:\mbb{R}_{+} \to \mbb{R}_{\geq 0}$ be a normalized operator monotone. Let $\rho_{AB} \in \Density(A \otimes B)$ and $\cE_{A \to A'}$ and $\cF_{B \to B'}$ be the adjoints of unital Schwarz maps. Then $\mu_{f}(A':B')_{\cE \otimes \cF(\rho_{AB})} \leq \mu_{f}(A:B)_{\rho}$. The same holds for $\mu_{f}^{\text{Lin}}(A:B)_{\rho}$.
\end{proposition}
\begin{proof}
    We prove the $\mu_{f}$ case. The proof for $\mu^{\text{Lin}}_{f}$ is identical. Let $\wt{A} \coloneq \supp(\rho_{A})$, $\wt{B} \coloneq \supp(\rho_{B})$. Let $V_{\wt{A} \to A}$ be the isometry from $\wt{A}$ to $A$ and similarly for $W_{\wt{B} \to B}$. Then $\rho_{\wt{A}\wt{B}} = (V \otimes W)^{\ast}\rho_{AB}(V \otimes W)$. Defining the channel $\cV(\cdot) \coloneq V \cdot V^{\ast}$ and similarly for $\cW$, we have $(\cE \otimes \cF)(\rho_{AB}) = (\cE \circ \cV \otimes \cF \circ \cW)(\rho_{\wt{A}\wt{B}})$. As the adjoint of $\cE$ is a unital Schwarz map, the adjoint of $\cE \circ \cV$ is a unital Schwarz map (Proposition \ref{prop:restricted-unital-Schwarz}) and similarly for $\cF \circ \cW$. Furthermore $\cE \circ \cV$ acts on the support of $\rho_{\wt{A}}$ and similarly for $\cF \circ \cW$ with $\rho_{\wt{B}}$. Thus, we satisfy the conditions of Proposition \ref{prop:DPI-for-f-correlation-with-supp-restriction}. Thus,
    \begin{align}
        \mu_{f}(A':B')_{(\cE \otimes \cF)(\rho_{AB})} = \mu_{f}(A':B')_{(\cE \circ \cV \otimes \cF \circ \cW)(\rho_{\wt{A}\wt{B}})} \leq \mu_{f}(\wt{A}:\wt{B})_{\rho_{\wt{A}\wt{B}}} = \mu_{f}(\rho_{AB}) \ , 
    \end{align}
    where the inequality is Proposition \ref{prop:DPI-for-f-correlation-with-supp-restriction} and the second equality is Proposition \ref{prop:iso-inv-of-q-max-corr}.
\end{proof}

Next, we establish all the $f$-maximal correlation coefficients recover the classical case. We remark that for $\mu_{f}$ this is an elementary argument. For $\mu_{f}^{\text{Lin}}$, this appears to require using a more advanced result we prove later in the work. This is not necessarily surprising as it implies the classical maximal correlation coefficient cannot change when optimizing over complex-valued functions, which also does not seem to have an elementary proof.\footnote{It is claimed in \cite{Beigi-2013a} to be easy to see $\mu_{f_{0}}^{\text{Lin}}(p_{XY}) = \mu(p_{XY})$, however, it seems one needs to appeal to the Schmidt coefficient characterization of the maximal correlation coefficient in that specific case as well.}
\begin{proposition}\label{prop:recovers-classical}
    Let $f:\mbb{R}_{+} \to \mbb{R}_{\geq 0}$ and $\rho_{AB} = p_{XY}$, i.e.~the state is classical. Then $\mu_{f}(X:Y)_{p} = \mu_{f}^{\text{Lin}}(X:Y)_{p} = \mu(X:Y)_{p}$ where $\mu(X:Y)_{p}$ is given in \eqref{eq:maximal-correlation-coefficient}.
\end{proposition}
\begin{proof}
    We prove the case $\mu_{f}(X:Y)_{p}$ and then extend this to $\mu_{f}^{\text{Lin}}(X:Y)_{p}$. Recalling \eqref{eq:maximal-correlation-coefficient}, by encoding the optimal choices of functions $f$ and $g$ as Hermitian operators as discussed below \eqref{eq:maximal-correlation-coefficient} and noting the quantum maximal correlation coefficient is an optimization, we conclude $\mu(X:Y)_{p} \leq \mu_{f}(X:Y)_{p}$. Thus, we aim to prove the other direction of the inequality. We will do this by showing that the optimizers $X$ and $Y$ can be achieved using classical functions. By an abuse of notation, we let $\cX$, $\cY$ be the alphabets of the support of $p_{X}$ and $p_{Y}$ respectively. Let $(X,Y)$ be the optimizers of $\mu_{f}(X:Y)_{p}$. Let $\Delta_{X}(\cdot) = \sum_{i \in \cX} \dyad{i} \cdot \dyad{i}$ be the completely dephasing channel in the computational basis on the $X$ space and similarly for $\Delta_{Y}$. Then the classical state is invariant under this dephasing, i.e.~$p_{XY} = (\Delta_{X} \otimes \Delta_{Y})(p_{XY})$. Then, as the completely dephasing channel is self-adjoint, by \eqref{eq:DPId-obj-func} and \eqref{eq:DPId-expect-const}, $\Delta_{X}(X)$ and $\Delta_{Y}(Y)$ are all that matter for the objective function and the expectations. Moreover, as $\Var_{f, p_{X}}[\Delta(X)] \leq \Var_{f,p_{X}}[X]$, we may conclude the optimizers are in fact also invariant under dephasing. Thus, $X = \sum_{i \in \cX} X_{i,i} \dyad{i}$ and similarly for $Y$. As $X$ is Hermitian, $X_{i,i}$ is real for all $i \in \cX$. Thus, we may define the real-valued function $f(i) = X_{i,i}$ for all $i \in \cX$. We may similarly define the function $g(j) = Y_{j,j}$ for all $j \in \cY$ to obtain functions $f,g$ such that $\mbb{E}_{p_{XY}}[fg] = \mu_{f}(X:Y)_{p}$. Thus, we conclude $\mu_{f}(X:Y)_{p} \leq \mu(X:Y)_{p}$, which establishes $\mu_{f}(X:Y)_{p} = \mu(X:Y)_{p}$.

    We now extend the above to $\mu_{f}^{\text{Lin}}(X:Y)_{p}$. Note the above construction on linear operators would result in \textit{complex}-valued functions and it is not clear a priori that optimizing over complex-valued rather than real-valued functions would not alter the optimal value. Rather than establish this, we note that as the above argument shows the optimizers commute with $p_{XY}$, the variance term is the same expression for all normalized operator monotone $f$. Moreover, later in the section we determine that the geometric mean maximal correlation coefficient is the same when optimized over linear operators or not (see \eqref{eq:GM-max-corr-doesn't-depend-on-lin} and the surrounding discussion), so we may conclude $\mu_{f}^{\text{Lin}}(X:Y)_{p} = \mu_{f}(X:Y)_{p} = \mu(X:Y)_{p}$ for all normalized operator monotone $f$, which completes the proof.
\end{proof}

\subsubsection{Relating \texorpdfstring{$f$}{}-Maximal Correlation Coefficients}\label{subsec:relation-between-maximal-correlation-coeffs}
We now establish relations between the $f$-maximal correlation coefficients in terms of the choice of function $f$ and determine conditions under which the ordering collapses. As already discussed, as the expectation induced by the $L_{f}^{2}(\sigma)$ space takes the same value for all choices of $f$, what varies for the $f$-maximal correlation coefficients as $f$ is varied is how the variance of the observables are measured. We thus can obtain a partial ordering on the quantum maximal correlation coefficients from the partial ordering on operator monotone functions.
\begin{proposition}\label{prop:relating-f-correlation-coefficients} ~
    \begin{enumerate}
        \item If $f_{1},f_{2}$ are normalized operator monotone functions such that $f_{1} \leq f_{2}$, then $\mu_{f_{1}}(A:B)_{\rho} \geq \mu_{f_{2}}(A:B)_{\rho}$ for all $\rho_{AB}$. The same relation holds for $\mu_{f}^{\text{Lin}}$.
        \item If $\mu_{f_{1}}(A:B)_{\rho}$ contains an optimizer $(X,Y)$ such that $[X,\rho_{A}] = 0 = [Y,\rho_{B}] =0$, then $\mu_{f_{1}}(A:B)_{\rho} = \mu_{f_{2}}(A:B)_{\rho} = \mu_{AM}(A:B)_{\rho}$. The same holds for $\mu^{\text{Lin}}_{f}$.
    \end{enumerate}
\end{proposition}
\begin{proof}
    We begin with Item 1. Let $X,Y$ optimize $\mu_{f_{2}}(\rho_{AB})$. Then $1 = \langle X , X \rangle_{f_{2},\rho_{A}} \geq \langle X , X \rangle_{f_{1},\rho_{A}} \geq 0$ by Proposition \ref{prop:ordering-of-NC-mult}. It follows one may scale $X$ by $1/\sqrt{\langle X , X \rangle_{f_{1},\rho_{A}}} \geq 1$ and similarly for $Y$ to obtain a feasible point for $\mu_{f_{1}}$. This constructs a feasible point for $\mu_{f_{1}}$ that obtains at least as large of a value as $\mu_{f_{2}}$. The same argument holds for optimizing over linear operators. Item 2 follows from the fact that when $[X,\sigma] = 0$, $\langle X , X \rangle_{f,\sigma} = \Tr[X^{\ast}X\sigma]$, noting in this case $[X^{\ast},\rho]=0$ as well, and using the definition of $f_{AM}$. 
\end{proof}

Combining Item 1 of the above proposition with \eqref{eq:ordering-of-means} and \eqref{eq:ordering-of-standard-monotones}, we obtain
\begin{align}\label{eq:ordering-of-maximal-corr-means}
    \mu_{AM} \leq \mu_{LM} \leq \mu_{GM} \leq \mu_{HM} \quad \text{and} \quad \mu_{AM} \leq \mu_{f} \leq \mu_{HM} \quad \forall f \in \cM_{St} \ . 
\end{align}
Item 2 of Proposition \ref{prop:relating-f-correlation-coefficients} tells us that, when optimizers commute for some choice of $f$, the hierarchy collapses down to the `minimal' maximal correlation coefficient, $\mu_{AM}$.  
Moreover, we can in fact prove the sufficient conditions for quantum maximal correlation coefficients being equal can in fact also be necessary conditions.
\begin{proposition}\label{prop:separating-mu-AM-and-mu-GM}
$\mu_{AM}(A:B)_{\rho} \leq \mu_{GM}(A:B)_{\rho}$ holds with equality if and only if there exist optimizers $(X,Y)$ for $\mu_{GM}$ that commute with $\rho_{A},\rho_{B}$ respectively.
\end{proposition}
\begin{proof}
     As the `if' condition is established by Item 2 in Proposition \ref{prop:relating-f-correlation-coefficients}, we only need to establish the other direction. The argument makes use of Definition \ref{def:f-quantum-max-corr} and $\mbf{J}_{GM,\sigma}$ versus $\mbf{J}_{AM,\sigma}$. The basic idea will be to show in this case we can argue directly from the Cauchy-Schwarz inequality and when it is saturated. First, using that $X \in \Herm$,
    \begin{align}
        \Vert X \Vert_{GM,\rho_{A}} = \Tr[X^{\ast}\rho^{1/2}X\rho^{1/2}] 
        =&\langle \rho^{1/2}X, X\rho^{1/2} \rangle \\ 
        \leq& \Vert \rho^{1/2}X \Vert_{2} \Vert X \rho^{1/2} \Vert_{2} \label{eq:norm-conversion-CS} \\
        =& \Vert X \rho^{1/2} \Vert_{2}^{2}
        = \Tr[\rho^{1/2}X^{\ast}X\rho^{1/2}] 
        = \Vert X \Vert_{AM,\rho_{A}} \ ,
    \end{align}
    where the inequality is Cauchy-Schwarz (CS) and the fourth equality uses that $\Vert \rho^{1/2}X \Vert_{2} = \Vert (\rho^{1/2}X)^{\ast} \Vert_{2}$ and that $X$ and $\rho$ are both Hermitian. The same idea holds for $\Vert Y \Vert_{GM,\rho_{B}}$. Now, given $\rho_{AB}$, define $\cF^{GM}_{\text{opt}}$ (resp.~$\cF^{AM}_{\text{opt}}$) as the set of optimizers for $\mu_{GM}(A:B)_{\rho}$ (resp.~$\mu_{AM}(A:B)_{\rho}$). Assume there does not exist $(\wt{X},\wt{Y}) \in \cF^{AM}_{\text{opt}}$ such that $[\wt{X},\rho_{A}] = 0 = [\wt{Y},\rho_{B}]$. Then \eqref{eq:norm-conversion-CS} is always strict as the Cauchy-Schwarz inequality is an equality if and only if both arguments arguments are the same up to a scalar. It follows there exists $\kappa_{X},\kappa_{Y} > 1$ such that $\Vert \kappa_{X} X \Vert_{\rho} = 1$, $\Vert \kappa_{Y} Y \Vert_{\rho} =1$, i.e. are feasible for $\mu_{GM}(A:B)_{\rho}$ and $\vert \Tr[\kappa_{X}X \otimes \kappa_{Y}Y^{\ast} \rho_{AB}] \vert > \mu_{AM}(A:B)_{\rho}$. As $\mu_{GM}(A:B)_{\rho}$ is a supremization, this proves the inequality is strict.
\end{proof}

Before moving forward, we note that quantum systems can be entangled and still have their correlation coefficients collapse. This is because a global symmetry may guarantee the optimizers commute with the local states and thus satisfy Proposition \ref{prop:relating-f-correlation-coefficients}. Our example is a generalization of an example in \cite{Beigi-2013a}.
\begin{proposition}
    Consider the $d$-dimensional isotropic state $\rho_{d,\lambda} \coloneq \lambda \widehat{\Phi}^{+} + (1-\lambda) \pi$ where $\widehat{\Phi}^{+}$ is the maximally entangled state. For all operator monotone $f$, 
    $$\mu_{f}(A:B)_{\rho_{d,\lambda}} = \mu_{f}^{\text{Lin}}(A:B)_{\rho_{d,\lambda}} = \mu_{AM}(A:B)_{\rho_{d,\lambda}} = \lambda \ . $$
\end{proposition}
\begin{proof}
    Note that the marginals are the maximally mixed state, so they commute with all linear operators. We therefore have
    $$\mu_{f}(\rho_{d,\lambda}) = \mu_{AM}(\rho_{d,\lambda}) = \max\{\Tr[\rho_{\lambda} X \otimes Y^{\ast}] : \; \Tr[X] = \Tr[Y] = 0 \, , \, \Vert X \Vert_{2} = \sqrt{d} = \Vert Y \Vert_{2} \} \ , $$
    where the first equality is Proposition \ref{prop:relating-f-correlation-coefficients} and the second is by direct calculation. Then
    \begin{align}
        \Tr[\rho_{\lambda} X \otimes Y^{\ast}] =& \lambda \Tr[X \otimes Y^{\ast}\widehat{\Phi}^{+}] + \frac{1-\lambda}{d}\Tr[X \otimes Y^{\ast}] \\ 
        =& \lambda/d \cdot \Tr[X^{\Trans}Y^{\ast}] \\
        \leq& \lambda/d \Vert X^{\Trans} \Vert_{2} \Vert Y \Vert_{2} \\ 
        =& \lambda \ ,
    \end{align}
    where we used the tracelessness of the operator, the transpose trick, and Cauchy-Schwarz. Now let $X = Y = Z$ where $Z = \sum_{i \in \{1,...,d\}} \lambda_{i} \dyad{i}$ where $\{\ket{i}\}_{i}$ is the basis the transpose is defined with respect to and 
    $$\lambda_{i} = \begin{cases} \sqrt{1 + \mbb{1}\{d \text{ is odd}\}\frac{1}{2\lfloor d/2 \rfloor}} & i \in \{1,...,\lfloor d/2 \rfloor\} \\
    -\sqrt{1 + \mbb{1}\{d \text{ is odd}\}\frac{1}{2\lfloor d/2 \rfloor}} & i \in \{d,d-1,...,d-\lfloor d/2 \rfloor + 1\} \\
    0 & \text{otherwise}
    \end{cases} \ . $$
    It follows for all dimensions, $Z$ is Hermitian, traceless, and $\Vert Z \Vert_{2} = \sqrt{\Tr(Z^{2})} = \sqrt{d}$ and since $X = Y$ the Cauchy-Schwarz inequality is saturated. Thus, this is achievable. The same argument holds for $\mu_{f}^{\text{Lin}}$.
\end{proof}

Lastly, we remark that there exist equivalence classes of choice of function $f$ for given $\mu_{f}(\rho_{AB})$. In particular, $\mu_{f_{k}}(\rho) = \mu_{f_{1-k}}(\rho) = \mu_{f_{k,sym}}$ as follows from $\Tr[X^{\ast}\sigma^{k}X\sigma^{1-k}] = \Tr[X\sigma^{k}X^{\ast}\sigma^{1-k}]=\Tr[X^{\ast}\sigma^{1-k}X\sigma^{k}]$ where we used that $X \in \Herm$ and the cyclicity of trace. In the linear operator relaxed version, $\mu_{f_{k},\sigma}^{\text{Lin}}(\rho) = \mu_{f_{1-k},\sigma}^{\text{Lin}}(\rho)$ because $\vert \Tr[\rho_{AB} X \otimes Y^{\ast}]\vert = \vert \ol{\Tr[\rho X^{\ast} \otimes Y]} \vert$, but it is unclear (and seemingly unlikely) if the symmetrized equivalence also holds. 

\subsection{Extreme Values and Classical Correlation}\label{subsec:extreme-values-and-ACD}
The classical maximal correlation coefficient is bounded between zero and one, takes the value zero if and only if the joint distribution is independent, and takes the value one if and only if perfect classical correlation can be extracted from it using local operations. Moreover, this final property is known to be equivalent to $p_{XY}$ being `decomposable' as we define subsequently. In this section, we extend all these classical results to the quantum setting for the appropriate range of operator monotone functions. We summarize the results of this subsection in the following theorem.
\begin{tcolorbox}[width=\linewidth, sharp corners=all, colback=white!95!black, boxrule=0pt,frame hidden, breakable=true]
\begin{theorem}\label{thm:extreme-values-summary} ~
    \begin{enumerate}
        \item \textbf{Independence-Detection:} For normalized operator monotone function $f$, $\mu_{f}(A:B)_{\rho} \geq 0$ with equality if and only if $\rho_{AB} = \rho_{A} \otimes \rho_{B}$. The same holds for $\mu_{f}^{\text{Lin}}(A:B)_{\rho}$.
        \item \textbf{Normalized:} For $f \geq f_{GM}$, $\mu_{f}(A:B)_{\rho} \leq \mu_{f}^{\text{Lin}}(A:B)_{\rho} \leq 1$.
        \item \textbf{Equivalences of Exact Correlation Extraction:} The following are equivalent:
        \begin{enumerate}[label=\roman*)]
            \item $\mu_{AM}(A:B)_{\rho}=1$,  
            \item there exist local two-outcome measurements, $\cM_{A \to X}$ and $\cN_{B \to X'}$ such that 
            $$(\cM \otimes \cN)(\rho_{AB}) = \chi^{\vert p}_{XX'} \text{ for } 0 < p < 1 \, , $$
            \item the previous item holds with projective two-outcome measurements,
            \item there exist decompositions of the spaces $A = A_{0} \oplus A_{1}$ and $B = B_{0} \oplus B_{1}$ such that
            $$\rho_{AB} = p \rho^{0}_{A_{0}B_{0}} + (1-p)\rho^{1}_{A_{1}B_{1}} + X + X^{\ast} \ , $$
            where $\rho^{i} \in \Density(A_{i} \otimes B_{i})$ for $i \in \{0,1\}$, $X \in \Lin(A_{0} \otimes B_{0}, A_{1} \otimes B_{1})$, and $p \in (0,1)$.
        \end{enumerate}
    \end{enumerate}
\end{theorem}
\end{tcolorbox}
We remark Items 1 and 2 were established for the case $f_{0}$ by Beigi \cite{Beigi-2013a}, but his proofs do not directly generalize. For Item 1, we provide an operational characterization of independence in the quantum setting (Lemma \ref{lem:indep-if-and-only-if-measurements}) and use it to lift the classical result. For Item 2, we establish a non-commutative extension of a probability-theoretic proof method.\footnote{In the subsequent section, we extend Beigi's proof method for proving the maximal correlation coefficient is bounded above by one. We can use that extended result to establish Item 2. However, that proof method is so operator-theoretic that it makes it unclear why $f_{GM}$ is central, which is why we choose our probabilistic proof method.} This will allow us to identify why $f_{GM}$ is so central to this result (Remark \ref{rem:importance-of-f-GM-for-probabilistic-proof}). The equivalence of Items 3)i) and 3)ii) rather directly follows from \cite{Beigi-2013a} and our identification of quantum maximal correlation coefficients. It is Items 3)iii) and 3)iv) that are our main technical contribution. It is in particular Item 3)iv) that allows us to identify a quantum generalization of the notion of a joint distribution being decomposable.

Note that Item 3)iv) of Theorem \ref{thm:extreme-values-summary} truly is a generalization of the classical case. Indeed, a simple example is the maximally entangled state of two qubits where $A_{0} = \linspan\{\ket{0}\} = B_{0}$, $A_{1} = \linspan\{\ket{1}\} = B_{1}$, $\rho^{i}_{A_{i}B_{i}} = \dyad{i}^{\otimes 2}$, $p = 1/2$, and $X = \frac{1}{2}\ket{11}\bra{00}$. More generally, for any $\lambda \in [0,1]$, $\lambda \widehat{\Phi}^{+} + (1-\lambda) Z_{A}\widehat{\Phi}^{+}Z_{A}$ satisfies this decomposition, which shows the equivalence can also hold for entangled \textit{mixed} states.

In the rest of the subsection, we develop the tools and provide the proofs to establish Theorem \ref{thm:extreme-values-summary}.

\subsubsection{Equivalence between Maximal Correlation Coefficient Being Zero and Independence}
As already mentioned, the following is known.
\begin{proposition}\label{prop:classical-max-corr-zero-means-independent}
    $\mu(X:Y)_{p}$ if and only if $p_{XY} = p_{X} \otimes p_{Y}$ . 
\end{proposition}
\noindent This property is special to the maximal correlation coefficient because the Pearson correlation coefficient (by which the classical maximal correlation coefficient may be induced) only detects linear correlation and thus can be zero for dependent random variables. To extend this to the quantum setting, we will reduce independence of quantum systems to independence of all distributions one could generate from measuring the quantum systems locally. In other words, we will establish and use the following lemma, which says a quantum state is correlated if and only if there exist two-outcome local measurements such that the measurement outcomes are correlated. This is intuitive but is established using the Holevo-Helstrom theorem and linear algebra (see Appendix \ref{app-subsec:max-corr-coeff-lemmata} for the formal proof). 
\begin{lemma}\label{lem:indep-if-and-only-if-measurements}
    Let $\rho_{AB}$ be finite-dimensional. There does not exist local two-outcome measurements $\cM_{A \to X}$, $\cN_{B \to Y}$ such that $\cM_{A \to X} \otimes \cN_{B \to Y}(\rho_{AB}) \neq \cM_{A \to X}(\rho_{A}) \otimes \cN_{B \to Y}(\rho_{B})$ if and only if $\rho_{AB} = \rho_{A} \otimes \rho_{B}$.
\end{lemma}

Using this lemma and the data processing inequality for quantum maximal correlation coefficients, we can lift the classical maximal correlation's property to all quantum maximal correlation coefficents.
\begin{proposition}[Item 1 of Theorem \ref{thm:extreme-values-summary}] For all normalized operator monotone $f$, $\mu_{f}(A:B)_{\rho} \geq 0$ with equality if and only if $\rho_{AB} = \rho_{A} \otimes \rho_{B}$. The same holds for $\mu_{f}^{\text{Lin}}(A:B)_{\rho}$.
\end{proposition}
\begin{proof}
    The lower bound of $0$ follows from Definition \ref{def:f-quantum-max-corr} having an absolute value. To establish the saturation condition, assume $\mu_{f}(A:B)_{\rho} = 0$. Then by DPI (Proposition \ref{prop:DPI-for-f-correlation}), for all two-outcome measurements $\cM_{A \to X},\cN_{B \to Y}$, it holds $\mu(X:Y)_{(\cM \otimes \cN)(\rho_{AB})} = 0$. As the classical maximal correlation coefficient takes the value zero only if the joint distribution is independent (Proposition \ref{prop:classical-max-corr-zero-means-independent}), we may conclude that for all two-outcome measurements, the resulting state $p_{XY} = (\cM \otimes \cN)(\rho_{AB})$ satisfies $p_{XY} = p_{X} \otimes p_{Y}$. By Lemma \ref{lem:indep-if-and-only-if-measurements}, this can only be the case if $\rho_{AB} = \rho_{A} \otimes \rho_{B}$. For the other direction, note if $\rho_{AB} = \rho_{A} \otimes \rho_{B}$, then for any feasible $(X,Y)$ of $\mu_{f}(A:B)_{\rho}$, we have
     \begin{align}
         \Tr[X \otimes Y^{\ast}\rho_{AB}] = \Tr[X\rho_{A}]\Tr[Y^{\ast}\rho_{B}] = 0 
     \end{align}
    as $\Tr[X\rho_{A}] = 0 = \Tr[Y\rho_{B}]$, so $\mu_{f}(A:B)_{\rho}= 0$. The same argument holds for $\mu^{\text{Lin}}_{f}$.
\end{proof}

\subsubsection{The Quantum Maximal Correlation Coefficients that are Bounded Above by One} \label{subsubsec:q-max-corr-bounded-above-by-one}
Next, we establish Item 2 of Theorem \ref{thm:extreme-values-summary}. The classical maximal correlation coefficient being bound above by one follows from the probabilistic form of the Cauchy-Schwarz inequality, $\vert \text{Cov}(X,Y) \vert \leq \sqrt{\text{Var}(X)\text{Var}(Y)}$ for random variables $X$ and $Y$ on the same probability space, and then generalizing to when $X$ and $Y$ are on different probability spaces by using the adjoint of the classical channel. In other words, in the classical setting, there is a direct probability-theoretic argument for the maximal correlation coefficient being bounded above one. Here we extend this probability-theoretic argument to the quantum regime.

We begin with a few preliminaries. We first need some identities related to the canonical purification of a state $\sigma_{A}$:
\begin{align}
    \psi_{AA'}^{\sigma} \coloneq \sigma^{1/2}_{A}\Phi^{+}_{AA'}\sigma^{1/2}_{A} \ . 
\end{align}
First, we show that for a quantum state $\rho_{AB}$, there is a unique channel $\cE_{A \to B}$ that takes the canonical purification of $\rho_{A}$ to $\rho_{AB}$. This is the quantum generalization of the fact that given $p_{XY}$, there always exists $W_{X \to Y}$ such that $(\id_{X} \otimes W)(\chi^{\vert p}_{XX}) = p_{XY}$.
\begin{proposition}\label{prop:every-joint-state-is-a-degraded-purif}
    $\rho_{AB} \in \Density(AB)$ if and only if there exists a quantum channel $\cE_{A' \to B}$ such that $(\text{id} \otimes \cE)(\psi^{\rho}_{AA'}) = \rho_{AB}$ where $\psi^{\rho}$ is the canonical purification of $\rho_{A}$. Moreover this channel is unique.
\end{proposition}
\begin{proof}
    ($\leftarrow$) If $\cE$ is as promised, then $(\text{id}_{A} \otimes \cE)(\psi^{\rho}) \in \Density(AB)$ by the CPTP property of the channel. \\
    ($\rightarrow$) Let $\rho_{AB} \in \Density(AB)$. Let $\phi^{\rho} \in D(ABE)$ be a purification of $\rho_{AB}$. It follows $\psi^{\rho}$ and $\phi^{\rho}$ are both purifications of $\rho_{A}$. Thus, there exists an isometry $V_{A' \to BE}$ such that $\ket{\phi^{\rho}} = V\ket{\psi^{\rho}}$. Then, $\rho_{AB} = \Tr_{E}[\phi^{\rho}] = \Tr_{E}[V\psi^{\rho}V^{\ast}]$. Defining $\cV(\cdot) \coloneq V \cdot V^{\ast}$ and then defining $\cE_{A \to B} \coloneq \Tr_{E} \circ \cV$, we have $\rho_{AB} = (\id_{A} \otimes \cE_{A \to B})(\psi^{\rho})$ as claimed. This completes the equivalence. To see that it is unique, note 
    $$\rho_{AB} = (\id_{A} \otimes \cE_{A' \to B})(\psi^{\rho}) = \sqrt{\rho}_{A}(\id_{A} \otimes \cE)(\dyad{\Phi})\sqrt{\rho}_{A} = \rho_{A}^{1/2}\Omega_{\cE}\rho_{A}^{1/2} \ . $$
    By the Choi isomorphism, $\Omega_{\cE}$ is unique to $\cE$.
\end{proof}
\begin{remark}
    Note that uniqueness is because we fixed working with the canonical purification, otherwise there would be an isometric degree of freedom.
\end{remark}
\noindent Furthermore, we may relate the canonical purification to the $\mbf{J}_{f_{GM},\rho}$ operator. Specifically, if we define the transpose via the eigenbasis of $\rho_{A}$, we see by direct calculation that the Choi operator of $\mbf{J}_{f_{GM},\rho}$ is the canonical purification of $\rho_{A}$:
\begin{align}\label{eq:choi-of-GM-J-operator}
    \Omega_{\mbf{J}_{f_{GM},\rho}} = (\id_{A} \otimes \mbf{J}_{f_{GM},\rho})(\Phi^{+}) = \rho_{A'}^{1/2}\Phi^{+}\rho_{A'}^{1/2} = \rho_{A}^{1/2}\Phi^{+}\rho_{A}^{1/2} \ .
\end{align}
Note that we used the transpose trick and that by assumption $\rho_{A} = \rho_{A}^{T}$. The advantage of these identifications is that for any $\rho_{AB}$, there exists a channel $\cE_{A \to B}$ such that 
\begin{align}\label{eq:identification-of-rhoAB-via-GM}
    \rho_{AB} = \Omega_{\cE \circ \mbf{J}_{f_{GM},\rho_{A}}} \ . 
\end{align}

Our remaining preliminaries are with regards to the `Petz recovery map' originally introduced in \cite{Petz-1986a}, and is a quantum generalization of the adjoint of a classical channel with respect to the $L^{2}(p)$-space.
\begin{definition}\label{def:Petz-recovery-map}
    Given a quantum channel $\cE_{A \to B}$ and a quantum state $\rho \in \Density(A)$, the Petz recovery map is
    \begin{align}\label{eq:Petz-recovery-defn}
        \cP_{\cE,\rho}(X) \coloneq \rho^{1/2}\cE^{\ast}\left( [\cE(\rho)]^{-1/2} X [\cE(\rho)]^{-1/2} \right)\rho^{1/2} \ .
    \end{align}
\end{definition}
\noindent Note that the Petz recovery map may be expressed as $\cP_{\cE,\rho} = \mbf{J}_{f_{GM},\rho} \circ \cE^{\ast} \circ \mbf{J}_{f_{GM},\cE(\rho)}^{-1}$, and it is through this identification it is relevant to the the $L^{2}_{f_{GM}}(\rho)$ measure space. What is important for the following is that $\cP_{\cE,\rho}^{\ast} = \mbf{J}_{f_{GM},\cE(\rho)}^{-1} \circ \cE \circ \mbf{J}_{f_{GM},\rho}$, which note, for a specific choice of $\cE$, is the map considered in the Choi operator in \eqref{eq:identification-of-rhoAB-via-GM} composed with $\mbf{J}_{f_{GM},\cE(\rho)}^{-1}$. Moreover, it satisfies being a contraction in the following sense.
\begin{proposition}\label{prop:adjoint-of-Petz-is-contraction}
    Let $\cE_{A \to B}$ be a quantum channel, $\rho_{A}$ be a quantum state, and $X \in \Lin(A)$. Then $\Vert \cP^{\ast}_{\cE,\rho}(X) \Vert_{f_{GM},\cN(\rho)} \leq \Vert X \Vert_{f_{GM},\rho}$.
\end{proposition}
\begin{proof}
    Using the definition of the inner product, $\mbf{J}^{-1}_{f_{GM},\rho}$ is self-adjoint with respect to the Hilbert-Schmidt inner product, and the data processing inequality for the $\mbf{J}_{f,\rho}^{-1}$ operator,
    \begin{align}
        \Vert \cP^{\ast}_{\cE,\rho}(X) \Vert_{f_{GM},\cN(\rho)} &= \sqrt{\langle \cP^{\ast}_{\cE,\rho}(X), \cP^{\ast}_{\cE,\rho}(X) \rangle_{f,\cN(\rho)}} \\
        &= \sqrt{\langle  \mbf{J}_{f_{GM},\rho}(X) , \cE^{\ast} \circ \mbf{J}_{f_{GM},\cE(\rho)}^{-1} \circ\cE \circ [\mbf{J}_{f_{GM},\rho}(X)] \rangle } \\
        &\leq \sqrt{\langle  \mbf{J}_{f_{GM},\rho}(X) , \mbf{J}_{f_{GM},\rho}^{-1} \circ [\mbf{J}_{f_{GM},\rho}(X)] \rangle } \\
        &= \Vert X \Vert_{f_{GM},\rho} \ .
    \end{align} 
\end{proof}

We now bound a large class of quantum maximal correlation coefficients.
\begin{proposition}[Item 2 of Theorem \ref{thm:extreme-values-summary}]
    Let $f$ be a normalized operator monotone function satisfying $f \geq f_{GM}$. Then 
    $$\mu_{f}(A:B)_{\rho} \leq \mu_{f}^{\text{Lin}}(A:B)_{\rho} \leq 1 \ . $$
\end{proposition}
\begin{proof}
    By the ordering on the maximal correlation coeffficients given in Proposition \ref{prop:relating-f-correlation-coefficients}, it suffices to prove the bound for $f$ being the geometric mean. By definition of the maximal correlation coefficient, we assume $A$ and $B$ are such that $\rho_{A}$ and $\rho_{B}$ are full rank. Let the transpose on $A$ be defined in the eigenbasis of $\rho_{A}$. Let $\cE$ be the channel such that $\rho_{AB} = \Omega_{\cE \circ \mbf{J}_{f_{GM}}}$ as exists by Lemma \ref{prop:every-joint-state-is-a-degraded-purif}. Then we have the equalities
    \begin{align}
        \Tr[X \otimes Y^{\ast}\rho_{AB}]
        &=  \Tr[X \otimes Y^{\ast}(\id_{A} \otimes \cE \circ \mbf{J}_{f_{GM},\rho})(\Phi^{+})] \label{eq:identify-rho-AB-with-Jfgm} \\
        &= \Tr[X \otimes Y^{\ast}(\id_{A} \otimes \mbf{J}_{f_{GM},\cE(\rho)} \circ \mbf{J}_{f_{GM},\cE(\rho)}^{-1} \circ \cE \circ \mbf{J}_{f_{GM},\rho})(\Phi^{+})] \\
        &= \Tr[X \otimes Y^{\ast}(\id_{A} \otimes \mbf{J}_{f_{GM},\cE(\rho)} \circ \cP_{\cE,\rho}^{\ast})(\Phi^{+})] \\
        &= \Tr[X \otimes Y^{\ast}\Omega_{\mbf{J}_{f_{GM},\cE(\rho)} \circ \cP_{\cE,\rho}^{\ast}}] \\ 
        &= \langle Y, \mbf{J}_{f_{GM},\cE(\rho)} \circ \cP_{\cE,\rho}^{\ast}(X^{T}) \rangle \\
        &= \langle Y, \cP_{\cE,\rho}^{\ast}(X^{T}) \rangle_{f_{GM},\cE(\rho)} \ , 
    \end{align}
    where the third equality is the definition of the Petz recovery map, the fourth is the definition of the Choi operator, and the fifth is \eqref{eq:action-of-Choi}.

    Now, using the Cauchy-Schwarz inequality and Proposition \ref{prop:adjoint-of-Petz-is-contraction},
    \begin{align}
         \vert \Tr[X \otimes Y^{\ast}\rho_{AB}] \vert =  \vert \langle Y, \cP_{\cE,\rho}^{\ast}(X^{T}) \rangle_{f_{GM},\cE(\rho)} \vert 
         &\leq \Vert Y \Vert_{f_{GM},\cE(\rho)} \Vert \cP_{\cE,\rho}^{\ast}(X^{T}) \Vert_{f_{GM},\cE(\rho)} \\ 
         &\leq \Vert Y \Vert_{f_{GM},\cE(\rho)} \Vert X^{T} \Vert_{f_{GM},\rho} \ . 
    \end{align}
    Finally, note that $\cE(\rho_{A}) = \rho_{B}$ and using the way the transpose has been chosen a direct calculation verifies $\Vert X^{T} \Vert_{f_{GM},\rho} = \Vert X \Vert_{f_{GM},\rho}$. Thus, we have for any feasible $X$ and $Y$,  $\vert \Tr[X \otimes Y^{\ast}\rho_{AB}] \vert \leq \Vert Y \Vert_{f_{GM},\cE(\rho)} \Vert X \Vert_{f_{GM},\rho} \leq 1$. This completes the proof.
\end{proof}

\begin{remark}\label{rem:importance-of-f-GM-for-probabilistic-proof}
While we have not yet introduced the technical details to make this clear, the reason $f_{GM}$ is central in the above argument is that we may always find a quantum channel $\cE$ such that $\rho_{AB} = \Omega_{\cE \circ \mbf{J}_{f_{GM},\rho_{A}}}$ so that we can switch to the $f_{GM}$ inner product. For other choices of operator monotone function $f$, this correspondence does not handle the full space of quantum states. Further details are provided in Remarks \ref{remark:bounding-contraction-coeff} and \ref{rem:bounding-other-max-corr-coeffs} once further mathematical details are provided.
\end{remark}

\subsubsection{Quantum Decomposability and Correlation Distillation} Ahlswede and G\'{a}cs defined the notion of a distribution being `decomposable' in the following manner. \begin{definition}\label{def:decomposable-distribution}
    \cite{Ahlswede-1976a} A joint distribution $p_{XY}$ is `decomposable' if there exists $A \subset \cX$, $B \subset \cY$ such that $0 < \Pr[x \in A], \Pr[y \in B] < 1$ and $x \in A$ if and only if $y \in B$. Otherwise, it is indecomposable.  
\end{definition}
\noindent This in effect says there are local events $A \subset \cX$ and $B \subset \cY$ such that neither happens with certainty, but either both $A$ and $B$ occur or neither $A$ nor $B$ occur. Given this, perhaps intuitively, Witsenhausen \cite{Witsenhausen-1975a} showed the following are equivalent:
\begin{enumerate}[itemsep=0pt]
    \item $\mu(X:Y)_{p} = 1$,
    \item there exist deterministic functions $f:\cX \to \{0,1\}$, $g:\cY \to \{0,1\}$ and parameter $p \in (0,1)$ such that $(f \otimes g)(p_{XY}) = p\dyad{0}^{\otimes 2}+(1-t)\dyad{1}^{\otimes 2}$, and 
    \item the joint distribution is decomposable.
\end{enumerate}
Motivated by this, we call Item 2 the `Witsenhausen property' of $\mu(p_{XY})$. 

Our goal is to generalize this equivalence of Witsenhausen's. This will make use of two lemmas. The first is the following technical result proven by Beigi.
\begin{lemma}\label{lem:Beigi-thm-5} \cite[Part of Theorem 5(b)]{Beigi-2013a}
    $\mu_{AM}(A:B)_{\rho} = 1$ if and only if there exist local measurements $\{M,\mbb{1}-M\}$, $\{N,\mbb{1}-N\}$ such that $\tr(\rho_{AB}M \otimes N) \in (0,1)$ and 
    $$\tr(\rho_{AB}(M \otimes (\mbb{1}-N)) = \tr(\rho_{AB}((\mbb{1}-M) \otimes N)) = 0 \ . $$
\end{lemma}
\begin{proof}
    The conditions in terms of a POVM are stated in \cite[Theorem 5]{Beigi-2013a} when $\mu_{f_{1}}(\rho_{AB})$ is achieved with a Hermitian operator. By an equivalence observed in Subsection \ref{subsec:relation-between-maximal-correlation-coeffs}, $\mu_{AM}(\rho_{AB}) = \mu_{f_{1}}(\rho_{AB}) \leq \mu_{f_{1}}^{\text{Lin}}(\rho_{AB})$ where the inequality becomes an equality when $\mu^{\text{Lin}}_{f_{1}}$ is achieved by a Hermitian operator. As Beigi showed $\mu^{\text{Lin}}_{f_{1}}(\rho_{AB}) \leq 1$, this is equivalent to stating the exact conditions for when $\mu_{AM}(\rho_{AB}) = 1$. This completes the proof.
\end{proof}

The second is a reduction from the existence of perfect correlation under two-outcome local measurements to the existence of projective two-outcome measurements, which we prove in Appendix \ref{app-subsec:max-corr-coeff-lemmata}.
\begin{proposition}\label{prop:POVM-to-PVM}
    If there exist two-outcome measurements with POVM elements $\{M,\mbb{1}_{A}-M\}$, $\{N,\mbb{1}_{B}-N\}$ such that 
    $$\Tr[\rho M \otimes (\mbb{1}_{B} -N)] = 0 = \Tr[\rho (\mbb{1}_{A} - M) \otimes N] \ , $$
    and $0 < \Tr[\rho M \otimes N] < 1$, then there exist projective measurements that do the same.
\end{proposition}

We now prove the quantum generalization of the equivalences.
\begin{proposition}[Item 3 of Theorem \ref{thm:extreme-values-summary}]\label{prop:mu-AM-equivalences}
    Let $\rho_{AB}$ be a quantum state. The following are equivalent
    \begin{enumerate}
        \item $\mu_{AM}(A:B)_{\rho}=1$,  
        \item there exist local two-outcome measurements, $\cM_{A \to X}$ and $\cN_{B \to X'}$ such that $(\cM \otimes \cN)(\rho_{AB}) = \chi^{\vert p}_{XX'}$ for $0 < p < 1$,
        \item the previous item holds with projective two-outcome measurements,
        \item there exist decompositions of the spaces $A = A_{0} \oplus A_{1}$ and $B = B_{0} \oplus B_{1}$ such that
        $$\rho_{AB} = p \rho^{0}_{A_{0}B_{0}} + (1-p)\rho^{1}_{A_{1}B_{1}} + X + X^{\ast} \ , $$
        where $\rho^{i} \in \Density(A_{i} \otimes B_{i})$ for $i \in \{0,1\}$, $X \in \Lin(A_{0} \otimes B_{0}, A_{1} \otimes B_{1})$, and $p \in (0,1)$.
    \end{enumerate}
\end{proposition}
\begin{proof}
    $(1 \iff 2)$ Let $\mu_{AM}(\rho_{AB}) = 1$. By Lemma \ref{lem:Beigi-thm-5}, 
    \begin{align}
        (\cM \otimes \cN)(\rho_{AB}) =& \Tr[\rho M \otimes N]\dyad{0} \otimes \dyad{0} + \Tr[((\mbb{1} - M) \otimes (\mbb{1}-N)\rho_{AB})]\dyad{1} \otimes \dyad{1} \\
        =& p\dyad{0}^{\otimes 2} + (1-p)\dyad{1}^{\otimes 2} \\
        =& \chi^{\vert p}_{XX'} \ ,
    \end{align}
    where $p \in (0,1)$. The other direction follows by noting Alice and Bob can both perform the same binary projective measurement on $\chi^{\vert p}$ to satisfy Item 3 of Lemma \ref{lem:Beigi-thm-5}.\\

    ($2 \iff 3$) This follows from Item 3 of Lemma \ref{lem:Beigi-thm-5} and Proposition \ref{prop:POVM-to-PVM}, which shows POVMs that achieve this strategy imply the existence of projective measurements that do the same.  \\
    
    ($3 \iff 4$) First, if such a decomposition exists, a direct calculation will verify the dichotomous, projective measurements $\{\Pi_{A_{0}},\Pi_{A_{1}}\}$, $\{\Pi_{B_{0}},\Pi_{B_{1}}\}$ will obtain $\chi^{\vert p}$. Thus, we just focus on the other direction. Assume such dichotomous, projective measurements exist, which we denote $\{P_{0},P_{1}\}$, $\{Q_{0},Q_{1}\}$. Now we define $A_{i} = \linspan(\supp(P_{i}))$ for $i \in \{0,1\}$ and similarly for $B_{i}$. This means $P_{0}$ is the projection onto $A_{0}$ and similarly for the others. Note that we may decompose the entire space as $A \otimes B = (A_{0} \otimes B_{0}) \oplus (A_{0} \otimes B_{1}) \oplus (A_{1} \otimes B_{0}) \oplus (A_{1} \otimes B_{1})$, which defines an ordering on the product basis vectors. We may express $\rho_{AB} = \sum_{i,i',j,j' \in \{0,1\}} E_{i,j} \otimes E_{i',j'} \otimes X_{ii',jj'}$ where $X_{ii',jj'} \in \Lin(A_{j} \otimes B_{j'}, A_{i} \otimes A_{i'})$. This is in effect writing it in block matrix form in terms of the subspaces of $A \otimes B$ given above. Since $\rho_{AB}$ is a density matrix, we know that $X_{ii',ii'} = q(i,i')\sigma_{A_{i}B_{i}}$ where $\sigma_{A_{i}B_{i}} \in \Density(A_{i} \otimes B_{i})$ and $\{q(i,i')\}_{i,i' \in \{0,1\}^{\times 2}}$ is a probability distribution. Now, $0 = \Tr[P_{0} \otimes Q_{1}\rho_{AB}] = t_{0,1}\Tr[ \sigma_{A_{0}B_{1}}]$ where the first equality is our assumption and the second is our block decomposition. It follows $t_{0,1} = 0$. By an identical argument $t_{1,0} = 0$. This means that $X_{01,01}$ and $X_{10,10}$ are both the zero matrix. By Sylvester's criterion, we know every column/row that includes an element of $X_{01,01}$, $X_{10,10}$ must be zero. In other words, for all $i,j$, 
    $$ X_{01,ij} = 0 \quad X_{10,ij} = 0 \quad X_{ij,01} = 0 \quad X_{ij,10} = 0 \ . $$
    The only blocks remaining are $\sigma_{A_{0}B_{0}},\sigma_{A_{1}B_{1}}$, $X_{00,11}$ and $X_{11,00}$. As $\rho_{AB}$ is Hermitian, $X_{00,11} = X^{\ast}_{11,00}$. Defining $p \coloneq t_{0,0}$, $\rho^{i}_{A_{i}B_{i}} \coloneq \sigma_{A_{i}B_{i}}$ and $X \coloneq X_{00,11}$ completes the proof.
\end{proof}

\subsection{Quantum Correlation Coefficients as Strong Monotones for Conversion under Local Processing}\label{subsec:strong-monotones}
As mentioned in the introduction, the property that classically makes the maximal correlation coefficient a strong monotone for transformations under local operations is that it \textit{tensorizes} over independent distributions, i.e. for $p_{XY} \otimes q_{X'Y'}$, $\mu(XX':YY')_{p \otimes q} = \max\{\mu(X:Y)_{p},\mu(X':Y')\}$. In \cite{Beigi-2013a}, Beigi showed this extends to $\mu^{\text{Lin}}_{f_{1}}(A:B)_{\rho}$. In this section, we extend this in a variety of manners.

First, we show that $\mu^{\text{Lin}}_{f_{k}}(A:B)_{\rho}$ tensorizes for all $k \in [0,1]$. That is, given Proposition \ref{prop:multiplicativity-of-J}, the tensorization property holds for all $\mu^{\text{Lin}}_{f}(A:B)$ that are induced by a multiplicative operator monotone function. This shows all these quantities are strong monotones for local operations:
\begin{tcolorbox}[width=\linewidth, sharp corners=all, colback=white!95!black, boxrule=0pt,frame hidden]
\begin{theorem}\label{thm:k-correlation-nec-for-local-processing}
    There exists $n \in \mbb{N}$ such that $\rho_{AB}^{\otimes n}$ can be converted to $\sigma_{A'B'}$ under local two-positive trace-preserving maps, only if
    $$\mu^{\text{Lin}}_{f_{k}}(A:B)_{\rho} \geq \mu^{\text{Lin}}_{f_{k}}(A':B')_{\sigma} \quad \forall k \in [0,1] \, . $$
\end{theorem}
\end{tcolorbox}

Second, we establish that all the Schmidt coefficients relevant to $\mu_{GM}(A:B)_{\rho}$ are monotones under local operations. This implies a new majorization condition for conversion under local operations that includes $\mu_{GM}$ acting as a monotone for conversion under local positive maps as a special case.
\begin{tcolorbox}[width=\linewidth, sharp corners=all, colback=white!95!black, boxrule=0pt,frame hidden]
\begin{theorem}\label{thm:mu-GM-as-monotone}
   Consider quantum states $\rho_{AB},\sigma_{A'B'},\tau_{A''B''}$. Define $\wt{\rho} \coloneq (\rho_{A} \otimes \rho_{B})^{-1/4}\rho_{AB}(\rho_{A} \otimes \rho_{B})^{-1/4}$ and similarly for $\wt{\sigma}$, $\wt{\tau}$. Let $\{\lambda_{i}\}_{i \in [d_{A}d_{B}]}$, $\{\omega_{j}\}_{j \in [d_{A'}d_{B'}]}$, $\{\zeta_{k}\}_{k \in [d_{A''}d_{B''}]}$ be the (ordered) Schmidt coefficients with respect to inner product spaces $(\Herm(A),\langle \cdot, \cdot \rangle)$, $(\Herm(B),\langle \cdot, \cdot \rangle)$ of $\wt{\rho}$, $\wt{\sigma}$, $\wt{\tau}$ respectively. Define $r^{\downarrow}$ be the vector of scalars $\{\lambda_{i}\zeta_{k}\}_{i,k}$ ordered to be decreasing. There exist positive, trace-preserving maps  $\cE_{AA'' \to A'}$ and $\cF_{BB'' \to B'}$ such that $(\cE \otimes \cF)(\rho \otimes \tau) = \sigma$ only if $r_{j} \geq \omega_{j}$ for all $j$ where $\omega_{j}$ may be embedded. \\
   
   In particular, there exists $n \in \mbb{N}$ such that $\rho_{AB}^{\otimes n}$ can be converted to $\sigma_{A'B'}$ under positive, trace-preserving maps $\cE_{A^{n} \to A'}$ and $\cF_{B^{n} \to B'}$ only if 
    \begin{align}
        \mu_{GM}(A:B)_{\rho} \geq \mu_{GM}(A':B')_{\sigma} \ . 
    \end{align}
\end{theorem}
\end{tcolorbox}

An application of these new results is that we obtain many new necessary conditions for `asymptotic common data,' as we briefly motivate before formally stating. Letting $\cZ = \{0,1\}$, it is natural to ask if one can extract $\chi^{\vert p}_{ZZ'}$ for some $p \in (0,1)$ asymptotically, i.e. does there exist a sequence of local operations $(\cE_{A^{n} \to \cZ})_{n}$, $(\cF_{B^{n} \to \cZ})_{n}$ and parameter $p \in (0,1)$ such that 
\begin{align}\label{eq:ACD}
    \lim_{n \to \infty} \Vert \cE_{A^{n} \to \cZ} \otimes \cF_{B^{n} \to \cZ}(\rho_{AB}^{\otimes n}) - \chi^{\vert p}_{ZZ'} \Vert_{1} = 0 \, .
\end{align}
In classical information theory, the classical special case of Item 3 of Theorem \ref{thm:extreme-values-summary} combined with the fact $\mu(X:Y)_{p}$ tensorizes is sufficient to conclude this is possible if and only if $\mu(X:Y)_{p}=1$. However, as Theorem \ref{thm:extreme-values-summary} shows $\mu_{AM}(A:B)_{\rho}$ characterizes the ability to extract exact correlation using local operations in the single-copy setting and we have no proof that it tensorizes,\footnote{Indeed, we conjecture it does not.} we cannot extend the classical result to the quantum setting. Thus, following \cite{Beigi-2013a}, we say $\rho_{AB}$ `asymptotically has common data' if there are sequences of quantum channels such that \eqref{eq:ACD} holds. \cite{Beigi-2013a} showed $\mu_{f_{0}}^{\text{Lin}}(A:B)_{\rho} = 1$ is a necessary condition for \eqref{eq:ACD} to hold. By the same argument as Beigi and our above results, we obtain the following.
\begin{tcolorbox}[width=\linewidth, sharp corners=all, colback=white!95!black, boxrule=0pt,frame hidden]
\begin{theorem}\label{thm:asymptotic-data}
    $\rho_{AB}$ asymptotically has common data only if $\mu_{f_{k}}^{\text{Lin}}(A:B)_{\rho} = 1$ for all $k \in (0,1/2)\cup(1/2,1)$ and $\mu_{GM}(A:B)_{\rho} = 1$.
\end{theorem}
\end{tcolorbox}
\begin{proof}
    In the case of $\mu_{GM}$, this follows from $\mu_{AM}(A^{n}:B^{n})_{\rho^{\otimes n}} \leq \mu_{GM}(A^{n}:B^{n})_{\rho^{\otimes n}} = \mu_{GM}(A:B)_{\rho}$ by \eqref{eq:ordering-of-maximal-corr-means} and Lemma \ref{lem:k-correlation-tensorize}. For the other cases, it follows from Lemma \ref{lem:k-correlation-tensorize} and the same proof as \cite[Theorem 9]{Beigi-2013a}, which relies on lemmata from that work.
\end{proof}
\noindent We remark the most important case of the above theorem is likely $\mu_{GM}$ given Proposition \ref{prop:separating-mu-AM-and-mu-GM} and the operational relevance of $\mu_{GM}$ in Section \ref{sec:quantum-chi-squared}.

In the rest of the section, we establish Theorems \ref{thm:k-correlation-nec-for-local-processing} and \ref{thm:mu-GM-as-monotone}. For Theorem \ref{thm:k-correlation-nec-for-local-processing}, this comes from first proving $\mu^{\text{Lin}}_{f}(A:B)_{\rho} \leq 1$ (Lemma \ref{lem:map-norm-is-1-for-k-correlation-coeff}) and then using this to prove the quantity tensorizes. Note that $\mu^{\text{Lin}}_{f}(A:B)_{\rho} \leq 1$ does not follow from Item 2 of Theorem \ref{thm:extreme-values-summary} as $x^{k} \not \geq x^{1/2}$ for all $x \geq 0$. Instead, we will use an operator-theoretic argument that generalizes the methods used by Beigi to address $\mu_{f_{0}}^{\text{Lin}}(A:B)_{\rho}$ \cite{Beigi-2013a}. In particular, we will generalize many of Beigi's identifications and use complex interpolation. For Theorem \ref{thm:mu-GM-as-monotone}, we use a significantly more direct proof method, which borrows ideas from \cite{Delgosha-2014a}.

\subsubsection{Establishing Theorem \ref{thm:k-correlation-nec-for-local-processing}}
We begin with some notation that we will need. As mentioned, this will stem from using complex interpolation on a relevant family of maps. To this end, we define the complex strip $\mrm{St} \coloneq \{z \in \mbb{C} : 0 \leq \mbb{Re}[z] \leq 1 \}$. Then we define the following family of maps, $\cT_{\tau,\gamma}(X) \coloneq \tau^{\gamma/2}X\tau^{(1-\gamma)/2}$ where we allow $\gamma \in \mrm{St}$. While this family of maps may seem strange, it will be natural to consider, which we briefly explain. Namely, when $\gamma \in [0,1]$, one may verify $\cT_{\tau,\gamma} =\mbf{J}_{f_{\gamma},\tau}^{1/2}$. This latter operator will naturally arise in the following lemma. The reason to introduce $\cT_{\tau,\gamma}$ over the complex strip more generally is so that we may apply complex interpolation in Lemma \ref{lem:map-norm-is-1-for-k-correlation-coeff}. Lastly, we note it is clear that for $\tau >0$, $\cT^{-1}_{\tau,\gamma}(X) =  \tau^{-(\gamma/2)}X\tau^{-(1-\gamma)/2}$ since this will indeed invert the operation. 

With these points established, we make the following formal identification, which converts evaluating $\mu_{f_{k}}(\rho_{AB})$ to evaluating the $2 \to 2$-norm on certain linear maps.
\begin{lemma}\label{lem:norm-bound}
    For $k \in [0,1]$, 
    \begin{equation}\label{eq:change-of-variables-for-Xk}
    \begin{aligned}
        \mu^{\text{Lin}}_{f_{k}}(A:B)_{\rho} = \sup & \; \Big\vert \Tr[\wt{X} \otimes \wt{Y}^{\ast} \wt{\rho}^{k}_{AB}] \Big\vert \\
        \text{s.t.} \; & \langle \rho_{A}^{1/2} , \wt{X}_{k} \rangle = 0 = \langle \rho_{B}^{1/2} , \wt{Y} \rangle \\
        & \Vert \wt{X} \Vert_{2} = 1 = \Vert \wt{Y} \Vert_{2} \ ,
    \end{aligned}
\end{equation}
where $\wt{X} \in \Lin(A), \wt{Y} \in \Lin(B)$. This implies $\mu_{f_{k}}^{\text{Lin}}(\rho_{AB}) \leq \Vert \Lambda_{\wt{\rho}_{k}} \Vert_{2 \to 2}$ where 
\begin{align}\label{eq:rho-tilde-k-map}
\Lambda_{\wt{\rho}_{k}} \coloneq \cT^{-1}_{\rho_{B},k} \circ \Lambda_{\rho} \circ \cT^{-1}_{\ol{\rho}_{A},1-k} \ , 
\end{align} 
and 
\begin{align}\label{eq:wt-rho-k-def}
    \wt{\rho}_{k} = (\rho_{A}^{-(1-k)/2} \otimes \rho_{B}^{-k/2}) \, \rho_{AB} (\rho_{A}^{-k/2} \otimes \rho_{B}^{-(1-k)/2}) \in \Lin(A \otimes B) \ . 
\end{align}
\end{lemma}
\begin{proof}
Following \cite{Beigi-2013a}, we can make a change of variables $\wt{X} \coloneq \rho_{A}^{k/2}X\rho_{A}^{(1-k)/2}$, $\wt{Y} \coloneq \rho_{B}^{k/2}Y\rho_{B}^{(1-k)/2}$ where we keep the $k$ dependence implicit. Then by direct calculations one may verify \eqref{eq:change-of-variables-for-Xk}. Note that we have to work with linear operator relaxation as for all $k \in [0,1]\setminus \{1/2\}$ this is not generically a Hermitian operator. The form of $\Lambda_{\wt{\rho}^{k}}$ follows an identical calculation to that before \cite[Eq. 29]{Delgosha-2014a}, which we provide for completeness:
\begin{equation}\label{eq:delgosha-calculation}
\begin{aligned}
    \wt{\rho}_{k} =& \sum_{i,j} \rho_{A}^{-(1-k)/2}\ket{i}\bra{j}\rho_{A}^{-k/2} \otimes \rho_{B}^{-k/2}\Lambda_{\rho}(\ket{i}\bra{j})\rho_{B}^{-(1-k)/2} \\
    =& \sum_{k,l,i,j} \dyad{k}\rho_{A}^{-(1-k)/2}\ket{i}\bra{j}\rho_{A}^{-k/2}\dyad{l} \otimes \rho_{B}^{-k/2}\Lambda_{\rho}(\ket{i}\bra{j})\rho_{B}^{-(1-k)/2} \\
    =& \sum_{k,l,i,j} \ket{k}\bra{i}\ol{\rho}_{A}^{-(1-k)/2}\ket{k}\bra{l}\ol{\rho}_{A}^{-k/2}\ket{j}\bra{l} \otimes \rho_{B}^{-k/2}\Lambda_{\rho}(\ket{i}\bra{j})\rho_{B}^{-(1-k)/2} \\
    =& \sum_{k,l,i,j} \ket{k}\bra{l} \otimes \rho_{B}^{-k/2}\Lambda_{\rho}(\ket{i}\bra{i}\ol{\rho}_{A}^{-(1-k)/2}\ket{k}\bra{l}\ol{\rho}_{A}^{-k/2}\ket{j}\bra{j})\rho_{B}^{-(1-k)/2} \\
    =& \sum_{i,j} \ket{k}\bra{l} \otimes \rho_{B}^{-k/2}\Lambda_{\rho}(\ol{\rho}_{A}^{-(1-k)/2}\ket{k}\bra{l}\ol{\rho}_{A}^{-k/2})\rho_{B}^{-(1-k)/2} \\
    =& (\id \otimes \cT^{-1}_{\rho_{B},k} \otimes \Lambda_{\rho} \otimes \cT_{\ol{\rho}_{A},1-k}^{-1})(\Phi^{+}) \ . 
\end{aligned}
\end{equation}
Finally,
\begin{equation}\label{eq:k-correlation-coeff-bound-by-map-norm}
\begin{aligned}
    \mu^{\text{Lin}}_{f_{k}}(A:B)_{\rho} \leq& \sup_{\Vert \wt{X} \Vert_{2} = 1 = \Vert \wt{Y} \Vert_{2}} \Big\vert \Tr[\wt{X} \otimes \wt{Y}\wt{\rho}_{k}] \vert 
    =  \sup_{\Vert \wt{X} \Vert_{2} = 1 = \Vert \wt{Y} \Vert_{2}} \Big\vert \Tr[\wt{X} \otimes \wt{Y} \Omega_{\Lambda_{\wt{\rho}_{k}}}] \Big\vert
    = \Vert \Lambda_{\wt{\rho}_{k}} \Vert_{2 \to 2} \ ,
\end{aligned}
\end{equation}
where we relaxed the supremization, used the definition of the Choi operator, and then applied \eqref{eq:map-norm-in-terms-of-Choi} where we have used $\Vert X \Vert_{2} = \Vert X^{T} \Vert_{2}$.
\end{proof}

Now we turn to our main technical lemma of this subsection, which will make use of the Hadamard's three-line theorem \cite{Reed-1975}.
\begin{proposition}\label{prop:Hadamard-3-line}
    Let $f:\mathrm{St} \to \mbb{C}$ be a bounded function that is holomorphic in the interior of $\mathrm{St}$ and continuous on the boundary. For $k \in \{0,1\}$ let $M_{k} = \sup_{t \in \mbb{R}} \vert f(k+it) \vert$. Then for every $0 \leq \theta \leq 1$, $\vert f(\theta) \vert \leq M_{0}^{1-\theta}M_{1}^{\theta}$.
\end{proposition}

\begin{lemma}\label{lem:map-norm-is-1-for-k-correlation-coeff}
For $k \in [0,1]$, $\Vert \Lambda_{\wt{\rho}_{k}} \Vert_{2 \to 2} = 1 = \lambda_{1}(\wt{\rho}_{k})$ where $\lambda_{i}(\wt{\rho}_{k})$ denote the Schmidt coefficients of $\wt{\rho}_{k}$ over the Hilbert spaces $(\Lin(A), \langle \cdot , \cdot \rangle)$ and $(\Lin(B), \langle \cdot , \cdot \rangle)$. Moreover, the operator norm is achieved with $X = \rho_{A}^{1/2}$, $Y = \rho_{B}^{1/2}$.
\end{lemma}
\begin{proof}
    As $(\Lin(A),\langle \cdot, \cdot \rangle)$ is a Hilbert space with the Schatten $2$-norm being the canonical norm, by Lemma \ref{lem:map-norms-as-optimizations}, for any linear map $\Phi$,
    \begin{align}
        \Vert \Phi \Vert_{2 \to 2} = \max_{\Vert X \Vert_{2} = 1 = \Vert Y \Vert_{2}} \vert \Tr[Y \Phi(X)] \vert 
    \end{align}
    where $X$ and $Y$ are linear operators. Our goal is to bound this operator norm in the case $\Phi$ is $\Lambda_{\wt{\rho}_{k}}$. To do this, we construct a family of maps that extend the maps $\{\Lambda_{\wt{\rho}_{k}}\}_{k \in [0,1]}$ to the complex strip, apply Hadamard's three-line theorem to a function induced by those maps, and then show for $k \in \{0,1\}$, $M_{k} \leq 1$ where $M_{k}$ is as defined in Proposition \ref{prop:Hadamard-3-line}.
    
    To construct extensions of the relevant maps to the complex strip, we define the map $T_{z} \coloneq \cT_{\rho_{B},z}^{-1} \circ \Lambda_{\rho} \circ \cT_{\ol{\rho}_{A},1-z}^{-1}$. Note that for complex number $z = a +ib$ and (full rank) density matrix $\tau$,
    \begin{align}
        \cT_{\tau,1-z}(X) = \tau^{-(1-a-ib)/2}X\tau^{-(a+ib)/2} = \tau^{-(1-a)/2}\tau^{ib/2}X\tau^{-(ib/2)}\tau^{-a/2} = (\cT^{-1}_{\tau,1-a} \circ \cU_{\tau,b})(X) \ ,  
    \end{align}
    where $\cU_{\tau,b}$ is the unitary channel defined by the unitary $U = \tau^{-ib/2}$, which is a unitary as $\tau$ is Hermitian. It follows that 
    \begin{align}\label{eq:Tz-expansion}
        T_{z} = \cU_{\rho_{B},\Im{z}} \circ \cT^{-1}_{\rho_{B},\Re{z}} \circ \Lambda_{\rho} \circ \cT^{-1}_{\tau,1-\Re{z}} \circ \cU_{\ol{\rho}_{A},\Im{z}} = \cU_{\rho_{B},\Im{z}} \circ \Lambda_{\wt{\rho}_{\Re{z}}} \circ \cU_{\ol{\rho}_{A},\Im{z}}  \ .
    \end{align}

    Next we use Hadamard's three-line theorem. Let $X$ and $Y$ such that $\Vert X \Vert_{2} = \Vert Y \Vert_{2} = 1$. Define the function on the complex strip by
    \begin{align}
        f(z) \coloneq \Tr[Y T_{z}(X)] \ . 
    \end{align}
    This function satisfies the conditions of Hadamard's three-line theorem, so we have for $\theta \in [0,1]$
    \begin{align}\label{eq:complex-interpolation-bound}
        \vert \Tr[Y T_{\theta}(X)] \vert \leq \left( \sup_{t \in \mbb{R}}  \left\vert\Tr[Y T_{it}(X)] \right\vert \right)^{1-\theta}\left( \sup_{t \in \mbb{R}}  \left\vert\Tr[Y T_{1+it}(X)] \right\vert \right)^{\theta} \ . 
    \end{align}
    Now,
    \begin{align}
        \vert \Tr[Y T_{it}(X)] \vert &\leq \sup_{\Vert X \Vert_{2} = 1 = \Vert Y \Vert_{2}} \vert \Tr[YT_{it}(X)] \vert \\
        &= \sup_{\Vert X \Vert_{2} = 1 = \Vert Y \Vert_{2}} \vert \Tr[Y \cU_{\rho_{B},t} \circ \Lambda_{\wt{\rho}_{0}} \circ \cU_{\ol{\rho}_{A},t}(X)] \vert \\
        &= \sup_{\Vert X \Vert_{2} = 1 = \Vert Y \Vert_{2}} \vert \Tr[ \cU_{\rho_{B},t}^{\ast}(Y) \circ \Lambda_{\wt{\rho}_{0}}( \cU_{\ol{\rho}_{A},t}(X))] \vert \\
        &= \sup_{\Vert X \Vert_{2} = 1 = \Vert Y \Vert_{2}} \vert \Tr[ Y \circ \Lambda_{\wt{\rho}_{0}}(X)] \vert \\
        &= \sup_{\Vert X \Vert_{2} = 1 = \Vert Y \Vert_{2}} \vert \Tr[ X \otimes Y \wt{\rho}_{0}] \vert,
    \end{align}
    where the inequality is our choice of $X$ and $Y$, the first equality is \eqref{eq:Tz-expansion}, the third is the unitary invariance of the Schatten $2$-norm, so we may absorb the unitary channels into the optimizations, and the final equality is \eqref{eq:map-norm-in-terms-of-Choi} and the norm being invariant under transpose. An identical argument shows $\vert \Tr[Y T_{1+it}(X)] \vert \leq \sup_{\Vert X \Vert_{2} = 1 = \Vert Y \Vert_{2}} \vert \Tr[ X \otimes Y \wt{\rho}_{1}] \vert$. Thus, we just need to bound these terms. 
    
    It was shown in the proof of \cite[Theorem 1]{Beigi-2013a} that $\sup_{\Vert X \Vert_{2} = 1 = \Vert Y \Vert_{2}} \vert \Tr[ X \otimes Y \wt{\rho}_{1}] \vert \leq 1$. For self-containment of our work, we show the proof for $\wt{\rho}_{0}$, which is nearly identical. Consider the Hilbert spaces $(\Lin(A), \langle \cdot , \cdot \rangle)$ and $(\Lin(B), \langle \cdot , \cdot \rangle)$ so that 
\begin{align}
    \wt{\rho}_{0} = \sum_{i} \lambda_{i} M_{i} \otimes N_{i} \ , 
\end{align}
where $\lambda_{1} \geq \lambda_{2} ... \geq 0$ are the Schmidt coefficients and $\{M_{i}\}_{i}$, $\{N_{i}\}_{i}$ are orthonormal bases (with respect to the Hilbert-Schmidt inner product) of $\Lin(A),\Lin(B)$ respectively. It follows that $\lambda_{1}(\wt{\rho}_{0}) = \Vert \Lambda_{\wt{\rho}_{0}} \Vert_{2 \to 2}.$ We thus just need to bound $\lambda_{1}(\wt{\rho}_{0})$. Then
\begin{equation}
\begin{aligned}
    \lambda_{1}(\wt{\rho}_{0}) =& \Tr[M_{1}^{\ast} \otimes N_{1}^{\ast}\wt{\rho}_{0}] \\
    =& \Tr[(M_{1}^{\ast}\rho_{A}^{-1/2}) \otimes (\rho_{B}^{-1/2}N_{1}^{\ast})\rho_{AB}] \\
    =& \Tr[((M_{1}^{\ast}\rho_{A}^{-1/2}) \otimes \mbb{1}_{B}) \rho_{AB}(\mbb{1}_{A} \otimes (\rho_{B}^{-1/2}N_{1}^{\ast}))] \\
    =& \langle \rho_{AB}^{1/2}((\rho_{A}^{-1/2}M_{1}) \otimes \mbb{1}_{B}), \rho_{AB}^{1/2}(\mbb{1}_{A} \otimes (\rho_{B}^{-1/2}N_{1}^{\ast})) \rangle \\
    \leq& \Tr[(M_{1}^{\ast}\rho_{A}^{-1/2} \otimes \mbb{1}_{B})\rho_{AB}(\rho_{A}^{-1/2}M_{1} \otimes \mbb{1}_{B})]^{1/2} \Tr[(\mbb{1}_{A} \otimes N_{1} \rho_{B}^{-1/2})\rho_{AB}(\mbb{1}_{A} \otimes \rho_{B}^{-1/2}N_{1}^{\ast})]^{1/2} \\
    =& \Tr[M_{1}^{\ast}\rho_{A}^{-1/2}\rho_{A}\rho_{A}^{-1/2}M_{1}]^{1/2}\Tr[N_{1} \rho_{B}^{-1/2}\rho_{B}\rho_{B}^{-1/2}N_{1}]^{1/2} \\
    =& \Tr[M_{1}^{\ast}M]^{1/2}\Tr[N_{1}^{\ast}N]^{1/2} \\
    =& 1 \ ,
\end{aligned}
\end{equation}
where the second equality is definition of $\wt{\rho}_{0}$, the third is cyclicity of trace, the fourth is re-writing in terms of HS inner product, the inequality is Cauchy-Schwarz, the fifth equality is the definition of partial trace, the sixth is canceling terms and the final is the orthonormality of the $\{M_{i}\}$ and $\{N_{i}\}$. 

Combining the abouve bounds with \eqref{eq:complex-interpolation-bound} and noting that our initial choice of $X$ and $Y$ in defining the function $f$ was arbitrary except that $\Vert X \Vert = \Vert Y \Vert =1$, we obtain for $\theta \in [0,1]$, $\Vert T_{\theta} \Vert_{2 \to 2} \leq 1$. As $T_{\theta} = \Lambda_{\wt{\rho}_{\theta}}$, this completes the proof.
\end{proof}

With the above established, we obtain a Schmidt coefficient characterization of the $\mu_{f_{k}}^{\text{Lin}}$ by following the method in \cite{Beigi-2013a}. As noted in \cite{Beigi-2013a}, this is equivalently a singular value characterization of $\Lambda_{\wt{\rho}_{k}}$ as follows from Proposition \ref{prop:Schmidt-to-sing}.

\begin{lemma}\label{lem:k-correlation-coeff-Schmidt-coeff-characterization}
    Let $\rho_{AB} \in \Density(AB)$, $k \in [0,1]$. The following are equivalent:
    \begin{enumerate}
        \item $\mu_{f_{k}}^{\text{Lin}}(A:B)_{\rho}$ , 
        \item $\lambda_{2}(\wt{\rho}_{k})$ where $\lambda_{i}(\wt{\rho}_{k})$ denote the Schmidt coefficients of $\wt{\rho}_{k}$ over the Hilbert spaces $(\Lin(A), \langle \cdot , \cdot \rangle)$ and $(\Lin(B), \langle \cdot , \cdot \rangle)$ , 
        \item the second singular value of $\Lambda_{\wt{\rho}_{k}}: \Lin(A) \to \Lin(B)$.
    \end{enumerate}
\end{lemma}
\begin{proof}
    Fix a $k \in [0,1]$ as the proof is identical in each case. We consider the Schmidt decomposition $
    \wt{\rho}_{k} = \sum_{i} \lambda_{i} M_{i} \otimes N_{i}$ where the $\lambda_{i}$, $\{M_{i}\}$, $\{N_{i}\}$ depend on the $k$, but we omit this dependence for notational simplicity. By Lemma \ref{lem:map-norm-is-1-for-k-correlation-coeff}, we know $\lambda_{1}(\wt{\rho}) = 1$ and $M_{1} = \rho_{A}^{1/2}$, $N_{1} = \rho_{B}^{1/2}$ without loss of generality.

    Therefore, for any feasible $\wt{X},\wt{Y}$ in \eqref{eq:change-of-variables-for-Xk},\begin{equation}\label{eq:Schmidt-coeff-bound}
    \begin{aligned}
        \vert \Tr[\wt{X} \otimes \wt{Y}\wt{\rho}_{AB}] \vert =& \vert \sum_{i \geq 1} \lambda_{i} \langle \wt{X}^{\ast}, M_{i} \rangle \langle \wt{Y}^{\ast}, N_{i} \rangle \vert \\
        =& \vert \sum_{i \geq 2} \lambda_{i} \langle \wt{X}^{\ast}, M_{i} \rangle \langle \wt{Y}^{\ast}, N_{i} \rangle \vert \\
        \leq& \left( \sum_{i \geq 2} \lambda_{i} \vert \langle \wt{X}^{\ast},  M_{i} \rangle \vert^{2} \right)^{1/2} \left( \sum_{i \geq 2} \lambda_{i} \vert \langle  \wt{Y}^{\ast}, N_{i} \rangle \vert^{2} \right)^{1/2} \\
        \leq& \lambda_{2} \left( \sum_{i \geq 2} \vert \langle \wt{X}^{\ast}, M_{i}  \rangle \vert^{2} \right)^{1/2} \left( \sum_{i \geq 2} \vert \langle \wt{Y}^{\ast},  N_{i} \rangle \vert^{2} \right)^{1/2} \\
        =& \lambda_{2}
    \end{aligned}
    \end{equation}
    where the first equality uses the Schmidt decomposition, the second equality uses $M_{1} = \rho^{1/2}_{A}$, $\langle \rho^{1/2}_{A} , \wt{X} \rangle = 0$ and similarly for the $B$ space, the first inequality is Cauchy-Schwarz using complex vectors $u,v$ via $u(i) = \sqrt{\lambda_{i}}\langle M_{i}^{\ast}, \wt{X}^{\ast} \rangle$ and similarly for $v$, the second inequality uses that $\lambda_{2} \geq \lambda_{2+j}$ for all $j \geq 0$, and the last equality uses $\sum_{i \geq 2} \vert \langle \wt{X}^{\ast}, M_{i}  \rangle \vert^{2} = \Vert \wt{X}^{\ast} \Vert_{2}^{2} = \Vert \wt{X} \Vert_{2}^{2} = 1$ and similarly for $\wt{Y}$ where the first equality uses the orthogonality property of $\wt{X}$ and that $\{M_{i}\}$ are an orthonormal basis and the final equality uses the assumption $\Vert \wt{X} \Vert_{2} = 1$. Moreover, if $\wt{X} = M_{2}$ and $\wt{Y} = N_{2}$, then the inequalities are equalities. This completes the Schmidt decomposition identification, and the singular value identification then follows from Proposition \ref{prop:Schmidt-to-sing}. This completes the proof.
\end{proof}

We now obtain that these correlation coefficients tensorize.
\begin{lemma}\label{lem:k-correlation-tensorize}
    Let $k \in [0,1]$. Consider $\rho_{AB},\sigma_{A'B'} \in \Density$. Then 
    $$\mu_{f_{k}}(AA':BB')_{\rho \otimes \sigma} = \max\{\mu_{f_{k}}(A:B)_{\rho},\mu_{f_{k}}(A':B')_{\sigma}\} \ . $$
\end{lemma}
\begin{proof}
    Fix a $k \in [0,1]$ as the proof is the same for each. Let $\lambda_{1} \geq \lambda_{2} \geq ...$ and $\omega_{1} \geq \omega_{2} \geq ...$ be the Schmidt coefficients of $\wt{\rho}_{k}$ and $\wt{\sigma}_{k}$. By Item 2 of Lemma \ref{lem:k-correlation-coeff-Schmidt-coeff-characterization}, $\mu_{f_{k}}(\rho_{AB} \otimes \sigma_{A'B'})$ is the second Schmidt coefficient of $\wt{\rho \otimes \sigma}_{k} = \wt{\rho}_{k} \otimes \wt{\sigma}_{k}$, where we have used $(\rho_{A} \otimes \sigma_{A'})^{\alpha} = \rho_{A}^{\alpha} \otimes \sigma_{A'}^{\alpha}$ for $\alpha \in [0,1]$. The second Schmidt coefficient of $\wt{\rho}_{k} \otimes \wt{\sigma}_{k}$ is 
    $$\max\{\lambda_{1}\omega_{2},\lambda_{2}\omega_{1}\} = \max\{\mu_{f_{k}}(A:B)_{\rho},\mu_{f_{k}}(A':B')_{\sigma}\}  \ , $$
    where we have used that the first Schmidt coefficient is always one (Lemma \ref{lem:map-norm-is-1-for-k-correlation-coeff}), and Item 2 of Lemma \ref{lem:k-correlation-coeff-Schmidt-coeff-characterization} again.
\end{proof}

Finally, we may prove that the $\mu_{f_{k}}$ capture necessary conditions for transformations under local processing.
\begin{proof}[Proof of Theorem \ref{thm:k-correlation-nec-for-local-processing}]
For any $n \in \mbb{N}$ and two-positive, trace-preserving maps $\cE_{A^{n} \to A'}$, $\cF_{B^{n} \to B'}$, let $\sigma_{A'B'} \coloneq (\cE \otimes \cF)(\rho_{AB}^{\otimes n})$. Then
$$\mu_{f_{k}}(A':B')_{\sigma} \geq \mu_{f_{k}}(A^{n}:B^{n})_{\rho^{\otimes n}} = \mu_{f_{k}}(A:B)_{\rho} \ , $$
where the first equality is Proposition \ref{prop:DPI-for-f-correlation} and the second inequality is Lemma \ref{lem:k-correlation-tensorize}.
\end{proof}

\subsubsection{Establishing Theorem \ref{thm:mu-GM-as-monotone} (Strengthened Results for \texorpdfstring{$\mu_{GM}$)}{}}\label{subsec:strengthened-results-for-GM}
In this subsection we establish that in fact all of the Schmidt coefficients of the operators that give rise to $\mu_{GM}(\rho_{AB})$ satisfy data processing under \textit{positive} maps. This allows us to establish Corollary \ref{thm:mu-GM-as-monotone}. We remark the fact that $\mu_{GM}$ admits the DPI under positive maps may be intuitive as, as we note later, $\mu_{GM}$ is deeply related to the Sandwiched divergence of order $2$, and it was shown in \cite{Muller-2017a} that sandwiched divergences satisfy DPI under positive maps.

For notational simplicity, for the rest of this section, given a quantum state $\rho$, we let $\wt{\rho}$ denote what was written as $\wt{\rho}_{1/2}$ in \eqref{eq:wt-rho-k-def}.
\begin{lemma}\label{lem:sandwiched-correlation-coeff-DPI}
    \sloppy Consider $\rho_{AB}$ and $\sigma_{AB'} \coloneq (\text{id}_{A} \otimes \cE)(\rho_{AB})$ where $\cE$ is a positive linear map. Then 
    \begin{align}
        \lambda_{i}(\wt{\rho}_{AB}) \geq \omega_{i}(\wt{\sigma}_{AB'}) \quad \forall i \ ,
    \end{align}
    where $\lambda_{1} \geq \lambda_{2} ... \geq 0$ are the Schmidt coefficients of $\wt{\rho}_{AB}$ for local Hilbert spaces $(\Herm(A),\langle \cdot, \cdot \rangle),(\Herm(B),\langle \cdot, \cdot \rangle)$ and similarly for $\omega_{i}$ and $\wt{\sigma}$. In particular, this implies that $\mu_{GM}(\rho_{AB}) \geq \mu_{GM}(\sigma_{AB'})$.
\end{lemma}
\begin{proof}
    Our proof method borrows ideas from \cite{Delgosha-2014a}. As is standard, e.g. \cite{Beigi-2013sandwiched}, for $\tau \in \Density_{+}(A)$ we define the map $\Gamma_{\tau}(X) \coloneq \tau^{1/2}X\tau^{1/2}$. One may then verify by direct calculation $\Gamma_{\tau}^{-1/2} = \cT_{\tau,1/2}^{-1}$. Therefore, by \eqref{eq:rho-tilde-k-map}, we know\begin{align}\label{eq:Omega-wt-rho-composition}
        \Omega_{\wt{\rho}} = \Gamma_{\rho_{B}}^{-1/2} \circ \Omega_{\rho} \circ \Gamma_{\ol{\rho}_{A}}^{-1/2} \ .
    \end{align} 
    Now, $\Omega_{\sigma} = (\text{id}_{A} \otimes \cE)(\rho_{AB}) = (\id_{A} \otimes \cE \circ \Omega_{\rho})(\dyad{\Phi^{+}})$. Thus, $\Omega_{\sigma} = \cE \circ \Omega_{\rho}$. Combining this with \eqref{eq:Omega-wt-rho-composition},
    \begin{align}
        \Omega_{\wt{\sigma}} =& \Gamma_{\Phi(\rho_{B})}^{-1/2} \circ \cE \circ \Omega_{\rho} \circ \Gamma_{\ol{\rho}_{A}}^{-1/2} \\
        =& \Gamma_{\Phi(\rho_{B})}^{-1/2} \circ \cE \circ \Gamma_{\rho_{B}}^{1/2} \circ \Gamma_{\rho_{B}}^{-1/2}  \circ \Omega_{\rho} \circ \Gamma_{\ol{\rho}_{A}}^{-1/2} \\
        \coloneqq& \; \wt{\cE} \circ \Omega_{\wt{\rho}} \ .
    \end{align}
    Now, by the min-max principle for singular values and the relationship between Schmidt coefficients and singular values, 
    \begin{align}
        \lambda_{i}(\wt{\sigma}) = s_{i}(\wt{\cE} \circ \Omega_{\wt{\rho}}) \leq s_{1}(\wt{\cE}) s_{i}(\Omega_{\wt{\rho}}) = \Vert \wt{\cE} \Vert s_{i}(\Omega_{\wt{\rho}}) =  \Vert \wt{\cE} \Vert \lambda_{i}(\wt{\rho}) \ .
    \end{align}
    We thus just need to bound the operator norm of $\wt{\cE}$. As $\cE$ is positive and $\Gamma_{\Phi(\rho)}^{-1/2},\Gamma^{1/2}_{\rho_{B}}$ are completely positive, we can conclude $\wt{\cE}$ is positive, so by the Russo-Dye theorem (see \cite[2.37 Theorem]{Bhatia-2009a}), $\Vert \wt{\cE} \Vert = 1$. This completes the proof.
\end{proof}

We use the above to prove Theorem \ref{thm:mu-GM-as-monotone}.
\begin{proof}[Proof of Theorem \ref{thm:mu-GM-as-monotone}]
    By \eqref{eq:wt-rho-k-def}, $\wt{\rho_{AB} \otimes \tau_{A''B''}} = \wt{\rho}_{AB} \otimes \wt{\tau}_{A''B''}$. Thus, the Schmidt coefficients of the former are the ordered products of pairs of Schmidt coefficients, i.e. $\{\lambda_{i}\zeta_{k}\}_{i,k}$ ordered. Thus applying Lemma \ref{lem:sandwiched-correlation-coeff-DPI} establishes the general claim. To obtain the in particular case, we let $\tau = \rho$ and then use that the first Schmidt coefficient is always one.
\end{proof}
Before moving forward, we make three remarks in regards to this theorem. First, a similar claim can be made for the adjoint of a unital Schwarz map and $\mu_{f_{1}}^{\text{Lin}}$ given Proposition \ref{prop:DPI-for-f-correlation} and \cite[Theorem 4]{Beigi-2013a}. Second, for pure states, a similar claim to Theorem \ref{thm:mu-GM-as-monotone} is given in \cite[Theorem 7]{George-2024-pure-state}. The difference is rather than capturing necessary conditions, that result captures the minimum possible error as measured under fidelity. Finally, because the Schmidt decomposition of $\wt{\rho_{AB}}$ with respect to $(\Herm(A), \langle \cdot , \cdot \rangle)$, $(\Herm(B), \langle \cdot, \cdot \rangle)$ also is a Schmidt decomposition with respect to $(\Lin(A), \langle \cdot , \cdot \rangle)$, $(\Lin(B), \langle \cdot, \cdot \rangle)$, we may conclude it is always the case
\begin{align}\label{eq:GM-max-corr-doesn't-depend-on-lin}
    \mu_{GM}(A:B)_{\rho} = \mu_{GM}^{\text{Lin}}(A:B)_{\rho} \ , 
\end{align} which we used in Proposition \ref{prop:recovers-classical}.

\begin{remark}[Standard Operator Monotone Maximal Correlation Coefficients] Following the proof method we have used, it is not hard to establish that many other maximal correlation coefficients admit a Schmidt coefficient interpretation.
\begin{proposition}\label{prop:standard-operator-monotone-maximal-corr}
    Let $f \in \cM_{\text{St}}$ such that $f_{GM} \leq f \leq f_{AM}$. Then $\mu_{f}(A:B)_{\rho} \leq 1$. Moreover, it is the second Schmidt coefficient of $\wt{\rho}_{f} \coloneq (\mbf{J}_{f,\rho_{A}}^{-1/2} \otimes \mbf{J}_{f,\rho_{B}}^{-1/2})(\rho_{AB})$ with respect to $(\Herm(A), \langle \cdot , \cdot \rangle)$, $(\Herm(B), \langle \cdot, \cdot \rangle)$.
\end{proposition}
However, since the $\mbf{J}_{f,\sigma}^{-1/2}$ operator is not multiplicative for $f \in \cM_{\text{St}}\setminus\{f_{GM}\}$ by Proposition \ref{prop:multiplicativity-of-J}, it is unclear what use this identification has. For completeness, we provide the proof of Proposition \ref{prop:standard-operator-monotone-maximal-corr} in Appendix \ref{app-subsec:max-corr-coeff-lemmata}.
\end{remark}

\section{Quantum \texorpdfstring{$\chi^{2}$}{}-Divergences}\label{sec:quantum-chi-squared}
We now turn to our second application of our $L^{2}_{f}(\sigma)$ spaces, input-dependent contraction coefficients of quantum $\chi^{2}$-divergences. Just as for quantum maximal correlation coefficients, our interest is lifting the classical theory to the quantum setting. This will be partitioned into three goals. The first goal will be understanding the correspondence between quantum maximal correlation coefficients and $\chi^{2}$ contraction coefficients. This is entirely resolved by Theorem \ref{thm:correspondence-between-contraction-coeffs-and-max-corr-coeffs} and its corollaries. This theorem follows from identifying the quantum $\chi^{2}$-divergences as the variance of non-commutative likelihood ratios (Section \ref{subsec:chi-squared-variance}), using this correspondence to convert $\chi^{2}$-contraction coefficients to operator norms (Lemma \ref{lem:contraction-as-map-norm}), and then introducing the quantum generalizations of couplings (Definition \ref{def:quantum-coupling}). The second goal is to establish new conditions for the saturation of the data processing inequality of quantum $\chi^{2}$-contraction coefficients, which is done in Section \ref{sec:chi-square-extreme-values-and-recoverability}. These results rely upon the eigenvalue characterization of input-dependent $\chi^{2}$-contraction coefficients, which is known (see e.g. \cite{Cao-2019a}), but we have to slightly improve upon the analysis, and thus this is briefly presented in Section \ref{sec:refined-eigenvalues}. Finally, our third goal is to clarify lifting the theory of convergence of time-homogeneous Markov chains \cite{GZB-preprint-2024} to the quantum setting, which is done in Section \ref{sec:time-homogeneous-Markov-chains} and builds on recent results of \cite{Beigi-2025a}.

\subsection{\texorpdfstring{$\chi^{2}$}{}-Divergences, Non-commutative Variance, and Data Processing}\label{subsec:chi-squared-variance}
In this section we introduce the quantum $\chi^{2}_{f}$ divergences of \cite{Temme-2010a} as well as some relevant quantities for investigating the data processing inequality. Beyond providing background, the main contribution is to show they may be expressed as non-commutative variances, which will be useful to resolve their relation to quantum maximal correlation coefficients.

We begin by recalling the classical $\chi^{2}$-divergence: $\chi^{2}(p \Vert q) = \sum_{x} q(x)^{-1}(p(x)-q(x))^{2} = \Tr[q^{-1}(p-q)^{2}]$. It follows that we could define a quantum $\chi^{2}$-divergence as $\Tr[\sigma^{-1}(\rho-\sigma)^{2}]$, but this makes a specific choice of division by $\sigma$. In \cite{Temme-2010a}, the authors observed that one can replace $\sigma^{-1}$ with any non-commutative division by $\sigma$ such that monotonicity holds and thus introduced the quantum $\chi^{2}$-divergences indexed by operator monotone functions $f$:
\begin{align}\label{eq:chi-squared-f-def}
    \chi^{2}_{f}(\rho \Vert \sigma) \coloneq \langle \rho-\sigma, \mbf{J}_{f,\sigma}^{-1}(\rho-\sigma) \rangle \ .
\end{align}
We note that given Fact \ref{fact:monotone-metric}, as has been observed previously, this means each $\chi^{2}_{f}$ really is a monotone metric evaluated at $\rho-\sigma$ : $\chi^{2}_{f}(\rho \Vert \sigma) = \gamma_{f,\sigma}(\rho - \sigma,\rho-\sigma)$.

As highlighted in the introduction, an alternative way of looking at the $\chi^{2}$-divergence is as the variance of the likelihood ratio, i.e. defining $r(x) \coloneq \frac{p(x)}{q(x)}$, $\chi^{2}(p \Vert q) = \sum_{x} q(x)(\frac{p(x)}{q(x)} -1) = \mbb{E}_{q}[r - \mbb{E}_{q}[r]] = \Var_{q}[r]$. This has been useful classically \cite{Raginsky-2016a}. The following proposition shows the quantum $\chi^{2}_{f}$-divergence maintains this property in the non-commutative setting. To the best of our knowledge, unlike the other two statements in the following proposition, this was not known previously.
\begin{proposition}\label{prop:chi-squared-properties} For operator monotone $f$,
    \begin{enumerate}   
        \item (Data-Processing) Let $\rho \ll \sigma$ and $\cE$ be the adjoint of a unital Schwarz map, then
    \begin{align}
        \chi^{2}_{f}(\cE(\rho) \Vert \cE(\sigma)) \leq \chi^{2}_{f}(\rho \Vert \sigma) \ .
    \end{align}
        \item (Variance Relation) If $f \in \cM_{\text{St}}$ and $\rho \ll \sigma$,
        $\chi^{2}_{f,\sigma}(\rho \Vert \sigma) = \Var_{f,\sigma}[\mbf{J}_{f,\sigma}^{-1}(\rho)]$.
        \item (Ordering) If $0 \leq f_{1} \leq f_{2}$, then $\chi^{2}_{f_{2}}(\rho \Vert \sigma) \leq \chi^{2}_{f_{1}}(\rho \Vert \sigma)$.
    \end{enumerate}
\end{proposition}
\begin{proof}
    Define $A \coloneq \rho - \sigma$. Then $\chi^{2}_{f}(\rho \Vert \sigma) = \langle A , \mbf{J}^{-1}_{f,\sigma} A \rangle$. For Item 1, as $\rho \ll \sigma$ we may restrict to the support of $\sigma$ without changing the calculation so that in effect $\sigma \in \Density_+$.  Then
    \begin{align}
        \chi^{2}_{f}(\cE(\rho) \Vert \cE(\sigma)) = \langle \cE(A) , \mbf{J}_{f,\cE(\sigma)}^{-1}(\cE(A)) \rangle =& \gamma_{\cE(\sigma)}(\cE(A),\cE(A)) \\
        \leq& \gamma_{f,\sigma}(A,A) = \langle \rho - \sigma , \mbf{J}_{f,\sigma}^{-1}(\rho - \sigma)\rangle = \chi^{2}_{f}(\rho \Vert \sigma) \ ,
    \end{align}
    where we used Proposition \ref{prop:DPI-for-J-op}. We prove Item 2 via two direct calculations. The first calculation is
    \begin{align}
        \chi^{2}_{f,\sigma}(\rho \Vert \sigma) = \langle \rho - \sigma, \mbf{J}_{f,\sigma}^{-1}(\rho - \sigma) \rangle = \langle \rho - \sigma, \mbf{J}^{-1}_{f,\sigma}(\rho) - \mbb{1} \rangle = \langle \mbf{J}_{f,\sigma}^{-1}(\rho), \rho - \sigma \rangle \ , 
    \end{align}
    where the second equality used that $\mbf{J}_{f,\sigma}^{-1}$ is linear and how it acts on operators that commute with $\sigma$, the third equality is cancelling terms and that $\mbf{J}_{f,\sigma}^{-1}$ is Hermitian-preserving, which we know holds if and only if $f$ is symmetry-inducing (Proposition \ref{prop:symmetry-inducing-equivalences}). The second calculation is
    \begin{align}
        \Var_{f,\sigma}[\mbf{J}_{f,\sigma}^{-1}(\rho)] =& \langle \mbf{J}_{f,\sigma}^{-1}(\rho), \mbf{J}_{f,\sigma}^{-1}(\rho) \rangle_{\mbf{J}_{f,\sigma}} - \langle \mbb{1}, \mbf{J}_{f,\sigma}^{-1}(\rho) \rangle_{\mbf{J}_{f,\sigma}}\langle \mbf{J}_{f,\sigma}^{-1}(\rho), \mbb{1} \rangle_{\mbf{J}_{f,\sigma}} \\
        =& \langle \mbf{J}_{f,\sigma}^{-1}(\rho), \rho \rangle - \langle \mbb{1}, \rho \rangle\langle \mbf{J}_{f,\sigma}^{-1}(\rho), \sigma \rangle \\
        =&  \langle \mbf{J}_{f,\sigma}^{-1}(\rho), \rho - \sigma \rangle \ ,
    \end{align}
    where the first equality follows from \eqref{eq:variance-property-step-1}, the second follows from the definition of the inversion of a map, the last uses linearity and Hermitian-preserving properties of $\mbf{J}_{f,\sigma}^{-1}$ and the linearity of trace. Combining these equalities completes the proof of Item 2. Note \eqref{eq:variance-property-step-1} relies on $f$ being normalized and we appealed to symmetry-inducing, which is why this selects for the family of standard operator monotones. Item 3 follows from Proposition \ref{prop:ordering-of-NC-mult}.
\end{proof}
Note that by \eqref{eq:ordering-of-means} and \eqref{eq:ordering-of-standard-monotones}, Proposition \ref{prop:ordering-of-NC-mult}, and Item 3 of the above proposition, we have the following known ordering:
\begin{align}\label{eq:ordering-of-chi-squareds}
    \chi^{2}_{AM} \leq \chi^{2}_{LM} \leq \chi^{2}_{GM} \leq \chi^{2}_{HM} \quad \text{and} \quad \chi^{2}_{AM} \leq \chi^{2}_{f} \leq \chi^{2}_{HM} \quad \forall f \in \cM_{St} \ .
\end{align}
The arithmetic, logarithmic, and geometric mean $\chi^{2}_{f}$ divergences will have operational relevance established in Section \ref{sec:time-homogeneous-Markov-chains}. $\chi^{2}_{GM}$ will be established to be particularly central to the theory of data processing throughout this section.

To end this subsection, we recall the definitions of the input-dependent and input-independent contraction coefficients for any $\chi_{f}^{2}$-divergence:
\begin{align}
    \eta_{\chi^{2}_{f}}(\cE,\sigma) =& \sup_{\rho \in \Density \, : \rho \neq \sigma } \frac{\chi^{2}_{f}(\cE(\rho) \Vert \cE(\sigma))}{\chi^{2}_{f}(\rho \Vert \sigma)} \label{eq:input-dependent-contraction-coeff} \\
    \eta_{\chi^{2}_{f}}(\cE) =& \sup_{\sigma \in \Density} \eta_{\chi^{2}_{f}}(\cE,\sigma) \ .
\end{align}
These quantities study if the data processing inequality of $\chi^{2}_{f}$ is strict, because by definition $\chi^{2}_{f}(\cE(\rho) \Vert \cE(\sigma)) \leq \eta_{\chi^{2}_{f}}(\cE,\sigma)\chi^{2}_{f}(\rho \Vert \sigma)$. It is these quantities which the following subsections predominantly study. Note that by Item 2 of Proposition \ref{prop:chi-squared-properties}, for $f \in \cM_{\text{St}}$, these quantities are `really' measuring how the variance of the non-commutative likelihood ratio decreases under the action of a quantum channel.

\subsection{Correspondence to Quantum Maximal Correlation Coefficients}\label{sec:correspondence-to-q-max-corr-coeffs}
In this section, we establish the quantum generalization of the correspondence between the maximal correlation coefficient and the $\chi^{2}$ input-dependent contraction coefficient. Namely, we establish the following, which generalizes the classical case stated in \eqref{eq:correlation-to-contraction}.
\begin{tcolorbox}[width=\linewidth, sharp corners=all, colback=white!95!black, boxrule=0pt,frame hidden]
\begin{theorem}\label{thm:correspondence-between-contraction-coeffs-and-max-corr-coeffs}
    Let $f \in \cM_{\text{St}}$. Let $\sigma_{A} \in \Density_{+}(A)$ and $\cE_{A \to B}$ be a positive, trace-preserving map. Let $H_{AB} \coloneq \Omega_{\cE \circ \mbf{J}_{f,\sigma}}$ where the Choi operator is defined in the eigenbasis of $\sigma_{A}$. $H_{AB}$ is a relaxed quantum coupling (Definition \ref{def:quantum-coupling}), i.e. satisfies $H_{A} = \sigma_{A}$ and $H_{B} = \cE(\sigma)$. Then,
    \begin{align}
        \sqrt{\eta_{\chi^{2}_{f}}}(\cE,\sigma) = \mu_{f}(A:B)_{H} \ ,
    \end{align}
    where the RHS is the maximal correlation coefficient extended to relaxed quantum couplings (Definition \ref{def:q-max-corr-coeff-on-coupling}).
\end{theorem}
\end{tcolorbox}
By Proposition \ref{prop:nec-cond-for-CP-ness}, we know $H_{AB}$ cannot generally be positive semidefinite. However, as explained in Remark \ref{rem:QSOT-from-NC-Prob}, they could always be identified as ``quantum states over time" in the sense of the quantum foundations community as expounded in  \cite{Leifer_2013,Fullwood_2022,parzygnat2023time}. Nonetheless, whenever $H_{AB}$ is positive semidefinite, $H_{AB}$ is physical in the traditional sense of being a quantum state. An example of when the coupling is guaranteed to be a quantum state is when $f= f_{HM}$ as $\mbf{J}_{f_{HM},\sigma}(X) = \int_{0}^{\infty} \exp(-t \sigma^{-1}/2)X\exp(-t \sigma^{-1}/2)$ \cite{Petz-2011a}, so it is the integral of the action of completely positive maps and thus $H_{AB}$ is always positive semidefinite. The more standard case of when $\mbf{J}_{f,\sigma}$ is when $f=f_{GM}$, from which we obtain the following, which generalizes \cite{Cao-2019a} by relaxing $\cE_{A \to B}(\sigma)$ needing to be full rank, relaxing that $\cE_{A \to B}$ need be a quantum channel, and establishing that this correspondence holds on the full state space.

\begin{tcolorbox}[width=\linewidth, sharp corners=all, colback=white!95!black, boxrule=0pt,frame hidden]
\begin{corollary}\label{cor:contraction-for-sandwiched-case}
    Let $\sigma \in \Density_{+}(A)$ and $\cE_{A \to B}$ be a positive, trace-preserving map. Then
    \begin{align}\label{eq:GM-contraction-as-maximal-corr}
        \sqrt{\eta_{\chi^{2}_{GM}}(\cE,\sigma)} = \mu_{GM}(A:B)_{(\text{id}_{A} \otimes \cE)(\wt{\psi}_{\sigma})} \ ,
    \end{align}
    where $\wt{\psi}_{\sigma}$ is an arbitrary purification of $\sigma$. Furthermore, for any quantum state $\rho_{AB}$, there exists a quantum channel $\cE_{A \to B}$ such that
    \begin{align}
        \mu_{GM}(A:B)_{\rho_{AB}} = \sqrt{\eta_{\chi^{2}_{GM}}(\cE,\rho_{A})} \ , 
    \end{align}
    i.e. this correspondence always exists.
\end{corollary}
\end{tcolorbox}
\noindent Because of this generic correspondence and that $\mu_{GM}(A:B)$ tensorizes (special case of Lemma \ref{lem:k-correlation-tensorize}), it immediately follows $\eta_{\chi^{2}_{GM}}$ tensorizes when the maps are quantum channels.
\begin{tcolorbox}[width=\linewidth, sharp corners=all, colback=white!95!black, boxrule=0pt,frame hidden]
\begin{corollary}
    $\eta_{\chi^{2}_{GM}}$ tensorizes for CPTP maps. That is, for CPTP maps $\cE_{A_0 \to B_0}$, $\cF_{A_1 \to B_1}$, $\sigma_{A_0}$, $\rho_{A_1}$,
    $$\eta_{\chi^{2}_{GM}}(\cE \otimes \cF,\sigma \otimes \rho) = \max\{\eta_{\chi^{2}_{GM}}(\cE,\sigma), \eta_{\chi^{2}_{GM}}(\cF,\rho)\} \ . $$
\end{corollary}
\end{tcolorbox}
\noindent Furthermore, because the correspondence between $\mu_{GM}$ and $\eta_{\chi^{2}_{GM}}$ holds for the entire quantum state space due to Proposition \ref{prop:every-joint-state-is-a-degraded-purif}, we directly obtain a form of `duality.'
\begin{tcolorbox}[width=\linewidth, sharp corners=all, colback=white!95!black, boxrule=0pt,frame hidden]
\begin{corollary}\label{cor:duality-of-contraction}
    Let $\rho_{AB}$ be a quantum state. There exist unique quantum channels $\cN_{A \to B}$ and $\cM_{B \to A}$ such that $\rho_{AB} = (\id_{A} \otimes \cN)(\psi_{\rho}) = (\cM \otimes \id_{B})(\wt{\psi}_{\rho})$ where $\psi_{\rho}$, $\wt{\psi}_{\rho}$ are the canonical purifications of $\rho_{A}$ and $\rho_{B}$ respectively and 
    $$\eta_{\chi^{2}_{GM}}(\cN,\rho_{A}) = \eta_{\chi^{2}_{GM}}(\cM,\rho_{B}) \ . $$ 
\end{corollary}
\end{tcolorbox}
\noindent The above corollaries are further examples of obtaining stronger results when the operator monotone function $f$ is the geometric mean analogous to stronger results we saw in this setting in Sections \ref{subsubsec:q-max-corr-bounded-above-by-one} and \ref{subsec:strengthened-results-for-GM}.

In the remainder of this subsection, we establish the above results. To do this, we introduce the `Heisenberg time-reversal' map, which generalizes the adjoint of the Petz recovery map used in Section \ref{subsec:extreme-values-and-ACD}. We then use this to establish the input-dependent contraction coefficients as operator norms over the same spaces as quantum maximal correlation coefficients, but evaluated on different maps. By identifying these different maps as (relaxed) quantum couplings, we are able to establish Theorem \ref{thm:correspondence-between-contraction-coeffs-and-max-corr-coeffs}. The corollaries then follow from using the extra identifications that follow from the geometric mean as used in Section \ref{subsubsec:q-max-corr-bounded-above-by-one}.

\paragraph{Heisenberg Time Reversal Map} We begin by introducing the `Heisenberg Time Reversal Map.' This has implicitly appeared at least in \cite{Cao-2019a} where it was concatenated with another map. 
\begin{tcolorbox}[width=\linewidth, sharp corners=all, colback=white!95!black, boxrule=0pt,frame hidden]
\begin{definition}\label{def:Heisenberg-time-reversal-map}
For operator monotone $f$, positive, trace-preserving map $\cE_{A \to B}$, and state $\sigma_{A}$, we define the `Heisenberg time-reversal' map 
\begin{align}\label{eq:Heis-reversal-map}
    \cR_{f,\cE,\sigma} \coloneq \mbf{J}_{f,\cE(\sigma)}^{-1} \circ \cE \circ \mbf{J}_{f,\sigma} \ .
\end{align}
\end{definition}
\end{tcolorbox}
\noindent We remark a direct calculation may verify $\cR_{GM,\cE,\sigma} = \cP^{\ast}_{\cE,\sigma}$, i.e. for $f$ being the geometric mean, we recover the adjoint of the Petz recovery map (with respect to Hilbert-Schmidt inner product). The `value' in this class of maps is that they are expectation-preserving under evolution in the technical sense given the following proposition. The expectation-preserving property justifies the claim these are the `Heisenberg' time-reversal maps as they act on the observables rather than the state.
\begin{proposition}\label{prop:reversal-map-expectation-preserving}
    For any operator monotone function $f$, positive, trace-preserving map $\cE_{A \to B}$ and state $\sigma_{A}$,
    \begin{enumerate} 
        \item $\cR_{f,\cE,\sigma}$ is unital, and
        \item for all $X \in \Lin(A)$, $\mbb{E}_{\cE(\sigma)}[\cR_{f,\cE,\sigma}(X)] = \mbb{E}_{\sigma}[X]$, i.e. it is expectation-preserving under evolution.
    \end{enumerate}
\end{proposition}
\begin{proof}
For Item 1, $\cR_{f,\cE,\sigma}(\mbb{1}) = (\mbf{J}^{-1}_{f,\cE(\sigma)} \circ \cE)(\sigma) = (\mbf{J}^{-1}_{f,\cE(\sigma)})\cE(\sigma) = \mbb{1}$ where the first equality is \eqref{eq:Heis-reversal-map} and that $[\sigma,\mbb{1}] = 0$ and the third is $\cE(\sigma)$ commutes with itself. For Item 2,
\begin{equation}
    \begin{aligned}
        \mbb{E}_{\cE(\sigma)}[\cR_{\cE,\sigma}(X)] &= \langle \mbb{1}, \mbf{J}_{\cE(\sigma)}^{-1} \circ \cE \circ \mbf{J}_{f,\sigma}(X) \rangle_{f,\cE(\sigma)} \\
        &= \langle \mbb{1}, \cE \circ \mbf{J}_{f,\sigma}(X) \rangle \\
        &=  \langle \mbb{1}, \mbf{J}_{f,\sigma}(X) \rangle \\
        &= \mbb{E}_{\sigma}[X] \ ,
    \end{aligned}
\end{equation}
where the first equality is by definition, the second is the definition of the inner product, the third is that the adjoint of $\cE$ is unital, and the last is by definition.
\end{proof}

\paragraph{$\chi^{2}$-Contraction Coefficients as Optimization Problems} We now characterize input-dependent contraction coefficients as operator norms on the Heisenberg time reversal maps and optimization programs that look similar to maximal correlation coefficients. This is the main technical lemma of this section, and only relies on $\cE$ being trace-preserving and positive.
\begin{lemma}\label{lem:contraction-as-map-norm}
    Let $f \in \cM_{\text{St}}$. Let $\cE_{A \to B}$ be a positive, trace-preserving map and $\sigma \in \Density_{+}(A)$.
    Then 
    $$\sqrt{\eta_{\chi^{2}_{f}}}(\cE,\sigma) = \Vert \cR_{\cE,\sigma} : \Herm_{0,\sigma}(A') \to \Herm_{0,\cN(\sigma)}(B') \Vert_{f,\sigma \to f,\cE(\sigma)} \ , $$ where the latter may be written as
    \begin{equation}\label{eq:contraction-coefficient-as-optimization}
        \begin{aligned}
            \max & \; \Big \vert \Tr[X \otimes Y^{\ast} \Omega_{\cE \circ \mbf{J}_{f,\sigma}}] \Big \vert \\
            \text{s.t.} \; & \langle \mbb{1}, X \rangle_{f,\sigma_{A}} = \langle \mbb{1}, Y \rangle_{f,\cE(\sigma)} = 0 \\
        \; &  \langle X, X \rangle_{f,\sigma_{A}} = \langle Y, Y \rangle_{f,\cE(\sigma)} = 1 \ ,
        \end{aligned}
    \end{equation}
    where the maximization is over $X \in \Herm(A')$, $Y \in \Herm(B')$ and \sloppy $A' \coloneq \linspan(\supp(\sigma)), B' \coloneq \linspan(\supp(\cE(\sigma))$.
\end{lemma}
\begin{proof}[Proof of Theorem \ref{lem:contraction-as-map-norm}]
The following proof may be seen as a non-commutative generalization of the proof of \cite[Theorem 3.2]{Raginsky-2016a}. The basic idea is to convert things into a form that we can identify it as an operator norm and then apply Lemma \ref{lem:map-norms-as-optimizations}.

Because we keep the function $f$ fixed throughout the proof, we will omit it when it is clear without it for notational simplicity.
\begin{align}
    \eta_{\chi^{2}_{f}}(\cE,\sigma) =& \sup_{\rho \in \Density \, : \rho \neq \sigma } \frac{\chi^{2}_{f}(\cE(\rho) \Vert \cE(\sigma))}{\chi^{2}_{f}(\rho \Vert \sigma)} \\
    =& \sup_{\rho \in \Density \, : \rho \neq \sigma} \frac{\Var_{\cE(\sigma)}[\mbf{J}_{\cE(\sigma)}^{-1} \circ \cE(\rho)]}{\Var_{\sigma}[\mbf{J}_{\sigma}^{-1}(\rho)]} \\
    =& \sup_{\substack{A \in \Herm \setminus \{0\}: \\ \mbb{E}_{\sigma}[A] = 1}} \frac{\Var_{\cE(\sigma)}[\cR_{\cE,\sigma}(A)]}{\Var_{\sigma}[A]}
    \ , 
\end{align}
where the equality is Item 2 of Proposition \ref{prop:chi-squared-properties} and the third equality is justified via
\begin{align}
    \langle \mbb{1}, \rho \rangle = 1 \iff \langle \mbb{1}, \mbf{J}_{f,\sigma} \circ \mbf{J}_{f,\sigma}^{-1} (\rho) \rangle = 1 \iff \langle \mbb{1}, \mbf{J}_{f,\sigma}^{-1} (\rho) \rangle_{f,\sigma} = 1 \ , 
\end{align}
and then defining $A \coloneq \mbf{J}_{f,\sigma}^{-1}(\rho)$ so that $\rho = \mbf{J}_{f,\sigma}(A)$. That the expectation is one is by Proposition \ref{prop:mult-and-div-under-trace} and that it is Hermitian follows from $\mbf{J}_{\sigma}^{-1}$ being Hermitian-preserving, which holds as we assumed $f$ is symmetry-inducing.

We now aim to replace the supremisation over $A$ with $H \in \Herm_{0,\sigma}(A) \setminus \{0\}$, where
\begin{align}\label{eq:orthogonal-Herm-space}
    \Herm_{0,\sigma}(A) = \{H \in \Herm(A): \mbb{E}_{\sigma}[H] = 0\} \ .
\end{align}
We do this by showing equivalence of supremising over either set. Let $H \in \Herm_{0,\sigma}(A)$, let $A_{H} \coloneq \mbb{1}_{A} + \ve H$ for any $\ve > 0$. It follows from the normalization of the density matrix $\mbb{E}_{\sigma}[A_{H}] = 1$.  We will now show $\Var_{\sigma}[A_{H}] = \ve^{2}\Var_{\sigma}[H]$ and $\Var_{\cE(\sigma)}[\cR_{\cE,\sigma}(A_{H})] = \ve^{2}\Var_{\sigma}[\cR_{\cE,\sigma}(H)]$. To establish the first equality,
\begin{align}
    \Var_{\sigma}[A_{H}] &= \langle A_{H}, A_{H} \rangle_{f,\sigma} - 1 \\
    &= \langle \mbb{1}, \mbb{1} \rangle_{f,\sigma} + \ve \langle \mbb{1}, H\rangle_{f,\sigma} + \ve \langle H,\mbb{1} \rangle_{f,\sigma} + \ve^{2} \langle H , H \rangle_{f,\sigma} - 1 \\
    &= 0 + \ve^{2} \langle H , H \rangle_{f,\sigma} \\
    &= \ve^{2}\Var_{\sigma}[H] \  ,
\end{align}
where: the first equality is Item 1 of Proposition \ref{prop:variance-properties}; the second is expanding terms, the third is canceling terms, Proposition \ref{prop:mult-and-div-under-trace}, and $H \in \Herm_{0,\sigma}(A)$; and the final equality uses $\ve^{2} \cdot 0 = 0$, Item 1 of Proposition \ref{prop:variance-properties}, and Proposition \ref{prop:mult-and-div-under-trace}.
To establish the second equality,
\begin{equation}\label{eq:AH-var-1}
\begin{aligned}
    &\langle \cR_{\cE,\sigma}(A_{H}), \cR_{\cE,\sigma}(A_{H}) \rangle_{f,\cE(\sigma)} \\
    =&  \langle \mbb{1} + \ve\cR_{\cE,\sigma}(H), \mbb{1} + \ve \cR_{\cE,\sigma}(H) \rangle_{f,\cE(\sigma)} \\
    =&  1 + \ve\langle \mbb{1}, \cR_{\cE,\sigma}(H) \rangle_{f,\cE(\sigma)} + \ve\langle \cR_{\cE,\sigma}(H), \mbb{1} \rangle_{f,\cE(\sigma)} + \ve^{2} \langle \cR_{\cE,\sigma}(H), \cR_{\cE,\sigma}(H) \rangle_{f,\cE(\sigma)} \ , 
\end{aligned}
\end{equation}
where the first equality is the linearity of $\cR_{\cE,\sigma}$ and that it is unital and the second is expanding terms. Now,
\begin{align}\label{eq:AH-var-2}
    \langle \mbb{1}, \cR_{\cE,\sigma}(H) \rangle_{f,\cE(\sigma)} = \mbb{E}_{\cE(\sigma)}[\cR_{\cE,\sigma}(H)] = \mbb{E}_{\sigma}[H] = 0 \ ,
\end{align}
where the first equality is definition, the second is by Proposition \ref{prop:reversal-map-expectation-preserving}, and the last is by our assumption on $H$, i.e. the definition of $\Herm_{0,\sigma}(A)$. Similarly,
\begin{align}\label{eq:AH-var-3}
    \langle \cR_{\cE,\sigma}(H), \mbb{1} \rangle_{f,\cE(\sigma)} = \langle \cR_{\cE,\sigma}(H), \cE(\sigma) \rangle =  \Tr[\cE \circ \mbf{J}_{f,\sigma} (H)] = \Tr[\mbf{J}_{f,\sigma}(H)] = 0 \ ,
\end{align}
where the first equality is the property of $\mbf{J}_{\cE(\sigma)}$, the second is the definition of $\cR_{\cE,\sigma}$ and Proposition \ref{prop:mult-and-div-under-trace}, the third is that $\cE^{\ast}$ is unital, and the last is $\mbb{E}_{\sigma}[H]=0$ by assumption. Combining these,
\begin{align*}
    \Var_{\cE(\sigma)}[\cR_{\cE,\sigma}(A_{H})] =& \langle \cR_{\cE,\sigma}(A_{H}), \cR_{\cE,\sigma}(A_{H}) \rangle_{f,\cE(\sigma)} - \left( \langle \mbb{1}, \cR_{\cE,\sigma}(A_{H}) \rangle_{f,\cE(\sigma)} \right)^{2} \\
    =& \langle \cR_{\cE,\sigma}(A_{H}), \cR_{\cE,\sigma}(A_{H}) \rangle_{f,\cE(\sigma)} - \left( \mbb{E}_{\sigma}[A_{H}] \right)^{2} \\
    =& (1 + \ve^{2} \langle \cR_{\cE,\sigma}(H), \cR_{\cE,\sigma}(H) \rangle_{f,\cE(\sigma)}) - 1 \\
    =&  \ve^{2} \langle \cR_{\cE,\sigma}(H), \cR_{\cE,\sigma}(H) \rangle_{f,\cE(\sigma)} \\
    =& \ve^{2}\Var_{\sigma}[\cR_{\cE,\sigma}(H)] \ ,
\end{align*}
where the first equality is Item 1 of Proposition \ref{prop:variance-properties}, the second is Proposition \ref{prop:reversal-map-expectation-preserving}, the third is construction of $A_{H}$ and combining \eqref{eq:AH-var-1},\eqref{eq:AH-var-2},\eqref{eq:AH-var-3} , the fourth is canceling terms, and the fifth is subtracting $0^{2}$, using that $\mbb{E}_{\sigma}[H] = 0$ and Item 1 of Proposition \ref{prop:variance-properties}. Thus, 
$$\sup_{\substack{A \in \Herm \setminus \{0\}: \\ \mbb{E}_{\sigma}[A] = 1}} \frac{\Var_{\cE(\sigma)}[\cR_{\cE,\sigma}(A)]}{\Var_{\sigma}[A]} \geq \sup_{H \in \Herm_{0}(A) \setminus \{0\}} \frac{\Var_{\cE(\sigma)}[\cR_{\cE,\sigma}(A)]}{\Var_{\sigma}[H]} \ . $$
We now prove the other containment. Let $A \in \Herm$ such that $\mbb{E}_{\sigma}[A] = 1$. Note this means $A \neq 0$. Define $G_{A} \coloneq A - \mbb{1}$. Then $\Var_{\sigma}[G_{A}] = \Var_{\sigma}[A]$ follows from a direct calculation using the definition of the inner product and Propositions \ref{prop:variance-properties} and \ref{prop:mult-and-div-under-trace}. One may show that $\Var_{\cE(\sigma)}[\cR_{\cE,\sigma}(G_{A})] = \Var_{\cE(\sigma)}[\cR_{\cE,\sigma}(A)]$ in a similar manner although using the same sort of identities as were used to establish $\Var_{\cN(\sigma)}[\cR_{\cE,\sigma}(A_{H})] =\ve^{2}\Var_{\sigma}[\cR_{\cE,\sigma}(H)]$. Thus, we have the other containment and may conclude
\begin{align}
    \eta_{\chi^{2}_{f}}(\cE,\sigma) =& \sup_{H \in \Herm_{0,\sigma}(A) \setminus \{0\}} \frac{\Var_{\cE(\sigma)}[\cR_{\cE,\sigma}(H)]}{\Var_{\sigma}[H]} \ .
\end{align}
Note that by Proposition \ref{prop:variance-properties} and the definition of $\Herm_{0}(A)$, this means $\Var_{\sigma}[H] = \Vert H \Vert^{2}_{\mbf{J}_{f,\sigma}}$. Moreover, as $\cR_{\cE,\sigma}$ is expectation-preserving (Proposition \ref{prop:reversal-map-expectation-preserving}), by Proposition \ref{prop:variance-properties}, $\Var_{\cE(\sigma)}[\cR_{\cE,\sigma}(H)] = \Vert \cR_{\cE,\sigma}(H) \Vert^{2}_{\mbf{J}_{f,\cE(\sigma)}}$. Thus, 
\begin{align}\label{eq:chi-squared-f-norm-ratio}
    \sqrt{\eta_{\chi^{2}_{f}}(\cE,\sigma)} =& \sup_{H \in \Herm_{0,\sigma}(A) \setminus \{0\}} \frac{\Vert \cR_{\cE,\sigma}(H) \Vert_{f,\cE(\sigma)}}{\Vert H \Vert_{f,\sigma}} \ ,
\end{align}
where we have moved the square out via monotonicity of the square function. As $\Herm_{0}(A)$ is a normed space and $\cR_{\cE,\sigma}: \Herm_{0,\sigma}(A) \to \Herm_{0,\cN(\sigma)}(B)$ by Proposition \ref{prop:reversal-map-expectation-preserving}, we may conclude that the RHS of the above equation is in fact the norm of $\cR_{\cE,\sigma}$ as a map between these two normed spaces (See \eqref{eq:map-norm-ratio-form}). Applying Lemma \ref{lem:map-norms-as-optimizations}, we obtain
\begin{align}
    \sqrt{\eta_{\chi^{2}_{f}}(\cE,\sigma)} =& \sup\Big\{ \Big \vert  \langle y, \cR_{\cE,\sigma}(x) \rangle_{f,\sigma} \Big \vert \, : \substack{ \, \Vert x \Vert_{f,\sigma} = 1 \, , \, \Vert y \Vert_{f,\cE(\sigma)} = 1 \\ x \in \Herm_{0,\sigma}(A) \, , \, y \in \Herm_{0,\cE(\sigma)}(B)}  \Big\} \ .
\end{align}
Finally, we simplify the objective:
\begin{align}
    \langle y, \cR_{\cE,\sigma}(x) \rangle_{f,\sigma} =& \langle y, \cE \circ \mbf{J}_{f,\sigma}(x) \rangle = \Tr[x^{\Trans} \otimes y^{\ast} \Omega_{\cE \circ \mbf{J}_{f,\sigma}}] \ ,
\end{align}
where the first equality is definition of $\cR$ and the inner product, the second is the action of a channel in terms of the Choi operator. Finally, we may remove the transpose on $x$ by choosing to define the transpose with respect to the basis of $\sigma$. This implies $\Tr[x \sigma] = \Tr[x^{\Trans} \sigma^{\Trans}] = \Tr[x^{\Trans} \sigma]$ and using \eqref{eq:MC-form} one may verify that $\langle x, \mbf{J}_{f,\sigma}(x) \rangle = \langle x^{\Trans}, \mbf{J}_{f,\sigma}(x^{\Trans}) \rangle$ for this choice. This completes the proof.
\end{proof}

\paragraph{Quantum Couplings and Establishing Correspondence} By looking at Definition \ref{def:f-quantum-max-corr} and Eq.~\eqref{eq:contraction-coefficient-as-optimization}, it immediately follows that $\sqrt{\eta_{\chi^{2}_{f}}}(\cE,\sigma)$ is directly equal to evaluating the $f$-maximal correlation coefficient on a state if and only if $\Omega_{\cE \circ \mbf{J}_{f,\sigma}} = \rho_{AB}$ such that $\rho_{A} = \sigma$ and $\rho_{B} = \cE(\sigma)$. In effect, this is asking when $\Omega_{\cE \circ \mbf{J}_{f,\sigma}} = \rho_{AB}$ is a a quantum generalization of a coupling, where recall that classically a coupling of two distributions $p_{X}$ and $q_{Y}$ is a joint distribution $r_{XY}$ such that $r_{X} = p_{X}$ and $r_{Y} = q_{Y}$. For this reason, we provide quantum generalizations of this notion.
\begin{definition}[Quantum Couplings]\label{def:quantum-coupling}
    Let $\rho_{A}$ and $\sigma_{B}$ be quantum states. A \textit{relaxed coupling} of $\rho_{A}$ and $\sigma_{B}$ is a Hermitian operator $H_{AB}$ such that $H_{A} = \rho_{A}$ and $H_{B} = \rho_{B}$. If a relaxed coupling is positive semidefinite, then we say it is a \textit{coupling}.
\end{definition}

Next we identify $\Omega_{\cE \circ \mbf{J}_{f,\sigma}}$ in Lemma \ref{eq:map-norm-as-optimization} as always being at least a relaxed coupling.
\begin{proposition}[$f$-couplings]\label{prop:f-couplings}
    Let $f \in \cM_{\text{St}}$, $\cE_{A \to B}$ be a positive, trace-preserving map, and $\sigma_{A} \in \Density(A)$. Define the Choi operator in the eigenbasis of $\sigma_{A}$. Then $\Omega_{\cE \circ \mbf{J}_{f,\sigma}}$ is a relaxed coupling of $\sigma_{A}$ and $\cE(\sigma)_{B}$.
\end{proposition}
\begin{proof}
    Let $\sigma = \sum_{i} \lambda_{i} \dyad{\nu_{i}}$. Then 
    \begin{align}
        \Omega_{\cE \circ \mbf{J}_{f,\sigma}} = \sum_{i,j} \ket{\nu_{i}}\bra{\nu_{j}} \otimes \cE(P_{f}(\lambda_{i},\lambda_{j})\ket{\nu_{i}}\bra{\nu_{j}})
        = \sum_{i,j} P_{f}(\lambda_{i},\lambda_{j})\ket{\nu_{i}}\bra{\nu_{j}} \otimes \cE(\ket{\nu_{i}}\bra{\nu_{j}}) \ ,
    \end{align}
    where we used the perspective function is a scalar to move it to the other term. Using the two equal representations respectively,
    \begin{align}
        \Tr_{A}[ \Omega_{\cE \circ \mbf{J}_{f,\sigma}}] &= \sum_{i,j} \Tr[\ket{\nu_{i}}\bra{\nu_{j}}] \otimes \cE(P_{f}(\lambda_{i},\lambda_{j})\ket{\nu_{i}}\bra{\nu_{j}}) = \sum_{i} \cE(P_{f}(\lambda_{i},\lambda_{i})\ket{\nu_{i}}\bra{\nu_{i}}) = \cE(\sigma_{A})  \\
        \Tr_{B}[ \Omega_{\cE \circ \mbf{J}_{f,\sigma}}] &= \sum_{i,j} P_{f}(\lambda_{i},\lambda_{j})\ket{\nu_{i}}\bra{\nu_{j}} \otimes \Tr[\cE(\ket{\nu_{i}}\bra{\nu_{j}})] = \sum_{i} P_{f}(\lambda_{i},\lambda_{i})\dyad{\nu_{i}} = \sigma_{A}  \ , 
    \end{align}
    where in each case we used linearity of $\cE$ and that $f(1) = 1$. In the second case, we used that $\cE$ is trace-preserving. 

    Finally, by Proposition \ref{prop:symmetry-inducing-equivalences}, since $f$ is symmetry-inducing, $\mbf{J}_{f,\sigma}$ is Hermitian-preserving. As $\cE$ is also Hermitian-preserving, $\cE \circ \mbf{J}_{f,\sigma}$ is Hermitian-preserving. As the Choi operator of a Hermitian-preserving map is Hermitian \cite{WatrousBook}, $\Omega_{\cE \circ \mbf{J}_{f,\sigma}}$ is Hermitian. This completes the proof.
\end{proof}

As $\Omega_{\cE \circ \mbf{J}_{f,\sigma}}$ is generally merely a relaxed coupling, we extend the definition of the maximal correlation coefficient to relaxed couplings.
\begin{definition}\label{def:q-max-corr-coeff-on-coupling}
    Let $f:$ be a normalized, operator monotone function. Let $H_{AB}$ be a relaxed quantum coupling of $\rho_{A}$ and $\rho_{B}$. The $f$-maximal correlation coefficient of the relaxed quantum coupling is
    \begin{equation}
    \begin{aligned}
    \label{eq:f-quantum-max-corr-on-coupling}
        \mu_{f}(A:B)_{H} \coloneq \max \; & \vert \Tr[X \otimes Y^{\ast}H_{\wt{A}\wt{B}}] \vert \\
        \text{s.t.} \; & \langle \mbb{1}, X \rangle_{f,\rho_{\wt{A}}} = \langle \mbb{1}, Y \rangle_{f,\rho_{\wt{B}}} = 0 \\
        \; &  \langle X, X \rangle_{f,\rho_{\wt{A}}} = \langle Y, Y \rangle_{f,\rho_{\wt{B}}} = 1 \ ,
    \end{aligned}
    \end{equation}
    where $\wt{A} \coloneq \supp(\rho_{A})$, $\wt{B} \coloneq \supp(\rho_{B})$, and the maximization is over $X \in \Herm(\wt{A})$, $Y \in \Herm(\wt{B})$.
\end{definition}
We note the following.
\begin{corollary}\label{cor:q-max-corr-coeff-props-on-relaxed-couplings}
    The quantum maximal correlation coefficients evaluated on relaxed quantum couplings are invariant under local isometries and satisfy data processing under local quantum channels.
\end{corollary}
\begin{proof}
    A direct calculation shows that if $H_{AB}$ is a coupling of $\rho_{A},\sigma_{B}$, then $(\cE \otimes \cF)(H_{AB})$ is a coupling of $\cE(\rho),\cF(\sigma)$. As we defined the generalization as a priori restricting to the support of the marginals, one can use all of the same proofs as in Section \ref{subsec:quantum-maximal-corr-coeff-and-DPI}. 
\end{proof}

Finally, we combine the above points to establish Theorem \ref{thm:correspondence-between-contraction-coeffs-and-max-corr-coeffs}.
\begin{proof}[Proof of Theorem \ref{thm:correspondence-between-contraction-coeffs-and-max-corr-coeffs}]
    Considering \eqref{eq:map-norm-as-optimization} and Definition \ref{def:f-quantum-max-corr}, if $\Omega_{\cE \circ \mbf{J}_{f,\sigma}}$ were a relaxed coupling, then \eqref{eq:map-norm-as-optimization} would be evaluating the maximal correlation coefficient on a relaxed coupling as defined in Definition \ref{def:quantum-coupling}. By Proposition \ref{prop:f-couplings}, this is indeed the case. This completes the proof.
\end{proof}

\begin{remark}[$f$-Couplings as Quantum States over Time]\label{rem:QSOT-from-NC-Prob}
    There is a somewhat extensive literature on ``quantum states over time" (see \cite{Leifer_2013,Fullwood_2022,parzygnat2023time} and references therein). At a formal level, the idea boils down to there existing a mapping from the space of Choi operators of quantum channels and unipartite quantum states, i.e. $\star: \Channel(A,B) \times \Density(A) \to \Herm(A \otimes B)$ where the image of this map is identified as the space of `quantum states over time.' It is not hard to show some pre-existing proposals are simply using a specific choice of $f$. To see this, first note that a direct calculation will verify $\mbf{J}_{f,\rho_{A} \otimes I} = \mbf{J}_{f,\rho_{A}} \otimes \id_{B}$. Then, as examples, using that $\mbf{J}_{f_{AM},\rho}[X] = \frac{1}{2}(\rho_{A}X + X\rho_{A})$ and $\mbf{J}_{f_{GM},\rho}[X] = \rho_{A}^{1/2}X\rho_{A}^{1/2}$, we may conclude that the $\star$ operation of Fullwood and Parzygnat is $\Omega_{\cE} \star \rho_{A} \coloneq \mbf{J}_{f_{AM},\rho} \otimes \id_{B}[\Omega_{\cE}]$ and the corresponding $\star$ operation of Leifer-Spekkens \cite{Leifer_2013} would be $\Omega_{\cE} \star \rho_{A} \coloneq \mbf{J}_{f_{GM},\rho_{A}}\otimes \id_{B}[\Omega_{\cE}]$. One can generalize this construction by picking $f \in \cM_{\text{St}}$ and defining $\star_{f}$ by the action $\Omega_{\cE} \star_{f} \rho_{A} = \mbf{J}_{f,\rho \otimes I}(\Omega_{\cE})$. Such a generalization relates to our $f$-couplings by noting
    \begin{align}
        \Omega_{\cE \circ \mbf{J}_{f,\rho}} = (\id_{A} \otimes \cE \circ \mbf{J}_{f,\rho_{A}})(\Phi^{+}) = (\mbf{J}_{f,\rho_{A}} \otimes \cE)(\Phi^{+}) = (\mbf{J}_{f,\rho_{A}} \otimes \id_{B})(\Omega_{\cE}) = \mbf{J}_{f,\rho_{A} \otimes I}(\Omega_{\cE}) \ ,
    \end{align}
    where we moved $\mbf{J}_{f,\rho_{A}}$ using that the Choi operator is defined in the eigenbasis of $\rho_{A}$. In other words, $f$-couplings may be viewed as formally motivated `quantum states over time,' although we are using the basis choice to avoid a transpose.\footnote{We thank Afham for asking us questions that made it clear to us this is what the quantum states over time researchers are working on.} 
    
    Lastly, we remark that as each standard operator monotone $f$ induces an alternative notion of Petz recovery map (see Definition \ref{def:Schrod-reversal}), this is likely the direct way of identifying well-defined quantum generalizations of Bayes' rule from `time-reversal maps' in \cite{parzygnat2023time}.
\end{remark}

\paragraph{Strengthened Correspondence for \texorpdfstring{$\eta_{\chi^{2}_{GM}}$}{}}
Here we provide the proof of Corollary \ref{cor:contraction-for-sandwiched-case} which was not explained at the start of the section.

\begin{proof}[Proof of Corollary \ref{cor:contraction-for-sandwiched-case}]
   The first claim follows from Theorem \ref{thm:correspondence-between-contraction-coeffs-and-max-corr-coeffs}, the identification in \eqref{eq:choi-of-GM-J-operator}, that purifications are equivalent under local isometries (cf.~\cite{WatrousBook}), and the geometric mean maximal correlation coefficient is invariant under local isometries even when extended to relaxed couplings (Corollary \ref{cor:q-max-corr-coeff-props-on-relaxed-couplings}). The `furthermore' statement holds because the identification in \eqref{eq:identification-of-rhoAB-via-GM} always exists for some choice of quantum channel by Proposition \ref{prop:every-joint-state-is-a-degraded-purif}.
\end{proof}

\begin{remark}\label{remark:bounding-contraction-coeff}
    In Section \ref{subsubsec:q-max-corr-bounded-above-by-one}, we bounded the maximal correlation coefficients evaluated on quantum states for $f \leq f_{GM}$ above by one using the adjoint of the Petz recovery map. It is not hard to see that proof method can be generalized using the Heisenberg time-reversal map in Definition \ref{def:Heisenberg-time-reversal-map}. However, for this to work for \textit{all} quantum states, one would need to relate $\rho_{AB}$ to the coupling induced by $\mbf{J}_{f,\rho_{A}}$ as done in \eqref{eq:identify-rho-AB-with-Jfgm} for $f = f_{GM}$. It does not seem possible to do this in general, which is why it is important to work with $f_{GM}$ and use the ordering between maximal correlation coefficients. 
\end{remark}

\subsection{Contraction Coefficients, the Schr\"{o}dinger Time Reversal Map, and Eigenvalues}\label{sec:refined-eigenvalues}

While, as seen in Lemma \ref{lem:contraction-as-map-norm}, the Heisenberg time reversal map is useful for clarifying the relation between $f$-maximal correlation coefficients, it turns out using the ``Schr\"{o}dinger" picture of time-reversal instead has its own uses. Namely, it allows us to identify the contraction coefficient $\eta_{\chi^{2}_{f}}(\cE,\sigma)$ as the eigenvalue of the Schr\"{o}dinger time-reversal map applied to $\cE(\sigma)$. This will be useful in studying the saturation of the data processing inequality in the subsequent subsection. This eigenvalue characterization was investigated in \cite{Cao-2019a}, and the eigenvalue characterization is more generally known, see e.g. \cite{Lesniewski-1999a,Temme-2010a, hiai2016contraction}. However, we refine \cite[Lemma 10]{Cao-2019a}, which will be useful for our main results, as well as clarify the relation between the Schr\"{o}dinger and Heisenberg time-reversal maps. For these reasons, we present this material briefly, but highlight what has been done previously so as to not overemphasize our contribution in this regard.

\paragraph{Schr\"{o}dinger Time Reversal Map (Recovery Maps)} We begin with the definition. 
\begin{tcolorbox}[width=\linewidth, sharp corners=all, colback=white!95!black, boxrule=0pt,frame hidden]
\begin{definition}\label{def:Schrod-reversal}
For an operator monotone $f$, positive, trace-preserving $\cE_{A \to B}$ and state $\sigma_{A}$, we define the ``Schr\"{o}dinger Time Reversal" map:
\begin{align}\label{eq:Schrod-reversal-map}
    \cS_{f,\cE,\sigma} \coloneq \mbf{J}_{f,\sigma} \circ \cE^{\ast} \circ \mbf{J}^{-1}_{f,\cE(\sigma)} \ .
\end{align}
\end{definition}
\end{tcolorbox}
\noindent Operationally, the `value' of these maps is they return $\cE(\sigma)$ to $\sigma$. It is in this sense that they are `Schr\"{o}dinger' time reversal maps as they act on the state rather than the observables. It is also in this sense that they can be viewed as recovery maps. In fact, $\cS_{GM,\cE,\sigma} = \cP_{\cE,\sigma}$, i.e. the Petz recovery map as may be verified via \eqref{eq:Petz-recovery-defn}. Given $\cR_{GM,\cE,\sigma} = \cP^{\ast}_{\cE,\sigma}$, this would suggest there exists a relation to the Heisenberg time reversal maps in general. These ideas are captured in the following proposition. 
\begin{proposition}\label{prop:schrodinger-reversal-map}
    For any operator monotone function $f$, positive map $\cE_{A \to B}$, and state $\sigma_{A}$,
    \begin{enumerate}
        \item If $\cE(\sigma)>0$, $\cS_{f,\cE,\sigma}(\cE(\sigma)) = \sigma$ 
        \item If $\sigma > 0$, $\cS_{f,\cE,\sigma}$ is the adjoint map of $\cR_{f,\cE,\sigma}$ with respect to HS inner product,
        \item If $\sigma > 0$, $\cS_{f,\cE,\sigma}$ is a trace-preserving map.
    \end{enumerate}
\end{proposition}
\begin{proof}
    Item 1 follows $ (\mbf{J}_{f,\sigma} \circ \cE^{\ast} \circ \mbf{J}^{-1}_{f,\cE(\sigma)})(\cE(\sigma)) = (\mbf{J}_{f,\sigma} \circ \cE^{\ast})(\mbb{1}) = \mbf{J}_{f,\sigma}(\mbb{1}) = \sigma$ where we used that $\cE^{\ast}$ is unital and that $\cE(\sigma),\sigma$ commute with themselves. 
    
    Item 2 follows from
    $$ \langle X , \cS_{f,\cE,\sigma}(Y) \rangle = \langle X , (\mbf{J}_{f,\sigma} \circ \cE^{\ast} \circ \mbf{J}^{-1}_{f,\cE(\sigma)})(Y) \rangle = \langle (\mbf{J}^{-1}_{f,\cE(\sigma)} \circ \cE^{\ast} \circ \mbf{J}_{f,\sigma})(X),Y \rangle = \langle \cR_{f,\cE,\sigma}(X),Y \rangle \ , $$
    where the first equality is \eqref{eq:Schrod-reversal-map}, the second is the definition of adjoint map used iteratively as well as the maps $\mbf{J}_{f,\sigma},\mbf{J}_{f,\sigma}^{-1}$ being self-adjointed (Proposition \ref{prop:J-operator-self-adjoint}), and the second equality is \eqref{eq:Heis-reversal-map}. Item 3 follows from Item 2 and that Proposition \ref{prop:reversal-map-expectation-preserving} shows $\cR_{f,\cE,\sigma}$ is unital.
\end{proof}

\paragraph{\texorpdfstring{$\chi^{2}$}{}-Contraction and Eigenvalues}
We now can state other characterizations of $\chi^{2}_{f}$ input-dependent contraction coefficients. 
\begin{lemma}\label{lem:contraction-coeff-to-eig}Let $\cE(\sigma) > 0,\sigma \in \Density_{+}(A)$, and $f$ be a symmetry-inducing operator monotone. Consider $\cS_{f,\cE,\sigma} \circ \cE: (\Herm(A), \langle \cdot , \cdot \rangle_{f,\sigma}^{\star}) \to (\Herm(A), \langle \cdot , \cdot \rangle_{f,\sigma}^{\star})$.
\begin{align}
    \eta_{\chi^{2}_{f}}(\cE,\sigma) = \lambda_{2}( \cS_{f,\cE,\sigma} \circ \cE) = \Vert \cE: \Herm_{0}(A) \to \Herm_{0}(B) \Vert_{\mbf{J}_{f,\sigma}^{-1} \to \mbf{J}_{f,\cE(\sigma)}^{-1}}^{2} \ ,
\end{align}
where $\Herm_{0}(A) \coloneq \{H \in \Herm(A) : \Tr[H] = 0\}$. Moreover, $\lambda_{1}(\cS_{f,\cE,\sigma} \circ \cE) = 1$ with eigenvector $\sigma$ and the remaining eigenvectors are traceless, Hermitian operators and the norm may be expressed as
\begin{equation}\label{eq:contraction-coefficient-as-another-optimization}
        \begin{aligned}
            \max & \; \Big \vert \Tr[X \otimes Y^{\ast} \Omega_{\mbf{J}_{f,\cE(\sigma)}^{-1}\circ \cE}] \Big \vert \\
            \text{s.t.} 
        \; &  \langle X, X \rangle_{f,\sigma_{A}}^{\star} = \langle Y, Y \rangle_{f,\cE(\sigma)}^{\star} = 1 \ \\
        \; & \Tr[X] = \Tr[Y] = 0 \ , 
        \end{aligned}
    \end{equation}
    where $X,Y$ are Hermitian.
\end{lemma}

Technically, Lemma \ref{lem:contraction-coeff-to-eig} is a slightly improved version of \cite[Lemma 10]{Cao-2019a}. This is because our version establishes that the relevant space is only the Hermitian operators, which will be useful subsequently, as well as provides equivalent expressions in terms of an operator norm and an optimization problem. To do this, our proof method slightly diverges from that in \cite{Cao-2019a}, but as the proof is long, less novel, and not crucial to understand for the rest of the work, we present it in Appendix \ref{subsec:lemmata-for-chi-square-contraction}. We do however note the following in the main text, which is used to establish the above, will be relevant later, and is slight generalizations of most items of \cite[Lemma 9]{Cao-2019a} with direct proofs.
\begin{proposition}\label{prop:Schrod-map-composed-properties}
    For operator monotone function $f$, $\sigma \in \Density_{+}(A)$, and $\cE_{A \to B}$ being the adjoint of a unital Schwarz map,
    \begin{enumerate}[itemsep=0pt]
        \item $\cS_{f,\cE,\sigma} \circ \cE$ is a positive semidefinite operator on Hilbert space $(\Lin(A),\langle \cdot , \cdot \rangle_{f,\sigma}^{\star})$. 
        \item $\sqrt{\cS_{f,\cE,\sigma} \circ \cE} \leq \id_{A}$ as a map on Hilbert space $(\Lin(A),\langle \cdot , \cdot \rangle_{f,\sigma}^{\star})$.
    \end{enumerate}
\end{proposition}
\begin{proof}
    First,
    \begin{align}
        \langle X , \cS_{f,\cE,\sigma} \circ \cE(X) \rangle_{f,\sigma}^{\star} = \langle X, \cE^{\ast} \circ J_{f,\cE(\sigma)}^{-1} \circ \cE(X) \rangle = \langle \cE(X), J_{f,\sigma}^{-1} \circ \cE(X) \rangle = \langle \cE(X), \cE(X) \rangle_{f,\sigma}^{\star} \geq 0 \ ,
    \end{align}
    where we used the definitions of the operators and that it is an inner product. This completes the first item. Second, as $\cS_{f,\cE,\sigma} \circ \cE$ is positive semidefinite, it is self-adjoint and admits a square root operator $\sqrt{\cS_{f,\cE,\sigma} \circ \cE} \eqqcolon \cM$.
    \begin{align}
        \Vert \cM(X) \Vert_{f,\sigma}^{\ast} = \sqrt{\langle X , \cS_{f,\cE,\sigma} \circ \cE(X) \rangle_{f,\sigma}^{\ast}} = \sqrt{\langle X, \cE^{\ast} \circ J_{f,\cE(\sigma)}^{-1} \circ \cE(X) \rangle} \leq \sqrt{\langle X, J_{f,\sigma}^{-1} (X)\rangle} = \Vert X \Vert_{f,\sigma}^{\star} \ ,
    \end{align}
    where the inequality is Proposition \ref{prop:DPI-for-J-op}.
\end{proof}

Furthermore, an unfortunate aspect of the above lemma is the support constraints on $\cE(\sigma)$. This can generally be avoided. This is because we can obtain the relation to the eigenvalue from Theorem \ref{lem:contraction-as-map-norm} and Item 3 of Proposition \ref{prop:schrodinger-reversal-map}. Arguably this suggests the `Heisenberg picture' is more fundamental. As the proof is just moving terms around in the inner product, the proof is also provided in Appendix \ref{subsec:lemmata-for-chi-square-contraction}, but we state the result here.
\begin{corollary}\label{cor:Schroding-er-map-without-rank-constraints}
    Let $f \in \cM_{\text{St}}$ and satisfy the conditions of Proposition \ref{prop:suff-conds-for-restricting-support}. Then for arbitrary positive, trace-preserving map $\cE_{A \to B}$ and $\sigma \in \Density_{+}(A)$, $\eta_{\chi^{2}_{f}}(\cE,\sigma)$ is the second eigenvalue of $\cS_{f,\cE,\sigma} \circ \cE: (\Herm(A), \langle \cdot , \cdot \rangle_{f,\sigma}^{\star}) \to (\Herm(A), \langle \cdot , \cdot \rangle_{f,\sigma}^{\star})$ where the first eigenvalue is always one with eigenvector $\sigma$, and the remaining eigenvectors are traceless.
\end{corollary}

 Before moving forward, we suggest what is likely the `best' interpretation of the lemma. The norm optimization program in Lemma \ref{lem:contraction-coeff-to-eig} is provided to highlight it does not relate to $f$-maximal correlation coefficients through the Schr\"{o}dinger picture, and thus the eigenvalue characterization is the most interesting characterization in the Schr\"{o}dinger picture. Next, as noted, we can view $\cS_{f,\cE,\sigma}$ as a `recovery map' with respect to a given inner product, so it is a measure of reversibility, albeit not in a physical sense as $\cS_{f,\cE,\sigma}$ need not be CP. However, by Proposition \ref{prop:schrodinger-reversal-map}, it is the adjoint map of $\cE$ on the Hermitian operators for the given inner product. Then, $\lambda_{2}(\cS_{f,\cE,\sigma} \circ \cE)$ is the square of the second largest singular value of $\cE$ on the Hermitian operators. It thus inherits the traditional, geometric intuition of the singular value: $\eta_{\chi^{2}_{f}}(\cE,\sigma)$ is measuring how $\cE$ shrinks a basis of the Hermitian operators under the inner product $\langle \cdot, \cdot \rangle_{f,\sigma}^{\star}$. The subsequent section will further build upon this idea.

\subsection{Extreme Values, Saturation of Data Processing, and Recoverability}\label{sec:chi-square-extreme-values-and-recoverability}
In this subsection, we are interested in when the data processing inequality is saturated. We begin this investigation by studying the conditions such that $\eta_{\chi^{2}}(\cE,\sigma) \in \{0,1\}$, i.e. the quantity saturates its extreme values. The study of $\eta_{\chi^{2}}(\cE,\sigma)=1$ will motivate the study of the $\chi^{2}_{f}$-divergences saturating the data processing inequality. Equivalent conditions to saturating data processing for each operator monotone $f$ is presented in Theorem \ref{thm:DPI-with-equality}. This is likely the result of broadest interest in this section. This theorem further implies results of previous works \cite{Gao-2023-sufficient-fisher,Jencova-2017a} as shown in Corollaries \ref{cor:suff-for-chi-sq} and \ref{cor:suff-for-SRD}. At the core of the section is the idea that while each $\chi^{2}_{f}$-divergence has its own special recovery map (the Schr\"{o}dinger Time Reversal map), the $\chi^{2}_{GM}$ recovery map (the Petz recovery map) is uniquely central to the theory of recoverability as has been known for decades \cite{petz1986sufficient}.

\paragraph{Contraction Coefficient Extreme Values}

We begin with simple facts about the bounds. Namely, under what conditions we can guarantee the contraction coefficient is within the intuitive parameter range $[0,1]$. Also note that when $\cE$ is a replacer channel, all information is destroyed, so Item 3 of the following proposition captures that the contraction coefficient is non-zero unless all information is destroyed.
\begin{proposition}\label{prop:f-chi-contraction-extreme-values} In all conditions below, let $f$ be an operator monotone, $\sigma \in \Density$ and $\cE$ be positive and trace-preserving. 
    \begin{enumerate}
        \item $0 \leq \eta_{\chi^{2}_{f}}(\cE,\sigma)$.
        \item If $\sigma \in \Density_+$ and (a) $\cE$ is the adjoint of a unital Schwarz map or (b) $f \in \cM_{\text{St}}$ satisfies Proposition \ref{prop:suff-conds-for-restricting-support}, then $\eta_{\chi^{2}_{f}}(\cE,\sigma) \leq 1$.
        \item If $\cE(\sigma) \in \Density_{+}$ or $f$ satisfies Proposition \ref{prop:suff-conds-for-restricting-support}, $\eta_{\chi^{2}_{f}}(\cE,\sigma) = 0$ if and only if $\cE$ is a replacer channel.
        \end{enumerate}
\end{proposition}
\begin{remark}\label{rem:bounding-other-max-corr-coeffs}
    Note Item 2 of the above cannot be used to establish $\mu_{f}(A:B)_{\rho} \leq 1$ for $f \leq f_{GM}$ as the correspondence between contraction coefficients and maximal correlation coefficients (Theorem \ref{thm:correspondence-between-contraction-coeffs-and-max-corr-coeffs}) generally uses relaxed quantum couplings (Definition \ref{def:quantum-coupling}) and does not generically cover the space of all joint quantum states for other choices of operator monotone $f$.
\end{remark}
\begin{proof}
    Item 1 follows from $\chi^{2}_{f,\cE(\sigma)} = \gamma_{f,\cE(\sigma)}(\cE(\rho-\sigma),\cE(\rho-\sigma)) \geq 0$ where the inequality is the non-negativity of $\langle A , A \rangle^{\star}_{f,\cE(\sigma)}$. \\

    To establish Item 2, in the case $\cE$ is the adjoint of a unital Schwarz map and $\sigma \in \Density_+$, DPI for $\chi^{2}_{f}$ holds (Item 1 of Proposition \ref{prop:chi-squared-properties}), so it follows via definition. Otherwise one must appeal to Corollary \ref{cor:Schroding-er-map-without-rank-constraints} which guarantees it is upper bounded by one. \\

    To establish Item 3, if $\cE$ is a replacer channel $\cR_{\tau}$, then 
    $$\chi^{2}_{f}(\cR_{\tau}(\rho) \Vert \cR_{\tau}(\sigma)) = \chi^{2}_{f}(\tau \Vert \tau) = \gamma_{f,\tau}(\tau - \tau, \tau - \tau) = \gamma_{f,\tau}(0,0) = 0 \ . $$
    On the other hand, if $\eta_{\chi^{2}_{f}}(\cE,\sigma) = 0$, then any sequence of density matrices $(\rho_{n})_{n}$ must be such that 
    $$\lim_{n \to \infty} \chi^{2}_{f}(\cE(\rho_{n}) \Vert \cE(\sigma)) = \lim_{n \to \infty} \gamma_{f,\cE(\sigma)}(\cE(\rho_{n})-\cE(\sigma),\cE(\rho_{n})-\cE(\sigma)),  \ . $$
    If $\cE(\sigma) \in \Density_+$, then by the non-negativity of a metric, $\gamma_{f,\cE(\sigma)}(X,X) = 0$ if and only if $X = 0$. Thus, all possible supremizing sequences are such that $\lim_{n \to \infty} \cE(\rho_{n}) = \cE(\sigma)$. We may therefore conclude $\cE$ maps all density matrices to $\cE(\sigma)$ as otherwise we would have a contradiction. Therefore, $\cE$ is $\cR_{\tau}$ where $\tau = \cE(\sigma)$. Similarly, if $f$ satisfies Proposition \ref{prop:suff-conds-for-restricting-support}, then we are only interested in $\Pi_{\supp(\cE(\sigma))}\cE(\rho_{n})\Pi_{\supp(\cE(\sigma))}$, so the same argument goes through. 
\end{proof}

We will now show that when considering the extreme values, $\chi^{2}_{GM}$ is the most important $\chi^{2}_{f}$ contraction coefficient as it captures the rest being equality as well. This will make use of the following lemma that relies on the eigenvalue interpretation of the contraction coefficient (Lemma \ref{lem:contraction-coeff-to-eig}). 
\begin{tcolorbox}[width=\linewidth, sharp corners=all, colback=white!95!black, boxrule=0pt,frame hidden]
\begin{lemma}\label{lem:chi-squared-contraction-unity}
    Let $\sigma \in \Density_{+}(A)$ and $f$ be symmetry-inducing. If (a) $\cE(\sigma) \in \Density_{+}$ and $\cE$ is the adjoint of a unital Schwarz map or (b) $f \in \cM_{\text{St}}$ satisfies Proposition \ref{prop:suff-conds-for-restricting-support} and $\cE$ is positive, then $\eta_{\chi^{2}_{f}}(\cE,\sigma) = 1$ if and only if there exists $\Density(A) \ni \rho \neq \sigma$ such that $\rho \ll \sigma$ and $(\cS_{f,\cE,\sigma} \circ \cE)(\rho) = \rho$.
\end{lemma}
\end{tcolorbox}
\begin{proof}
     If there exists $\rho \neq \sigma \in \Density$ such that $(\cS_{f,\cE,\sigma} \circ \cE)(\rho) = \rho$ then, by \eqref{eq:contraction-Schrodinger-Jf-inv},
    \begin{align*}
        1 \geq \eta_{\chi^{2}_{f}}(\cE,\sigma) = \sup_{\rho \neq \sigma \in \Density} \frac{\langle \rho - \sigma, (\cS_{f,\cE,\sigma} \circ \cE)(\rho - \sigma) \rangle_{f,\sigma}^{\star}}{\langle \rho - \sigma, \rho - \sigma \rangle_{f,\sigma}^{\star}} \geq \frac{\langle \rho - \sigma, \rho - \sigma \rangle_{f,\sigma}^{\star}}{\langle \rho - \sigma, \rho - \sigma \rangle_{f,\sigma}^{\star}} = 1 \ ,
    \end{align*}
    so this is a sufficient condition. To see it is also necessary, assume $\eta_{\chi^{2}_{f}}(\cE,\sigma) = 1$. We start with the case $\cE(\sigma) \in \Density_+$. By Lemma \ref{lem:contraction-coeff-to-eig}, this means there exists $X \in \Herm(A)$ that is traceless such that $(\cS_{f,\cE,\sigma} \circ \cE)(X) = X$. We will use this to build $\rho \neq \sigma$. First, we note that $X$ is linearly independent of $\sigma$ as we verify by contradiction with orthonormality. Assume $X = \lambda \sigma$ for some $\lambda \neq 0$. Then,
    \begin{align}
        \langle X , \sigma \rangle_{f,\sigma}^{\star} = \langle X , \Pi_{\supp(\sigma)} \rangle = \lambda \Tr[\sigma] \neq 0 \ , 
    \end{align}
    which contradicts orthonormality. Define $\zeta \coloneq \Vert K \Vert_{\infty}$ and let $\omega$ be such that $\lambda_{\min}(\sigma) > \omega > 0$ which exists as $\sigma$ is full rank. Then we have the following implications:
    \begin{align}
        \mbb{1}_{A} \geq \frac{1}{\zeta}X \geq -\mbb{1}_{A} \Rightarrow \lambda_{\min}(\sigma)\mbb{1}_{A} > \omega \mbb{1}_{A} \geq \frac{\omega}{\zeta}X \Rightarrow \lambda_{\min}(\sigma)\mbb{1}_{A} - \frac{\omega}{\zeta}X \geq 0 \Rightarrow \sigma - \frac{\omega}{\zeta}X \geq 0 \ .
    \end{align}
    Thus, we could define $\ve \coloneq \frac{\omega}{\zeta} > 0$ and  $\rho \coloneq (\sigma - \ve X)$, which is unit trace as $X$ is traceless. Note that $\rho \neq \sigma$ as $X$ and $\sigma$ are linearly independent and $\ve > 0$. Finally,
    \begin{align}
        (\cS_{f,\cE,\sigma} \circ \cE)(\rho) =& (\cS_{f,\cE,\sigma} \circ \cE)(\sigma - \ve X) = \sigma - \ve X = \rho    \ ,
    \end{align}
    where we just used the definition of $\rho$ and the linearity of the map. This completes the proof when $\cE(\sigma) \in \Density_+(B)$. In the case $f$ satisfies Proposition \ref{prop:suff-conds-for-restricting-support}, one does the same argument, but appeals to Corollary \ref{cor:Schroding-er-map-without-rank-constraints}.
\end{proof}

\begin{theorem}\label{thm:contraction-coeff-equal-unity}
   Let $\sigma \in \Density_{+}(A)$. Let $\cE$ be 2-positive and TP such that $\cE(\sigma) > 0$ then $\eta_{\chi^{2}_{f}}(\cE,\sigma)=1$ for all operator monotone $f$ if and only if $\eta_{\chi^{2}_{GM}}(\cE,\sigma)= 1$. Similarly, for $\cE$ positive and TP, $\eta_{\chi^{2}_{GM}}(\cE,\sigma) =1$ if and only if $\eta_{\chi^{2}_{f}}(\cE,\sigma)=1$ for all $f \in \cM_{\text{St}}$ satisfying Proposition \ref{prop:suff-conds-for-restricting-support}.
\end{theorem}
\begin{proof}
    In either case, the more general condition is trivial, so we just need to prove $\eta_{\chi^{2}_{GM}}(\cE,\sigma)=1$ implies the same for the others. Assume  $\eta_{\chi^{2}_{GM}}(\cE,\sigma)=1$ is the case. By Lemma \ref{lem:chi-squared-contraction-unity}, this implies there is $\sigma \neq \rho \ll \sigma$, such that $(\cP_{\cE,\sigma} \circ \cE)(\rho) = \rho$. Note that $\cP_{\cE,\sigma}$ is a CPTP map on the support of $\sigma$ \cite[Section 12.3]{Wilde-Book}. Thus, 
    \begin{align}
        \eta_{\chi^{2}_{f}} \geq \frac{\chi^{2}_{f}(\cE(\rho)\Vert \cE(\sigma))}{\chi^{2}_{f}(\rho \Vert \sigma)} \geq  \frac{\chi^{2}_{f}((\cP_{\cE,\sigma} \circ \cE)(\rho)\Vert (\cP_{\cE,\sigma} \circ \cE)(\sigma))}{\chi^{2}_{f}(\rho \Vert \sigma)} = 1 \ , 
    \end{align}
    where the first inequality that the contraction coefficient is a supremum and the second inequality is DPI. Using Item 2 of Proposition \ref{prop:f-chi-contraction-extreme-values} completes the proof.
\end{proof}

\paragraph{Saturation of Data Processing Inequality} Note that one may be less interested in the contraction coefficient than when the DPI inequality of $\chi^{2}_{f}$ divergence itself is saturated. This also follows from the eigenvalue characterization in Lemma \ref{lem:contraction-coeff-to-eig}.
\begin{tcolorbox}[width=\linewidth, sharp corners=all, colback=white!95!black, boxrule=0pt,frame hidden]
\begin{theorem}\label{thm:DPI-with-equality}
    Let $f$ be a symmetry-inducing operator monotone, $\rho \ll \sigma$, and $\cE$ be a 2-positive, trace-preserving map. Then $\chi^{2}_{f}(\cE(\rho)\Vert\cE(\sigma)) = \chi^{2}_{f}(\rho \Vert \sigma)$ if and only $(\cS_{f,\cE,\sigma} \circ \cE)(\rho) = \rho$. Moreover, for $f = f_{GM}$, the same holds where $\cE$ need only be positive, TP.
\end{theorem}
\end{tcolorbox}
\begin{proof}[Proof of Theorem \ref{thm:DPI-with-equality}]
    For the first direction, if $(\cS_{f,\cE,\sigma} \circ \cE)(\rho)=\rho$ where $\rho \ll \sigma$, then
    \begin{align}
        \chi^{2}_{f}(\cE(\rho) \Vert \cE(\sigma)) = \langle \rho - \sigma, \cS_{f,\cE,\sigma}(\rho-\sigma)\rangle_{f,\sigma}^{\star} = \langle \rho - \sigma , \rho - \sigma \rangle_{f,\sigma}^{\star} = \chi^{2}_{f}(\rho \Vert \sigma) \ ,
    \end{align}
    where the first equality is by the same argument as obtaining the numerator of \eqref{eq:contraction-Schrodinger-Jf-inv}, the second equality is Item 1 of Proposition \ref{prop:schrodinger-reversal-map} where we can treat $\sigma$ as full rank as $\rho \ll \sigma$, and the last equality is by the same argument as obtaining the denominator of \eqref{eq:contraction-Schrodinger-Jf-inv}.
    
     For the second direction, let $\chi^{2}_{f}(\cE(\rho) \Vert \cE(\sigma)) = \chi^{2}_{f}(\rho \Vert \sigma)$. Define $A \coloneq \rho - \sigma$. Then 
    \begin{align}\label{eq:saturating-DPI-step-1}
        1 = \frac{\chi^{2}_{f}(\cE(\rho) \Vert \cE(\sigma))}{\chi^{2}_{f}(\rho \Vert \sigma)} = \frac{\langle A , (\cS_{f,\cE,\sigma} \circ \cE)(A)\rangle_{f,\sigma}^{\star}}{\langle A , A \rangle_{f,\sigma}^{\star}}
    \end{align}
    Re-arranging, we have $\langle A, \cS_{f,\cE,\sigma} \circ \cE(A) \rangle_{f,\sigma}^{\star} = \langle A, A \rangle_{f,\sigma}^{\star}$. As $\cS_{f,\cE,\sigma} \circ \cE \leq \id_{A}$ on the respective space (Proposition \ref{prop:Schrod-map-composed-properties}), $A$ is an eigenvector of $\cS_{f,\cE,\sigma} \circ \cE$ with an eigenvalue of one. Thus,
    \begin{align}
        \rho - \sigma = A = (\cS_{f,\cE,\sigma} \circ \cE)(A) = (\cS_{f,\cE,\sigma} \circ \cE)(\rho) - \sigma \ , 
    \end{align}
    where the last equality uses Item 1 of Proposition \ref{prop:schrodinger-reversal-map} as we may treat $\sigma$ as full rank as $\rho \ll \sigma$. Then canceling the $\sigma$ on either side completes the proof.
\end{proof}

An application of Theorem \ref{thm:DPI-with-equality} is that it implies in some cases equality holding for all divergences of a given family is equivalent to just the corresponding $\chi^{2}$-divergence satisfying equality, which is equivalent to the Petz recovery map preserving some different state. We collect these corollaries here.

First, we obtain the following, which may be seen as recovering part of \cite[Theorem 8]{jenvcova2012reversibility}. As such, it also may be seen as a variant of \cite[Corollary 4.7]{Gao-2023-sufficient-fisher} with more general choices of dynamics.
\begin{tcolorbox}[width=\linewidth, sharp corners=all, colback=white!95!black, boxrule=0pt,frame hidden]
\begin{corollary}\label{cor:suff-for-chi-sq}
    Let $\cE$ be the adjoint of a unital Schwarz map and $\rho \ll \sigma$. The following are equivalent:
    \begin{enumerate}
        \item $\chi^{2}_{f}(\cE(\rho)\Vert \cE(\sigma)) = \chi^{2}_{f}(\rho \Vert \sigma)$ for all operator monotone $f$
        \item $\chi_{GM}^{2}(\cE(\rho)\Vert \cE(\sigma)) = \chi_{GM}^{2}(\rho \Vert \sigma)$
        \item $(\cP_{\cE,\sigma} \circ \cE)(\rho) = \rho$.
    \end{enumerate}
\end{corollary}
\end{tcolorbox}
\begin{proof}
    $1) \Rightarrow 2)$ as it is more specific. $2) \Rightarrow 3)$ by Theorem \ref{thm:DPI-with-equality}, and $3) \Rightarrow 1)$ by the data processing inequality (Proposition \ref{prop:DPI-for-J-op}).
\end{proof}

Second we obtain the following, which says the Sandwiched R\'{e}nyi divergences $\wt{D}_{\alpha}(\rho \Vert \sigma)$ (see Definition \ref{def:SRD}) saturate the data processing inequality under a positive map if and only if the $\chi^{2}_{GM}$ divergence does. To the best of our knowledge, this result is new, but may also be obtained by combining  the results of \cite{Jencova-2017a} and \cite{Muller-2017a}. This is not necessarily surprising as the proof of Theorem \ref{thm:DPI-with-equality} is similar to that of \cite[Lemma 8]{Jencova-2017a}.
\begin{tcolorbox}[width=\linewidth, sharp corners=all, colback=white!95!black, boxrule=0pt,frame hidden]
\begin{corollary}\label{cor:suff-for-SRD}
    Let $\cE$ be a positive, trace-preserving map and $\rho \ll \sigma$. The following are equivalent:
    \begin{enumerate}
        \item $\wt{D}_{\alpha}(\cE(\rho) \Vert \cE(\sigma)) = \wt{D}_{\alpha}(\rho \Vert \sigma)$ for all $\alpha \in (1/2,1)\cup(1,\infty)$
         \item $\chi_{GM}^{2}(\cE(\rho)\Vert \cE(\sigma)) = \chi_{GM}^{2}(\rho \Vert \sigma)$
        \item $(\cP_{\cE,\sigma} \circ \cE)(\rho) = \rho$.
    \end{enumerate}
\end{corollary}
\end{tcolorbox}
\begin{proof}
    $1) \Rightarrow 2)$ follows from the fact that 1) implies the equality for $\wt{D}_{2}$ and a direct calculation will  verify $\chi^{2}_{GM}(\rho \Vert \sigma) = \exp(\wt{D}_{2}(\rho \Vert \sigma)) - 1$. $2) \Rightarrow 3)$ follows from Theorem \ref{thm:DPI-with-equality}. $3) \Rightarrow 1)$ follows from the DPI for positive, trace-preserving maps \cite{Muller-2017a}.
\end{proof}

\subsection{Convergence Rates of Time Homogeneous Markov Chains}\label{sec:time-homogeneous-Markov-chains}
In this subsection, we consider what the previous sections imply about time homogeneous Markov chains, which can be natural noise models in large depth communication networks and memory devices. We extend results of \cite{GZB-preprint-2024} to the quantum setting for the recent Hirche-Tomamichel quantum $f$-divergences making use of a result in \cite{Beigi-2025a}. This shows another way in which they seem the `correct' or `ideal' quantum generalization of $f$-divergences (Theorem \ref{thm:mixing-rate}). We also obtain bounds on the rate of data processing for Sandwiched R\'{e}nyi divergences \cite{Wilde-2014a,Muller-Lennert-2013a} in terms of $\eta_{\wt{\chi}^{2}}(\cE,\pi)$. Before presenting the results, we begin with a detailed motivation that explains what this section aims to resolve.

\subsubsection{Motivation: Mixing Times and Rate of Data Processing}
Recall that in the previous section, we saw $\eta_{\chi^{2}_{f}}(\cE,\sigma) = 1$ if and only if the dynamics of $\cE_{A \to B}$ are physically reversible, via $\cP_{\cE,\sigma}$, for two distinct quantum states, namely $\sigma$ and $\rho \neq \sigma$. Now assume that $\sigma \coloneq \pi$ is the \textit{unique} fixed point of $\cE_{A \to A}$.\footnote{Recall there is at least one fixed point \cite[Theorem 4.24]{WatrousBook}.} It is worthwhile to know how quickly $\cE^{n} \coloneq \circ_{i \in [n]} \cE$ will take any input $\rho$ to $\pi$ as this places limits on distinguishability. Bounds on how fast such a process occurs is known as mixing times and has applications in computer science for sampling algorithms \cite{levin2017markov,Temme-2010a}. We can obtain the following bounds on the mixing time, which is effectively a straightforward variant of \cite[Theorem 9]{Temme-2010a} making use of Proposition \ref{prop:Df-lb-in-chi-square} that generalizes classical results of \cite{GZB-preprint-2024} (See Appendix \ref{subsec:time-homogeneous-MC-lemmata} for the proof).
\begin{proposition}\label{prop:L1-mixing-times}
    Let $\pi \in \Density(A)$ be the unique, full rank fixed point of $\cE$. For any $f \in \cM_{\text{St}}$ and $n \in \mbb{N}$,
    \begin{align}
        \Vert \cE^{n}(\rho) - \pi \Vert_{1} \leq \eta_{\chi^{2}_{f}}(\cE,\pi)^{n/2}\sqrt{\chi^{2}_{f}(\rho,\pi)} \leq \sqrt{\frac{2}{\lambda_{\min}(\pi)}}\eta_{\chi^{2}_{f}}(\cE,\pi)^{n/2} \ .
    \end{align}
    In other words, for $\Vert \cE^{n}(\rho) - \pi \Vert_{1} \leq \delta$ for all $\rho \in \Density(A)$, it suffices for $n \geq \log(\frac{2}{\delta^{2}\lambda_{\min}})/\log(1/\eta_{\chi^{2}_{f}}(\cE,\pi))$.
\end{proposition}

This previous proposition tells us that $\chi^{2}_{f}$-divergence contraction coefficients are useful for bounding the loss of distinguishability under the $L_{1}$-norm, which is relevant to hypothesis testing via the Holevo-Helstrom theorem. However, note that, at least asymptotically, most information processing tasks are characterized by quantities induced by the relative entropy. It thus is perhaps more important to understand mixing times in terms of the relative entropy or other monotonic divergences. Recently, \cite{GZB-preprint-2024} considered the problem of measuring the rate at which a classical time homogeneous Markov chain $\cW^{\circ n}$ with unique fixed point distribution $\pi$ under $f$-divergences. They showed for a large class of $f$-divergences, if the rate is to be considered in terms of contraction coefficients, the rate of contraction is bounded above by $\eta_{\chi^{2}}(\cW,\pi)$. This is appealing as it is efficient to compute and is the fastest it could be using contraction coefficients. \cite{GZB-preprint-2024} also extended this to Petz $f$-divergences \cite{HIAI_2011,Hiai-2017a} in terms of what can now be seen to be $\chi^{2}_{HM}(\cE,\pi)$, but with no promises that it is fastest or it is efficient to compute. This exemplifies our poor understanding of quantum time homogeneous Markov chains and quantum $f$-divergences. This subsection aims to make use of the tools we developed in Section \ref{sec:quantum-chi-squared} to resolve this.

\subsubsection{Further Background on Families of Quantum Divergences}
For clarity, we briefly present further background and standard definitions on families of quantum divergences we will use subsequently. Classically, the central quantity for (asymptotic) information theory is the relative entropy $D(p \Vert q)$. There exist two main families that generalize this quantity. The first family is the $f$-divergences $D_{f}(p \Vert q) \coloneq \sum_{x} q(x) f\left(\frac{p(x)}{q(x)} \right)$ where $f$ is a convex function. These all satisfy the data processing inequality and recover the relative entropy by setting $f = t\log(t)$ (See \cite{polyanskiy-2023a} for further detail). The second family are the R\'{e}nyi divergences $D_{\alpha}(p \Vert q) \coloneq \frac{1}{\alpha-1}\log(\sum_{x} p(x)^{\alpha}q(x)^{1-\alpha})$ where $\alpha \in (0,1) \cup (1,\infty)$. These also all satisfy the data processing inequality and recover the relative entropy by taking the limit $\alpha \to 1$ (See \cite{Tomamichel-2016a} for further information). 

In quantum information theory, as quantum states $\rho$ and $\sigma$ do not necessarily commute and thus cannot be simultaneously diagonalized, one can construct multiple generalizations of the R\'{e}nyi and $f$-divergences that recover the classical case in the same manner as there are multiple quantum $\chi^{2}$-divergences. We will be interested in the following families of quantum divergences, which all satisfy the data processing inequality (see citations for further information).
\begin{definition} (Petz $f$-divergences \cite{Petz-1986a,HIAI_2011, Hiai-2017a}) \label{def:Petz-f-div} Let $f:(0,\infty) \to \mbb{R}$ be a convex function, $\rho \ll \sigma \in \Density(A)$. The Petz $f$-divergence is
\begin{align}
    \ol{D}_{f}(\rho \Vert \sigma) \coloneq \langle \sigma^{1/2}, f(L_{\rho}R_{\sigma^{-1}})\sigma^{1/2} \rangle \ , 
\end{align}
where $L_{\rho},R_{\sigma^{-1}}$ are the left and right multiplication operators \eqref{eq:left-and-right-mult-operators} and the function of these is well defined (Proposition \ref{prop:func-of-mod-op}).
\end{definition}

\begin{definition}
    (Hirche-Tomamichel $f$-divergences \cite{Hirche-2024a}) \label{def:HT-divergences} Let $f:(0,\infty) \to \mbb{R}$ be convex, twice-differentiable, and satisfy $f(1) = 0$. Then for quantum states $\rho$ and $\sigma \in \Density(A)$, the HT $f$-divergence is
    \begin{align}
        D_{f}(\rho \Vert \sigma) \coloneq \int_{1}^{\infty} f''(\gamma)E_{\gamma}(\rho \Vert \sigma) + \gamma^{-3}f''(\gamma^{-1})E_{\gamma}(\sigma \Vert \rho) d\gamma \ 
    \end{align}
    whenever the integral is finite and otherwise set to $+\infty$ where, for any $\gamma \geq 1$,
    \begin{align}
        E_{\gamma}(\rho \Vert \sigma) \coloneq \Tr[\rho - \gamma \sigma]_{+} - (\Tr[A - \gamma \sigma])_{+} \ , 
    \end{align}
    and $A_+$ denotes the positive part of the eigen-decomposition of a Hermitian operator $A$.
\end{definition}
\begin{definition}(Sandwiched R\'{e}nyi Divergences \cite{Wilde-2014a,Muller-2017a}) \label{def:SRD} Let $\alpha \in (0,1) \cup (1,\infty)$ and $\rho,\sigma \in \Density(A)$. Then
\begin{align}
    D_{\alpha}(\rho \Vert \sigma) \coloneq \begin{cases}
        \frac{1}{\alpha-1}\log\Vert \sigma^{\frac{1-\alpha}{2\alpha}}\rho\sigma^{\frac{1-\alpha}{2\alpha}} \Vert^{\alpha}_{\alpha} & \text{if } (\alpha < 1 \wedge \Tr[\rho \sigma] \neq 0) \vee \rho \ll \sigma \\
        +\infty & \text{otherwise} \ ,
    \end{cases}
\end{align}
where $\Vert \cdot \Vert_{\alpha}$ denotes the $\alpha$-Schatten norm (see e.g. \cite{Tomamichel-2016a}).
\end{definition}

What is important for us is the relation between these families of divergences and specific instances of the quantum $\chi^{2}_{f}$-divergences. Classically, the $\chi^{2}$-divergence may be recovered as the `local behaviour' of any $f$-divergence \cite{Sason-2018a} (see also \cite{polyanskiy-2023a}). Quantitative bounds on approximating an $f$-divergence with the $\chi^{2}-$divergence for most $f$-divergences was established in \cite{GZB-preprint-2024}, which ultimately implies it controls the rate at which a classical, time-homogeneous Markov chain converges according to a large family of contraction coefficients.

In the quantum scenario, the relation is somewhat murkier. \cite{GZB-preprint-2024} established quantitative bounds on approximating the Petz $f$-divergences by what may be identified as $\chi^{2}_{HM}$ using $\mbf{J}^{-1}_{f_{HM},\sigma}[X] = \frac{1}{2}\left(\sigma^{-1}X + X\sigma^{-1}\right)$ \cite{Petz-2011a}. The authors further bound the rate of convergence for time-homogeneous Markov chains according to Petz $f$-divergences. However, they did not establish that this $\chi^{2}_{HM}$ is the local behaviour of the Petz $f$-divergences in terms of the contraction coefficient $\eta_{\chi^{2}_{HM}}(\cE,\sigma)$. Interestingly, $\chi^{2}_{HM}(\rho \Vert \sigma) = \ol{D}_{x^{2}}(\rho \Vert \sigma) - 1$ where $\ol{D}_{x^{2}}$ is the Petz $f$-divergence for $f(x) = x^{2}$ as may be verified by direct calculation using Definition \ref{def:Petz-f-div} and Proposition \ref{prop:func-of-mod-op}. In a similar manner, a direct calculation will verify $\chi^{2}_{GM}(\rho \Vert \sigma) = \exp(\wt{D}_{2}(\rho \Vert \sigma)) - 1$, but, to the best of our knowledge, no further relations between $\chi^{2}_{GM}(\rho \Vert \sigma)$ and $\wt{D}_{2}(\rho \Vert \sigma)$ are known. In a more positive direction, \cite[Theorem 2.8]{Hirche-2024a} showed that the $\chi^{2}_{LM}$-divergence corresponds to the local behaviour of the HT $f$-divergences and \cite[Theorem 5.9]{Hirche-2024a} generalized the quantitative bounds between $f$-divergences and the $\chi^{2}$-divergence of \cite{GZB-preprint-2024} to the HT $f$-divergences and the $\chi^{2}_{LM}$-divergence. Given its historical importance, we remark that the $\chi^{2}_{LM}$-divergence is a long-studied quantity. Namely $\chi^{2}_{LM}(\rho \Vert \sigma) = \gamma_{LM,\sigma}(\rho,\rho) \eqqcolon K^{Bo}_{D}(\rho,\sigma)$ where the final expression is the ``Bogoliubov-Kubo-Mori inner product" \cite[Eq. 7]{Petz-1996a}, which is useful in information geometry and considering quantum Fisher information (see e.g.~\cite{petz1994geometry,Lesniewski-1999a,Gao-2023-sufficient-fisher} and references therein).

In total, while we do not have a complete picture, the HT $f$-divergences recover the majority of the classical theory, but other $\chi^{2}_{f}$ divergences are relevant for other families of quantum divergences. Table \ref{tab:div-locality-summary} summarizes these correspondences.
\begin{table}[H]
    \centering
    \begin{tabular}{l|l} 
    Quantum Divergence Family & Relevant $\chi^{2} _{f}$-Divergence  \\ \hline
    Petz $f$-Divergences \cite{HIAI_2011,Hiai-2017a} & $\chi^{2}_{HM}(\rho \Vert \sigma) = \ol{D}_{x^{2}}(\rho \Vert \sigma) - 1$ \\
    Hirche-Tomamichel $f$-Divergences \cite{Hirche-2024a} & $\chi^{2}_{LM}(\rho \Vert \sigma) = \int_{0}^{\infty} \Tr[\rho(\sigma + s\mbb{1})^{-1}\rho(\sigma + s\mbb{1})^{-1}] ds - 1$ \\ 
    Sandwiched $\alpha$-R\'{e}nyi Divergences \cite{Wilde-2014a,Muller-Lennert-2013a} & $\chi^{2}_{GM}(\rho \Vert \sigma) = \exp(\wt{D}_{2}(\rho \Vert \sigma)) - 1$
\end{tabular}
    \caption{\footnotesize Summary of the families of quantum divergences considered in this section and the seemingly relevant $\chi^{2}_{f}$-divergence.}
    \label{tab:div-locality-summary}
\end{table}

\subsubsection{Bounds on Rates of Data Processing}
Here we establish the main result of this subsection, which shows we can bound the rate of contraction for time homogeneous Markov chains as measured by various families of quantum divergences using corresponding input-dependent contraction coefficients of quantum $\chi^{2}$-divergences (Theorem \ref{thm:mixing-rate}). We also strengthen this result for the relative entropy specifically (Theorem \ref{thm:rel-ent-mixing}).

To establish these results, we will make use of a few definitions. First, we recall the definition for input-dependent and input-independent contraction coefficients more generally than just for $\chi^{2}$-divergences.
\begin{definition}
    Assume $\mbb{D}(\cdot \Vert \cdot): \Density(A) \times \Density(A) \to \mbb{R}$ satisfies the data processing inequality, i.e. $\mbb{D}(\cE(\rho) \Vert \cE(\sigma)) \leq \mbb{D}(\rho \Vert \sigma)$ for all $\cE \in \Channel(A,B)$ and $\rho,\sigma \in \Density(A)$ such that $\mbb{D}(\rho \Vert \sigma) < +\infty$. Then the input-dependent and input-independent contraction coefficents are given respectively by
    \begin{align}
        \eta_{\mbb{D}}(\cE,\sigma) \coloneq \sup_{\substack{\rho \neq \sigma: \\ 0 < \mbb{D}(\rho \Vert \sigma) < + \infty}} \frac{\mbb{D}(\cE(\rho) \Vert \cE(\sigma))}{\mbb{D}(\rho \Vert \sigma)} \\
        \eta_{\mbb{D}}(\cE) \coloneq \sup_{\sigma \in \mbb{D}} \eta_{\mbb{D}}(\cE,\sigma) \ .
    \end{align}
\end{definition}
\noindent We remark the above definition captures more than divergences. For example, it captures the trace distance.

We also note the following relation that follows from definitions and results of \cite{Hirche-2024a}:
\begin{align}\label{eq:HT-f-div-contraction-ordering}
    \eta_{\chi^{2}_{LM}}(\cE,\sigma) \leq \eta_{\chi^{2}_{LM}}(\cE) \leq \eta_{D_{f}}(\cE) \leq \eta_{\Tr}(\cE) \ , 
\end{align}
where $D_{f}$ denotes a member of the HT $f$-divergences (Definition \ref{def:HT-divergences}). As this relation is known to hold for classical $f$-divergences, this suggests it `ought' to be the correct quantum $f$-divergences to recover the results of \cite{GZB-preprint-2024}.

Second, we call an $f$-divergence `Pinsker-satisfying' if there exists $L_{f} > 0$ such that $\mbb{D}_{f}(\rho \Vert \sigma) \geq \frac{L_{f}}{2}\text{TD}(\rho,\sigma)^{2}$. We can then define the set of functions 
\begin{align}
    \cF_{\text{Pin}} \coloneq \left\{f: \begin{matrix} f(1) = 0 \, , \, f''(1) > 0 \, , \text{convex, Pinsker-satisfying} \, , \\ \text{and twice continuously differentiable}
    \end{matrix} \right\} \ .
\end{align}
Third, we will specify the notion of convergence of a time-homogeneous quantum Markov chain that we are interested in. 
\begin{definition}
    We say a quantum channel $\cE_{A \to A}$ is mixing if there exists a unique quantum state $\pi$ such that $\lim_{n \to \infty} \Vert \cE^{\circ n}(\rho) - \pi \Vert_{1} = 0$ for all $\rho \in \Density(A)$.
\end{definition}

Lastly, we will make use of the following lemma.
\begin{lemma}(\cite[Theorem 5.9]{Beigi-2025a})
    Let $\rho,\sigma \in \Density(A)$. Define $I = [e^{-D_{\max}(\sigma \Vert \rho)},e^{D_{\max}(\rho \Vert \sigma)}]$. Then
    \begin{align}
    \frac{\kappa_{f}^{\downarrow}(\rho,\sigma)}{2}\chi^{2}_{LM}(\rho \Vert \sigma) \leq D_{f}(\rho \Vert \sigma) \leq \frac{\kappa_{f}^{\uparrow}(\rho,\sigma)}{2}\chi^{2}_{LM}(\rho \Vert \sigma) \ ,
    \end{align}
    where 
    \begin{align}
        \kappa_{f}^{\uparrow}(\rho,\sigma) &\coloneq \max_{\substack{t \in [0,1] \\ \gamma \in I}} f''(1+ t(\gamma-1)) \\
        \kappa_{f}^{\downarrow}(\rho,\sigma) &\coloneq \min_{\substack{t \in [0,1] \\ \gamma \in I}} f''(1+ t(\gamma-1)) \ .
    \end{align}
\end{lemma}

With these definitions provided, we state the theorem, which we note covers the three families of divergences summarized in Table \ref{tab:div-locality-summary}.
\begin{tcolorbox}[width=\linewidth, sharp corners=all, colback=white!95!black, boxrule=0pt,frame hidden]
\begin{theorem}\label{thm:mixing-rate}
    Let $\cE$ be mixing with fixed point $\pi$. 
    \begin{enumerate}
        \item \cite{GZB-preprint-2024} If $f \in \cF_{\text{Pin}}$ and either $\vert f''(0) \vert < +\infty$ or $\pi$ is full rank, $\lim_{n \to \infty} \left[\eta_{\ol{D}_{f}}(\cE^{n},\pi)^{1/n} \right] \leq \eta_{\chi^{2}_{HM}}(\cE,\pi)$.
        \item If $\alpha \in [1/2,\infty)$, $\lim_{n \to \infty} [\eta_{\wt{D}_{\alpha}}(\cE,\pi)^{1/n}] \leq \eta_{\chi^{2}_{GM}}(\cE,\pi)$.
        \item If $f \in \cF_{\text{Pin}}$ and $\pi$ is full rank, $\lim_{n \to \infty} [\eta_{D_{f}}(\cE,\pi)^{1/n}] \leq \eta_{\chi^{2}_{LM}}(\cE,\pi) \leq \eta_{f}(\cE,\pi)$. That is, it converges the fastest it could converge according to the theory of contraction coefficients.
    \end{enumerate}
    Moreover, the first and third bounds can be tight classically.
\end{theorem}
\end{tcolorbox}
Before providing the proof of Theorem \ref{thm:mixing-rate}, we note that while the above theorem is very general in terms of the divergences it bounds, perhaps the most important divergence to consider is the relative entropy. This is because it induces the quantities that capture the asymptotic rate of various information processing tasks. As $D(\rho \Vert \sigma) = \lim_{\alpha \to 1} \wt{D}_{\alpha}(\rho \Vert \sigma) = D_{t\log t}(\rho \Vert \sigma) = \ol{D}_{t \log t}(\rho \Vert \sigma)$ (see \cite{Tomamichel-2016a,Hiai-2017a,Hirche-2024a} respectively), this implies that its rate of convergence is bounded by three possible $\chi^{2}_{f}$ contraction coefficients. We can in fact prove even in a non-asymptotic sense the relative entropy contracts as the minimum of a larger class of $\chi^{2}_{f}$ contraction coefficients.
\begin{tcolorbox}[width=\linewidth, sharp corners=all, colback=white!95!black, boxrule=0pt,frame hidden,breakable]
\begin{theorem}\label{thm:rel-ent-mixing}
    Let $\cE_{A \to B}$ be a quantum channel and $\sigma \in \Density(A)$. Define $\cG \coloneq \{f \in \cM_{St}: f_{HM} \leq f \leq f_{LM}\}$. 
    \begin{align}
        \eta_{D}(\cE,\sigma) \leq \frac{\ln(2)}{\lambda_{\min}(\sigma)}\eta_{\chi^{2}_{f}}(\cE,\sigma) \quad \forall f \in \cG \ . \label{eq:rel-ent-bound}
    \end{align}
    Similarly, if $\cE(\pi) = \pi$, for all $\rho \ll \pi$
    \begin{align}
        D(\cE^{n}(\rho) \Vert \pi) \leq \frac{2}{\lambda_{\min}(\pi)} \eta_{\chi^{2}_{f}}(\cE,\pi)^{n} \ .
    \end{align}
\end{theorem}
\end{tcolorbox}
We note for $d \geq 2$, $\frac{\ln(2)}{\lambda_{\min}(\sigma)} \geq 2\ln(2) \approx 1.4$. Even classically it is known $\eta_{\chi^{2}}(\cE,\sigma) < \eta_{D}(\cE,\sigma)$ \cite{Anantharam-2013a}, so any scaling term must be strictly greater than one for \eqref{eq:rel-ent-bound} to hold. \eqref{eq:rel-ent-bound} also recovers the classical result \cite[Corollary 1]{Makur-2020a} where we note that work uses base $e$. Moreover, we stress that by \eqref{eq:HT-f-div-contraction-ordering}, for sufficient depth these are the tightest bounds one can obtain for relative entropy.
\begin{proof}[Proof of Theorem \ref{thm:rel-ent-mixing}]
    First, letting $f \in \cG$ and $\rho \ll \sigma$,
    \begin{equation}
    \begin{aligned}
        D(\cE(\rho) \Vert \cE(\sigma)) \leq \chi^{2}_{LM}(\cE(\rho)  \cE(\sigma)) \leq& \chi^{2}_{f}(\cE(\rho) \Vert \cE(\sigma)) \\\leq& \eta_{\chi^{2}_{f}}(\cE,\sigma) \chi^{2}_{f}(\rho \Vert \sigma) \\
        \leq& \eta_{\chi^{2}_{f}}(\cE,\sigma) \chi^{2}_{HM}(\rho \Vert \sigma) \\
        \leq& \eta_{\chi^{2}_{f}}(\cE,\sigma) \frac{2}{\lambda_{\min}(\sigma)} \text{TD}(\rho,\sigma)^{2} \\
        \leq& \eta_{\chi^{2}_{f}}(\cE,\sigma) \frac{\ln(2)}{\lambda_{\min}(\sigma)} D(\rho \Vert \sigma)
    \end{aligned}
    \end{equation}
    where the first inequality is well known (e.g. \cite{Beigi-2025a}), the second inequality is by \eqref{eq:ordering-of-chi-squareds}, the third inequality is by definition, the fourth inequality is again by \eqref{eq:ordering-of-chi-squareds}, the fifth inequality is Proposition \ref{prop:Df-lb-in-chi-square}, and the final inequality is Pinsker's inequality (for relative entropy). To complete the proof of \eqref{eq:rel-ent-bound}, we divide both sides of the above inequality by $D(\rho \Vert \sigma)$ and supremize over $\rho \neq \sigma$ such that $D(\rho \Vert \sigma) < +\infty$, which is equivalent to $\rho \neq \sigma$ and $\rho \ll \sigma$.

    For the other case, by the above inequalities,
    \begin{align}
        D(\cE^{n}(\rho) \Vert \pi) = D(\cE^{n}(\rho) \Vert \cE^{n}(\pi)) \leq \frac{2}{\lambda_{\min}(\sigma)}\eta_{f}(\cE^{n},\pi)\text{TD}(\rho,\sigma)^{2} \leq \frac{2}{\lambda_{\min}(\sigma)}\eta_{f}(\cE,\pi)^{n} \ ,
    \end{align}
    where we used $\cE(\pi) = \pi$ and that $\eta_{f}(\cF \circ \cE,\sigma) \leq \eta_{f}(\cF,\cE(\sigma)) \eta_{f}(\cE,\sigma)$.
\end{proof}

With the relative entropy strengthening established, we prove Theorem \ref{thm:mixing-rate}, which is morally similar to the above proof except we don't have linear bounds between the contraction coefficients.
\begin{proof}[Proof of Theorem \ref{thm:mixing-rate}]
    The first item is proven in \cite{GZB-preprint-2024}. We now establish the third item. That $\eta_{\chi^{2}_{LM}}(\cE,\sigma) \leq \eta_{f}(\cE,\sigma)$ is \cite[Corollary 4.2]{Hirche-2024a}. Thus it suffices to prove the other bound.
    \begin{equation}
    \begin{aligned}
        &\eta_{D_{f}}(\cE^{n},\pi) \\
        =& \sup_{\rho \in \Density(A): 0 < D_{f}(\rho \Vert \pi) < +\infty} \frac{D_{f}(\cE^{n}(\rho) \Vert \cE^{n}(\pi))}{D_{f}(\rho \Vert \pi)} \\
        \leq& \frac{4}{\lambda_{\min}(\pi)} \sup_{\rho \in \Density(A): 0 < D_{f}(\rho \Vert \pi) < +\infty} \frac{D_{f}(\cE^{n}(\rho) \Vert \cE^{n}(\pi))}{\chi^{2}_{LM}(\rho \Vert \pi)} \\
        \leq & \frac{4}{\lambda_{\min}(\pi)} \left[\sup_{\rho \in \Density(A): 0 < D_{f}(\rho \Vert \pi) < +\infty} \kappa^{\uparrow}_{f}(\cE^{n}(\rho),\cE^{n}(\pi))\right] \sup_{\rho \in \Density(A): 0 < D_{f}(\rho \Vert \pi) < +\infty}\frac{\chi^{2}_{LM}(\cE^{n}(\rho) \Vert \cE^{n}(\pi))}{\chi^{2}_{LM}(\rho \Vert \pi)} \\
        =& \frac{4}{\lambda_{\min}(\pi)} \left[\sup_{\rho \in \Density(A): 0 < D_{f}(\rho \Vert \pi) < +\infty} \kappa^{\uparrow}_{f}(\cE^{n}(\rho),\cE^{n}(\pi))\right] \eta_{\chi^{2}_{LM}}(\cE^{n},\pi) \\
        \leq& \frac{4}{\lambda_{\min}(\pi)} \left[\sup_{\rho \in \Density(A): 0 < D_{f}(\rho \Vert \pi) < +\infty} \kappa^{\uparrow}_{f}(\cE^{n}(\rho),\cE^{n}(\pi))\right] \eta_{\chi^{2}_{LM}}(\cE,\pi)^{n} \ ,
    \end{aligned}
    \end{equation}
    where the first inequality is Proposition \ref{prop:Df-lb-in-chi-square} the second inequality is \cite[Theorem 5.9]{Beigi-2025a}, the third uses that $\kappa^{\uparrow}_{f}(\cE^{n}(\rho),\cE^{n}(\pi))$ and $\frac{\chi^{2}_{LM}(\cE^{n}(\rho) \Vert \cE^{n}(\pi))}{\chi^{2}_{LM}(\rho \Vert \pi)}$ are non-negative functions over the set being optimized, and the final inequality uses the submultiplicativity of the input-dependent contraction coefficient and that $\cE(\pi) = \pi$. 

    Now we need to address the remaining supremum term. Clearly it must converge given its definition and the mixing property. Here we give an argument using Lemma \ref{lem:Dmax-topol-equiv-to-sigma-often}. Consider the intervals 
    $$I_{n} \coloneq [e^{-D_{\max}(\pi \Vert \cE^{n}(\rho))},e^{D_{\max}(\cE^{n}(\rho) \Vert \pi)}] \ . $$  First, by the assumption $\cE$ is mixing, for all $\delta > 0$, there exists $n_{\delta} \in \mbb{N}$ such that for all $n \geq n_{\delta}$, $\Vert \cE^{\circ n}(\rho) - \pi \Vert_{1} \leq \delta$ for all $\rho \in \Density(A)$. As $\pi$ is full rank, by Lemma \ref{lem:Dmax-topol-equiv-to-sigma-often}, the same claim may be made for $D_{\max}(\cE^{n}(\rho) \Vert \pi)$. Moreover,  as $\pi > 0$, \cite[Lemma 67]{GZB-preprint-2024} shows there exists $n_{0} \in \mbb{N}$ such that $\cE^{n}(\rho) > 0$ for all $n \geq n_{0}$. It follows $\cE^{n}(\rho) \gg \pi$ for all $\rho$ and $n \geq n_{0}$, so we may again apply Lemma \ref{lem:Dmax-topol-equiv-to-sigma-often} bound $D_{\max}$ by the trace distance.\footnote{The upper bound in Lemma \ref{lem:Dmax-topol-equiv-to-sigma-often} involves the minimum eigenvalue, but by the mixing property the minimum eigenvalue must converge to $\lambda_{\min}(\pi)$. This is sufficient to note as the eigenvalues are well-behaved under perturbations by Weyl's inequalities \cite{Bhatia-1997a}.} In other words, $\lim_{n \to \infty} D_{\max}(\cE^{n}(\rho)\Vert \pi) = 0 = \lim_{n \to \infty} D_{\max}(\pi \Vert \cE^{n}(\rho))$. Thus, for all $m \in \mbb{N}$, there exists $n_{m} \in \mbb{N}$ such that $I_{n} \subset (1-\frac{1}{m}, 1+\frac{1}{m})$ for all $n \geq n_{m}$. This proves $\lim_{n \to \infty} \left[ \sup_{\rho \in D(A)} \kappa^{\uparrow}_{f}(\cE^{n}(\rho),\cE^{n}(\pi)) \right] = f''(1)$. With this addressed, taking the $n^{th}$ root to both sides and taking the limit $n \to \infty$  obtains the result. The proof for the sandwiched divergence is effectively identical by using Lemma \ref{lem:sandwiched-to-sandwiched-chi-square-bound}. The tightness claim follows from \cite{Makur-2020a}.
\end{proof}

\paragraph{On Faithfulness of the Above Results}
A natural concern with the above results is that they will be trivial when the contraction coefficient is one. However, it is easy to see from Section \ref{sec:chi-square-extreme-values-and-recoverability} that in fact $\chi^{2}_{GM}(\cE,\sigma)$ is trivial only if \textit{any} contraction coefficient bound would be trivial.
\begin{proposition}\label{prop:faithfulness-of-chi-squared-GM}
    If $\eta_{\chi^{2}_{GM}}(\cE,\sigma) = 1$, then  $\eta_{\mbb{D}}(\cE,\sigma) = 1$ for any $\mbb{D}(\cdot \Vert \cdot): \Density(A) \times \Density(A) \to \mbb{R}$ satisfying the data processing inequality.
\end{proposition}
\begin{proof}
    By Lemma \ref{lem:chi-squared-contraction-unity}, if $\eta_{\chi^{2}_{GM}}(\cE,\sigma) = 1$, then there exists $\rho \neq \sigma$ such that $(\cP_{\cE,\sigma} \circ \cE)(\rho) = \rho$ where $\cP_{\cE,\sigma}$ is the Petz recovery map. It follows
    \begin{align}
        \mbb{D}(\rho \Vert \sigma) \geq \mbb{D}(\cE(\rho) \Vert \cE(\sigma)) \geq \mbb{D}((\cP_{\cE,\sigma} \circ \cE)(\rho) \Vert (\cP_{\cE,\sigma} \circ \cE)(\sigma)) = \mbb{D}(\rho \Vert \sigma) \ ,  
    \end{align}
    where we used the data processing inequality. It follows $\mbb{D}(\rho \Vert \sigma) = \mbb{D}(\cE(\rho) \Vert \cE(\sigma))$. The definition of (input-dependent) contraction coefficient completes the proof.
\end{proof}
This shows that if contraction coefficients can be used to obtain mixing times of a time-homogeneous Markov chain at all, it can be done using $\eta_{\chi^{2}_{GM}}(\cE,\sigma)$.

\paragraph*{On Strengthening the Above Results} The above discussion as well as Item 3 of Theorem \ref{thm:mixing-rate} and Theorem \ref{thm:rel-ent-mixing} show that we have in some respects reached the strongest results contraction coefficients can tell us for studying discrete-time time-homogeneous Markov chains without using more structure. 
To see this, note that submultiplicativity of the contraction coefficients does not help because one will always collapse to the original input-dependent contraction coefficient if $\cE(\pi) = \pi$:
$$\mbb{D}(\cE^{\circ 2}(\rho) \Vert \cE^{\circ 2}(\pi)) \leq \mbb{D}(\cE(\rho) \Vert \cE(\pi))\eta_{\mbb{D}}(\cE,\cE(\pi)) = \mbb{D}(\cE(\rho) \Vert \cE(\pi))\eta_{\mbb{D}}(\cE,\pi) \ . $$
Therefore, $\eta_{\mbb{D}}(\cE,\pi)^{n}$ really is the quantity that captures convergence for arbitrary inputs and we have captured how it behaves at least for Hirche-Tomamichel $f$-divergences. It seems the best way to strengthen these bounds is to consider multiple iterations of the channel, i.e. $\eta_{\mbb{D}}(\cE^{\circ k},\rho)$ for $k > 1$.\\

The rest of the section establishes bounds used in establishing Theorems \ref{thm:mixing-rate} and \ref{thm:rel-ent-mixing}. Some results may be of independent interest.

\subsubsection{Inequalities between Divergences}\label{sec:divergences-and-TD-inequalities}
In this paragraph, we obtain relations between divergences and the trace distance. First we obtain a reverse Pinsker inequality for max divergence.
\begin{lemma}\label{lem:Dmax-topol-equiv-to-sigma-often}
    Let $\rho,\sigma \in \Density(A)$ and $\rho \ll \sigma$. Then 
    \begin{align}
        \frac{2}{\ln(2)}\text{TD}(\rho,\sigma)^{2} \leq D_{\max}(\rho \Vert \sigma) \leq \frac{1}{\ln(2)} \lambda_{\min}(\sigma)^{-1}\text{TD}(\rho,\sigma)
    \end{align}
\end{lemma}
\begin{proof}
    The lower bound is a corollary of Pinsker's inequality \cite{WatrousBook}, i.e. $D_{\max}(\rho \Vert \sigma) \geq D(\rho \Vert \sigma) \geq \frac{2}{\ln(2)}\text{TD}(\rho,\sigma)^{2}$. Thus, we focus on the upper bound.

    First, recall $D_{\max}(\rho \Vert \sigma) = \log(\Vert \sigma^{-1/2}\rho\sigma^{-1/2}\Vert_{\infty})$. Now we know that $\Vert \sigma^{-1/2}\rho\sigma^{-1/2}\Vert_{\infty} \geq 1$ as follows from $D_{\max}(\rho \Vert \sigma) \geq 0$. Moreover as $\sigma^{-1/2}\rho\sigma^{-1/2}$ is Hermitian, its operator norm is its maximum eigenvalue. Let $\ket{\psi} \in \supp(\sigma)$ be the corresponding eigenvector. We then have
    \begin{align}
        \ln(\Vert \sigma^{-1/2}\rho\sigma^{-1/2}\Vert_{\infty}) =& \ln(\langle \dyad{\psi}, \sigma^{-1/2}\rho\sigma^{-1/2}\rangle) \\
        =& \ln(1 + \langle \dyad{\psi}, \sigma^{-1/2}\rho\sigma^{-1/2}-I\rangle) \\
        \leq& \ln(1 + \vert\langle \dyad{\psi}, \sigma^{-1/2}\rho\sigma^{-1/2}-I\rangle\vert) \\
        \leq& \vert\langle \dyad{\psi}, \sigma^{-1/2}\rho\sigma^{-1/2}-I\rangle\vert \\
        =& \vert\langle \sigma^{1/2}\sigma^{-1/2}\dyad{\psi}\sigma^{-1/2}\sigma^{1/2}, \sigma^{-1/2}\rho\sigma^{-1/2}-I\rangle\vert \\
        =& \vert \langle \sigma^{-1/2}\dyad{\psi}\sigma^{-1/2}, \rho-\sigma \rangle) \vert \ ,
    \end{align}
    where the second equality is $\Tr[\dyad{\psi}] = 1$, the first inequality is that $\ln$ monotonically increases, the second inequality is $\ln(1+x) \leq x$ for $x > -1$, the third equality is that $\sigma^{1/2}\sigma^{-1/2} = \Pi_{\supp(\sigma)}$, the fourth equality is $\sigma^{1/2} \cdot \sigma^{1/2}$ is a self-adjoint map (with respect to HS inner product) and simplifying using $\rho \ll \sigma$.
    
    Now as $0 \leq \dyad{\psi} \leq I$, $0 \leq \sigma^{-1/2}\dyad{\psi}\sigma^{-1/2} \leq \lambda_{\min}(\sigma)^{-1}\Pi_{\supp(\sigma)}$. Therefore,
    \begin{align}
        \vert\langle \sigma^{-1/2}\dyad{\psi}\sigma^{-1/2}, \rho-\sigma \rangle\vert = &\lambda_{\min}(\sigma)^{-1} \vert\langle \lambda_{\min}(\sigma)\sigma^{-1/2}\dyad{\psi}\sigma^{-1/2}, \rho-\sigma \rangle\vert \\
        \leq& \lambda_{\min}(\sigma)^{-1} \sup_{0 \leq M \leq I} \vert \langle M, \rho - \sigma \rangle \vert \\
        =& \lambda_{\min}(\sigma)^{-1} \text{TD}(\rho,\sigma) \ , 
    \end{align}
    where the final equality is the variational characterization of the trace distance, see e.g. \cite{Khatri-2020a}. This completes the proof.
\end{proof}

Next we aim to bound divergences below by $\chi^{2}_{HM}$, the largest quantum $\chi^{2}$-divergence. This may be seen as a quantum generalization of a result of \cite{Sason-2014a}. To that end, we begin with the following observation that is well known in the classical setting for probability distributions.
\begin{proposition}\label{prop:spec-to-TD}
    For $\rho,\sigma \in \Density(A)$, $\Vert \rho - \sigma \Vert_{\infty} \leq \frac{1}{2}\Vert \rho - \sigma \Vert_{1}$. 
\end{proposition}
\begin{proof}
    We have
    \begin{align}
        \Vert \rho - \sigma \Vert_{\infty} =&\sup_{\ket{\psi} \in \cS(A)} \vert \Tr[\psi(\rho - \sigma)] \vert 
        \leq \sup_{0 \leq \Lambda \leq \mbb{1}} \vert \Tr[\Lambda (\rho - \sigma)] \vert 
        = \text{TD}(\rho,\sigma) \ , 
    \end{align}
    where the first equality uses that as $\rho-\sigma$ is Hermitian, it's infinity norm is the maximum eigenvalue, the inequality is as we have relaxed the optimization, and the final equality is the variational characterization of trace distance.
\end{proof}

Using this, we obtain the following, which is known in the classical case \cite{Sason-2014a}. Given \eqref{eq:ordering-of-chi-squareds}, it is also a strengthening of a result in \cite{GZB-preprint-2024} as well as of \cite[Eq.~5.69]{Beigi-2025a} where the latter follows from \eqref{eq:ordering-of-chi-squareds}.
\begin{proposition}\label{prop:Df-lb-in-chi-square}
    For $P,Q \in \Pos(A)$, $\chi^{2}_{HM}(P \Vert Q) \leq \frac{\Vert P-Q \Vert_{\infty} \Vert P - Q \Vert_{1}}{\lambda_{\min}(Q)}$. Moreover, for $\rho,\sigma \in \Density(A)$, $\chi^{2}_{HM}(\rho \Vert \sigma) \leq \frac{2 \text{TD}(\rho,\sigma)^{2}}{\lambda_{\min}(\sigma)}$.
\end{proposition}
\begin{proof}
    \begin{align}
        \chi^{2}_{HM}(P \Vert Q) = \Tr[Q^{-1}(P-Q)^{2}] \leq& \lambda_{\min}(Q)^{-1} \Tr[(P-Q)^{2}] \\ 
        =& \lambda_{\min}(Q)^{-1} \big\vert \Tr[(P-Q)^{2}] \big\vert \\
        \leq & \lambda_{\min}(Q)^{-1} \Vert P - Q \Vert_{\infty} \Vert P - Q \Vert_{1} \ , 
    \end{align}
    where the first inequality follows from $Q^{-1} \leq \lambda_{\min}^{-1}(Q)I$, the equality is the non-negativity of the squared operator and the final inequality is H\"{o}lder's inequality. Finally, when $P = \rho$, $Q = \sigma$, we may apply Proposition \ref{prop:spec-to-TD}. 
\end{proof}

Finally, this establishes our lower bounds.
\begin{lemma}[Lower bounds]\label{lem:chi-square-lb-in-terms-of-HM}
Let $\rho,\sigma \in \Density(A)$. 
    \begin{enumerate}
        \item Let $f:(0,\infty) \to \mbb{R}$ be convex, twice continuously differentiable, and Pinsker-satisfying. Then $\mbb{D}_{f}(\rho \Vert \sigma) \geq \frac{L_{f} \cdot \lambda_{\min}(\sigma)}{4}\chi^{2}_{HM}(\rho \Vert \sigma)$.
        \item Let $\alpha \in [1/2,+\infty]$. Then $\wt{D}_{\alpha}(\rho \Vert \sigma) \geq \frac{\lambda_{\min}(\sigma)}{2}\chi^{2}_{HM}(\rho \Vert \sigma)$.
    \end{enumerate}
\end{lemma}
\begin{proof}
    For the first item, this follows immediately from the assumption on being Pinsker-satisfying and Proposition \ref{prop:Df-lb-in-chi-square}. For the second item, we use
    \begin{align}
        \wt{D}_{\alpha}(\rho \Vert \sigma) \geq \wt{D}_{1/2}(\rho \Vert \sigma) = \log(F(\rho,\sigma)^{-1}) \geq \frac{1}{\ln(2)}(1 - F(\rho,\sigma)) \geq \frac{1}{\ln(2)}\text{TD}(\rho,\sigma)^{2} \ , 
    \end{align}
    where the first inequality we used that the Sandwiched R\'{e}nyi divergence monotonically increases in $\alpha$, the equality is well-known and may be verified via the definition, the second inequality is $\ln(x) \geq 1 - \frac{1}{x}$, and the final inequality is the Fuchs-van de Graaf inequality. Applying Proposition \ref{prop:Df-lb-in-chi-square} completes the proof.
\end{proof}

\paragraph{Upper Bound for Sandwiched R\'{e}nyi Divergence}
In this paragraph, we show how to upper bound Sandwiched R\'{e}nyi divergences by the $\chi_{GM}^{2}$-divergence. To do this, we require some preliminaries. First we define the Sandwiched quantum mean, $\wt{Q}_{\alpha}(\rho \Vert \sigma) = \Tr[\left( \sigma^{\frac{1-\alpha}{2\alpha}}\rho\sigma^{\frac{1-\alpha}{2\alpha}}\right)^{\alpha}]$. Note that $\wt{D}_{\alpha}(\rho \Vert \sigma) = \frac{1}{\alpha-1}\log(\wt{Q}_{\alpha}(\rho \Vert \sigma))$ for all $\alpha \in (0,1) \cup (1,\infty)$ and $\wt{\chi}^{2}(\rho \Vert \sigma) = \wt{Q}_{2}(\rho \Vert \sigma) - 1$ as follows from Definition \ref{def:SRD}. We also know the Pinching inequality for $\alpha > 1$: 
\begin{align}\label{eq:sandwiched-mean-pinch-property}
    \vert \text{spec}(\sigma) \vert^{\alpha} \wt{Q}_{\alpha}(\cP_{\sigma}(\rho)\Vert \sigma) \geq \wt{Q}_{\alpha}(\rho \Vert \sigma) \geq \wt{Q}_{\alpha}(\cP_{\sigma}(\rho)\Vert \sigma) \ , 
\end{align}
where $\cP_{\sigma}$ is the pinching map with respect to $\sigma$ (see \cite{Tomamichel-2016a}). Next we define the Sandwiched Hellinger divergence, which was introduced in \cite{Beigi-2025a}:
\begin{equation}\label{eq:sandwiched-hell-in-terms-of-mean}
    \wt{H}_{\alpha}(\rho \Vert \sigma) = \frac{1}{1-\alpha}\left(1 - \wt{Q}_{\alpha}(\rho \Vert \sigma) \right) \ . 
\end{equation}
Note this is defined such that $\wt{D}_{\alpha}(\rho \Vert \sigma) = \frac{1}{\alpha-1}\log(1+(\alpha-1)\wt{H}_{\alpha}(\rho \Vert \sigma))$, which is the correspondence between the classical R\'{e}nyi and Hellinger divergences (see e.g. \cite{Sason-2016-f-div-ineqs}). We remark the Hellinger divergences are $f$-divergences defined via the function $f(t) \coloneq \frac{t^{\alpha}-1}{\alpha-1}$. We now have the definitions we need.
\begin{lemma}\label{lem:sandwiched-to-sandwiched-chi-square-bound}
Let $\rho,\sigma \in \Density(A)$ such that $\rho \ll \sigma$. Let $p_{\rho,\sigma}$ be the distribution defined by the spectrum of $\cP_{\sigma}(\rho)$ and similarly $q$ with respect to $\sigma$. The following bound holds.
    \begin{align}
        \wt{D}_{\alpha}(\rho \Vert \sigma) \leq \begin{cases} 
           \frac{\kappa_{\alpha}^{\uparrow}(p_{\rho,\sigma}, q)}{\ln(2)} \wt{\chi}^{2}(\rho \Vert \sigma) & \alpha > 2 \\
           \frac{2}{\ln(2)}\wt{\chi}^{2}(\rho \Vert \sigma) & \alpha \leq 2 \ ,
        \end{cases}
    \end{align}
    where
    \begin{align}
        1 \leq \kappa^{\uparrow}_{\alpha}(p_{\rho,\sigma},q) = \alpha \max_{x} \left(\frac{p_{\rho,\sigma}(x)}{q(x)}\right)^{\alpha -2} \ .
    \end{align}
\end{lemma}
\begin{proof}
    We start with the case $\alpha > 2$ and then explain the modifications in the other case. Now,
    \begin{align}
        \wt{D}_{\alpha}(\rho \Vert \sigma) =& \frac{1}{n}\left[ \wt{D}_{\alpha}(\rho^{\otimes n} \Vert \sigma^{\otimes n}) \right] \\
        \leq&  \frac{1}{n}\left[ \wt{D}_{\alpha}(\cP_{\sigma}(\rho)^{\otimes n} \Vert \sigma^{\otimes n}) + \log\vert \text{spec}(\sigma^{\otimes n})\vert  \right] \\
        =&  \wt{D}_{\alpha}(\cP_{\sigma}(\rho) \Vert \sigma) + \frac{\log\vert \text{spec}(\sigma^{\otimes n})\vert}{n} \\
        \leq&  \frac{1}{\ln(2)}\wt{H}_{\alpha}(\cP_{\sigma}(\rho) \Vert \sigma) + \frac{\log\vert \text{spec}(\sigma^{\otimes n})\vert}{n} \\
        =& \frac{1}{\ln(2)}H_{\alpha}(p_{\rho,\sigma} \Vert q) + \frac{\log\vert \text{spec}(\sigma^{\otimes n})\vert}{n} \\
        \leq & \kappa_{\alpha}^{\uparrow}(p_{\rho,\sigma}, q)\chi^{2}(p_{\rho,\sigma}\Vert q) + \frac{\log\vert \text{spec}(\sigma^{\otimes n})\vert}{n} \\
        =& \frac{\kappa_{\alpha}^{\uparrow}(p_{\rho,\sigma}, q)}{\ln(2)}\chi_{GM}^{2}(\cP_{\sigma}(\rho) \Vert \sigma) + \frac{\log\vert \text{spec}(\sigma^{\otimes n})\vert}{n} \\
        \leq & \frac{\kappa_{\alpha}^{\uparrow}(p_{\rho,\sigma}, q)}{\ln(2)}\chi_{GM}^{2}(\rho \Vert \sigma) + \frac{\log\vert \text{spec}(\sigma^{\otimes n})\vert}{n} \ ,
    \end{align}
    where the first and second equality are additivity over tensor products \cite{Tomamichel-2016a}, the first inequality is the pinching inequality \eqref{eq:sandwiched-mean-pinch-property}, the second inequality is $\log(1+x) \leq x$, the third equality is because the Sandwiched Hellinger evaluates on commuting states to the classical Hellinger on those spectra, the fourth is \cite[Theorem 31]{GZB-preprint-2024}, the fourth equality is because on commuting states $\chi^{2}_{GM}(\rho \Vert \sigma) = \chi^{2}(p \Vert q)$ where $p,q$ are the spectra of $\rho$ and $\sigma$, and the final equality is by \eqref{eq:sandwiched-mean-pinch-property}. As this held for all $n \in \mbb{N}$ and $\log \vert \text{spec}(\sigma^{\otimes n})\vert = O(\log(n))$, considering the limit $n \to \infty$ establishes the first case in the proof up to the definition of $\kappa^{\uparrow}_{\alpha}$. By \cite{GZB-preprint-2024}, 
    $$\kappa^{\uparrow}(p_{\rho,\sigma},q) \coloneq \max_{x,t \in [0,1]} \alpha\left(1-t + t\frac{p_{\rho,\sigma}(x)}{q(x)}\right)^{\alpha -2} = \alpha \max_{x} \left(\frac{p_{\rho,\sigma}(x)}{q(x)}\right)^{\alpha -2}   \ , $$
    where the equality uses that there must exist $x$ such that the ratio is at least 1 and $t^{\alpha-2}$ is monotonically increasing for $\alpha > 2$. This completes the first case. In the case $\alpha \leq 2$, we may use monotonicity of the Sandwiched divergences, i.e. $\wt{D}_{\alpha}(\rho \Vert \sigma) \leq \wt{D}_{2}(\rho \Vert \sigma)$ and then do the same argument, but $\kappa^{\uparrow}_{2}(\mbf{p},\mbf{q}) = 2$, which simplifies the claim.
\end{proof}

\section{Computing Operator Monotone-Induced Correlation Quantities}\label{sec:computability}
In this section, we establish the operationally relevant contraction and correlation coefficients $\eta_{\chi^{2}_{f}}(\cE,\sigma)$, $\mu_{f}(A:B)_{\rho}$, and $ \mu_{f_{k}}^{\text{Lin}}(A:B)_{\rho}$ are efficient to compute  (Theorems \ref{thm:contraction-coeff-computability} and \ref{thm:f-max-corr-computability} respectively). We do this by constructing efficient algorithms. From a mathematical perspective, computing $\mu(X:Y)_{p}$ is well-known to be efficient as it is the second eigenvalue of $\wt{p}_{XY} = p_{X}^{-1/2}p_{XY}p_{Y}^{-1/2}$ with respect to the Euclidean inner product, and so we might hope this extends to the quantum setting. From a practical perspective, this section is motivated from these quantities being relevant for determining mixing times (Proposition \ref{prop:L1-mixing-times}, Theorem \ref{thm:rel-ent-mixing}), bounding the rate of contraction (Theorem \ref{thm:mixing-rate}), and placing limits on convertability under local operations (Theorem \ref{thm:k-correlation-nec-for-local-processing}). It therefore may be useful to be able to compute these quantities. Indeed, the main practically relevant conclusion of this section is the following:
\begin{tcolorbox}[width=\linewidth, sharp corners=all, colback=white!95!black, boxrule=0pt,frame hidden]
\begin{theorem}\label{thm:computable-mixing-times}
    If $\cE_{A \to A}$ has a unique, full rank fixed point $\pi_{A}$, Algorithm \ref{alg:contraction-coeff} is an efficient method for computing mixing time bounds under trace distance (resp.~relative entropy) via Proposition \ref{prop:L1-mixing-times} (resp~Theorem \ref{thm:rel-ent-mixing}). Moreover, this method is trivial for $f=f_{GM}$ only if any bound using contraction coefficients would be trivial. 
\end{theorem}
\end{tcolorbox}
\begin{proof}
    By Proposition \ref{prop:L1-mixing-times} and Theorem \ref{thm:rel-ent-mixing}, to bound the mixing time under trace distance or relative entropy, it suffices to bound $\eta_{\chi^{2}_{f}}$ for some standard operator monotone $f$. By Theorem \ref{thm:contraction-coeff-computability}, this chosen quantity can be computed efficiently. To establish the `moreover' statement, by Lemma \ref{lem:chi-squared-contraction-unity}, $\eta_{\chi^{2}_{f_{GM}}}(\cE,\sigma) = 1$ only if there exists $\rho \neq \sigma$ such that the Petz recovery map composed with the channel preserves both $\rho$ and $\sigma$. By Proposition \ref{prop:faithfulness-of-chi-squared-GM}, this implies all other contraction coefficients are unity.
\end{proof}
\noindent The above theorem is appealing for a variety of technical reasons. First, it is always the case that either the mixing time is finite (in particular by choosing $f_{GM}$) or one needs more information about the channel than a classical description of it and its fixed point. Second, the fact the algorithm works for all standard operator monotone $f$ is useful because the choice of $f$ that results in the $\chi^{2}_{f}$ contraction coefficient being minimal is known to depend on the choice of quantum channel \cite[Section 5.2]{hiai2016contraction}.

Mathematically, the proof idea of this section is simple so long as one has the appropriate linear algebraic background as we now explain. For standard operator monotone functions, We have seen that $\eta_{\chi^{2}_{f}}$, $\mu_{f}(A:B)_{\rho}$, and $\mu_{f_{k}}^{\text{Lin}}(A:B)_{\rho}$ can be characterized as eigenvalues (Corollary \ref{cor:Schroding-er-map-without-rank-constraints}, Proposition \ref{prop:standard-operator-monotone-maximal-corr}, and Lemma \ref{lem:k-correlation-coeff-Schmidt-coeff-characterization} respectively). It is a textbook result (e.g. \cite{Johnston-2021a}) that one can compute the eigenvalues of a finite-dimensional map using the inner product. The nuance is that to be efficiently computable, the relevant map and inner product need to be efficient to compute as well. Establishing that the relevant map and inner product are efficient to compute is thus the technical work in this section. 

\begin{remark}[Relation to Spectral Gap Methods]\label{rem:relation-to-spectral-gap-methods}
    In spectral gap methods for mixing times, one takes the map representing the dynamics, say $\cE$, and checks the minimal difference between the identity channel and the map representing the dynamics (or some relevant map induced by it) when excluding the known fixed point. This is a way of extracting the second eigenvalue of the dynamical map (or the relevant map induced by it) as the minimizer will be the corresponding eigenvector. As,  by Corollary \ref{cor:Schroding-er-map-without-rank-constraints}, the $\chi^{2}_{f}$ input-dependent contraction coefficients are the second eigenvalue of $\cS_{f,\cE,\sigma} \circ \cE$, the ability to compute the contraction coefficient for many $f$ (Theorem \ref{thm:contraction-coeff-computability}) implies the ability to compute \textit{a family} of spectral gaps and take the smallest to get the strongest bound in this methodology. This is an analogous idea to \cite{generalized-QSL}, which optimizes the operator monotone function to alter the Riemannian metric relevant to calculating the quantum speed limit.  
    
    We note that an important aspect of taking Theorem \ref{thm:computable-mixing-times} as a spectral gap method is that Corollary \ref{cor:Schroding-er-map-without-rank-constraints} already guarantees the first eigenvalue is one and is relevant to the contraction coefficient. In the quantum setting, unless one appeals to Corollary \ref{cor:Schroding-er-map-without-rank-constraints}, it is not obvious what is the appropriate map induced by the dynamics nor the appropriate inner product spaces to consider for extracting an eigenvalue to obtain a generic spectral gap method. Indeed, in \cite{Temme-2010a}, the authors apply a spectral gap method directly to an alternative map where they try and extract an eigenvalue in terms of the Hilbert-Schmidt inner product. This forces them to work with unital channels as they otherwise cannot guarantee their induced map has a spectrum bounded above by one (See \cite[Proposition 19]{Temme-2010a} and discussion thereafter).
\end{remark}

\subsection{Preliminaries}
In this subsection, we briefly provide background on the theory of computability and abstract linear algebra needed in this section not needed in previous sections.

\paragraph{Computability} As is standard, we take a classical algorithm to be efficient if its time complexity scales polynomially in the input description \cite{Sipser-2005a,Arora-2009a}. We remark being efficient is distinct from necessarily being practical, e.g. semidefinite programs are often efficient in this sense \cite{Nemirovski-2004a}, but aren't always practical to use in quantum information theory as the number of constraints grows exponentially in the number of qubits. A major appeal of the algorithms we present in this section is that they are built from standard subroutines (matrix multiplication, trace, etc.) that are known to be efficient in the way described.

The algorithms we introduce will also rely upon computing the matrix monotone function $f$. Thus, our claims of efficiency essentially will rely upon our ability to compute $f$. To offload the notion of efficiency of $f$, we use the notion of a poly-time computable function.
\begin{definition}
    \cite{Braverman-2005a}  The time complexity of number $a \in \mbb{R}$, $T_{a}(n)$, is the best time complexity for computing a $2^{-n}$ approximation of $a$. Let $S \subseteq I \subset \mbb{R}$. We say a continuous function $f:I \to \mbb{R}$ is poly-time computable on $S$ if there exists a polynomial $p$ such that $T_{f(x)}(n) \leq p(n)$ for all $x \in S$.
\end{definition}
The important point for us is that if $f$ is poly-time computable on $S$, if we pick up some error that scales as a polynomial $p$ in the input description, we only need to increase the precision of $f$ by $\log(p(n))$ to remove this extra error. Therefore, $f$ being poly-time computable allows efficiency of our algorithms to be established by counting the subroutines. We remark this means we are not doing a full analysis on the numerical error/numerical stability. We briefly discuss numerical stability in Section \ref{subsec:num-stability}.

\subsubsection{More Linear Algebraic Background}\label{sec:more-lin-alg-bg}
As mentioned, the key idea for establishing the results of this section will be that we have identified the relevant quantities as eigenvalues and thus we simply need methods for computing the relevant eigenvalues. In finite dimensions, there is a standard method for obtaining the eigenvalues of a linear map $\cT: (V, \langle \cdot , \cdot \rangle_{V}) \to (V, \langle \cdot, \cdot \rangle_{V})$ as we briefly explain. Given an orthonormal basis (ONB) $\{\mbf{e}_{i}\}_{i}$ for inner product space $(V, \langle \cdot , \cdot \rangle_{V})$, one computes the `standard matrix $T$ of $\cT$' via $T_{i,j} = \langle \mbf{e}_{i}, \cT(\mbf{e}_{j}) \rangle_{V}$. Then the eigenvalues of $\cT$ are the eigenvalues of $T$, so one may solve for the eigenvalues of $\cT$ from $T$. We refer the reader to \cite{Johnston-2021a} for further information and examples of the standard matrix.

Given previous sections, the relevant inner product spaces of this section will involve the vector space of Hermitian or linear operators. Here we introduce the basic ideas we will use to construct an ONB for Hermitian and linear operators. First, we recall the generalized Gell-Mann operator basis \cite{Kalev-2014a}: \begin{align}
    G_{n,m} = \begin{cases}
        \frac{1}{\sqrt{2}}(\ket{n}\bra{m}+\ket{m}\bra{n}) & n < m \\
        \frac{i}{\sqrt{2}}(\ket{n}\bra{m}-\ket{m}\bra{n}) & m < n
    \end{cases}
\end{align} 
for $n,m \in [d]$ and $G_{n,n} = \frac{1}{\sqrt{n(n+1)}}(\sum_{k = 1}^{n} \dyad{k} - n\dyad{n+1})$ for $n \in [d-1]$. The generalized Gell-Mann matrices form an orthonormal basis of the traceless, Hermitian operators with respect to the HS inner product as may be verified by direct calculation. Then one may convert this into a (non-orthonormal) basis of the Hermitian operators (with respect to the HS inner product) by adding the square root of any density matrix. That this is true may be verified as follows. Any density matrix has unit trace and thus must be linearly independent of the generalized Gell-Man operator basis. As $\Herm(A)$ has dimension $\vert A \vert^{2}$ and $\Herm_{0}(A)$ has dimension $\vert A \vert^{2} -1$, we may conclude this is a basis of $\Herm(A)$. Finally, such a basis is also a basis for the linear operators as it also has dimension $\vert A \vert^{2}$. We will use this basis throughout this section.

The remaining tool we will make use of is the Gram-Schmidt process, which is well-known, see e.g. \cite[Theorem 1.4.6]{Johnston-2021a}. This in particular allows us to convert a basis of a vector space into an orthonormal basis.
\begin{proposition}[Gram-Schmidt Process]\label{prop:GS-process}
    Let $(V, \langle \cdot, \cdot \rangle_{V})$ be an inner product space. Let $\{v_{1},...,v_{n}\} = B \subset V$ be a linearly independent set of vectors. Define
    \begin{align}
        w_{j} \coloneq v_{j} - \sum_{i=1}^{j-1} \langle u_{i}, v_{j} \rangle_{V} \cdot u_{i} \text{ and } u_{j} \coloneq w_{j}/\Vert w_{j} \Vert_{V} \quad \forall j \in \{1,...,n\} \ . 
    \end{align}
    Then $B' \coloneq \{u_{1},...,u_{n}\}$ is an orthonormal basis of $\linspan(B)$ with respect to inner product $\langle \cdot , \cdot \rangle_{V}$.
\end{proposition}
This will be used in particular to convert the basis we just described to an ONB for whichever inner product space we need to consider.

\subsection{Computing Contraction Coefficients \texorpdfstring{$\eta_{\chi^{2}_{f}(\cE,\sigma)}$}{}}
Here we show that we can efficiently compute the relevant contraction coefficients. The basic idea is that since Lemma \ref{lem:contraction-coeff-to-eig} establishes a relation between $\eta_{\chi^{2}_{f}}(\cE,\sigma)$ and the second eigenvalue of $\cS_{f,\cE,\sigma} \circ \cE$ on the inner product space $(\Herm(A), \langle \cdot , \cdot \rangle_{f,\sigma}^{\star})$ we just need to compute this eigenvalue. As described in Algorithm \ref{alg:contraction-coeff}, we will use the Gram-Schmidt process to efficiently find an ONB with respect to the inner product space $(\Herm(A), \langle \cdot , \cdot \rangle_{f,\sigma}^{\star})$ and then use this ONB to compute the standard matrix of $\cS_{f,\cE,\sigma} \circ \cE$. This allows us to determine the eigenvalue. The efficiency claims then follow from $f$ being poly-time computable.
\begin{algorithm}
\caption{Computing $\eta_{\chi^{2}_{f}}(\cE,\sigma)$}\label{alg:contraction-coeff}
\begin{algorithmic}[1]
\Procedure{GetContractionCoefficient}{$f$,$\cE$,$\sigma$} \Comment{Computes $\eta_{\chi^{2}_{f}}(\cE,\sigma)$}
    \State $\{e_{j}\}_{j \in [d^{2}_{A}]} \leftarrow \text{GetONB}(f,\cE,\sigma)$
    \For{$(i,j) \in [d_{A}^{2}]^{\times 2}$}\Comment{Generates Standard Matrix}
        \State $T_{i,j} \leftarrow \langle e_{i}, (\cS_{f,\cE,\sigma} \circ \cE)(e_{j}) \rangle_{f,\sigma}^{\star}$
    \EndFor
    \State $\{\lambda_{i}\}_{i \in [d_{A}^{2}]} \leftarrow \text{Eig}(T)$ \Comment{Obtains Eigenvalues}
    \State \textbf{return} $\lambda_{2}$
\EndProcedure 
\Procedure{GetONB}{$f$,$\sigma$}\Comment{Obtains ONB for $(\Herm(A),\langle \cdot , \cdot \rangle_{f,\sigma}^{\star})$}
\State $\{v_{j}\}_{j \in [d_{A}^{2}]} \leftarrow \{\sigma^{1/2}\} \cup \{G_{n,m}\}_{n,m}$ 
\For{$j \in [d_{A}^{2}]$} \Comment{Gram-Schmidt Process}
    \State $w_{j} \leftarrow v_{j} - \sum_{i=1}^{j-1} \langle e_{i}, v_{j}\rangle_{f,\sigma}^{\star}e_{i}$
    \State $e_{j} \leftarrow w_{j}/\sqrt{\langle w_{j},w_{j} \rangle_{f,\sigma}^{\star}}$
\EndFor
\State \textbf{Return} $\{e_{j}\}_{j \in [d^{2}_{A}]}$.
\EndProcedure
\end{algorithmic}
\end{algorithm}

\begin{tcolorbox}[width=\linewidth, sharp corners=all, colback=white!95!black, boxrule=0pt,frame hidden]
\begin{theorem}\label{thm:contraction-coeff-computability}
    Let $f \in \cM_{St}$ and be poly-time computable. Let $\sigma \in \Density_{+}(A)$ and either $\cE(\sigma) > 0$ or $f$ satisfies the conditions of Proposition \ref{prop:suff-conds-for-restricting-support}. Then Algorithm \ref{alg:contraction-coeff} computes $\eta_{\chi^{2}}(\cE,\sigma)$ with time complexity $O(\max\{d_{A}^{9},d_{A}^{2}d_{B}^{5}\})$. In other words, computing $\eta_{\chi^{2}_{f}}(\cE,\sigma)$ is efficient under these conditions.
\end{theorem}
\end{tcolorbox}

We remark that the scaling of this algorithm's complexity is clearly not ideal from a practical perspective. However, we note one particularly nice property of this algorithm: faithfulness. Given Proposition \ref{prop:faithfulness-of-chi-squared-GM}, if there exists a contraction coefficient that can give bounds on the mixing time, $\eta_{\chi^{2}_{GM}}(\cE,\pi)$ can. Given its possible practical utility, we highlight this as its own theorem before providing the proof of Theorem \ref{thm:contraction-coeff-computability}.
\begin{proof}[Proof of Theorem \ref{thm:contraction-coeff-computability}]
    (\textbf{Correctness}) \sloppy Given the assumptions stated in the theorem statement, by Lemma \ref{lem:contraction-coeff-to-eig} or Corollary \ref{cor:Schroding-er-map-without-rank-constraints}, we are interested in the second eigenvalue of $\cS_{f,\cE,\sigma} \circ \cE: (\Herm(A), \langle \cdot , \cdot \rangle_{f,\sigma}^{\star}) \to (\Herm(A), \langle \cdot , \cdot \rangle_{f,\sigma}^{\star})$, so we need to show Algorithm \ref{alg:contraction-coeff} computes this. At a high level, the algorithm computes an ONB for $(\Herm(A), \langle \cdot , \cdot \rangle_{f,\sigma}^{\star})$, computes the standard matrix $T$ of $\cS_{f,\cE,\sigma} \circ \cE$ on this basis, and then solves for the eigenvalues of $T$, which are the eigenvalues of $\cS_{f,\cE,\sigma} \circ \cE$. The following explains this in more detail.

    First, we explain the correctness of the call to GetONB, i.e. that it provides an ONB $(\Herm(A), \langle \cdot , \cdot \rangle_{f,\sigma}^{\star})$. As explained in Section \ref{sec:more-lin-alg-bg}, $\{v_{j}\}_{j \in [d_{A}^{2}]} = \{\sigma_{A}^{1/2}\} \cup \{G_{n,m}\}$ form a basis of $\Herm(A)$. As the Gram-Schmidt process (Proposition \ref{prop:GS-process}) takes a set of linearly independent vectors $\{v_{i}\}_{i}$ to a set of vectors that are an ONB of $\linspan(\{v_{i}\}_{i})$ with respect to the given inner product, the $\{w_{j}\}_{j \in [d_{A}^{2}]}$ in Algorithm \ref{alg:contraction-coeff} form an ONB for the inner product space $(\Herm(A), \langle \cdot , \cdot \rangle_{f,\sigma}^{\star})$.

    Second, as $\{w_{j}\}_{j \in [d_{A}^{2}]}$ form an ONB of $(\Herm(A), \langle \cdot , \cdot \rangle_{f,\sigma}^{\star})$, by definition of the $T_{i,j}$, the matrix $T \coloneq \sum_{i,j} T_{i,j} \vert i \rangle \langle j \vert$ is the standard matrix encoding of the linear transformation $\cS_{f,\cE,\sigma} \circ \cE$ on this basis. As the eigenvalues of $T$ correspond to the eigenvalues of $\cS_{f,\cE,\sigma}\circ \cE$, $\lambda_{2}(T) = \lambda_{2}(\cS_{f,\cE,\sigma}\circ \cE)$. This completes the proof of correctness.

    \textbf{(Efficiency)} We break the analysis into three pieces: pre-computation considerations, the complexity of GetONB computation, and computing the standard matrix. We then combine these to complete the analysis.
    
    We begin by some preliminary pre-computation considerations. We presume the $\{\lambda_{i}\}$ of $\sigma$ are pre-computed early in the algorithm. Such pre-computation has time complexity $O(d_{A}^{3})$ via the singular value decomposition (SVD) algorithm. Similarly, we assume the Kraus operators of $\cE$ are pre-computed. The time complexity of this depends on how $\cE$ is initially encoded, but it may be found from the Choi operator by getting the eigenvectors and converting them into matrices. Obtaining the eigenvectors is $O((d_{A}d_{B})^{3})$ by the SVD as the Choi operator is a $d_{A}d_{B} \times d_{A}d_{B}$ matrix. Converting a vector of length $d_{A}d_{B}$ into a $d_{B} \times d_{A}$ matrix is $O(d_{A}d_{B})$ and as there are at most $O(d_{A}d_{B})$ Kraus operators via this construction, it is at most time $O((d_{A}d_{B})^{2})$. Finally, to compute $\cE(X)$ for Hermitian $X$, one may use the Kraus operators. As there are at most $d_{A}d_{B}$ Kraus operators and matrix multiplication of a $d_{B} \times d_{A}$ matrix by a $d_{A} \times d_{A}$ matrix is $O(d_{A}^{2}d_{B})$, this is $O(d_{A}^{3}d_{B}^{2})$. The eigenvalues of $\cE(\sigma)$ may be stored. This completes analysis of the pre-computation.

    Next, we analyze the GetONB subroutine. Defining the $\{v_{j}\}$ takes $O(d_{A}^{2})O(d_{A}^{2}) = O(d_{A}^{4})$ steps. In each step of the Gram-Schmidt process, one has to compute the inner product $j-1$ times. By Lemma \ref{lem:efficiency-of-inner-product}, computing the inner product once is $O(d_{A}^{5})$. Thus, we have $0 \cdot O(d_{A}^{5}) + 1 \cdot O(d_{A}^{5}) + \hdots + (d^{2}-1) \cdot O(d_{A}^{5}) = \frac{1}{2}d_{A}^{2}(d_{A}^{2}-1)O(d_{A}^{5}) = O(d_{A}^{9})$.
    
    Now we analyze the GetContractionCoefficient subroutine. First, note that we already showed computing $\cE(e_{j})$ takes at most $O(d_{A}^{4}d_{B}^{3})$. Next, we have $\cS_{f,\cE,\sigma} = \mbf{J}_{f,\sigma} \circ \cE^{\ast} \circ \mbf{J}^{-1}_{f,\cE(\sigma)}$. As is well known, $\cE^{\ast}$ can be computed with the Kraus operators of $\cE$ and, by an identical argument as time complexity for $\cE(X)$, computing $\cE^{\ast}(Y)$ will take at most $O(d_{A}^{2}d_{B}^{3})$ time. By Lemma \ref{lem:efficiency-of-Jfpsigma}, computing $\mbf{J}_{f,\sigma}$ and $\mbf{J}_{f,\cE(\sigma)}^{-1}$ take time $O(d_{A}^{5})$ and $O(d_{B}^{5})$ respectively. Thus, computing $(\cS_{f,\cE,\sigma} \circ \cE)(e_{j})$ takes $O(d_{A}^{3}d_{B}^{2}) + O(d_{B}^{5}) + O(d_{A}^{2}d_{B}^{3}) + O(d_{A}^{5}) = O(\max\{d_{A}^{3}d^{2}_{B},d^{2}_{A}d_{B}^{3},d_{A}^{5},d_{B}^{5}\}) = O((\max\{d_{A},d_{B}\})^{5})$. By Lemma \ref{lem:efficiency-of-inner-product}, computing the inner product takes $O(d_{A}^{5})$, so we may conclude computing $\langle e_{i}, (\cS_{f,\cE,\sigma} \circ \cE)(e_{j}) \rangle_{f,\sigma}^{\star}$ still has complexity $O((\max\{d_{A},d_{B}\})^{5})$. This must be done $d_{A}^{2}$ times to compute $T$, so this obtains a complexity of $O(d_{A}^{2}(\max\{d_{A},d_{B}\})^{5})$ for computing $T$. Finally, computing the eigenvalues of $T$ via SVD has time complexity $O( (d_{A}^{2})^{3}) = O(d_{A}^{6})$. As $d_{A}^{6} \leq d_{A}^{7} \leq d_{A}^{2}(\max\{d_{A},d_{B}\})^{5}$, this time complexity is already dominated. 

    Putting the complexity analysis of precomputation, getONB, and computing $T$ together, we have 
    $$O(d_{A}^{4}d_{B}^{3}) + O(d_{A}^{9}) + O(d_{A}^{2}(\max\{d_{A},d_{B}\})^{5}) = O(\max\{d_{A}^{9},d_{A}^{2}d_{B}^{5}\}) \ , $$
   where we have used if $d_{A} \geq d_{B}$, then $d_{A}^{9}$ dominates and $d_{A}^{2}d_{B}^{5} > d_{A}^{4}d_{B}^{3}$ if and only if $d_{B} > d_{A}$ as a direct calculation will verify. This completes the proof.
\end{proof}
\begin{remark}
    Note that this shows for computing mixing times of a time-homogeneous Markov chain $\cE_{A \to A}$ where $d_{A} = d_{B}$, the Gram-Schmidt process is asymptotically the most time demanding aspect of the algorithm.
\end{remark}

\paragraph{Non-Standard Subroutines}
As the proof shows, Theorem \ref{thm:contraction-coeff-computability} requires two non-standard subroutines. The first is computing the inner product $\langle \cdot, \cdot \rangle_{f,\sigma}^{\star}$ and the second is calculating $\cS_{f,\cE,\sigma}$ on an input operator. Given the set $(f,\cE,\sigma)$, computing $\cS_{f,\cE,\sigma}=\mbf{J}_{f,\sigma} \circ \cE^{\ast} \circ \mbf{J}^{-1}_{f,\cE(\sigma)}$ reduces to being able to compute $\mbf{J}^{p}_{f,\sigma}(X)$ for $p \in \{-1,1\}$ and arbitrary $\sigma,X$. Here we provide algorithms and complexity analysis for these subroutines.

\begin{algorithm}
\caption{Computing $\langle X , Y \rangle_{\mbf{J}_{f,\sigma}^{p}}$}\label{alg:inner-product}
\begin{algorithmic}[1]
\Procedure{InnerProduct}{$f$,$p$,$X$,$Y$,$\sigma$}
    \State $(\{\lambda_{i}\},\{\ket{\nu_{i}}\}_{i \in [d]}) \leftarrow \text{SVD}(\sigma)$ \Comment{Singular Value Decomposition}
    \For{$(i,j) \in [d]^{\times 2}$}
        \State $Z^{\sigma}_{i,j} \leftarrow P_{f}(\lambda_{i},\lambda_{j})^{p}\bra{\nu_{i}}Y\ket{\nu_{j}}$ \Comment{$\mbf{J}_{f,\sigma}^{p}(Y)$ in Basis of $\sigma$}
        \State ${X^{\ast}}^{\sigma}_{i,j} \leftarrow \bra{\nu_{i}}X^{\ast}\ket{\nu_{j}}$ \Comment{$X^{\ast}$ in Basis of $\sigma$}
    \EndFor
    \State \textbf{return} $\Tr[{X^{\ast}}^{\sigma}Z^{\sigma}]$
\EndProcedure 
\end{algorithmic}
\end{algorithm}

\begin{lemma}\label{lem:efficiency-of-inner-product}
    Let $f$ be poly-time computable and $p \in \{-1,-1/2,1/2,1\}$. Let $A \cong \mbb{C}^{d}$,  $X,Y \in \Lin(A)$, and $\sigma \in \Density_{+}(A)$. Then $\langle A , B \rangle_{\mbf{J}_{f,\sigma}^{p}}$ has time complexity $O(d^{5})$ and thus is efficient to compute using Algorithm \ref{alg:inner-product}.
\end{lemma}
\begin{proof}
    (\textbf{Correctness}) We want to compute $\Tr[X^{\ast}\mbf{J}_{f,\sigma}^{p}(Y)]$. By \eqref{eq:Hadamard-prod-form}, if we can write $Y$ in the basis of $\sigma$, $Y^{\sigma}$, we can compute $Z \coloneq \mbf{J}_{f,\sigma}^{p}(Y)$ via $Y(i,j) \coloneq P_{f}(\lambda_{i},\lambda_{j})^{p}Y^{\sigma}(i,j)$. This is Line 4 of Algorithm \ref{alg:inner-product}. Then, by writing $X^{\ast}$ also in the basis of $\sigma$, ${X^{\ast}}^{\sigma}$, we can compute $\langle X, Y \rangle_{\mbf{J}^{p}_{f,\sigma}} = \Tr[{X^{\ast}}^{\sigma}Z^{\sigma}]$, which is what we return. \\

    \noindent (\textbf{Efficiency}) Note $d$ is the rank of $\sigma$. Thus the singular value decomposition has time complexity $O(d^{3})$. As $f$ is efficient and $x^{p}$ is efficient, computing $P_{f}(\lambda_{i},\lambda_{j})^{p}$ is efficient (and not a function of the dimension). Using that $\bra{\nu_{j}} Y \ket{\nu_{i}} = \Tr[\ket{\nu_{i}}\bra{\nu_{j}}Y]$, which is matrix multiplication followed by the trace, we can conclude computing $Z^{\sigma}_{i,j}$ and $X^{\sigma}_{i,j}$ is $O(d^{2}) + O(d^{3})$. As this is done $d^{2}$ times, computing $Z^{\sigma},X^{\sigma}$ is $O(d^{5})$. Finally, computing $\Tr[{X^{\sigma}}^{\ast}Z^{\sigma}]$ is $O(d^{3})$ for the same reason. Thus we have $O(d^{5})$, which is a polynomial in the dimension.
\end{proof}
\begin{remark}
    The above proof works for more general $p$, but we selected the cases interesting for this work.
\end{remark}

\begin{algorithm}
    \caption{Computing $\mbf{J}_{f,\sigma}^{p}(X)$}\label{alg:computing-Jfpsigma}
    \begin{algorithmic}[1]
    \Procedure{Jfpsigma}{$f$,$p$,$X$,$\sigma$} \Comment{Computing $\mbf{J}_{f,\sigma}^{p}(X)$}
    \State $(\{\lambda_{i}\},\{\ket{\nu_{i}}\}_{i \in [d]}) \leftarrow \text{SVD}(\sigma)$ \Comment{Singular Value Decomposition}
    \For{$(i,j) \in [d]^{\times 2}$}
        \State $Z^{\sigma}_{i,j} \leftarrow P_{f}(\lambda_{i},\lambda_{j})^{p}\bra{\nu_{i}}X\ket{\nu_{j}}$ \Comment{$J_{f,\sigma}^{p}(X)$ in Basis of $\sigma$} 
        \State $U_{i,j} \leftarrow \langle i \vert \nu_{j} \rangle$ \Comment{Unitary from $\sigma$ basis to computational basis.}
    \EndFor
    \State \textbf{return} $U Z^{\sigma} U^{\ast}$
\EndProcedure 
\end{algorithmic}
\end{algorithm}

\begin{lemma}\label{lem:efficiency-of-Jfpsigma}
    Let $f$ be poly-time computable and $p \in \{-1,-1/2,1/2,1\}$. Let $A \cong \mbb{C}^{d}$,  $X \in \Lin(A)$, and $\sigma \in \Density_{+}(A)$. Then $\mbf{J}_{f,\sigma}^{p}$ has time complexity $O(d^{5})$ and thus is efficient to compute using Algorithm \ref{alg:computing-Jfpsigma}.
\end{lemma}
\begin{proof}
    (\textbf{Correctness}) We want to compute $\mbf{J}_{f,\sigma}^{p}(X)$. Algorithm \ref{alg:computing-Jfpsigma} does this by computing it in the basis of $\sigma$ and then applying the unitary that will return it in the computational basis. By \eqref{eq:Hadamard-prod-form}, $Z^{\sigma} \coloneq \mbf{J}_{f,\sigma}^{p}(Y)$ via $Z(i,j) \coloneq P_{f}(\lambda_{i},\lambda_{j})^{p}\langle \nu_{i} \vert X \vert \nu_{j} \rangle$. This is Line 4 of Algorithm \ref{alg:computing-Jfpsigma}. The only issue is since we assume $\sigma,X$ are provided in the computational basis, we should return $Z^{\sigma}$ in the computational basis. The unitary that takes the basis of $\sigma$ to the computational basis has entries $U_{i,j} = \langle i \vert \nu_{j} \rangle$ as a standard calculation will verify. Thus, Line 5 of Algorithm \ref{alg:computing-Jfpsigma} constructs this unitary. Finally, this means returning $UZ^{\sigma}U^{\ast}$ returns $\mbf{J}^{p}_{f,\sigma}(X)$ in the computational proof as promised. \\

    \noindent (\textbf{Efficiency}) Note $d$ is the rank of $\sigma$. Thus the singular value decomposition has time complexity $O(d^{3})$. As $f$ is efficient and $x^{p}$ is efficient, computing $P_{f}(\lambda_{i},\lambda_{j})^{p}$ is efficient (and not a function of the dimension). Using that $\bra{\nu_{j}} X \ket{\nu_{i}} = \Tr[\ket{\nu_{i}}\bra{\nu_{j}}X]$, which is matrix multiplication followed by the trace, computing $Z^{\sigma}_{i,j}$ is $O(d^{2}) + O(d^{3}) = O(d^{3})$. Computing $U_{i,j}$ is $O(d)$.
    As both of these are done $d^{2}$ times, computing $\mbf{J}_{f,\sigma}^{p}(X)$ is $O(d^{5})$. Thus we have $O(d^{5})$, which is a polynomial in the dimension, and thus efficient.
\end{proof}
\subsection{Computing Quantum Maximal Correlation Coefficients \texorpdfstring{$\mu_{f}(A:B)_{\rho}, \mu_{f}^{\text{Lin}}(A:B)_{\rho}$}{}} 
In this subsection, we show one can compute the quantities $\mu_{f}(\rho_{AB})$ or $\mu^{\text{Lin}}_{f}(\rho_{AB})$ so long as we have guaranteed they are bounded above by one. The basic idea is that under these conditions, we can again reduce the problem to an eigenvalue problem, which we may solve with linear algebraic methods. The main bottleneck then becomes obtaining the maps whose eigenvalues we wish to compute. This seems to require an SDP, which limits the practicality of computing $\mu_{f}(\rho_{AB})$ as the dimension grows, but does not limit the efficiency.

The first step is to express maximal correlation coefficients in terms of eigenvalues. While we focused on the Schmidt coeffficient/singular value characterization of $f$-maximal correlation coefficients in Section \ref{sec:q-maximal-correlation-coeff}, we can of course convert this to an eigenvalue characterization.
\begin{proposition}\label{prop:max-corr-to-eig} Let all inner product spaces be implicitly with respect to the HS inner product.
    \begin{enumerate}
        \item For $f \in \cM_{St}$ such that $f_{GM} \leq f \leq f_{AM}$, $\mu_{f}(\rho_{AB})$ is the square root of the second eigenvalue of  $\Lambda_{\wt{\rho}_{f}}^{\ast} \circ \Lambda_{\wt{\rho}_{f}}: \Herm(A) \to \Herm(A)$.
        \item For $k \in [0,1]$, $\mu^{\text{Lin}}_{f_{k}}(\rho_{AB})$ is the square root of the second eigenvalue of $\Lambda_{\wt{\rho}_{k}}^{\ast} \circ \Lambda_{\wt{\rho}_{k}}: \Lin(A) \to \Lin(A)$.
    \end{enumerate}
\end{proposition}
\begin{proof}
    We begin with Item 1. By Proposition \ref{prop:standard-operator-monotone-maximal-corr}, $\mu_{f}(\rho_{AB})$ is the second Schmidt coefficient of $\wt{\rho}_{f}$ where both spaces are with respect to the HS inner product. As in the proof of Proposition \ref{prop:standard-operator-monotone-maximal-corr}, let $\Lambda_{\wt{\rho}_{f}}$ be the map such that $\wt{\rho}_{f} = \Omega_{\Lambda_{\wt{\rho}_{f}}}$. By Proposition \ref{prop:Schmidt-to-sing}, $\mu_{f}(\rho_{AB})$ is the second singular value of $\Lambda_{\wt{\rho}_{f}}$ (with respect to the HS inner product). By definition of singular values, this is the square root of the second eigenvalue of $\Lambda_{\wt{\rho}_{f}}^{\ast} \circ \Lambda_{\wt{\rho}_{f}}$ (with respect to the HS inner product). The proof of Item 2 is identical using Lemma \ref{lem:k-correlation-coeff-Schmidt-coeff-characterization}.
\end{proof}
As we have already have shown how to extract eigenvalues through the standard matrix, the limiting factor is computing these maps and their duals, e.g. $\Lambda_{\wt{\rho}_{f}},\Lambda_{\wt{\rho}_{f}}^{\ast}$. Computing the dual of a linear map $\cE$ is easy if one has a description of $\cE$ from which it is easy to extracts it Kraus operators (e.g. if the description is the Choi matrix) as then $\cE^{\ast}(Y) = \sum_{k} A_{k}^{\ast}YB_{k}$ if $\cE(X) = \sum_{k} A_{k}XB_{k}^{\ast}$. Thus, the limiting factor is in fact being able to compute the initial map.

Recall from \eqref{eq:rho-tilde-k-map} and the proof of Proposition \ref{prop:standard-operator-monotone-maximal-corr},
\begin{align}
\Lambda_{\wt{\rho}_{k}} \coloneq \cT^{-1}_{\rho_{B},k} \circ \Lambda_{\rho} \circ \cT^{-1}_{\ol{\rho}_{A},1-k} \quad \Lambda_{\wt{\rho}_{f}} = \mbf{J}_{f,\rho_{B}}^{-1/2} \circ \Lambda_{\rho} \circ (\mbf{J}_{f,\rho_{A}}^{-1/2})^{T}
\end{align} 
 where recall $\cE^{T}$ defines the `transposed' map,\footnote{As a reminder, the transposed map may be defined via the Kraus decomposition of the initial map: $\cE^{T}(Y) \coloneq \sum_{k} A^{T}_{k}Y\ol{B}_{k}$ if $\cE = \sum_{k} A_{k}XB_{k}^{\ast}$.} so it can be obtained from the Kraus decomposition of $\mbf{J}_{f,\rho_{A}}^{-1/2}$. It follows that given $\rho_{AB}$, one may compute the relevant $\cT^{-1}_{\tau,k}$, $\mbf{J}^{-1/2}_{f,\tau}$ maps. Therefore, the remaining limitation is being able to calculate $\Lambda_{\rho}$. The following lemma shows we can do this.
 \begin{lemma}\label{lem:channel-extractor}
     For $\rho_{AB}$, $\Lambda_{\rho} = \cN_{A \to B} \circ \Gamma_{\rho_{A}^{T}}$  where $\Gamma_{\rho_{A}}(X) = \rho_{A}^{1/2}X\rho_{A}^{1/2}$ and $\cN$ is the unique channel in Proposition \ref{prop:every-joint-state-is-a-degraded-purif}. Moreover, $\cN$ is the optimizer for the SDP
     \begin{equation}\label{eq:channel-extractor-SDP}
     \begin{aligned}
         \min & \; t \\
         & \; -t\mbb{1} \leq \rho_{AB} - \Tr[\Phi^{+}_{A\ol{A}} \otimes X_{AB}(\psi_{\rho}^{T}) \otimes \mbb{1}_{\ol{A}B})] \leq t\mbb{1} \\
         & \Tr_{B}[X] = \mbb{1}_{A'} \\
         & X_{A'B} \geq 0 \ ,
     \end{aligned}
     \end{equation}
     where the transpose in the SDP is with respect to the joint space $AA'$.
 \end{lemma}
 \begin{proof}
     By Proposition \ref{prop:every-joint-state-is-a-degraded-purif}, there is a unique channel $\cN_{A \to B}$ such that $(\id_{A} \otimes \cN)(\psi_{\rho}) = \rho_{AB}$. Now noting $\psi_{\rho} \coloneq \rho_{A}^{1/2}\Phi^{+}_{A\ol{A}}\rho_{A}^{1/2} = \Gamma_{\rho_{A}^{T}}(\Phi^{+})$ where we used the transpose trick and that the square root operation commutes with the transpose. Thus we may conclude $\Lambda_{\rho} = \cN \circ \Gamma_{\rho_{A}^{T}}$. This completes the first point. To obtain the second point, note for any channel $\cE_{A \to B}$, we have some state $\sigma^{\cE}_{AB} = (\id_{A} \otimes \cE)(\psi_{\rho})$. As Proposition \ref{prop:every-joint-state-is-a-degraded-purif} tells us $\cN$ exists and is unique, $\sigma^{\cE}_{AB} = \rho_{AB}$ if and only if $\cE = \cN$. Thus, minimizing any definite measure of difference between $\sigma^{\cE}_{AB}$ and $\rho_{AB}$ over all channels $\cE$ will result in the optimizer $\cN$. The SDP in \eqref{eq:channel-extractor-SDP} is $\min\{ \Vert \rho_{AB} - (\id_{A} \otimes \cE)(\psi_{\rho}) \Vert_{\infty} : \cE \in \Channel(A,B) \} $ as follows:
     \begin{equation}
     \begin{aligned}
         &\min\{ \Vert \rho_{AB} - (\id_{A} \otimes \cE)(\psi_{\rho}) \Vert_{\infty} : \cE \in \Channel(A',B) \} \\
         \iff &\min\{t: -t\mbb{1} \leq \rho_{AB} - (\id_{A} \otimes \cE)(\psi_{\rho}) \leq t\mbb{1} : \cE \in \Channel(A',B) \} \\
         \iff &\min\{t: -t\mbb{1} \leq \rho_{AB} - \Tr_{AA'}[\Omega_{\id \otimes \cE}(\psi_{\rho}^{T} \otimes \mbb{1}_{})] \leq t\mbb{1} : \cE \in \Channel(A',B) \} \\
         \iff &\min\{t: -t\mbb{1} \leq \rho_{AB} - \Tr_{AA'}[\Phi^{+}_{A\ol{A}} \otimes \Omega_{\cE}(\psi_{\rho}^{T} \otimes \mbb{1}_{\ol{A}B})] \leq t\mbb{1} : \cE \in \Channel(A',B) \} \\
         \iff &\min\{t: -t\mbb{1} \leq \rho_{AB} - \Tr_{AA'}[\Phi^{+}_{A\ol{A}} \otimes X(\psi_{\rho}^{T} \otimes \mbb{1}_{\ol{A}B})] \leq t\mbb{1} : X_{AB} \geq 0 \, , \, \Tr_{B}[X] \geq 0 \} \ ,
     \end{aligned}
     \end{equation}
     where the first equivalence is because the difference is Hermitian, so its infinity norm is the largest magnitude eigenvalue, the second is \eqref{eq:action-of-Choi}, the third is that $\Phi^{+}_{(A'A)(\ol{A}'\ol{A})} = \Phi^{+}_{A'\ol{A}'} \otimes \Phi^{+}_{A\ol{A}}$, and the fourth is the necessary and sufficient conditions of the Choi operator of a quantum channel. Thus, we have shown that we have an SDP whose unique minimizer is $\cN$. This completes the proof.
 \end{proof}

\begin{remark}
    Note that the need for this SDP is a fully quantum aspect. If $\rho_{AB} = p_{XY}$, then Bayes rule would be enough to determine the relevant channel $W_{Y \vert X}$, but it also wouldn't be necessary to find as one can simply compute the singular values of the matrix $p_{X}^{-1/2}p_{XY}p_{Y}^{-1/2}$ directly to determine $\mu(p_{XY})$.
\end{remark}

With Lemma \ref{lem:channel-extractor}, we may conclude the $\mu_{f}(\rho_{AB})$ can be calculated by appealing to an SDP subroutine if given $\rho_{AB}$. We give one such approach in the following proof, but there may be other approaches.
\begin{tcolorbox}[width=\linewidth, sharp corners=all, colback=white!95!black, boxrule=0pt,frame hidden]
\begin{theorem}\label{thm:f-max-corr-computability}
    Let $f$ be a normalized, operator monotone. If $f$ is poly-time computable, given $\rho_{AB} \in \Density(AB)$, $\mu_{f}(A:B)_{\rho}$ for $f \in \cM_{\text{St}}$ and $\mu^{\text{Lin}}_{f_{k}}(A:B)_{\rho}$ for $k \in [0,1]$ are efficient to compute. 
\end{theorem}
\end{tcolorbox}
\begin{proof}
    We explain the case where one is given $(\cE_{A \to B},\rho_{A})$. If one is instead provided $\rho_{AB}$, one must extract $\cE$ using Lemma \ref{lem:channel-extractor}.
    
    We begin with how to compute $\mu_{f}(\rho_{AB})$ and then explain the modifications for $\mu_{f_{k}}^{\text{Lin}}(\rho_{AB})$. By Proposition \ref{prop:max-corr-to-eig}, our goal is to compute the square root of the second eigenvalue of $\Lambda^{\ast}_{\wt{\rho}_{f}} \circ \Lambda_{\wt{\rho}_{f}}$ with respect to the HS inner product. Thus, we need to compute the standard matrix of this linear map on a basis of $\Herm(A)$ with respect to the HS inner product and obtain its second largest eigenvalue. The set $\{\pi_{A}^{1/2}\} \cup \{G_{n,m}\}_{n,m}$ forms an ONB of $\Herm(A)$ with respect to HS inner product, so we do not need to appeal to the Gram-Schmidt process. Thus all we need to do is be able to compute $\Lambda^{\ast}_{\wt{\rho}_{f}} \circ \Lambda_{\wt{\rho}_{f}}$ on this basis to build the standard matrix. By \eqref{eq:Lambda-wt-rho-f-structure}, $\Lambda_{\wt{\rho}_{f}} = \mbf{J}_{f,\rho_{B}}^{-1/2} \circ \cN \circ \Gamma_{\rho_{A}^{T}} \circ (\mbf{J}_{f,\rho_{A}}^{-1/2})^{T}$, so it suffices to obtain the Kraus operators of each of these maps individually and compose them. Therefore, we show how to compute these maps individually.
    
    One may compute the Choi operators $\Omega_{\mbf{J}_{f,\rho_{A}}^{-1/2}},\Omega_{\mbf{J}_{f,\rho_{B}}^{-1/2}}$  using \eqref{eq:Hadamard-prod-form}. As \eqref{eq:Hadamard-prod-form} involves computing the given map $d_{A}^{2}$ times, by appealing to Lemma \ref{lem:efficiency-of-Jfpsigma}, computing $\Omega_{\mbf{J}_{f,\rho_{A}}^{-1/2}},\Omega_{\mbf{J}_{f,\rho_{B}}^{-1/2}}$ may be computed in time $O(d_{A}^{7})$. We may then extract the Kraus operators by taking the singular value decomposition (SVD) and reshaping the matrices. This has time complexity $O(d_{A}^{6}) + O(d_{A}^{4}) = O(d_{A}^{6})$. Therefore, we may obtain the Kraus decompositions efficiently in time $O(d^{7}_{A})$. 
     
    Now let $(\{A_{x}\}_{x \in \cX},\{B_{x}\}_{x \in \cX})$, $(\{C_{y}\}_{y \in \cY},\{D_{y}\}_{y \in \cY})$, $(\{E_{z}\}_{z \in \cZ},\{F_{z}\}_{z \in \cZ})$ denote the Kraus representations ${\cN}$, $\mbf{J}_{f,\rho_{A}}^{-1/2}$, and $\mbf{J}_{f,\rho_{B}}^{-1/2}$ respectively.\footnote{Note $A_{x} = B_{x}$ for all $x \in \cX$ as $\cN$ is completely-positive.} We then may define the Kraus operators of $\Lambda_{\wt{\rho}_{f}}$:
    \begin{align}
         G_{x,y,z} \leftarrow  E_{z}A_{x}{\rho_{A}^{T}}^{1/2}C_{y}^{T} \quad \quad L_{x,y,z} \leftarrow D^{T}_{y}{\rho_{A}^{T}}^{1/2}B_{x}F_{z} \quad \forall (x,y,z) \in \cX \times \cY \times \cZ \, .
    \end{align}
    It follows the Kraus representation of $\Lambda_{\wt{\rho}_{f}}^{\ast} \circ \Lambda_{\wt{\rho}_{f}}$ is given by
    \begin{align}
        \left(\{M_{w}\}_{w \in \cW} \coloneq \{G_{x_{2},y_{2},z_{2}}^{\ast}G_{x_{1},y_{1},z_{1}}\}_{\substack{(x_{1},y_{1},z_{1}) \\ (x_{2},y_{2},z_{2})}} \quad \{R_{w}\}_{w \in \cW}\coloneq \{L_{x_{2},y_{2},z_{2}}^{\ast}L_{x_{1},y_{1},z_{1}}\}_{\substack{(x_{1},y_{1},z_{1}) \\ (x_{2},y_{2},z_{2})}} \right) \ .
    \end{align}
    To bound the number of Kraus operators, note $\cX$ depends on the Kraus representation obtained/given, but without loss of generality $\vert \cX \vert \leq d_{A}d_{B}$ by obtaining the canonical Kraus representation as we have done in the other cases. Thus, we can take the bound $\vert \cX \times \cY \times \cZ \vert \leq  d_{A}d_{B} \cdot d_{A}^{2} \cdot d_{B}^{2} = d_{A}^{3}d_{B}^{3}$ without loss of generality. Thus, this method results in at most $d_{A}^{6}d_{B}^{6}$ Kraus operators for $\Lambda_{\wt{\rho}_{f}}^{\ast} \circ \Lambda_{\wt{\rho}_{f}}$.\footnote{We note this is clearly suboptimal as its canonical Kraus representation would have at most $d_{A}^{2}$ Kraus operators.} Thus, one may compute the action of this map on the ONB of $\Herm(A)$ to obtain its standard matrix and solve for its eigenvalues. By Proposition \ref{prop:max-corr-to-eig}, taking the square root of the second largest eigenvalue obtains $\mu_{f}(\rho_{AB})$. This completes the proof for $\mu_{f}(\rho_{AB})$. 
    
    For the case of $\mu^{\text{Lin}}_{f_{k}}(\rho_{AB})$, by Proposition \ref{prop:max-corr-to-eig}, we need to compute the second eigenvalue of $\Lambda^{\ast}_{\wt{\rho}_{k}} \circ \Lambda_{\wt{\rho}_{k}}: \Lin(A) \to \Lin(A)$. One can efficiently generate an ONB of $(\Lin(A),\langle \cdot, \cdot \rangle)$, e.g. the generalized Gell Man matrices together with the square root of any density matrix will work as before. Thus, one just needs to be able to compute $\Lambda^{\ast}_{\wt{\rho}_{k}} \circ \Lambda_{\wt{\rho}_{k}}$. From Lemma \ref{lem:norm-bound}, 
    $$\Lambda_{\wt{\rho}_{k}} = \cT^{-1}_{\rho_{B},k} \circ \cN_{A \to B} \circ \Gamma_{\rho_{A}^{T}} \circ \cT^{-1}_{\ol{\rho}_{A},1-k} \ , $$
    where we remind the reader that $\cT^{-1}_{\tau,\gamma}(X) = \tau^{-\gamma/2}X\tau^{-(1-\gamma)/2}$. Letting $(A_{k},B_{k})$ denote the Kraus operators of $\cN$, a direct calculation will determine that the left and right Kraus operators of $\Lambda_{\wt{\rho}_{k}}$, denoted $G_{k}$ and $L_{k}$ respectively, are therefore
    \begin{align}
        G_{k} = \rho_{B}^{-k/2}A_{k}\sqrt{\rho_{A}^{T}} \ol{\rho}_{A}^{-(1-k)/2} \quad L_{k} = \rho_{B}^{-(1-k)/2}B_{k}\sqrt{\rho_{A}^{T}}^{\ast}{\ol{\rho}_{A}^{-k/2}}^{\ast}
    \end{align}
    
    The rest of the proof follows the same idea as the previous case. This completes the proof.
\end{proof}

\begin{remark} In the case $f = f_{GM}$, an alternative method is to get $\cN$ via Lemma \ref{lem:channel-extractor} and then apply Theorem \ref{thm:contraction-coeff-computability} for $\eta_{GM}(\cN,\rho_{A})$ given Corollary \ref{cor:contraction-for-sandwiched-case}.
\end{remark}

\subsection{On Numerical Stability and Practicality}\label{subsec:num-stability}
In this section, we have established the efficiency of computing the contraction coefficients and correlation coefficients. However, we have omitted the more practical consideration of error propagation and numerical stability. For accessibility, an initial implementation of the algorithms described in the programming language Julia \cite{bezanson2017julia} is publicly provided at \href{https://github.com/qit-george/OperatorMonotoneCorrelationTools}{this GitHub repository}. However, we have found these straightforward implementations to be too numerically unstable for random states and choices of channels for the authors to be comfortable drawing further insights from the numerical implementation. We suggest two causes of this error: the Gram-Schmidt process and the use of Kraus operators.

With regards to the Gram-Schmidt process, it is known that when the inner product space is $\mbb{R}^{n}$ equipped with the Euclidean inner product, the Gram-Schmidt process propagates error in a way that can be somewhat improved using the `modified Gram-Schmidt process' \cite{bjorck1967solving}. However, other methods are known to outperform even the modified Gram-Schmidt process in terms of numerical stability under these conditions (See \cite{stewart2022numerical} for further information on numerical methods). As our algorithms go beyond the traditional use case of the Euclidean inner product on $\mbb{R}^{n}$ as we need to consider more general inner product spaces on linear operators, it is unclear how we can currently remove the use of the (modified) Gram-Schmidt process. Furthermore, as our algorithms deviate from these previous analyses by using more general inner product spaces, a full study of the numerical stability cannot be directly lifted from previous work. 

Second, our algorithms above require extracting the Kraus operators from the Choi representation. This requires using the algorithm for SVD, which itself has numerical error,\footnote{We remark one reason the first author chose the programming language Julia is that it is open access. However, we found that Julia's linear algebra package seems to, at the time of writing this paper, be less exact at the singular value decomposition compared to Matlab at least in certain simple cases. As such, from a practical standpoint, it is possible one can somewhat improve the stability simply by switching to Matlab.} and then computing the action of the map, which will have further numerical error. It follows this likely propagates more error than necessary. A possible alternative approach would be to work solely with Choi operators. 

Given the scope of this work, we leave the numerical error analysis and development of a robust algorithmic implementation to future work.

\section{Conclusion and Outlook}\label{sec:conclusion}

In this work we have lifted the majority of the classical information-theoretic framework of the maximal correlation coefficient and its relation to the $\chi^{2}$-divergence's input-dependent contraction coefficient to the quantum setting. This has resolved a variety of gaps in our understanding of the quantum generalizations of these two concepts. In particular, we have:
\begin{enumerate}[itemsep=0pt]
\item determined the relevant quantum maximal correlation coefficients for one-shot and asymptotic distillation of a single bit of correlation (Theorems \ref{thm:extreme-values-summary} and \ref{thm:asymptotic-data});
\item established a family of strong monotones for converting quantum states under local operations (Theorem \ref{thm:k-correlation-nec-for-local-processing}) as well as a related more generalization majorization result (Theorem \ref{thm:mu-GM-as-monotone}). These results generalize a main result of \cite{Beigi-2013a};
\item established the general relation between quantum maximal correlation coefficients and input-dependent quantum $\chi^{2}$ contraction coefficients (Theorem \ref{thm:correspondence-between-contraction-coeffs-and-max-corr-coeffs});
\item established new equivalent conditions to a quantum $\chi_{f}^{2}$-divergence saturating the data processing inequality (Theorem \ref{thm:DPI-with-equality});
\item clarified the framework of mixing times of time-homogeneous quantum Markov chains by verifying that the HT $f$-divergences recover the classical asymptotic relation of \cite{GZB-preprint-2024} (Theorem \ref{thm:mixing-rate});
\item and established that the majority of the quantities investigated in this paper are efficient to compute (Theorems \ref{thm:contraction-coeff-computability} and \ref{thm:f-max-corr-computability}), which results in a generic algorithm for computing mixing times of primitive quantum channels (Theorem \ref{thm:computable-mixing-times}).
\end{enumerate}

However, there remain at least two interesting open problems in this framework: \\
\begin{itemize}
    \item \textbf{Does there exist a quantum state $\rho_{AB}$ such that ${\mu_{AM}(A:B)_{\rho} < 1}$ and ${\mu_{GM}(A:B)_{\rho}=1}$?} Operationally, combining Corollary \ref{cor:contraction-for-sandwiched-case}, Theorem \ref{thm:contraction-coeff-equal-unity}, and Item 3 of Theorem \ref{thm:extreme-values-summary}, this would show there is a separation between reversibility and correlation distillation under local operations that does not exist classically. Namely, it would be equivalent to there exist quantum states $\rho_{AB} = (\id_{A} \otimes \cE)(\psi_{\rho_{A}})$ and $\sigma_{A} \neq \rho_{A}$ such that  $\sigma_{A} = (\cP_{\cE,\rho_{A}} \circ \cE)(\sigma_{A})$, but there does not exist local operations to convert $\rho_{AB}$ into perfect correlation $\chi^{\vert p}_{XX'}$ for $p \not \in \{0,1\}$. This is not possible classically given that all maximal correlation coefficients are the same on classical distributions (Proposition \ref{prop:recovers-classical}).
    \item \textbf{Does there exist a quantum state ${\rho_{AB}}$ such that ${\mu_{AM}(A:B) < 1}$ and ${\lim_{n \to \infty} \mu_{AM}(A^{n}:B^{n})_{\rho^{\otimes n}}} = 1$?} Note by Theorem \ref{thm:asymptotic-data} this could only be true if the first question were answered in the affirmative. An answer to this second question in the affirmative would then show that correlation distillation with local operations is fundamentally different than in the classical case.
\end{itemize}

Beyond these open problems, we note that a key technical and conceptual step of this work was to elucidate the general framework of non-commutative $L^{2}(p)$ spaces induced by operator monotone functions $f$, the $L^{2}_{f}(\sigma)$ spaces, and show they have uses for lifting other statistical quantities to quantum mechanics. It would be interesting to see which canonical results in classical probability theory lift to the quantum setting through this framework and for which operator monotone function $f$. It is our hope this can further be used in other applications of quantum information processing where the classical information-theoretic methods benefit from using an operator-theoretic approach via the $L^{2}(p)$ space.

\section{Acknowledgements}
We thank Salman Beigi for helpful discussions as well as feedback on a preliminary draft of this work, Mil\'{a}n Mosonyi for a variety of helpful comments including bringing to our attention \cite[Theorem 8]{jenvcova2012reversibility} as well as that \cite{hiai-2012quasi} introduced the same definition of covariance as in this work, Afham for asking IG about the relation of this work to `quantum states over time,' John Goold for asking about the relation between Theorem \ref{thm:computable-mixing-times} and spectral gap methods, Daniel Stilck Fran\c{c}a for a discussion on the relation between input-dependent $\chi^{2}$ contraction coefficients and spectral gap methods, Ryuji Takagi for bringing \cite{generalized-QSL} to our attention, and Matthew Simon Castaneda Tan for related discussions about quantum divergences. This research is supported by the Ministry of Education, Singapore, through grant T2EP20124-0005 and the NRF Investigatorship award (NRF-NRFI10-2024-0006).

\bibliography{References.bib}

\appendix

\section{Technical Lemmata}

\subsection{Derivation of \eqref{eq:func-expansion}}\label{app:func-expansion}
Here we derive \eqref{eq:func-expansion} in detail for completeness. It is a special case of a known expansion of the function of the `relative modular operator' $L_{P}R_{Q^{-1}}$ (see e.g. \cite{HIAI_2011,Hiai-2017a}), so we provide a derivation of this more general form.

We begin with some preliminaries on functional calculus of linear maps. Recall from standard quantum information theory textbooks, e.g. \cite{WatrousBook}, that for $I \subset \mbb{C}$, function $f:I \to \mbb{C}$, and a matrix that is a normal operator $N$ with spectral decomposition $N = \sum_{i} \lambda_{i} \dyad{e_{i}}$ such that $\{\lambda_{i}\}_{i} \subset I$, we define $f(N) = \sum_{i} f(\lambda_{i}) \dyad{e_{i}}$. This is functional calculus for matrices (see e.g. \cite[Chapter 3]{hiai2014introduction}). However, functional calculus works for `normal operators' in the general sense of Hilbert spaces (see e.g. \cite{Conway-1985a}), which means we can in particular apply it to linear maps $\Phi: \Lin(A) \to \Lin(A)$ where $A$ is finite-dimensional. Namely, let $\{\zeta_{k}\}_{k} \subset \mbb{C}$, $\{Z_{k}\}_{k} \subset \Lin(A)$ be the eigenvalues and eigenvectors of $\Phi$, i.e. $\Phi(Z_{k}) = \zeta_{k} Z_{k}$ for each $k$. Then for $f:I \to \mbb{C}$ such that $\{\zeta_{k}\} \subset I$, \begin{align}\label{eq:function-of-lin-map}
    f(\Phi) \coloneq \sum_{k} f(\lambda_{k})Z_{k} 
\end{align} 
is well-defined and defines the linear map $f(\Phi): \Lin(A) \to \Lin(A)$ whenever $\Phi$ has eigenvalues that lay inside the interval $I$. We will use this to prove \cite[Eq.~2.1]{HIAI_2011}, which recovers \eqref{eq:func-expansion} via the special case $P = Q = \sigma$ and $f$ as given in Definition \ref{def:J-operator}.

\begin{proposition} (\cite[Eq.~2.1]{HIAI_2011}) \label{prop:func-of-mod-op}
    For positive operators $P,Q$ with spectral decompositions $P = \sum_{i} \lambda_{i} \dyad{e_{i}} \eqqcolon \sum_{i} \lambda_{i}P_{i}$, $Q = \sum_{j} \omega_{j} \dyad{f_{j}} \eqqcolon \sum_{j} \omega_{j} Q_{j}$, define the `relative modular operator' $L_{P}R_{Q^{-1}}$ which uses the left and right multiplication operators as defined in \eqref{eq:left-and-right-mult-operators}. Let $f: I \to \mbb{C}$ such that $\{\lambda_{i}/\omega_{j}\}_{i,j} \subset I$. Then
    \begin{align}\label{eq:func-of-mod-op}
        f(L_{P}R_{Q^{-1}}) = \sum_{i,j} f(\lambda_{i}/\omega_{j})L_{P_{i}}R_{Q_{j}} \, .
    \end{align}
\end{proposition}
\begin{proof}
    This is a straightforward generalization of the proof of \cite[Lemma 42]{Leditzky2025}. We will find a decomposition of $f(L_{P}R_{Q^{-1}})$ using functional calculus and then show the RHS of \eqref{eq:func-of-mod-op} defines the same map. To this end, we first prove $\{\ket{e_{i}}\bra{f_{j}}\}_{i,j}$ is a set of eigenvectors for $L_{P}R_{Q^{-1}}$. Namely,
    \begin{align}
        L_{P}R_{Q^{-1}}(\ket{e_{i}}\bra{f_{j}})= P\ket{e_{i}}\bra{f_{j}}Q^{-1} = \lambda_{i}\ket{e_{i}}\bra{f_{j}}\omega_{j}^{-1} = (\lambda_{i}/\omega_{j})\ket{e_{i}}\bra{f_{j}} \ , 
    \end{align}
    where we just used the definition of left and right multiplication operators and the assumed spectral decomposition of the operators. As $\{\ket{e_{i}}\bra{f_{j}}\}_{i,j}$ form a basis for $\Lin(A)$, this defines all the eigenvectors of $L_{P}R_{Q^{-1}}$. Thus, by \eqref{eq:function-of-lin-map} along with our assumption on $f$ and $\{\lambda_{i}/\omega_{j}\}_{i,j}$, $f(L_{P}R_{Q^{-1}}) = \sum_{i,j} f(\lambda_{i}/\omega_{j})\ket{e_{i}}\bra{f_{j}}$.

    Similarly, define the map $\Phi = \sum_{i,j} f(\lambda_{i}/\omega_{j})L_{P_{i}}R_{Q_{j}}$. Then for any $\ket{e_{i}}\bra{f_{j}}$,
    \begin{align}
        \Phi(\ket{e_{i}}\bra{f_{j}}) =& \sum_{i',j'} f(\lambda_{i'}/\omega_{j'})L_{P_{i'}}R_{Q_{j'}}(\ket{e_{i}}\bra{f_{j}}) \\
        =& \sum_{i',j'} f(\lambda_{i'}/\omega_{j'}) P_{i'}\ket{e_{i}}\bra{f_{j}}Q_{j'} \\
        =& f(\lambda_{i}/\omega_{j}) \ket{e_{i}}\bra{f_{j}} \\
        =& f(L_{P}R_{Q^{-1}})(\ket{e_{i}}\bra{f_{j}}) \ ,
    \end{align}
    where we used the spectral decomposition again and the previous calculation's result. As $\{\ket{e_{i}}\bra{f_{j}}\}_{i,j}$ are a basis for $\Lin(A)$, we may conclude that $\Phi = f(L_{P}R_{Q^{-1}})$, which concludes the proof.
\end{proof}

\subsection{Proof of Lemma \ref{lem:map-norms-as-optimizations}}
\begin{proof}[Proof of Lemma \ref{lem:map-norms-as-optimizations}]
    By \cite[Corollary 6.7]{Conway-1985a}, for any $x \in X$, 
    $$\Vert \Lambda(x) \Vert_{Y} = \sup\{ \vert f(\Lambda(x)) \vert : f \in Y^{\ast} \, , \Vert f \Vert_{Y^{\ast}} \leq 1 \} \ , $$
    where $Y^{\ast}$ is the dual space of $Y$. As $\mathscr{Y}$ is assumed to be a Hilbert space, the Riesz Representation theorem, e.g. \cite{Conway-1985a}, guarantees for each $f \in Y^{\ast}$ such that $\Vert f \Vert_{Y^{\ast}} = a$, there is $y_{f} \in Y$ with $\Vert y_{f} \Vert_{Y} = a$ such that $\vert f(\Lambda(x)) \vert = \vert \langle y_{f}, \Lambda(x) \rangle \vert$. Therefore, using the definition of the operator norm,
    $$ \Vert \Lambda \Vert_{\mathscr{X} \to \mathscr{Y}} = \sup \{ \vert \langle y, \Lambda(x) \rangle_{Y} \vert : x \in X, \, y \in Y, \, \Vert x \Vert_{X} \leq 1 \, , \Vert y \Vert_{Y} \leq 1 \} \ . $$
    As we work in finite dimensions, $\Lambda$ is bounded, so by Cauchy-Schwarz, one may verify $\langle y, \Lambda(x) \rangle_{Y} : X \times Y \to \mbb{R}$ is continuous. As the closed unit balls are compact in finite dimensions, we may replace the supremum with a maximum. Finally, one can see the maximizers have to have unit norm as otherwise one can increase the value by rescaling them.
\end{proof}

\subsection{Lemmata for Quantum Maximal Correlation Coefficients}\label{app-subsec:max-corr-coeff-lemmata}
In this appendix we collect various further results needed for the study of quantum maximal correlation coefficients.
\begin{proposition}\label{prop:restricted-unital-Schwarz}
    If $\cM_{A' \to A}$ is a unital Schwarz map and $V_{\wt{A} \to A}$ is an isometry, then the map $\cV^{\ast} \circ \cM(\cdot) \coloneq V^{\ast} \cM(\cdot) V$ is a unital Schwarz map.
\end{proposition}
\begin{proof}
    To see that $\cV^{\ast} \circ \cM$ is unital, $\cV^{\ast} \circ \cM(\mbb{1}_{A'}) = \cV^{\ast}(\mbb{1}_{A}) = V^{\ast}\mbb{1}_{A}V = \mbb{1}_{\wt{A}}$ where we used $\cM$ is unital and that $V$ is an isometry.

    To see that it is a Schwarz map, let $X \in \Lin(A')$. Then, using the Schwarz inequality \eqref{eq:Schwarz-Inequality},
    \begin{align}
        \cM(X^{\ast}X) \geq \cM(X)^{\ast}\cM(X) \Rightarrow V^{\ast}\cM(X^{\ast}X)V \geq V^{\ast}\cM(X)^{\ast}\cM(X)V \ , 
    \end{align}
    where we used $\cV^{\ast}$ is a completely positive map. Moreover, as $I_{A} \geq VV^{\ast}$,
    \begin{align}
        V^{\ast}\cM(X)^{\ast}\cM(X)V \geq V^{\ast}\cM(X)^{\ast}VV^{\ast}\cM(X)V = [(\cV \circ \cM)(X)]^{\ast} [(\cV \circ \cM)(X)] \ . 
    \end{align}
    Combining these inequalities and noting $V^{\ast}\cM(X^{\ast}X)V = \cV^{\ast} \circ \cM(X^{\ast}X)$, we have shown $\cV^{\ast} \circ \cM$ satisfies the Schwarz inequality and thus is a Schwarz map.
\end{proof}

\begin{proof}[Proof of Proposition \ref{prop:standard-operator-monotone-maximal-corr}]
    First we transform the problem into a $2\to 2$-norm as we have previously. We have $\langle \mbb{1}, X \rangle_{\mbf{J}_{f,\sigma}} = \langle \sigma^{1/2}, \mbf{J}^{1/2}_{f,\sigma}(X) \rangle$ by Proposition \ref{prop:mult-and-div-under-trace}. We also have
    \begin{align}
        \langle X , X \rangle_{\mbf{J}_{f,\sigma}} = \langle X , (\mbf{J}_{f,\sigma}^{1/2} \circ \mbf{J}_{f,\sigma}^{1/2}(X) \rangle = \langle \mbf{J}_{f,\sigma}^{1/2}(X), \mbf{J}_{f,\sigma}^{1/2}(X) \rangle \ ,
    \end{align}
    where we used that $f$ is symmetry-inducing, $X$ is Hermitian, and $\mbf{J}_{f,\sigma}^{1/2}$ is self-adjoint with respect to Hermitian operators under the Hilbert-Schmidt inner product as follows from Proposition \ref{prop:J-operator-self-adjoint}. Lastly,
    \begin{align}
        \Tr[X \otimes Y\rho_{AB}] = \Tr[(\mbf{J}_{f,\sigma}^{-1/2} \circ \mbf{J}_{f,\sigma}^{1/2})(X) \otimes (\mbf{J}_{f,\sigma}^{-1/2} \circ \mbf{J}_{f,\sigma}^{1/2})(Y)\rho_{AB}] = \Tr[\mbf{J}_{f,\sigma}^{1/2}(X) \otimes \mbf{J}_{f,\sigma}^{1/2}(Y) \wt{\rho}_{f}] \ ,
    \end{align}
    where we again used Proposition \ref{prop:J-operator-self-adjoint}. Thus, by direct calculation, we have
    \begin{equation}\label{eq:mu_f-st-opt}
    \begin{aligned}
        \mu_{f}(A:B)_{\rho} = \sup & \; \Big\vert \Tr[\wt{X} \otimes \wt{Y}^{\ast} \wt{\rho}_{f}] \Big\vert \\
        \text{s.t.} \; & \langle \rho_{A}^{1/2} , \wt{X} \rangle = 0 = \langle \rho_{B}^{1/2} , \wt{Y} \rangle \\
        & \Vert \wt{X} \Vert_{2} = 1 = \Vert \wt{Y} \Vert_{2} \ ,
    \end{aligned}
\end{equation}
where $\wt{X} \in \Herm(A), \wt{Y} \in \Herm(B)$ and $\wt{\rho}_{f} = (\mbf{J}_{f,\rho_{A}}^{-1/2} \otimes \mbf{J}_{f,\rho_{B}}^{-1/2})(\rho_{AB})$.

Next, we let $\Lambda_{\wt{\rho}_{f}}$ be the map such that $\wt{\rho}_{f} = \Omega_{\Lambda_{\wt{\rho}_{f}}}$ and show it has a singular value of one, which shows $\wt{\rho}_{f}$ has a Schmidt coefficient of one with respect to $(\Herm(A), \langle \cdot , \cdot \rangle)$, $(\Herm(B), \langle \cdot, \cdot \rangle)$ by Proposition \ref{prop:Schmidt-to-sing}. Recall\footnote{Or verify using \eqref{eq:action-of-Choi} and the transpose trick.} that for a sequence of linear maps $\cF \circ \Lambda \circ \cE$, the Choi operator is $(\cE^{T} \otimes \cF)\Omega_{\Lambda}$ where $\cE^{T}(Y) \coloneq \sum_{k} A^{T}_{k}Y\ol{B}_{k}$ if $\cE = \sum_{k} A_{k}XB_{k}^{\ast}$, which defines the `transpose' of a linear map. Therefore, 
\begin{align}\label{eq:Lambda-wt-rho-f-structure}
    \Lambda_{\wt{\rho}_{f}} = \mbf{J}_{f,\rho_{B}}^{-1/2} \circ \Lambda_{\rho} \circ (\mbf{J}_{f,\rho_{A}}^{-1/2})^{T} \ . 
\end{align} 
We can thus show $\Lambda_{\wt{\rho}_{f}}$ has a singular value $1$. Namely, defining the transpose with respect to the eigenbasis of $\rho_{A}$,
\begin{equation}
\begin{aligned}
    \Lambda_{\wt{\rho}_{f}}(\rho_{A}^{1/2}) =& \Tr_{A}[{\rho_A^{1/2}}^{T} \otimes I_{B} \Omega_{\Lambda_{\wt{\rho}_{f}}}] \\
    =& \Tr_{A}[\rho_A^{1/2}\otimes I_{B}(\mbf{J}_{f,\rho_{A}}^{-1/2} \otimes \mbf{J}_{f,\rho_{B}}^{-1/2})(\rho_{AB})] \\
    =& \Tr_{A}[\mbf{J}_{f,\rho_{A}}^{-1/2}(\rho_A^{1/2}) \otimes I_{B}(\text{id}_{A} \otimes \mbf{J}_{f,\rho_{B}}^{-1/2})(\rho_{AB})] \\
    =& \Tr_{A}[(\id_{A} \otimes \mbf{J}_{f,\rho_{B}}^{-1/2}(\rho_{AB})] \\
    =& (\mbf{J}_{f,\rho_{B}}^{-1/2})(\rho_{B}) \\
    =& \rho_{B}^{1/2} \ ,
\end{aligned}
\end{equation}
where the first equality is \eqref{eq:action-of-Choi}, the second is our identifications, the third is $\mbf{J}_{f,\rho_{A}}^{-1/2}$ is self-adjoint on Hermitian operators, and the fourth and sixth are because $\rho_{A}$ and $\rho_{B}$ commute with themselves. By Proposition \ref{prop:Schmidt-to-sing}, this means $\wt{\rho}_{f}$ has a Schmidt coefficient with respect to $(\Herm(A), \langle \cdot , \cdot \rangle)$, $(\Herm(B), \langle \cdot, \cdot \rangle)$ of one. That is, we may write $\wt{\rho}_{f} = \sum_{i} \lambda_{i} M_{i} \otimes N_{i}$ where $\{\lambda_{i}\}_{i}$ are real and decreasing and $\{M_{i}\}_{i} \subset \Herm(A)$, $\{N_{i}\}_{i} \in \Herm(B)$ are orthonormal bases (with respect to the Hilbert-Schmidt inner product) and we know there is an $i$ such that $\lambda_{i} M_{i} \otimes N_{i} = \rho_{A}^{1/2} \otimes \rho_{B}^{1/2}$.

Now, if there exists $i'$ such that $\lambda_{i'} > 1$, then $(M_{i'},N_{i'})$ is feasible for \eqref{eq:mu_f-st-opt} and, using orthonormality,
\begin{align}
    \Tr[M_{i'} \otimes N_{i'} \wt{\rho}_{f}] = \sum_{i} \lambda_{i} \langle M_{i'}^{\ast}, M_{i} \rangle \langle N_{i'}^{\ast},N_{i}\rangle = \lambda_{i} \leq \mu_{f}(\rho_{AB}) \ ,
\end{align}
so $\mu_{f}(A:B)_{\rho} > 1$. However, for $f \geq f_{GM}$, $1 \geq \mu_{f_{1/2}}^{\text{Lin}}(A:B)_{\rho} \geq \mu_{GM}(A:B)_{\rho} \geq \mu_{f}(A:B)_{\rho}$ where the first inequality is Lemma \ref{lem:map-norm-is-1-for-k-correlation-coeff} and the second is \eqref{eq:ordering-of-maximal-corr-means}. This completes the proof.
\end{proof}

\subsubsection{Lemmata for Extreme Values}
For clarity, we recall the Holevo-Helstrom theorem.
\begin{proposition}\label{prop:holevo-helstrom}
    Let $\rho^{0},\rho^{1} \in \Density(A)$ and $\lambda \in [0,1]$. For any binary POVM $\{M_{0},M_{1}\}$, the success probability of determining the value $0$ or $1$ correctly is bounded above by
    \begin{align}
      \lambda \Tr[M_{0}\rho^{0}] + (1-\lambda) \Tr[M_{1}\rho^{1}] \leq \frac{1}{2} + \frac{1}{2}\Vert \lambda \rho^{0} + (1-\lambda)\rho^{1} \Vert_{1} \ . 
    \end{align}
    Moreover, there exists a measurement that achieves this bound. 
\end{proposition}
Note that because $\Vert X \Vert_{1} = 0$ if and only if $X = 0$. This shows the optimal measurement correctly guesses the value of the index with probability greater than $1/2$ unless $\lambda = 1/2$ and $\rho^{0} = \rho^{1}$.

\begin{proof}[Proof of Lemma \ref{lem:indep-if-and-only-if-measurements}]
    The first direction just follows from if $\rho_{AB} = \rho_{A} \otimes \rho_{B}$, then $(\cM \otimes \cN)(\rho_{A} \otimes \rho_{B}) = p_{X} \otimes p_{Y}$. We thus focus on proving the other direction. That is, we assume there does not exist local measurements $\cM_{A \to X},\cN_{B \to Y}$ such that $p_{XY} \neq p_{X} \otimes p_{Y}$. The basic proof idea is to build measurements that correlate the two spaces until we have enough structure that we can use properties of vector spaces to prove that $\rho_{AB} = \rho_{A} \otimes \rho_{B}$.
    
    For a two-outcome measurement channel $\cM_{A \to X}$, define $\rho^{\cM}_{XB} \coloneq (\cM_{A \to X} \otimes \id_{B})(\rho_{AB})$. First we show that for any two-outcome measurement on Alice's side such that $\cM(\rho_{A})$ is full rank, it must be the case $\rho^{\cM}_{XB} = \rho^{\cM}_{X} \otimes \rho^{\cM}_{B}$ or we contradict our assumption. Let $\rho_{XB}^{\cM} = \sum_{x \in \{0,1\}} p(x)\dyad{x} \otimes \rho^{x}_{B}$ where $p(0) \in (0,1)$. Then by Proposition \ref{prop:holevo-helstrom}, there exists a two-outcome measurement on $B$ that outputs a guess $\hat{x}$ such that $\Pr[x = \hat{x}] = \frac{1}{2} + \frac{1}{2} \Vert p(0)\rho^{0}_{B} - p(1)\rho^{1}_{B} \Vert_{1} > 1/2$ unless $\rho^{0}_{B} = \rho^{1}_{B}$. Thus either $\rho^{\cM}_{XB} = \rho_{X}^{\cM} \otimes \rho_{B}^{\cM}$ or Bob can apply the Holevo-Helstrom measurement to result in an outcome stored in register $X'$ that is positively correlated with the value stored in register $X$. Thus, either $\rho^{\cM}_{XB} = \rho_{X}^{\cM} \otimes \rho_{B}^{\cM}$ or we contradict our assumption. Thus, for any two-outcome measurement on Alice's side such that $\cM(\rho_{A})$ is full rank it must be the case $\rho^{\cM}_{XB} = \rho^{\cM}_{X} \otimes \rho^{\cM}_{B}$. Note that if Alice's two-outcome measurement is not full rank, then one of the outcomes happens with probability one and thus $\rho^{\cM}_{XB} = \rho^{\cM}_{X} \otimes \rho_{B}$. Therefore, for any two-outcome measurement on Alice's side, the resulting state is independent, i.e. $\rho^{\cM}_{XB} = \rho^{\cM}_{X} \otimes \rho^{\cM}_{B}$.
    
    We next show that $\rho_{B}^{\cM}$ must be independent of the choice of $\cM$. Let there exist two-outcome measurements $\cM_{A \to X}$, $\cM'_{A \to X}$ such that $\rho_{B}^{\cM} \neq \rho_{B}^{\cM'}$. Again by Proposition \ref{prop:holevo-helstrom}, $\rho^{\cM}_{B}$ and $\rho^{\cM'}_{B}$ may be successfully distinguished with some probability strictly greater than half. Therefore, let $\ol{\cM}_{A \to X}$ be the measurement that flips a coin to use $\cM$ or $\cM'$, storing the choice as $x \in \{0,1\}$ and throws out the value of the measurement outcome. Then, $\rho^{\ol{\cM}}_{XB} = \frac{1}{2}\dyad{0} \otimes \rho^{\cM}_{B} + \frac{1}{2}\dyad{1} \otimes \rho^{\cM'}_{B}$. Again using Proposition \ref{prop:holevo-helstrom}, one may conclude that there exists a measurement on Bob's side that correlates the $X$ system, which contradicts the assumption. Therefore $\rho_{XB}^{\cM} = \rho_{X}^{\cM} \otimes \rho_{B}$ for all two-outcome measurements.
    
    The rest of the proof is using linear algebra to show that $\rho_{XB}^{\cM} = \rho_{X}^{\cM} \otimes \rho_{B}$ for all two-outcome $\cM$ implies $\rho_{AB} = \rho_{A} \otimes \rho_{B}$. As Bob's state is $\rho_{B}$ independent of Alice's measurement outcome, Bob's conditional state is $\rho_{B}$ scaled by the probability of Alice getting the outcome corresponding to the POVM element. As a POVM element is $0 \leq M \leq \mbb{1}$, it follows $\Tr_{A}[M \otimes \mbb{1}_{B}\rho_{AB}] = \Tr[M \otimes \mbb{1}_{B}\rho_{AB}]\rho_{B}$ for all $0 \leq M \leq \mbb{1}$. Using any Hermitian operator decomposes into the difference of two positive operators, and that any linear operator is decomposed into two Hermitian operators where one is multiplied by an imaginary number, we may use linearity of partial trace to ultimately conclude 
    \begin{align}\label{eq:linear-operator-restriction}
        \Tr_{A}[Z \otimes \mbb{1}\rho_{AB}] = \Tr[Z \otimes \mbb{1}\rho_{AB}]\rho_{B} \quad \forall Z \in \Lin(A) \ .
    \end{align}
    Denote the Schmidt decomposition of $\rho_{AB}$ with respect to Hilbert spaces $(\Lin(A),\langle \cdot, \cdot \rangle)$, $(\Lin(B), \langle \cdot, \cdot \rangle)$ as $\rho_{AB} = \sum_{i} \lambda_{i} N_{i} \otimes M_{i}$ where $\lambda_{i} \geq 0$, $\{N_{i}\}_{i}$ and $\{M_{i}\}_{i}$ are orthonormal sets. Then by the orthonormality,
    \begin{align}\label{eq:indep-lin-const-1}
        \Tr_{A}[N_{i}^{\ast} \otimes \mbb{1}\rho_{AB}] = \lambda_{i} M_{i} \quad \forall i \ .
    \end{align}
    Using that by the Schmidt decomposition $\Tr[N_{i}^{\ast} \otimes \mbb{1} \rho_{AB}] = \lambda_{i}\Tr[M_{i}]$ and \eqref{eq:linear-operator-restriction},
    \begin{align}\label{eq:indep-lin-const-2}
        \Tr_{A}[N_{i}^{\ast} \otimes \mbb{1}\rho_{AB}] = \lambda_{i}\Tr[M_{i}]\rho_{B} \quad \forall i \ . 
    \end{align}
    Combining \eqref{eq:indep-lin-const-1} and \eqref{eq:indep-lin-const-2}, one obtains $\lambda_{i}M_{i} = \lambda_{i}\Tr[M_{i}]\rho_{B}$ for all $i$. Moreover, taking the partial trace over $B$, the Schmidt decomposition shows $\rho_{A} = \sum_{i}\lambda_{i}\Tr[M_{i}]N_{i}.$ Combining these points, $\rho_{AB} = \sum_{i} \lambda_{i}N_{i} \otimes M_{i} = \sum_{i} \lambda_{i}\Tr[M_{i}]N_{i} \otimes \rho_{B} = \rho_{A} \otimes \rho_{B}$. 
\end{proof}

\begin{proposition}[Restatement of Proposition \ref{prop:POVM-to-PVM}]
    If there exist two-outcome measurements with POVM elements $\{M,\mbb{1}_{A}-M\}$, $\{N,\mbb{1}_{B}-N\}$ such that 
    $$\Tr[\rho M \otimes (\mbb{1}_{B} -N)] = 0 = \Tr[\rho (\mbb{1}_{A} - M) \otimes N] \ , $$
    and $0 < \Tr[\rho M \otimes N] < 1$, then there exist projective measurements that do the same.
\end{proposition}
\begin{proof}
    We will show such projectors must exist. By assumption, $0 < \Tr[\rho M \otimes N] \leq \Tr[\rho_{A} M] < 1$. Define $p_{0} \coloneq \Tr[\rho_{A} M]$ and $\rho^{0}_{B} \coloneq \frac{1}{p_{0}}\Tr_{A}[M \otimes \mbb{1}_{B} \rho_{AB}]$. $p_{1}$ and $\rho^{1}_{B}$ may be defined in a similar fashion. Then,
    \begin{align}
        p_{0} = \Tr[M\rho_{A}] = \Tr[M \otimes \mbb{1}\rho_{AB}] =& \Tr[M \otimes N \rho] + \Tr[M \otimes (\mbb{1}_{B} - N)\rho_{AB}] \\
        =& \Tr[M \otimes N \rho] \\
        =& p_{0}\Tr[N\rho^{0}] \ ,
    \end{align}
    where we used the assumption in the second line. It follows $\Tr[N \rho^{0}] = 1$. On the other hand,
    \begin{align*}
        \Tr[\rho^{1}_{B}N] = \Tr[\Tr_{A}[(\mbb{1}-M)\rho_{AB}]N] = \Tr[(\mbb{1}-M)\otimes N \rho_{AB}] = 0 \ ,
    \end{align*}
    where the last equality is our assumption. Thus $\rho^{0}_{B},\rho^{1}_{B}$ are perfectly discriminable with $\{N,\mbb{1}-N\}$. By the Holevo-Helstrom theorem, $\{\rho^{0}_{B},\rho^{1}_{B}\}$ must be mutually orthogonal. Thus we may replace $\{N,\mbb{1}_{B}-N\}$ with $\{\Pi_{\supp(\rho^{0})},\mbb{1}_{B} - \Pi_{\supp(\rho^{0})}\}$ and obtain the same probabilities. Noting that we never relied on $M$ be a non-projective POVM element, we may use the same argument to replace $\{M,\mbb{1}_{A} - M\}$ with the projector onto the support of $\wt{\rho}^{0}_{A} \coloneq \frac{1}{\Tr[\Pi_{\supp(\rho^{0})}\rho_{B}]}\Tr_{B}[\mbb{1}_{A} \otimes \Pi_{\supp(\rho^{0})}\rho_{AB}]$ and its completion. This completes the proof.
\end{proof}

\subsection{Lemmata for \texorpdfstring{$\chi^{2}$}{} Contraction Coefficients}\label{subsec:lemmata-for-chi-square-contraction}

\begin{proof}[Proof of Lemma \ref{lem:contraction-coeff-to-eig}]
    Following \cite{Cao-2019a}  exactly, we have
    \begin{align}
        \eta_{\chi^{2}_{f}}(\cE,\sigma) =& \sup_{\rho \neq \sigma \in \Density} \frac{\chi_{f}^{2}(\cE(\rho) \Vert \cE(\sigma))}{\chi^{2}_{f}(\rho \Vert \sigma)} \\ 
        =& \sup_{\rho \neq \sigma \in \Density} \frac{\langle \rho - \sigma, (\cE^{\ast} \circ \mbf{J}_{f,\cE(\sigma)}^{-1} \circ \cE)(\rho - \sigma) \rangle}{\langle \rho - \sigma, \mbf{J}_{f,\sigma}^{-1}(\rho - \sigma) \rangle} \label{eq:contraction-Schrodinger-step-1} \\
        =& \sup_{\rho \neq \sigma \in \Density} \frac{\langle \rho - \sigma, (\cS_{f,\cE,\sigma} \circ \cE)(\rho - \sigma) \rangle_{f,\sigma}^{\star}}{\langle \rho - \sigma, \rho - \sigma \rangle_{f,\sigma}^{\star}} \label{eq:contraction-Schrodinger-Jf-inv} \\
        =& \sup_{\substack{0 \neq H \in \Herm: \\ \Tr[H]=0}} \frac{\langle H, (\cS_{f,\cE,\sigma} \circ \cE)(H) \rangle_{f,\sigma}^{\star}}{\langle H, H \rangle_{f,\sigma}^{\star}} \label{eq:chi-square-contraction-as-petz-ratio} \ ,
    \end{align}
    where the first equality is \eqref{eq:chi-squared-f-def} and the definition of adjoint map, the second equality is $\text{id} = \mbf{J}_{f,\sigma}^{-1} \circ \mbf{J}_{f,\sigma}$, the definition of $\cS_{f,\cE,\sigma}$ and the definition of the inner product, and the last equality is $\rho-\sigma$ is traceless and Hermitian.
    
    We now establish the claims about eigenvalues of the map. This uses a similar, but modified, proof of \cite[Lemma 10]{Cao-2019a}, which can be viewed as a less refined version of our result. By,
    \begin{align}\label{eq:Schrodinger-thm-item-2-eigenvector}
        \langle \sigma, (\cS_{f,\cE,\sigma} \circ \cE)(\sigma) \rangle_{f,\sigma}^{\star} = \langle \sigma, (\cE^{\ast} \circ \mbf{J}_{f,\cE(\sigma)}^{-1} \circ \cE)(\sigma) \rangle= \langle \cE(\sigma), \mbb{1}_{\supp(\cE(\sigma))} \rangle = 1 \ ,
    \end{align}
    where the first equality is $\mbf{J}_{f,\sigma}^{-1} \circ \mbf{J}_{f,\sigma} = \id_{\supp(\sigma)}$ and $\sigma$ already is projecting onto that subspace, the second equality is definition of adjoint map and $\mbf{J}_{f,\cE(\sigma)}^{-1}(\cE(\sigma)) = \mbb{1}_{\supp(\cE(\sigma))}$ and the final equality is $\cE$ is trace-preserving. Combining this with Item 2 of Proposition \ref{prop:Schrod-map-composed-properties}, we have $\cS_{f,\cE,\sigma} \circ \cE$ as an operator from inner product space from $(\Lin(A), \langle \cdot , \cdot \rangle_{f,\sigma}^{\star})$ into itself has eigenvalues between zero and one with its largest eigenvector being $\sigma$. Note the same claims hold if we replace the space of linear operators with the space of Hermitian operators.

    We now consider $$\cS_{f,\cE,\sigma} \circ \cE: (\Herm(A), \langle \cdot , \cdot \rangle_{f,\sigma}^{\star}) \to (\Herm(A), \langle \cdot , \cdot \rangle_{f,\sigma}^{\star}) \ , $$
    where the output space is appropriate as $\cS_{f,\cE,\sigma} \circ \cE$ is Hermitian-preserving as follows from Proposition \ref{prop:symmetry-inducing-equivalences}. This map is also self-adjoint on this space as as it is positive semidefinite (Proposition \ref{prop:Schrod-map-composed-properties}).
    
    As $\cS_{f,\cE,\sigma} \circ \cE$ is self-adjoint over the Hermitian operators, we may take its spectral decomposition, $\cS_{f,\cE,\sigma} \circ \cE = \sum_{i \in [d_{A}^{2}]} \lambda_{i} \nu_{i}$ where we know $\{\lambda_{i}\}_{i} \subset [0,1]$, $\{\nu_{i}\}_{i} \subset \Herm(A)$, $\lambda_{1} = 1$, and $\nu_{1} = \sigma$. Then, by the orthonormality, $0=\langle \nu_{i}, \sigma \rangle_{f,\sigma}^{\star} = \Tr[\nu_{i}]$ where we used Proposition \ref{prop:mult-and-div-under-trace} and that $\sigma > 0$. This allows us to conclude $\{\nu_{i}\}_{i \in [d_{A}^{2}]\setminus \{1\}}$ is an orthonormal set of traceless Hermitian operators. Using that the space of traceless Hermitian operators has dimension $d_{A}^{2} -1$, they form a basis of the traceless Hermitian operators. It follows any feasible $H$ in \eqref{eq:chi-square-contraction-as-petz-ratio} can be decomposed in this basis, so we may conclude $\eta_{\chi^{2}_{f}}(\cE,\sigma) \leq \lambda_{2}$. However, $\nu_{2}$ is a traceless Hermitian operator and thus feasible for \eqref{eq:chi-square-contraction-as-petz-ratio}. Thus, $\eta_{\chi^{2}_{f}}(\cE,\sigma) = \lambda_{2}$. 

    Now we prove Item 3, which follows a very similar method to the proof of Theorem \ref{lem:contraction-as-map-norm}. Starting from \eqref{eq:contraction-Schrodinger-step-1} but already having replaced $\rho - \sigma$ with traceless $0 \neq H \in \Herm(A)$, 
    \begin{align}
        \eta_{\chi^{2}_{f}}(\cE,\sigma) =& \sup_{\substack{0 \neq H \in \Herm: \\ \Tr[H]=0}} \frac{\langle H, (\cE^{\ast} \circ \mbf{J}_{f,\cE(\sigma)}^{-1}\circ \cE)(H) \rangle}{\langle H, \mbf{J}_{f,\sigma}^{-1}(H)\rangle } 
        = \sup_{0 \neq H \in \Herm_{0}(A)} \frac{\langle \cE(H),\cE(H) \rangle_{\mbf{J}_{f,\cE(\sigma)}^{-1}}}{\langle H, H\rangle_{f,\sigma}^{\star} } \ ,
    \end{align}
    where the second equality is just using the definition of adjoint map and the inner product and giving a name for the space $\Herm_{0}(A)$. Moreover as $\cE$ is Hermitian-preserving and trace-preserving, $\cE$ maps $\Herm_{0}(A)$ to $\Herm_{0}(B)$. Combining these points, we have
    \begin{align}
        \eta_{\chi^{2}_{f}}(\cE,\sigma)
        = \sup_{H \in \Herm_{0}(A)\setminus \{0\}} \frac{\langle \cE(H),\cE(H) \rangle_{\mbf{J}_{f,\cE(\sigma)}^{-1}}}{\langle H, H\rangle_{f,\sigma}^{\star} } 
        = \left(\sup_{H \in \Herm_{0}(A)\setminus \{0\}} \frac{\Vert \cE(H)\Vert_{\mbf{J}_{f,\cE(\sigma)}^{-1}}}{\Vert H \Vert^{2}_{\mbf{J}_{f,\sigma}^{-1}}}\right)^{2} \ ,  
    \end{align}
    where the second equality uses the monotonicity of the square. As $\cE$ is Hermitian-preserving and trace-preserving, the argument of the square is the operator norm given in the Lemma statement. Doing an identical argument as to in Theorem \ref{lem:contraction-as-map-norm} using Lemma \ref{lem:map-norms-as-optimizations} completes the proof. 
\end{proof}

\begin{proof}[Proof of Corollary \ref{cor:Schroding-er-map-without-rank-constraints}]
Starting from \eqref{eq:chi-squared-f-norm-ratio}, by our assumption on $f$,
\begin{align}
    \eta_{\chi^{2}_{f}}(\cE,\sigma) =& \sup_{H \in \Herm_{0,\sigma}(A) \setminus \{0\}} \frac{\Vert \cR_{f,\cE,\sigma}(H) \Vert^{2}_{\mbf{J}_{f,\cE(\sigma)}}}{\Vert H \Vert^{2}_{\mbf{J}_{f,\sigma}}} \\
    =& \sup_{H \in \Herm(A)\setminus\{0\}: \Tr[\sigma H]= 0} \frac{\langle \cR_{f,\cE,\sigma}(H),\cR_{f,\cE,\sigma}(H) \rangle_{\mbf{J}_{f,\cE(\sigma)}} }{\langle H , H \rangle_{\mbf{J}_{f,\sigma}}} \\
    =& \sup_{H \in \Herm(A)\setminus\{0\}: \Tr[\sigma H]= 0} \frac{\langle \cR_{f,\cE,\sigma}(H),(\id_{\supp(\cE(\sigma))} \circ \cE \circ \mbf{J}_{f,\sigma})(H) \rangle }{\langle (\mbf{J}_{f,\sigma}^{-1} \circ \mbf{J}_{f,\sigma})(H) , \mbf{J}_{f,\sigma}(H) \rangle} \\
    =& \sup_{H \in \Herm(A)\setminus\{0\}: \Tr[\sigma H]= 0} \frac{\langle (\cE \circ \mbf{J}_{f,\sigma})(H),(\mbf{J}_{f,\cE(\sigma)}^{-1} \circ \cE \circ \mbf{J}_{f,\sigma})(H) \rangle }{\langle (\mbf{J}_{f,\sigma}^{-1} \circ \mbf{J}_{f,\sigma})(H) , \mbf{J}_{f,\sigma}(H) \rangle} \\
    =& \sup_{G \in \Herm(A)\setminus\{0\}: \Tr[G]= 0} \frac{\langle \cE(G),(\mbf{J}_{f,\cE(\sigma)}^{-1} \circ \cE)(G) \rangle }{\langle \mbf{J}_{f,\sigma}^{-1}(G) , G \rangle} \\
    =& \sup_{G \in \Herm_{0}(A)\setminus\{0\}: \Tr[G]=0} \frac{\langle G,(\cE^{\ast} \circ \mbf{J}_{f,\cE(\sigma)}^{-1} \circ \cE)(G) \rangle }{\langle G, G \rangle_{f,\sigma}^{\star}} \\ 
    =& \sup_{G \in \Herm_{0}(A)\setminus\{0\}: \Tr[G]=0} \frac{\langle G,(\cS_{f,\cE,\sigma} \circ \cE)(G) \rangle_{f,\sigma}^{\star} }{\langle G, G \rangle_{f,\sigma}^{\star}} \ ,
\end{align}
where the second equality is definition of the canonical norms for an inner product space, the third uses the definition of the inner product spaces and $\cR_{f,\cE,\sigma}$ and the fact that $\mbf{J}_{f,\cE,\sigma}^{-1}$ projects onto the space $\text{id}_{\supp(\cE(\sigma))}$ projects onto by assumption on $f$ and Proposition \ref{prop:suff-conds-for-restricting-support}. The fourth equality uses $\mbf{J}_{f,\cE(\sigma)}^{-1}$ is self-adjoint with respect to the Hilbert-Schmidt inner product on Hermitian matrices (Proposition \ref{prop:J-operator-self-adjoint}). The fifth equality is by defining $G \coloneq \mbf{J}_{f,\sigma}(H)$ which is Hermitian since $f$ is symmetry-inducing, is non-zero since $H$ is non-zero, $\sigma$ is full rank with respect to $A$, and \eqref{eq:Hadamard-prod-form}, and $\Tr[G] = \Tr[\mbf{J}_{f,\sigma}(H)] = \Tr[\sigma H] = 0$ where we used Proposition \ref{prop:mult-and-div-under-trace} and the assumption on $H$. The rest of the proof may be argued in an identical fashion to Lemma \ref{lem:contraction-coeff-to-eig} via the spectral decomposition of $\cS_{f,\cE,\sigma} \circ \cE$ where note that \eqref{eq:Schrodinger-thm-item-2-eigenvector} does not depend on $\cE(\sigma)$ being full rank.
\end{proof}

\subsubsection{Lemmata for Time Homogeneous Markov Chains}\label{subsec:time-homogeneous-MC-lemmata}

\begin{proof}[Proof of Proposition \ref{prop:L1-mixing-times}]
    We have
    \begin{align}
        \Vert \cE^{n}(\rho - \pi) \Vert_{1}^{2} \leq \chi^{2}_{f}(\cE^{n}(\rho),\cE^{n}(\pi)) 
        =& \chi^{2}_{f}(\cE^{n}(\rho),\cE^{n}(\pi)) \\
        \leq& \eta_{\chi^{2}_{f}}(\cE,\cE^{n-1}(\pi))\chi^{2}_{f}(\cE^{n-1}(\rho),\cE^{n-1}(\pi)) \\
        =& \eta_{\chi^{2}_{f}}(\cE,\pi)\chi^{2}_{f}(\cE^{n-1}(\rho),\cE^{n-1}(\pi)) \\
        \leq& \eta_{\chi^{2}_{f}}(\cE,\pi)^{n} \chi^{2}_{f}(\rho,\sigma) \ ,
    \end{align}
    where the first inequality is \cite[Lemma 5]{Temme-2010a}, the second inequality is the definition of contraction coefficient, the second equality is the assumption $\cE(\pi) = \pi$, and the final inequality is doing the same argument iteratively. Taking the square root completes the first inequality. Using \eqref{eq:ordering-of-chi-squareds}, one may alternatively apply Proposition \ref{prop:Df-lb-in-chi-square} and that $\text{TD}(\rho,\sigma) \leq 1$ to remove the dependence on $\chi_{f}^{2}$. Re-ordering terms completes the proof.
\end{proof}

\end{document}

%% file: Packages.tex
\usepackage[utf8]{inputenc}
\usepackage[margin=.75in]{geometry} 
\usepackage{fancyhdr}
\usepackage[dvipsnames]{xcolor}

\usepackage{mathpazo}
\usepackage{amssymb}
\usepackage{amsmath}
\usepackage{mathtools}
\usepackage{relsize}
\usepackage{hyperref}
\usepackage{amsthm}
\usepackage{physics}
\usepackage{enumitem}
\usepackage[most]{tcolorbox}
\usepackage{mathrsfs}
\usepackage{algorithm}
\usepackage{algpseudocode}
\usepackage{authblk}
\usepackage{subcaption}

\usepackage{matlab-prettifier}

\usepackage{float} 
\allowdisplaybreaks 
\usepackage{hyperref} 
\usepackage{mwe}
\usepackage{listing}
\usepackage[en-US]{datetime2}
\usepackage{placeins}
\usepackage{capt-of}
\usepackage{tikz}
\usepackage{array}

\usepackage{comment}
\usepackage[sort&compress,numbers]{natbib}

%% file: Notations.tex
\newcommand{\ve}{\varepsilon}
\newcommand{\mrm}{\mathrm}
\newcommand{\mbb}{\mathbb}
\newcommand{\mbf}{\mathbf}

\newcommand{\wt}{\widetilde}
\newcommand{\ol}{\overline}

\newcommand{\supp}{\mrm{supp}}
\newcommand{\linspan}{\mrm{span}}

\newcommand{\Var}{\mrm{Var}}

\renewcommand{\tr}{\operatorname{Tr}}

\renewcommand\bra[1]{{\langle{#1}|}}
\renewcommand\ket[1]{{|{#1}\rangle}}

\newcommand{\cA}{\mathcal{A}}

\newcommand{\cE}{\mathcal{E}}
\newcommand{\cF}{\mathcal{F}}
\newcommand{\cG}{\mathcal{G}}

\newcommand{\cM}{\mathcal{M}}
\newcommand{\cN}{\mathcal{N}}

\newcommand{\cP}{\mathcal{P}}

\newcommand{\cR}{\mathcal{R}}
\newcommand{\cS}{\mathcal{S}}
\newcommand{\cT}{\mathcal{T}}
\newcommand{\cU}{\mathcal{U}}
\newcommand{\cV}{\mathcal{V}}
\newcommand{\cW}{\mathcal{W}}
\newcommand{\cX}{\mathcal{X}}
\newcommand{\cY}{\mathcal{Y}}
\newcommand{\cZ}{\mathcal{Z}}

\newcommand{\Lin}{\mathrm{L}}
\newcommand{\Trans}{\mathrm{T}}
\newcommand{\Pos}{\mathrm{Pos}}
\newcommand{\Herm}{\mathrm{Herm}}
\newcommand{\Channel}{\mathrm{C}}

\newcommand{\Density}{\mathrm{D}}

\newcommand{\id}{\textrm{id}}
\newcommand{\Pd}{\mathrm{Pd}}

\newenvironment{mylist}[1]{\begin{list}{}{
	\setlength{\leftmargin}{#1}
	\setlength{\rightmargin}{0mm}
	\setlength{\labelsep}{2mm}
	\setlength{\labelwidth}{8mm}
	\setlength{\itemsep}{0mm}}}
	{\end{list}}

\newcounter{problemcounter}

\theoremstyle{definition}
\newtheorem{theorem}{Theorem}
\newtheorem*{theorem*}{Theorem}
\newtheorem{result}[theorem]{Result}
\newtheorem{lemma}[theorem]{Lemma}

\newtheorem{definition}[theorem]{Definition}

\newtheorem{corollary}[theorem]{Corollary}

\newtheorem{remark}[theorem]{Remark}
\newtheorem{proposition}[theorem]{Proposition}
\newtheorem{fact}[theorem]{Fact}